\documentclass[11pt]{article}
\usepackage[T1]{fontenc}
\usepackage{lmodern}
\usepackage{multicol}
\usepackage{caption}
\usepackage{float,needspace}
\usepackage{comment}
\usepackage{amsmath}
\usepackage{tabularx}
\usepackage{bm}
\usepackage[caption=false]{subfig}

\usepackage{amsmath,amsfonts,amsthm,amssymb,graphics,graphicx,mathtools,epstopdf,etoolbox,indentfirst,enumitem,mleftright,xcolor,booktabs,footnote,relsize,microtype}

\usepackage{tikz}
\usetikzlibrary{arrows.meta}

\makeatletter
\renewcommand{\paragraph}{%
  \@startsection{paragraph}{4}%
  {\z@}{1.25ex \@plus 1ex \@minus .2ex}{-1em}%
  {\normalfont\normalsize\bfseries}%
}
\makeatother

\let\originalleft\left
\let\originalright\right
\renewcommand{\left}{\mathopen{}\mathclose\bgroup\originalleft}
\renewcommand{\right}{\aftergroup\egroup\originalright}

\usepackage{hyperref} 
\usepackage{doi} 

\definecolor{darkmagenta}{rgb}{0.85, 0, 0.45}
\hypersetup{
breaklinks=true,
colorlinks = true,
citecolor= blue,
urlcolor= blue,
linkcolor = blue
}

\usepackage{pict2e,picture}
\makeatletter
\DeclareRobustCommand{\customarrow}[1][]{%
  \check@mathfonts
  \if\relax\detokenize{#1}\relax
    \settowidth{\dimen@}{$\m@th\rightarrow$}%
  \else
    \setlength{\dimen@}{#1}%
  \fi
  \sbox\z@{\usefont{U}{lasy}{m}{n}\symbol{41}}%
  \begin{picture}(\dimen@,\ht\z@)
  \roundcap
  \put(\dimexpr\dimen@-.7\wd\z@,0){\usebox\z@}
  \put(0,\fontdimen22\textfont2){\line(1,0){\dimen@}}
  \end{picture}%
}
\makeatother

\newcommand{\floor}[1]{\left\lfloor #1 \right\rfloor}
\newcommand{\ket}[1]{\left| #1 \right>}
\newcommand{\bra}[1]{\left< #1 \right|}
\newcommand{\ketbra}[2]{\ket{#1} \!\! \bra{#2}}
\newcommand{\pure}[1]{\ketbra{#1}{#1}}
\newcommand{\tr}[2][]{\operatorname{Tr}_{#1}\left[#2\right]} 
\newcommand{\Tr}{\operatorname{Tr}} 
\newcommand{\reg}{\mathrm{reg}} 
\newcommand{\binh}{h_{\mathrm{bin}}} 

\newcommand{\CPTP}[2]{\operatorname{CPTP}(#1 
,
#2)} 
\newcommand{\defvar}{=}
\newcommand{\dop}[1]{\operatorname{S}_{#1}} 
\newcommand{\eps}{\varepsilon}
\newcommand{\expval}{\mathop{\mathbb{E}}} 
\newcommand{\freq}{\operatorname{freq}}
\newcommand{\trnorm}[1]{\left\lVert#1\right\rVert_{\mathrm{tr}}}
\newcommand{\id}{\mathbb{I}} 
\newcommand{\idmap}{\operatorname{id}} 

\newcommand{\norm}[1]{\left\lVert#1\right\rVert} 
\newcommand{\pd}{P} 
\newcommand{\pert}{\mu} 
\newcommand{\Pos}{\operatorname{Pos}} 
\newcommand{\smf}[1]{\log({1}/{#1}^2)} 
\newcommand{\vanf}{\mathtt{van}} 
\newcommand{\suchthat}{\text{ s.t.}} 
\newcommand{\supp}{\operatorname{supp}} 
\newcommand{\fullsupp}{\mathrm{supp}} 

\newcommand{\Renyi}{R\'{e}nyi}
\newcommand{\mbf}[1]{\mathbf{#1}} 
\newcommand{\bsym}[1]{\boldsymbol{#1}} 
\newcommand{\Dargs}[2]{\left( #1 \middle\Vert #2\right)} 
\newcommand{\mapA}{\mathcal{A}}
\newcommand{\mapD}{\mathcal{D}}
\newcommand{\mapE}{\mathcal{E}}
\newcommand{\mapF}{\mathcal{F}}
\newcommand{\mapL}{\mathcal{L}}
\newcommand{\mapM}{\mathcal{M}}
\newcommand{\mapN}{\mathcal{N}}
\newcommand{\mapP}{\mathcal{P}}
\newcommand{\mapR}{\mathcal{R}}
\newcommand{\mapT}{\mathcal T}

\newcommand{\Dmax}{D_\mathrm{max}}

\DeclareMathOperator*{\mylimsup}{limsup}
\DeclareMathOperator*{\myliminf}{liminf}

\newcommand{\esmooth}{\overline{\eps}}
\newcommand{\realstate}{\rho^\mathrm{real}}
\newcommand{\idealstate}{\rho^\mathrm{ideal}}
\newcommand{\inM}{M} 
\newcommand{\inN}{N}
\newcommand{\inQ}{Q}
\newcommand{\allN}{\mbf{\inN}}
\newcommand{\allQ}{\mbf{\inQ}}
\newcommand{\outQ}{O}
\newcommand{\stabE}{R}

\newcommand{\commN}{\widehat{\inN}}
\newcommand{\commO}{\outQ}
\newcommand{\commC}{\CS}

\newcommand{\CS}{C}

\newcommand{\cS}{c}
\newcommand{\alphCS}{\mathcal C}

\newcommand{\outputset}{\mathcal{O}_{\mathrm{any}}}
\newcommand{\lse}{\operatorname{logsumexp}}
\newcommand{\ubndaf}[1][f]{R_{\alpha,\mapE,\mapF}(#1)} 
\newcommand{\alphastar}{\alpha} 

\newcommand{\ptest}{\gamma}
\newcommand{\gen}{\mathtt{gen}}
\newcommand{\test}{\mathtt{test}}
\newcommand{\testQ}{\varnothing}
\newcommand{\Acc}{\mathsf{Acc}}
\newcommand{\Bad}{\mathsf{Bad}}
\newcommand{\ancA}{M^{\prime\A}}
\newcommand{\ancB}{M^{\prime\B}}

\newcommand{\constval}{k}

\newcommand{\base}{b}

\DeclareRobustCommand{\rchi}{{\mathpalette\irchi\relax}}
\newcommand{\irchi}[2]{\raisebox{\depth}{$#1\chi$}} 

\newtheorem{theorem}{Theorem}[section]
\newtheorem{lemma}[theorem]{Lemma}
\newtheorem{corollary}[theorem]{Corollary}

\newtheorem{fact}[theorem]{Fact}
\newtheorem{postulates}[theorem]{Postulate Set}
\newtheorem{remark}[theorem]{Remark}
\newtheorem*{informal*}{Informal Statement}

\theoremstyle{definition} 
\newtheorem{definition}{Definition}[section]

\floatstyle{ruled}
\newfloat{prot}{tbp}{lop}
\floatname{prot}{Protocol}
\definecolor{revisionblue}{rgb}{0,0.25,0.65}

\newcommand{\A}{\mathrm{A}}
\newcommand{\B}{\mathrm{B}}

\newcommand{\memA}{M^\A}
\newcommand{\memB}{M^\B}
\newcommand{\idealA}{S^\A}
\newcommand{\idealB}{S^\B}
\newcommand{\junkA}{J^\A}
\newcommand{\junkB}{J^\B}

\newcommand{\sfP}{\mathsf{P}}
\newcommand{\sfQ}{\mathsf{Q}}
\newcommand{\sfV}{\mathsf{V}}
\newcommand{\game}{\mathsf{G}}
\newcommand{\strat}{\mathfrak{S}}
\newcommand{\ddP}{\mathfrak{P}_m}
\newcommand{\diP}{\mathfrak{P}'}
\newcommand{\bCS}{\mbf C}
\newcommand{\bcS}{\mbf c}
\newcommand{\boutQ}{\mbf O}
\newcommand{\bK}{\mbf K}
\newcommand{\bX}{\mbf X}
\newcommand{\bA}{\mbf A}
\newcommand{\bk}{\mbf k}
\newcommand{\bx}{\mbf x}
\newcommand{\ba}{\mbf a}

\newcommand{\bmapE}{\bsym{\mapE}}
\newcommand{\bmapF}{\bsym{\mapF}}
\newcommand{\Dsharp}{D^\#}
\newcommand{\Qsharp}{Q^\#}
\newcommand{\compexp}[1]{\operatorname{CompExp}^{\#}_{#1}(\bmapE,\bmapF;S_\Omega)}
\newcommand{\compexpsand}[1]{\operatorname{CompExp}_{#1}(\mapE^{\otimes n},\mapF^{\otimes n};S_\Omega)}
\newcommand{\epstol}{\eps_{\mathrm{tol}}}
\newcommand{\hon}{\mathrm{hon}}
\newcommand{\dhon}{\delta_{\hon}}

\newcommand{\bC}{\mbf C}

\newcommand{\bR}{\mbf R}
\newcommand{\bL}{\mbf L}
\newcommand{\bT}{\mbf T}

\newcommand{\bB}{\mbf B}
\newcommand{\bY}{\mbf Y}
\newcommand{\bmsg}{\mbf m}
\newcommand{\by}{\mbf y}
\newcommand{\bv}{\mbf v}
\newcommand{\br}{\mbf r}

\newcommand{\bc}{\mbf c}

\newcommand{\bz}{\mbf z}
\newcommand{\bZ}{\mbf Z}
\newcommand{\renD}{\mathbb{D}_\alpha}

\newcommand{\poly}{\mathrm{poly}}

\newcommand{\Dgeni}[1]{\Dsharp_2\left(\sigma^{{#1},\mapE_{#1}}_{|\gen}\middle\Vert\sigma^{{#1},\mapF_{#1}}_{|\gen}\right)}
\newcommand{\Dgen}{\Dsharp_2\left(\sigma^{\oplus,\mapE}_{|\gen}\middle\Vert\sigma^{\oplus,\mapF}_{|\gen}\right)}
\newcommand{\copper}{{estimator}} 
\newcommand{\vecD}{v^\mathrm{Div}}
\newcommand{\vecDc}[1]{\vecD(#1)}

\newcommand{\suppOmega}{\Omega_{\mathrm{supp}}}
\newcommand{\cmapM}{\overline{\mapM}}
\newcommand{\newmapE}{\widetilde{\mapE}}
\newcommand{\newmapF}{\widetilde{\mapF}}

\allowdisplaybreaks

\begin{document}

\title{Flexible rigidity: Compiling generic trusted-preparation protocols into sequential device-independent protocols}

\author{
Amir Arqand \\
\small Institute for Quantum Computing \\
\small University of Waterloo \\
\small \texttt{aarqand@uwaterloo.ca} \\
\and
Srijita Kundu \\
\small Quantum Computing Research Centre \\
\small Hon Hai (Foxconn) Research Institute \\
\small \texttt{srijita.kundu@foxconn.com.sg} \\
\and
Ernest Y.-Z. Tan \\
\small Department of Physics, \\
\small National University of Singapore \\
\small \texttt{phytyze@nus.edu.sg}
}
\date{}

\maketitle

\begin{abstract}
We give a generic compiler for turning a wide class of device-dependent protocols with trusted state preparation into sequential device-independent protocols. In order for our compiler to apply, the trusted-preparation protocol needs to satisfy certain minimal conditions: having an entanglement-based realization in the first place, and the entanglement-based realization corresponding to the optimal strategy of any non-local game that satisfies a fairly weak rigidity condition. We apply our generic compiler to give sequential device-independent protocols for uncloneable encryption, copy-protection of point functions, and blind randomness expansion.

Our technical contribution is a new class of rigidity theorems based on R\'{e}nyi divergences, with flexible cryptographic applications. They allow our compiler to achieve a constant \emph{rate} (i.e., the ratio of trusted states prepared to device uses), while tolerating a constant level of noise per device usage. This overcomes a fundamental scaling obstacle faced by previous rigidity theorems for trace distance or fidelity. For the generic sequential case, our proof is based on the sharp-divergence chain rule by Fawzi and Fawzi (2021), which suffices for a constant but fairly low rate. In the special case when the devices retain no memory, we improve the rates by deriving a novel subadditivity property for a modified channel sandwiched divergence, which may be of independent interest. Finally, as a prospect for future improvements, we also prove a chain rule for the iterated-mean divergences from Brown, Fawzi and Fawzi (2021).
\end{abstract}

\section{Introduction}\label{sec:intro}
\subsection{Background}
Device-independent (DI) cryptography aims to establish security from the observed behaviour of untrusted quantum devices, subject to assumptions about their allowed communication. A central idea used in many device-independent security proofs is that if the input-output distribution of a quantum device is found to win a non-local game with close to optimal probability, and the non-local game has a property called \emph{rigidity}, then the devices must be implementing states and measurements close to the optimal strategy for that game. Therefore, if the devices pass a test indicating that they can win the game with high probability, a security statement that holds for the ideal states and measurements for the game must hold for whatever is being implemented by the unknown device.

The basic form of rigidity suffices when we want to implement a protocol that only needs to use the devices once, corresponding to playing a single copy of the non-local game. But for most protocols, multiple uses of the devices (or the use of a larger device that plays many copies of a game in parallel) is required. A central difficulty in such a setting is that the devices may retain quantum memory and correlate their behaviour across an entire protocol.\footnote{Throughout the entirety of this work, when discussing ``device memory'' we mean only memory across the rounds of a single protocol instance. We do not claim security against devices that are reused for a new protocol instance while still retaining memory from the previous instance, which faces issues discussed in~\cite{BCK13}.} Consequently, a security statement for one use of a device does not immediately give a security statement for many uses over a protocol. A general method for making this transition is particularly valuable: it allows the analysis of individual rounds to be separated from the problem of controlling correlations across rounds.

\paragraph{Entropy accumulation.} Entropy accumulation provides such a method for various security proofs that can be based on conditional entropy. The entropy accumulation theorem of Dupuis, Fawzi and Renner~\cite{DFR20} bounds the smooth min-entropy generated by a sequential process in terms of single-round conditional entropies, under an appropriate Markov condition. Its applications to device-independent quantum key distribution (QKD) and randomness expansion yield modular security proofs against general attacks~\cite{ADF+18,ARV19, DF19}. The entropy accumulation framework has been refined in subsequent years: relaxing the Markov condition~\cite{MFSR24}, improving bounds by studying R\'enyi entropy directly~\cite{arx_AHT24}, and incorporating constraints on the marginal states entering the channels~\cite{arx_AT25}.

These theorems reduce a global entropy estimate to single-round quantities, but applying them to a cryptographic task still requires a reduction from its security criterion to a conditional entropy bound. This is sufficient for applications like QKD and randomness expansion, but there are other quantum cryptographic tasks that involve different security requirements. For example, uncloneable encryption \cite{BL20} asks whether two separated recipients can both recover a message after a quantum ciphertext has been split between them and the decryption key has been released. This security condition cannot be reduced to an entropy bound, due to involving the joint guessing probability of two spatially separated parties.  

\paragraph{Parallel repetition.} Existing DI protocols for general cryptographic tasks such as certified deletion \cite{KT23}, uncloneable encryption \cite{KT25}, and bounded-storage oblivious transfer \cite{BCJK24} have thus used parallel repetition arguments instead. Like in entropy accumulation-based approaches, their security proofs use self-testing results at a single-copy game level, but require parallel repetition theorems for non-standard forms of non-local games that capture the corresponding attacks, for the global analysis. Although the general parallel repetition approach remains similar, proving a custom parallel repetition theorem for each task is quite challenging. The proof techniques also influence the formulations of the tasks: the certified deletion construction in \cite{KT23} retains some permanently private randomness for its anchoring argument, while the uncloneable encryption construction in \cite{KT25} uses a key-release procedure with independent fresh randomness for the two recipients, yielding potentially different decryption keys.

\paragraph{Computational remote state preparation.} A different approach to replacing trusted state preparation was developed by Gheorghiu, Metger and Poremba~\cite{GMP22}. Their remote state preparation protocol allows a fully classical client to certify the preparation of a batch of Wiesner states by a quantum server, assuming the quantum hardness of learning with errors. This provides a modular way to replace trusted preparation in several cryptographic protocols, including uncloneable encryption and point-function copy protection, but does require the computational assumption. Moreover, their construction has an inverse-polynomial \emph{rate}, i.e., polynomially many device uses are required in order to prepare one trusted Wiesner state, and the amount of noise their protocols can tolerate is also inverse-polynomial in the number of device uses.

These developments motivate a general method for converting a trusted-preparation cryptographic protocol to a device-independent one, such that the security condition of the DI protocol may be directly related to that of the device-dependent (DD) one. 

\paragraph{Limitations of self-testing results.} We note that parallel rigidity or self-testing results \cite{Col17, McK17} do give a generic way to convert a security condition on trusted states to a security condition on untrusted states that sufficiently violate a Bell inequality, albeit one that is not very practically useful for cryptographic purposes. The idea for parallel self-testing results is that if some joint state and measurements win $(\omega^*(\game)-\epstol)n$ copies out of $n$ parallel copies of a game $\game$ (which has quantum value $\omega^*(\game)$) with high probability, then the states and measurements must be close (in trace distance or vector $\ell_2$ distance) to $n$ copies of the ideal state and measurements for $\game$. Theoretically, this means that the probability of any event on the real states in the protocol is close to the probability on these ideal states, as long as the real states pass the self-test with high probability.

However, there are challenges in trying to apply a parallel self-testing based approach to a cryptographic protocol. First, ``close'' in the context of these results means $\poly(n,\epstol)$ close. This means that the result is only non-trivial when $\epstol$ is at least polynomially small in $n$, and in many cases, even that is not sufficient. This is because the number of device uses often coincides with the security parameter $\lambda$ of a protocol, and in cryptographic applications, we want the probability of a ``bad'' event to be exponentially small in $\lambda$. This means the tolerance $\epstol$ of the protocol has to in fact be exponentially small in $n$.

But this tolerance parameter $\epstol$ determines how much noise we can tolerate from devices that are producing noisy versions of the intended states, but are behaving honestly otherwise. Suppose a non-local game $\game$ requires perfect EPR pairs for its optimal strategy, and a real device is preparing states that are $\delta$-noisy EPR pairs instead. The winning probability for a single copy of $\game$ is then going to be $\omega^*(\game)-\delta$, and it is only going to win, say, $(\omega^*(\game)-2\delta)n$ out of $n$ games with high probability. This means that if we don't want to abort all the time\footnote{DI protocols typically abort if the self-testing check does not pass.} in protocols using real-world devices, we need to set $\epstol=2\delta$, where $\delta$ is going to be a constant independent of $n$ (since real-world devices do not know or care how many times you use them in a cryptographic protocol).
This makes the existing parallel self-testing results unusable in such scenarios.

We emphasize an important point that the $\poly(n, \epstol)$ closeness we get from these results is not merely an artifact of their proof techniques, but a fundamental limitation of trace-distance-based measures in this context. 
Specifically, observe that the trace distance between $n$ IID $\delta$-noisy EPR pairs and $n$ perfect EPR pairs approaches $1$ ``exponentially fast'' (in $n$), in the sense of the quantum Chernoff bound~\cite{ACM+07,NS09}. Despite this, the probability of a ``bad'' event on $n$ $\delta$-noisy EPR pairs is not always different from the probability of the same event on noiseless EPR pairs by an additive amount near $1$, as we shall soon see. A similar point holds for fidelity-based measures such as purified or angular distance; in the subsequent discussions we shall mostly only mention trace distance, with the understanding that purified or angular distance have similar behaviour.

So the general method we are looking for to convert DD security to DI security should ideally tolerate constant noise while having a high rate, and in order to do so, it must incorporate self-testing in a way that does not simply involve the trace distance between the $n$-copy real and ideal states. We note that other prospects have been suggested based on the ``averaged'' single-copy distance~\cite{GSD22}, or by only keeping a few copies for later cryptographic use~\cite{BWB25} --- however, thus far these approaches still face challenges arising from issues similar to those described above.

\subsection{Our results}
\subsubsection{DD-to-DI compilers}
In this work, we give the first generic compiler from device-dependent protocols with trusted preparation, to sequential device-independent protocols, which have both constant rate, and tolerate a constant amount of noise in the honest case.

The type of DD protocol we consider is one where the honest party (Alice) samples some classical randomness, prepares $m$ copies of some trusted pure state corresponding to this, sends these to Bob (who is the potentially dishonest party), and then Alice and Bob run some kind of security experiment. We assume $\Bad$ is some bad event in the security experiment whose probability we are trying to minimize in a protocol design. This is a fairly generic template that many cryptographic protocols with one honest party fall into, as we shall see with our applications. We denote this type of protocol when Alice sends $m$ copies of the trusted state by $\ddP$. 

$\ddP$ needs to satisfy a few additional conditions in order for the compiler to work, which we state informally below.
\begin{postulates}[Informal version of Postulate Sets~\ref{post:DD} and \ref{post:selftest}]\label{post:informal}
~\\[-5mm]
\begin{enumerate}[label=\textup{(\roman*)},ref=\textup{(\roman*)}]
\item \label{post:dep-informal}The security of $\ddP$ must hold if the trusted states are acted upon by depolarizing noise. Specifically, for any depolarizing parameter $0< \mu < 1$, the probability of the bad event in the security experiment under this depolarizing noise, and under all allowed attacks by Bob, must satisfy
\[ \Pr_{\ddP}[\Bad] \leq q_{\mathrm{dep}}(\mu)^m\]
for some $q_{\mathrm{dep}}(\mu) < 1$.
\item\label{post:ent-informal} The ensemble of states used by Alice in $\ddP$ must have an entanglement-based realization.
\item\label{post:selftest-informal} The entanglement-based realization must correspond to the optimal strategy for a self-testable non-local game $\game$. In particular, the self-testing result for this $\game$ must give extraction isometries that depend only on the measurements in a given strategy for $\game$.
\end{enumerate}
\end{postulates}
We now explain the intuition behind these postulates. Postulate~\ref{post:informal}\ref{post:dep-informal} is not necessarily required for any DD-to-DI compiler, but is an artifact of our proof technique, for reasons we explain later. However, this condition is not actually hard to satisfy, and in fact it holds by default whenever the depolarizing noise channel is an allowed attack by Bob (and Alice's marginal is invariant under depolarizing noise). This is because, instead of Alice sending states that have been depolarized to Bob, Bob could himself apply the depolarizing noise channel to the states he receives from Alice, and then carry out the continuation of the security experiment. In such a case, for all $\mu$, we simply have $q_{\mathrm{dep}}(\mu)=q_{\mathrm{dep}}(0)$.

Postulates~\ref{post:informal}\ref{post:ent-informal} and \ref{post:informal}\ref{post:selftest-informal} on the other hand are necessary in some form, in order for a DI implementation of $\ddP$ to exist at all, since device-independence essentially requires self-testing of an entangled state. Postulate~\ref{post:selftest-informal} is a self-testing statement, though with a slightly unusual framing. Namely, in much of the existing literature, self-testing results say that there must exist local isometries that bring the state and measurements of a strategy that wins $\game$ with close-to-ideal probability, close to the ideal state and measurements for $\game$, tensored with a junk state. We require that the isometries only depend on the measurements in the strategy, though the junk state and closeness parameter can depend on the state in the strategy. But we observe that many self-testing results in the literature do actually have this stronger form, even if they are not explicitly stated to be so. Self-testing for the CHSH game is explicitly stated in this form in \cite{RUV13, RKB18}; it can be checked that self-testing results for some other games such as Magic Square also hold in this form, by examining their proofs, e.g. in \cite{WBM+16}.

We now state our main result for compiling protocols $\ddP$ of the above form into DI protocols.
\begin{theorem}[Sequential DD-to-DI conversion; informal version of Theorems~\ref{th:complete} and \ref{th:DI}]\label{th:DI-informal}
Let $\ddP$ be a DD protocol satisfying Postulate Set~\ref{post:informal}, with a non-local game $\game$. Then, there is a sequential DI implementation $\diP$ of $\ddP$ which first uses a sequential state preparation procedure using devices compatible with $\game$, which either accepts or aborts. If $\diP$ accepts, it then does the cryptographic continuation of $\ddP$ (on which the event $\Bad$ is defined). For every $k>0$, there exist $\mu = \Theta(k)$, $\gamma \in (1/2,1)$ with $\gamma\to1$ as $k\to0$, $\epstol \in (0, (1-\gamma)/2)$, such that $\diP$ uses the devices $n$ times, with $m=\lfloor(1-\gamma)n/2\rfloor$, and
\begin{itemize}
\item Honest devices that implement the ideal strategy for $\game$ with $\dhon$ noise, such that $\gamma\dhon<\epstol$, satisfy
\[ \Pr_{\diP}[\Acc] \geq 1-\exp(-\Omega(m)).\]
\item For any state and measurements implemented by the device,
\[ \Pr_{\diP}[\Acc\land\Bad] \leq (2^kq_{\mathrm{dep}}(\mu))^{m/2}.\]
\end{itemize}
If $\limsup_{\pert\to0}q_{\mathrm{dep}}(\pert)\leq q$ for some $0<q<1$, then for every $0<k<-\log q$ the constants can be chosen so that $2^kq_{\mathrm{dep}}(\pert)<1$, giving an exponentially small soundness error in $m$, just like $\ddP$.
\end{theorem}
We emphasize that in the above theorem, the parameters $\mu, \gamma, \epstol$ depend only on $k$, whose exponentiation provides the multiplicative factor in the base of the soundness probability, and not on $n$. However, as $k\to0$, the amount of honest noise $\dhon < \frac{1-\gamma}{2\gamma}$ that the protocol can tolerate also goes to $0$ (independently of $n$). This behaviour of $\dhon$ with $k$ is necessary: if the amount of noise is high enough that the devices only win $\game$ with the classical value $\omega(\game)$, then no self-testing is possible, and therefore security cannot hold. Therefore, the more secure we want the protocol to be, the less noise it can tolerate.

We also note that the rate of the protocol goes to $0$ as well when $k$ goes to $0$. $\diP$ makes $n$ uses of the device and continues with the protocol $\ddP$, which means it is producing $m$ usable copies of the trusted state in DD protocol (with the other device uses being consumed for self-testing). The rate of the protocol is therefore
\[ \frac{m}n = \frac{1-\gamma}{2}.\]
Since $1-\gamma\to 0$ as $k\to0$, the rate must also go to $0$. Unlike in the case of noise tolerance, this behaviour is not necessary. Indeed, there exist sequential DI protocols for QKD where the rate remains constant even as the security parameter goes to $0$ \cite{ADF+18,ARV19, DF19, arx_AHT24}.\footnote{Technically, the way \emph{rate} is defined in QKD is not quite analogous to how we use it here. In QKD, if we extract $m$ bits that are $\lambda$-close to being uniform and unknown to the adversary, from $n$ uses of the devices, then the rate is $m/n$ and the security parameter is $\lambda$. QKD protocols typically give us $m = \Omega(n)-O(\log(1/\lambda))$, which means the rate is independent of the security parameter as long as we pick $\lambda=2^{-O(n)}$, which is typically the case.}

We have an alternative result where the rate does stay constant as $k$ goes to $0$, but we currently need some extra assumptions in this case, specifically limiting the memory of the devices between uses.
\begin{postulates}[Informal version of Postulate Set~\ref{post:no-memory}] \label{post:memory-informal}
The devices act on an initial state consisting of $n$ registers that can be arbitrarily entangled, and additionally purified to some arbitrary external register. In each round, the devices act solely on one of these registers, and all registers generated as a result of that round are not supplied as ``memory'' into the device actions in any other round. 
\end{postulates}

\begin{remark}\label{remark:no-memory}
We highlight that our no-memory condition in this work does \emph{not} assume the device behaviour is independent across rounds. Rather, we do allow the devices to act on states that are arbitrarily entangled across the rounds, but we do not allow the devices to pass memory from round to round via post-measurement states. Abstractly, we model this by declaring that the channels act via a tensor-product structure rather than sequential composition; this is the standard model used in e.g.~device-dependent QKD, and which we depict in Fig.~\ref{fig:tensor-product-channels} below.
\end{remark}

We have an improved DD-to-DI conversion theorem when the DI implementation satisfies the above postulate.
\begin{theorem}[Memoryless DD-to-DI conversion; informal version of Theorem~\ref{th:scalingdemo}]\label{th:memory-informal}
Let $\ddP$ be a DD protocol satisfying Postulate Set~\ref{post:informal}, with a non-local game $\game$, and let $\diP$ be its sequential implementation as in Theorem~\ref{th:DI-informal}. If the implementation of $\diP$ satisfies Postulate~\ref{post:memory-informal}, then, for every $k>0$, and every $\ptest \in (0,1)$, there exist $\mu = \Theta(k), \epstol \in (0, (1-\gamma)/2)$, $\base\in (0,1)$ with $\base\to1$ as $k\to0$, such that $\diP$ uses the devices $n$ times, with $m=\lfloor(1-\gamma)n/2\rfloor$, and
\begin{itemize}
\item Honest devices that implement the ideal strategy for $\game$ with $\dhon$ noise, such that $\gamma\dhon<\epstol$, satisfy
\[ \Pr_{\diP}[\Acc] \geq 1-\exp(-\Omega(m)).\]
\item For any state and measurements implemented by the device,
\[ \Pr_{\diP}[\Acc\land\Bad] \leq (2^kq_{\mathrm{dep}}(\mu))^m + b^m.\]
\end{itemize}
\end{theorem}

\begin{remark}
Aside from the Postulate Sets~\ref{post:informal} and \ref{post:memory-informal} required for the DD protocol, the DI protocols corresponding to Theorems~\ref{th:DI-informal} and \ref{th:memory-informal} need to satisfy some conditions at the implementation level. First, the protocols have to be sequential: which means the inputs and outputs for the $i$-th use of the devices (called the $i$-th round) have to be obtained before the $(i+1)$-th round begins. Second, Alice and Bob's devices have to be shielded from each other during the implementation of each round, so their devices only act on their respective local quantum registers --- this is generally a requirement for DI protocols, and more broadly it is usually also needed for the security of the ``base'' DD protocol as well.
\end{remark}

\subsubsection{Comparison theorems}\label{subsec:compare-result}
Our DD-to-DI conversion theorems are based on abstract theorems for comparing the output of two sets of channels that act on a starting state in sequence, in the style of entropy accumulation theorems. Specifically, we bound divergence measures between a ``real state'', which is the product of applying a channel sequence $\mapE_1, \ldots, \mapE_n$ to a starting state, and an ``ideal state'' which is the product of applying another channel sequence $\mapF_1,\ldots,\mapF_n$ on a starting state. It was observed in works such as~\cite{arx_MD23,arx_MD24} that such divergence bounds can have powerful operational implications, and this served as the inspiration for our work, though our proof methods and applications are rather different.

We have two such abstract comparison theorems, one which is used to obtain Theorem~\ref{th:DI-informal} and one used to obtain Theorem~\ref{th:memory-informal}. The channel structures for both theorems are depicted in Figures~\ref{fig:sequential-channels} and \ref{fig:tensor-product-channels} respectively.

\paragraph{Sequential comparison theorem.} Figure~\ref{fig:sequential-channels} depicts the general sequential channel structure. There is a common starting state $\rho^0_{\inM_1R}$ on which both the real and ideal sequence of channels act; this state need not be pure, and no channel in the sequence acts on $R$. Both channels $\mapE_i$ and $\mapF_i$ act on the round-$i$ state in the register $M_i$, and produce two ``protected'' registers $C_i$ and $O_i$, and a continuing ``memory'' $M_{i+1}$ for the next round. The register $C_i$ is classical and all $C_i$-s lie in the same set $\mathcal{C}$, while $O_i$ could be anything. Neither of these protected registers is touched by any channel $\mapE_j, \mapF_j$ for $j > i$. The real and ideal state are defined as
\[ \realstate = \left((\mapE_n\circ\ldots\circ\mapE_1)\otimes\idmap_R\right)[\rho^0_{\inM_1R}], \qquad \idealstate = \left((\mapF_n\circ\ldots\circ\mapF_1)\otimes\idmap_R\right)[\rho^0_{\inM_1R}].\]

At the end of the entire sequence, we define a global event $\Omega$ on the frequency distribution observed in the $C_1, \ldots, C_n$ registers (i.e.~the number of occurrences of each symbol, divided by $n$). $S_\Omega$ is used to denote the set of probability distributions on $\mathcal{C}$, which would give frequency distributions consistent with the event $\Omega$. We condition the real state on the event $\Omega$.

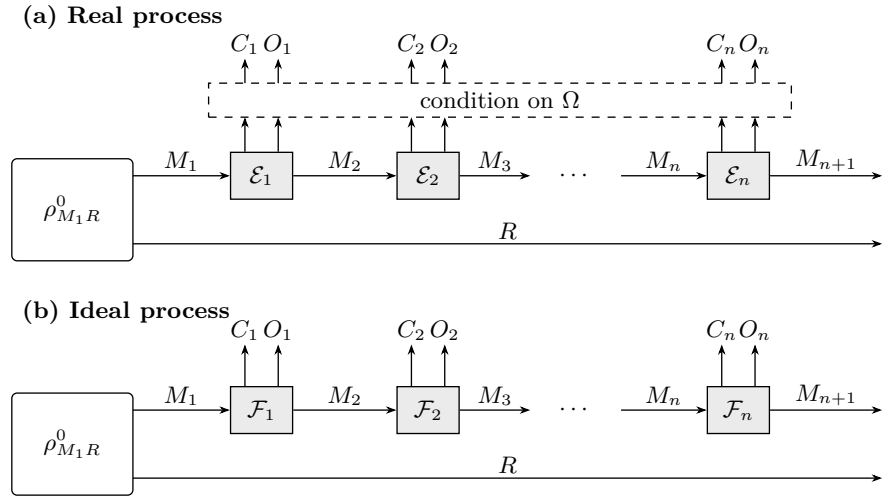
\begin{figure}[!ht]
\centering
\begin{tikzpicture}[
  x=1cm, y=1cm, font=\footnotesize,
  channel/.style={draw, fill=black!8, line width=.5pt, minimum width=.82cm, minimum height=.62cm, inner sep=2pt},
  state/.style={draw, rounded corners=2pt, line width=.5pt, fill=white, inner sep=3pt},
  wire/.style={-{Stealth[length=1.4mm,width=1mm]}, line width=.5pt},
  register/.style={inner sep=1.5pt},
  panel/.style={anchor=west, font=\footnotesize\bfseries}
]
\begin{scope}
\node[panel] at (-1.2,2.1) {(a) Real process};
\node[state, minimum width=1.6cm, minimum height=1.35cm] (input) at (-.4,-.45) {$\rho^0_{\inM_1R}$};
\node[channel] (round1) at (2.1,0) {$\mapE_{1}$};
\node[channel] (round2) at (4.3,0) {$\mapE_{2}$};
\node[channel] (roundn) at (8.4,0) {$\mapE_{n}$};
\node[draw,dashed,line width=.5pt,minimum width=7.7cm,minimum height=.45cm,inner sep=2pt] (condition) at (5.25,1) {condition on $\Omega$};
\foreach \i/\x in {1/2.1,2/4.3,n/8.4} {
  \draw[wire] ([xshift=-2.2mm]round\i.north) -- (\x-.22,.775);
  \draw[wire] ([xshift=2.2mm]round\i.north) -- (\x+.22,.775);
  \draw[wire] (\x-.22,1.225) -- (\x-.22,1.55) node[register,above] {$\CS_{\i}$};
  \draw[wire] (\x+.22,1.225) -- (\x+.22,1.55) node[register,above] {$\outQ_{\i}$};
}
\draw[wire] (.4,0) -- node[register,above] {$\inM_1$} (round1.west);
\draw[wire] (round1.east) -- node[register,above] {$\inM_2$} (round2.west);
\draw[wire] (round2.east) -- node[register,above] {$\inM_3$} (5.65,0);
\node at (6.25,0) {$\cdots$};
\draw[wire] (6.85,0) -- node[register,above] {$\inM_n$} (roundn.west);
\draw[wire] (roundn.east) -- node[register,above] {$\inM_{n+1}$} (10.3,0);
\draw[wire] (.4,-.9) -- node[register,above,pos=.5] {$R$} (10.3,-.9);
\end{scope}
\begin{scope}[yshift=-3.1cm]
\node[panel] at (-1.2,1.3) {(b) Ideal process};
\node[state, minimum width=1.6cm, minimum height=1.35cm] (input) at (-.4,-.45) {$\rho^0_{\inM_1R}$};
\node[channel] (round1) at (2.1,0) {$\mapF_{1}$};
\node[channel] (round2) at (4.3,0) {$\mapF_{2}$};
\node[channel] (roundn) at (8.4,0) {$\mapF_{n}$};
\foreach \i in {1,2,n} {
  \draw[wire] ([xshift=-2.2mm]round\i.north) -- ++(0,.55) node[register,above] {$\CS_{\i}$};
  \draw[wire] ([xshift=2.2mm]round\i.north) -- ++(0,.55) node[register,above] {$\outQ_{\i}$};
}
\draw[wire] (.4,0) -- node[register,above] {$\inM_1$} (round1.west);
\draw[wire] (round1.east) -- node[register,above] {$\inM_2$} (round2.west);
\draw[wire] (round2.east) -- node[register,above] {$\inM_3$} (5.65,0);
\node at (6.25,0) {$\cdots$};
\draw[wire] (6.85,0) -- node[register,above] {$\inM_n$} (roundn.west);
\draw[wire] (roundn.east) -- node[register,above] {$\inM_{n+1}$} (10.3,0);
\draw[wire] (.4,-.9) -- node[register,above,pos=.5] {$R$} (10.3,-.9);
\end{scope}
\end{tikzpicture}
\caption{Sequential channel structure. The real and ideal processes start from the same state $\rho^0_{\inM_1R}$, which need not be pure. All channels act as the identity on $R$. Each round passes $\inM_{i+1}$ to the next round and produces the protected registers $\CS_i\outQ_i$, of which $\CS_i$ is classical. The dashed box denotes conditioning the real output on $\Omega$, which is defined on $\CS_1, \ldots, \CS_n$ only.
}
\label{fig:sequential-channels}
\end{figure}
Our abstract comparison theorem, for the general sequential case, which is stated below, compares the sandwiched Renyi divergence $D_\alpha$ of order $\alpha$, for $\alpha>1$, between the real and ideal states. The full set of postulates required on the channels in order for the Theorem to apply are stated in Postulate Set~\ref{posts:seq}. In the subsequent discussion, we will use the somewhat informal terminology that an operator $A$ is ``support-contained in'' an operator $B$ if $\supp(A) \subseteq \supp(B)$.
\begin{theorem}[Sequential comparison theorem; informal version of Theorem~\ref{th:main}]\label{th:main-informal}
Suppose for the sequence of channels described above, $\mapE_i[\sigma]$ is support-contained in $\mapF_i[\sigma]$ for any state $\sigma$ and all $i\in[n]$, and some other mild conditions are satisfied.
Then, if $p_\Omega=\Pr_{\realstate}[\Omega]>0$, and $\alpha >1$, we have
\[ D_\alpha(\realstate_{|\Omega}\Vert\idealstate) \leq
\sup_{S_\Omega}\left[\sum_{i=1}^nd^{\#}_\alpha(\mapE_i, \mapF_i)\right] +\frac\alpha{\alpha-1}\log\frac1{p_\Omega}.\]
Here $d^{\#}_\alpha(\mapE_i, \mapF_i)$ is a divergence-based quantity depending only on the $i$-th round channels $\mapE_i, \mapF_i$, and a distribution over the set $\mathcal{C}$; a supremum over such distributions that are contained in $S_\Omega$ is taken on the outside.
\end{theorem}
This theorem upper bounds the $D_\alpha$ between the real state conditioned on $\Omega$ and the ideal state in terms of the sum of single-round divergences, as well as the real probability of $\Omega$. This is analogous to entropy accumulation theorems, which work for similar channel structures, and \emph{lower bound} the conditional entropy of a conditioned real state after a sequential process, in terms of the sum of single-round conditional entropies, and the same conditioning probability. The single-round quantities $d^{\#}_\alpha(\mapE_i, \mapF_i)$ are defined in terms of a quantity called the sharp-divergence, which was introduced in~\cite{FF21sharp}.

\paragraph{Flexible rigidity.} For the application to Theorem~\ref{th:DI-informal}, the real sequence of channels $\mapE_1, \ldots, \mapE_n$ is taken to be the real channels implemented by the devices in the protocol. The ideal channels $\mapF_1, \ldots, \mapF_n$ on the other hand are taken to be channels that replace part of the states in these devices with (a depolarized version of) the ideal state from the game $\game$. The depolarization condition is needed in order to make sure that the condition that $\mapE_i[\sigma]$ is support-contained in $\mapF_i[\sigma]$ is satisfied, since depolarized states have full support, and this is ultimately why we need Postulate~\ref{post:informal}\ref{post:dep-informal}.

The event $\Omega$ that we use in the application (which corresponds to the $\Acc$ event in Theorem~\ref{th:DI-informal}) is that on a sufficiently high fraction of the rounds in which we play $\game$, the winning condition of $\game$ is satisfied. Thus this theorem gives us an upper bound on the $D_\alpha$ between the real states in the protocol conditioned on them empirically winning $\game$ with high probability, and the state of the protocol if it had been carried out with $m$ copies of the depolarized ideal state for $\game$.

This is precisely the kind of self-testing result for multi-copy states we lacked before. The reason is that unlike trace-distance-based measures, divergences are actually good at handling small deviations in each of the $n$ copies. If the device state in each generation (i.e., non-test) round is $\delta$-far from the ideal state, the $D_\alpha$ between the conditioned real and ideal states would be $\poly(\delta)m$, just like in the case of trace distance. However, this bound on $D_\alpha$ is much more useful in terms of bounding an event probability on the conditioned real state in terms of the same probability on the ideal state. Specifically for $\alpha=2$,
\[ \Pr_{\realstate_{|\Omega}}[\Bad] \leq \sqrt{2^{D_2(\realstate_{|\Omega}\Vert\idealstate)}\Pr_{\idealstate}[\Bad]}.\]
For other values of $\alpha$, we can get a roughly similar relation between event probabilities w.r.t.~the two states by routing through a quantity known as the smoothed max-divergence. As long as $\Pr_{\idealstate}[\Bad]$ is exponentially small in $m$, and the divergence can be made ${km}$ for an arbitrarily small constant $k$, this lets us get the exponentially small upper bounds on $\Pr[\Acc\land\Bad]$ in Theorems~\ref{th:DI-informal} and \ref{th:memory-informal}. As we saw in the statements of those theorems, $k$ can be taken to be arbitrarily small, as long as we're willing to work with a small (but still independent-of-$n$) rate and error tolerance. Therefore, with the real and ideal channel identifications described above, Theorem~\ref{th:main-informal} in the sequential case, and Theorem~\ref{th:tensor-informal} that we state next for the tensor-product case, are the rigidity results with flexible cryptographic applications that we were looking for.

\paragraph{Tensor-product comparison theorem.} Figure~\ref{fig:tensor-product-channels} depicts the tensor-product channel structure, for which we prove our next comparison theorem. In this case, there is a (not necessarily pure) starting state $\rho^0_{\allN\bR}$ on registers $\allN=N_1\ldots N_n, \bR=R_1\ldots R_n$, and channels $\mapE, \mapF$ that act on the $N_i$ registers. The figure depicts map $\mapE_i$ and $\mapF_i$ acting on register $N_i$, but these channels are all isomorphic copies of the same two channels. $\rho^0_{\allN\bR}$ may be entangled across the $\allN$ registers, but the channels are tensor-product across these. Each channel $\mapE_i$ and $\mapF_i$ produces the equivalent of the protected registers $C_iO_i$ in the sequential setting, but no continuing memory registers for the other channels to act on.

The real final state $\realstate$ we consider is still
\[ \realstate = (\mapE^{\otimes n}\otimes\idmap_{\bR})[\rho^0_{\allN\bR}],\]
but the ideal state is instead of the form
\[ \idealstate = (\mapF^{\otimes n}\otimes\idmap_\bR)[\nu^0_{\allN\bR}],\] 
where the initial state $\nu^0_{\allN\bR}$ can depend on $\rho^0_{\allN\bR}$ and all the channels involved --- essentially, the comparison theorem will show that there \emph{exists} some such $\nu^0_{\allN\bR}$ such that the divergence between the real and ideal states can be bounded. As in the sequential setup, we will have an event $\Omega$ here defined on the frequency distribution on the $C_i$ registers, and here we will also consider a final measurement channel $\mapM$ on both $\realstate$ and $\idealstate$.

\begin{figure}[!ht]
\centering
\begin{tikzpicture}[
  x=1cm, y=1cm, font=\footnotesize,
  channel/.style={draw, fill=black!8, line width=.5pt, minimum width=.82cm, minimum height=.62cm, inner sep=2pt},
  state/.style={draw, rounded corners=2pt, line width=.5pt, fill=white, inner sep=3pt},
  wire/.style={-{Stealth[length=1.4mm,width=1mm]}, line width=.5pt},
  register/.style={inner sep=1.5pt},
  panel/.style={anchor=west, font=\footnotesize\bfseries}
]
\begin{scope}
\node[panel] at (-1.2,1.4) {(a) Real process};
\node[state, minimum width=1.8cm, minimum height=2.6cm] (input) at (-.3,-.35) {$\rho^0_{\allN\bR}$};
\node[channel, minimum height=2.6cm] (measurement) at (8.2,-.35) {$\mapM$};
\node[draw,dashed,line width=.5pt,align=center,minimum width=1.35cm,minimum height=2.6cm,inner sep=2pt] (condition) at (6,-.35) {condition\\on $\Omega$};
\node[channel] (round1) at (3,.55) {$\mapE_{1}$};
\draw[wire] (.6,.55) -- node[register,above] {$\inN_{1}$} (round1.west);
\draw[wire] (round1.east) -- node[register,above] {$\CS_{1}\outQ_{1}$} (5.325,.55);
\draw[wire] (6.675,.55) -- (7.79,.55);
\node[channel] (roundn) at (3,-.65) {$\mapE_{n}$};
\draw[wire] (.6,-.65) -- node[register,above] {$\inN_{n}$} (roundn.west);
\draw[wire] (roundn.east) -- node[register,above] {$\CS_{n}\outQ_{n}$} (5.325,-.65);
\draw[wire] (6.675,-.65) -- (7.79,-.65);
\node at (3,-.05) {$\vdots$};
\draw[wire] (.6,-1.35) -- node[register,above,pos=.38] {$\bR$} (5.325,-1.35);
\draw[wire] (6.675,-1.35) -- (7.79,-1.35);
\draw[wire] (measurement.east) -- node[register,above] {$F$} (10.3,-.35);
\end{scope}
\begin{scope}[yshift=-3.65cm]
\node[panel] at (-1.2,1.4) {(b) Ideal process};
\node[state, minimum width=1.8cm, minimum height=2.6cm] (input) at (-.3,-.35) {$\nu^0_{\allN\bR}$};
\node[channel, minimum height=2.6cm] (measurement) at (8.2,-.35) {$\mapM$};
\node[channel] (round1) at (3,.55) {$\mapF_{1}$};
\draw[wire] (.6,.55) -- node[register,above] {$\inN_{1}$} (round1.west);
\draw[wire] (round1.east) -- node[register,above] {$\CS_{1}\outQ_{1}$} (7.79,.55);
\node[channel] (roundn) at (3,-.65) {$\mapF_{n}$};
\draw[wire] (.6,-.65) -- node[register,above] {$\inN_{n}$} (roundn.west);
\draw[wire] (roundn.east) -- node[register,above] {$\CS_{n}\outQ_{n}$} (7.79,-.65);
\node at (3,-.05) {$\vdots$};
\draw[wire] (.6,-1.35) -- node[register,above,pos=.38] {$\bR$} (7.79,-1.35);
\draw[wire] (measurement.east) -- node[register,above] {$F$} (10.3,-.35);
\end{scope}
\end{tikzpicture}
\caption{Tensor-product channel structure. The round channels are isomorphic copies of $\mapE$ or $\mapF$, acting on distinct input registers $\inN_i$. There are different input states $\rho^0_{\allN\bR}$ and $\nu^0_{\allN\bR}$; the latter depends on the channels, $\rho^0_{\allN\bR}$ and $\mapM$. The dashed box denotes conditioning the real output on $\Omega$, before applying the same measurement channel $\mapM$ in both processes.}
\label{fig:tensor-product-channels}
\end{figure}

We now state our comparison theorem for tensor-product channels. The full set of conditions required for this theorem to apply are given in Postulate Set~\ref{posts:overall}.
\begin{theorem}[Tensor-product comparison theorem; informal version of Theorem~\ref{th:robustselftest}]\label{th:tensor-informal}
Suppose the channels $\mapE$ and $\mapF$ described are such that $\mapE[\sigma]$ is support-contained in $\mapF[\sigma]$ for all $\sigma$, and some other mild conditions are satisfied. Then, for every initial state $\rho^0_{\allN\bR}$, every measurement channel $\mapM$, there exists a state $\nu^0_{\allN\bR}$ such that, with $\idealstate$ defined using $\nu^0_{\allN\bR}$, and $p_\Omega=\Pr_{\realstate}[\Omega] > 0$, and $\alpha>1$,
\[ D_\alpha(\mapM[\realstate_{|\Omega}]\Vert\mapM[\idealstate]) \leq n\cdot\sup_{S_\Omega}d_\alpha(\mapE,\mapF) + \frac\alpha{\alpha-1}\log\frac1{p_\Omega}.\]
Here $d_\alpha(\mapE,\mapF)$ is a divergence-based quantity depending on a distribution over the set $\mathcal{C}$; a supremum over such distributions that are contained in $S_\Omega$ is taken on the outside. 
\end{theorem}
Unlike Theorem~\ref{th:main-informal}, the single-round quantities $d_\alpha$ in the above theorem are defined in terms of sandwiched divergence. Unlike sharp-divergences, sandwiched divergences are better behaved in the limit $\alpha\to1$, which is the main property that allows us to not have to take $\gamma\to1$ when we want $k\to0$, when we apply Theorem~\ref{th:tensor-informal} to get Theorem~\ref{th:memory-informal}. This is akin to how analysis is done in e.g.~\cite{DFR20, arx_AHT24} to get good asymptotic key rates for QKD from entropy accumulation theorems, while using a small fraction of rounds for self-testing.

We further comment on how Theorem~\ref{th:tensor-informal} can be interpreted as a flexible rigidity theorem like Theorem~\ref{th:main-informal}, in order to get Theorem~\ref{th:memory-informal}. Firstly, the formal version of Theorem~\ref{th:memory-informal} actually allows distinct channels in each round, but this can be modelled by the single pair of channels in Theorem~\ref{th:tensor-informal}; see Remark~\ref{remark:IIDchanns} for further details. Secondly, the ideal state in Theorem~\ref{th:memory-informal} is different from the ideal state in Theorem~\ref{th:main-informal}, since it does not use the original state $\rho^0_{\allN\bR}$. But recall that we took the ideal channel in the rigidity application to be a channel that replaces some of the registers of its input state with the ideal state for $\game$, and we only needed this replacement property in our argument. The fact that the ideal channel was acting on the same initial state as the real channel was not really important. Finally, the fact that Theorem~\ref{th:tensor-informal} has a measurement channel at the end is also fine, since we ultimately want to relate the probabilities of $\Bad$ under $\realstate_{|\Omega}$ and $\idealstate$ anyway; therefore, we can use Theorem~\ref{th:tensor-informal} in the memoryless rigidity application in the same way as Theorem~\ref{th:main-informal}.

\subsubsection{Applications}
In order to show the applicability of our DD-to-DI conversion theorems, we construct DI protocols for three cryptographic tasks: uncloneable encryption, copy-protection of multi-bit point functions, and blind randomness expansion. The analysis of all these protocols is done in the general sequential setting, applying Theorem~\ref{th:DI-informal}, though it is also possible to do the analysis in the memoryless case using Theorem~\ref{th:memory-informal}, in order to get better rates.

For all three applications, we give sequential protocols, which have a sequential state preparation phase in the style of the protocol from Theorem~\ref{th:DI-informal}, followed by a cryptographic continuation corresponding to the protocol in question. The sequential state preparation phase necessarily requires Alice to send inputs and get outputs from Bob round by round, which means the number of rounds of interaction is linear in the security parameter. This is a limitation of our sequential approach over the parallel repetition-based approach to DI cryptography~\cite{KT23, KT25, BCJK24}, which require only constantly many rounds, and can handle the devices getting all the inputs and producing all the outputs at once.

\paragraph{Uncloneable encryption.} In uncloneable encryption (UE), introduced in~\cite{BL20}, Alice, who holds a secret key, sends a quantum encryption of a message to Bob using the key. Bob may later split the encrypted ciphertext between two decoders, and the idea is that even if both decoders then receive the secret key, both of them should not be able to recover the message from it. There are two notions of security in UE: search security (also known as simply uncloneability), and decision security (also known as uncloneable-indistinguishability). It has been known how to achieve search security for UE since~\cite{BL20}; decision security was only shown recently in a number of works, of which we use~\cite{CLX26}.

Our result for uncloneable encryption is stated below.
\begin{theorem}[Combination of Theorems~\ref{th:UE} and \ref{cor:decision}]\label{th:UE-informal}
There exist sequential DI protocols for UE of arbitrary-length messages with search security, and for one-bit messages with decision security.
\end{theorem}
The search-secure protocol is obtained by using the construction of Theorem~\ref{th:DI-informal} as a blackbox on the Wiesner-state-based protocol from~\cite{BL20}, along with self-testing for the CHSH game. The decision-secure protocol is obtained by applying the construction of~\cite{CLX26} to the search-secure protocol.

This is the first DI protocol for standard uncloneable encryption with information-theoretic security. \cite{KT25} gave a DI protocol for a variant known as uncloneable encryption with variable keys, although as mentioned before, their construction did work in the parallel setting.

\paragraph{Copy-protection of point functions.}  In quantum copy protection, which was introduced in \cite{Aar09}, Alice sends a quantum state $\sigma_f$ as a copy-protected program for a function $f$ to Bob, who should be able to evaluate the $f$ on any input with high probability by doing a measurement on $\sigma_f$. The security requirement is that if Bob splits the copy-protected $\sigma_f$ between two decoders, both of them should not be able to evaluate the function on their respective registers with probability higher than a trivial value, when inputs come from a specified distribution.

We follow \cite{GMP22} in getting copy protection for a specific class of functions (multi-bit point functions) under a specific distribution, using UE as a blackbox. Our result for copy protection, using the search UE protocol from Theorem~\ref{th:UE-informal}, is stated below.
\begin{theorem}[Informal version of Theorem~\ref{th:copy}]
There exists a sequential DI protocol for quantum copy protection of multi-bit point functions.
\end{theorem}
This is the first DI protocol for copy protection of any class of functions with information-theoretic security. Unlike \cite{GMP22}, we do not use a random oracle in our construction and instead give a hash function in the plain. It is possible to get security with this modification, under the standard security definition for copy protection.

\paragraph{Blind randomness expansion.} In blind randomness expansion, which was introduced in \cite{MS17}, Alice uses an untrusted device to generate randomness that remains secret from Bob, who supplies the device and participates in the protocol.  Metger, Fawzi, Sutter, and Renner~\cite[Section~5.1]{MFSR24} gave the first blind randomness expansion protocol secure against general sequential attacks, using generalized entropy accumulation. Their CHSH-based construction produces linearly many output bits from a polylogarithmic seed. Our construction uses the Magic Square game instead, and consumes linear input randomness. We do not obtain an improvement in randomness consumption or security guarantees over \cite{MFSR24}; the purpose of our construction is only to show an application for our divergence-based approach. As in~\cite[Definition~11 of the full version]{JK22}, our security definition for randomness expansion needs only to hold conditional on the input randomness, without requiring the output alone to be longer than the seed.

Unlike UE, the security of our DI blind randomness expansion protocol is not obtained from a DD protocol in a blackbox way by applying Theorem~\ref{th:DI-informal}: the concept of blind randomness expansion may not be very sensible in a DD setting. Instead our security proof is done by using a property of the ideal Magic Square states shown in~\cite{FM18}, and identifying the real and ideal channels (which replaces some registers with the ideal Magic Square states) so that we can apply Theorem~\ref{th:main}, just like in the proof of Theorem~\ref{th:DI-informal}.

\subsection{Our techniques}
\label{subsec:techniques}
We briefly describe how we obtain our main results. First we will give a rough sketch of how both comparison theorems, Theorems~\ref{th:tensor-informal} and \ref{th:main-informal} are proved (using, for consistency across this subsection $\Omega=\Acc$). Then we will describe how these comparison theorems are used to prove the DD-to-DI compilers, Theorems~\ref{th:DI-informal} and \ref{th:memory-informal}: this will consist of upper bounding the single-round quantities from the comparison theorems under the channel identifications as described in Section~\ref{subsec:compare-result}.

\subsubsection{Chain rules}
The first step is finding an appropriate \emph{chain rule}. Consider some quantum registers $\bA = A_1 A_2 \dots A_n$, and consider some {\Renyi} entropy or divergence defined on them --- let us conceptually express this via the highly informal notation $\renD(\bA)$ (ignoring, for instance, that divergences are defined on {pairs} of states rather than a single state). For most {\Renyi} divergences $\renD(\bA)$, it is ``usually'' (though not always) possible to use it to upper bound the smooth max-divergence, which provides a path towards rigidity theorems of our desired form.
A chain rule typically consists of statements of the form that, if the registers $\bA$ have some structure such as e.g.~being produced sequentially by a series of quantum channels, then\footnote{One might hope for a stronger statement with some kind of equality rather than inequality;
unfortunately this does not usually hold outside specific cases such as (classical) Kullback-Leibler divergence and (quantum) von Neumann entropy.}
\begin{align}\label{eq:vibe_chain}
\renD(\bA) \leq \sum_i \sup \renD(A_i),
\end{align}
where the supremum is usually taken over some rather ``large'' set of states: for instance, the set of all possible output states of some channels. This breaks down the global term $\renD(\bA)$ into contributions from the local terms $\renD(A_i)$. Technically, for the channel structure in this work the latter terms would also have further ``memory'' registers $M_i$ to be passed into the later channels, but we ignore this point for this overview.

Unfortunately,~\eqref{eq:vibe_chain} is too optimistic a goal for many quantum {\Renyi} divergences. For instance, the sandwiched divergence typically instead has a \emph{regularized} quantity on the right-hand-side of that bound, defined via a somewhat elaborate asymptotic limit that we find challenging to analyze. Therefore, the first conceptual point in our analysis is to identify a suitable quantum {\Renyi} divergence and chain rule, which has sufficiently tractable quantities on both sides of~\eqref{eq:vibe_chain}. 

The first option we consider is the sharp-divergence $\Dsharp_\alpha$~\cite{FF21sharp}, which has such a chain rule and has been applied for cryptography in e.g.~\cite{arx_MD23}. It has the drawback of converging to the Belavkin--Staszewski (BS) divergence as $\alpha\to 1$, which is slightly ill-behaved and places some barriers on the protocol rates we can obtain this way.
(The geometric divergence~\cite{FF21geom} also has such a chain rule, but yields even worse bounds as it is larger than the BS divergence for all $\alpha>1$.)

As a second option, we turn to prospects for overcoming the regularization challenges for the sandwiched divergence. Our key realization here is to instead try slightly modifying the \emph{left}-hand-side of~\eqref{eq:vibe_chain} in that case. Specifically, on that side we grant ourselves the freedom to choose an arbitrary ``initial state'' in the second argument, as discussed above in Fig.~\ref{fig:tensor-product-channels}. 

From an information-theory standpoint, this basically corresponds to proposing a modified notion of channel divergence between $\mapE^{\otimes n}\otimes\idmap_{\bR}$ and $\mapF^{\otimes n}\otimes\idmap_\bR$. Usually, the channel divergence is defined as the largest divergence (between output states) obtained by maximizing over all input states to both channels; here, we still maximize over input states to the first channel but \emph{minimize} over input states to the second channel.\footnote{This is distinct from the ``minimax divergences'' proposed in~\cite{FC26}, which has a significantly different structure of the input registers and optimizations.} This grants the resulting quantity some properties similar to one version of {\Renyi} conditional entropy, which minimizes over a state in the second argument. Specifically, for sandwiched divergences it allows us to apply a de~Finetti argument similar to~\cite{GW21,MFSR24} to tame its regularization behaviour. The price we pay for this idea is that this channel divergence is no longer ``manifestly'' smaller than its regularized version, unlike the usual channel divergence; we handle this by introducing a suitable final measurement to induce the desired behaviour. 

Overall, our main contribution here is the fact that this modified channel divergence is, apart from a final measurement, basically upper bounded by $n$ times of the usual channel divergence between $\mapE$ and $\mapF$ \emph{without regularization}. We call this an ``almost-subadditivity'' property, since the left- and right-hand sides of~\eqref{eq:vibe_chain} correspond here to slightly different quantities, but this suffices for our intended applications. An interesting feature is that it is ``non-asymptotic'' despite the de~Finetti argument, in that it holds for every $n$ with no fudge terms --- in particular, it evades the explicit dimension dependency that is usually induced by de~Finetti arguments.

\subsubsection{Incorporating testing}

A bound of the form~\eqref{eq:vibe_chain} still usually does not suffice for cryptography, since the right-hand-side is often an ``unrestricted'' supremum over all output states of some channels, yielding trivial bounds in protocol contexts. The key insight of a series of works in the topic of \emph{entropy accumulation}~\cite{DFR20,arx_AHT24} is that such apparently-unhelpful chain rules can usually be converted into a property of ``statistical estimators'' in a protocol. By conducting a ``test'' involving these estimators, one can then obtain nontrivial bounds on $\renD(\bA)$ conditioned on the test accepting. In the context of the protocol, this $\renD(\bA)$ will be the divergence between the real state and a suitable ideal state.

More precisely: suppose that alongside the quantum registers $\bA$, we have classical registers $\bC = C_1 C_2 \dots C_n$, jointly generated sequentially by a series of quantum channels. In a cryptographic protocol, we can use these channels to model ``rounds'' of the protocol; e.g.~we take some entangled state and with some probability ``test'' it by playing a single-copy non-local game on it (storing the outcome in $C_i$).
Now consider \emph{any} functions $f_i : \mathcal{C}_i \to \mathbb{R}$, where $\mathcal{C}_i$ denotes the alphabet of $C_i$. Suppose we could use the chain rule~\eqref{eq:vibe_chain} to show
\begin{align}\label{eq:vibe_statlinear}
\expval_{\bcS} \left[\renD(\bA)_{|\bcS} - \sum_i f_i(c_i) \right] 
\overset{?}{\leq}
\sum_i \sup \expval_{c_i} \left[\renD(A_i)_{|c_i} - f_i(c_i) \right],
\end{align}
where by $\renD(\bA)_{|\mathbf{c}}$ we mean the divergence of the state(s) conditioned on $\bC = \mathbf{c}$, and analogously for $\renD(A_i)_{|c_i}$. The suprema are taken over similar sets as in~\eqref{eq:vibe_chain}, inheriting this from the chain rule.

Observe that each expectation term on the right can be informally viewed as a vague ``quality measure'' of how well $ f_i(c_i)$ estimates $\renD(A_i)_{|c_i}$ in that round --- if the expectation is large, it ``underestimates'' $\renD(A_i)_{|c_i}$; if the expectation is zero, it is a ``good estimate'' of $\renD(A_i)_{|c_i}$, and so on (in expectation over $c_i$). By taking the supremum of this expectation in each round, we are quantifying how this estimate performs for the ``worst possible'' outputs from the channel in that round. 

Critically,~\eqref{eq:vibe_statlinear} would then tell us that as long as we know the ``quality'' of the single-round estimators $f_i(c_i)$ (via the terms on the right-hand-side), then we can bound the ``quality'' of the global estimator $\sum_i f_i(c_i)$ for $\renD(\bA)_{|\bcS}$ (on the left-hand-side). Here, the fact that the inequality is in the $\leq$ direction reflects that in this work we mainly want ``upper estimators'' on $\renD(\bA)_{|\bcS}$, so we mostly only need to bound the ``quality'' in that direction. Keep in mind that $\sum_i f_i(c_i)$ is an entirely classical value that could be physically observed and computed in the protocol: this would thus be a very practical prospect for estimating $\renD(\bA)_{|\bcS}$. Our task is hence reduced to choosing ``good'' single-round estimators $f_i(c_i)$, which seems more tractable as they can be understood by studying the rounds individually, even if we do have to account for their performance on the ``worst-case'' output states each round. 

At this point we should clarify that for most {\Renyi} divergences, we do not literally obtain a bound of the form~\eqref{eq:vibe_statlinear}, other than at $\alpha=1$. Rather, the chain rules~\eqref{eq:vibe_chain} instead usually yield bounds of the form
\begin{align}\label{eq:vibe_statlme}
\log_{B_\alpha} \expval_{\bcS} \left[B_\alpha^{\renD(\bA)_{|\bcS} - \sum_i f_i(c_i) }\right] 
\leq
\sum_i \sup \log_{B_\alpha} \expval_{c_i} \left[B_\alpha^{\renD(A_i)_{|c_i} - f_i(c_i)} \right],
\end{align}
i.e.~we have ``exponential mean'' or ``log-mean-exponential'' expressions with respect to some base ${B_\alpha}>0$ (which depends on $\alpha$). However, since logarithm and exponential are just an inverse pair of monotone functions, this does not particularly affect any of the intuition we discussed above. For readers more familiar with martingale proof techniques, we note that (for ${B_\alpha}>1$) this is equivalent to a multiplicative bound
\begin{align}
\expval_{\bcS} \left[B_\alpha^{\renD(\bA)_{|\bcS} - \sum_i f_i(c_i) }\right] 
\leq
\prod_i \sup \expval_{c_i} \left[B_\alpha^{\renD(A_i)_{|c_i} - f_i(c_i)} \right],
\end{align}
a form that was used in e.g.~the \emph{quantum probability estimation} framework of~\cite{ZFK20}. To emphasize the estimator interpretation, the functions $f_i$ were called ``quantum estimation score-systems'' in~\cite{arx_AHT24}, via analogy to ``quantum estimation factors'' (based on the multiplicative form) in~\cite{ZFK20}.

Returning to the protocol, let us suppose it conducts a ``test'' that decides whether to accept based on the observed values $\bcS$. (In fact one could even consider more general prospects, but we defer this to Sec.~\ref{subsec:futurework}.) 
We would expect that the resulting value of $\renD(\bA)_{|\Acc}$ could then be bounded in terms of the $\sum_i f_i(c_i)$ value conditioned on accepting, if it is a ``good'' estimator.
It turns out this can be turned into an actual proof strategy. We first modify the left-hand-side of~\eqref{eq:vibe_statlme} by ``truncating'' the expectation (viewed as a sum over $\bcS$) to only the $\bcS$ values that the protocol accepts, which is basically possible at the cost of a correction that depends on $\Pr[\Acc]$. With this we can then crudely bound all the surviving $\sum_i f_i(c_i)$ terms with the worst (largest) value that would be accepted. We then flip this largest value onto the right-hand-side, obtaining  
\begin{align}
\renD(\bA)_{|\Acc}
\leq
&\sum_i \sup \log_{B_\alpha} \expval_{c_i} \left[B_\alpha^{\renD(A_i)_{|c_i} - f_i(c_i)} \right] \nonumber\\
&\quad+ \left[\text{largest $\sum_i f_i(c_i)$ value conditioned on accepting}\right] \nonumber\\
&\quad+ \frac{\alpha}{\alpha-1}\log\frac1{\Pr[\Acc]}. \label{eq:vibe_truncated}
\end{align}
Even putting aside the reasoning leading here, the result is somewhat intuitive --- the first term is inherited from~\eqref{eq:vibe_statlme} and reflects that a ``low-quality'' {\copper} would yield a worse bound, the second term imposes that we have conditioned on some particular $\sum_i f_i(c_i)$ values, and the third term corrects for conditioning on an event.

\subsubsection{Choosing the estimators}
\label{subsubsec:vibes_opt}

Of course, the catch is that we have not yet specified a recipe for how to actually find choices of $f_i$ that are ``good'' single-round estimators. While one could hand-craft them, it is surprisingly easy to find the ``best'' choice (in a certain sense) via minimax theorems. 
The important point to observe is that the bound~\eqref{eq:vibe_truncated} holds \emph{for any} $f_i$, regardless of how ``bad'' they are as estimators --- their ``quality'' has been captured simply by the presence of the first term. We can thus minimize the right-hand-side over all possible choices of $f_i$. It turns out that under fairly mild conditions on the accept criterion, the resulting minimax problem cleanly simplifies using techniques from convex duality, yielding a rather intuitive bound.

Specifically, let us suppose for simplicity that the accept condition in the protocol is that it first computes the frequency distribution observed on the classical registers $\bCS$, then accepts if and only if this frequency distribution lies in some closed convex neighbourhood $S_\Acc$ of the ``noiseless'' ideal distribution.\footnote{This suffices to include accept conditions such as e.g.~the frequency distribution being within trace distance $\epstol$ of the ideal distribution. Though more broadly, any accept condition involving only the $\bCS$ values can be ``coarsened'' into this form by setting $S_\Acc$ to be the convex hull of all accepted frequency distributions.} For this accept condition, minimizing over the estimators $f_i$ yields a clean bound of the form
\begin{align}\label{eq:vibe_Dbound}
\renD(\bA)_{|\Acc}
\leq
\sup \sum_i \left[\text{single-round quantities}\right] + \frac{\alpha}{\alpha-1}\log\frac1{\Pr[\Acc]},
\end{align}
where the supremum is ``effectively restricted'' in a certain informal sense, which we shall now discuss. Note that in this expression all the dependency on $f_i$ has been ``optimized out'', i.e.~we never need to consider explicit $f_i$ values from this point forward (though doing so can have applications we mention in Sec.~\ref{subsec:futurework}). 

\begin{remark}
The supremum being outside the sum in~\eqref{eq:vibe_Dbound} reflects a subtle tightening in our results compared to previous work such as~\cite{DFR20,arx_AHT24}. In slightly more detail: throughout~\eqref{eq:vibe_chain}--\eqref{eq:vibe_statlme}, the various suprema were originally ``uncoupled'', in that there was a separate supremum for each channel/round. In~\eqref{eq:vibe_Dbound} however, all of these suprema become ``coupled'' in the process of a minimax swap, resulting in a single supremum of a sum across the rounds. For the purposes of~\eqref{eq:vibe_penalty} below, we implicitly relax it slightly to a coarser picture by defining a single ``representative'' channel that suffices to describe any of the individual channels, an approach that was also taken in e.g.~\cite{DFR20,arx_AHT24} to simplify the analysis.
\end{remark}

We now elaborate on the ``single-round quantities'' in~\eqref{eq:vibe_Dbound}. Following our protocol descriptions above, let us say that in test rounds (where the nonlocal game is played) the game outcome is recorded in $\CS_i$, and for the remaining rounds, we call them generation rounds and the $\CS_i$ register for that round is set to a special symbol $\gen$. 
After a few relaxations and simplifications, including using a single ``representative'' channel to decouple the optimizations across rounds, the supremum in~\eqref{eq:vibe_Dbound} basically reduces to the form
\begin{align}
\label{eq:vibe_penalty}
\sup_{\rchi\in \outputset}  \sup_{\sfQ \in S_\Acc} n\,  \left[ 
\sfQ(\gen) \renD(A)_{|\gen} 
- \frac{1}{\alpha-1} D\Dargs{\sfQ}{\sfP}  \right],
\end{align}
where the terms can be, informally, interpreted as follows. The set $\outputset$ denotes all output states that could have been produced in \emph{any single round} of the protocol (ignoring a convex-hull issue we explain in Sec.~\ref{subsec:bound_compexp}); we thus omit the round subscripts $i$ in the remaining notation. For each state $\rchi \in \outputset$, the term $\renD(A)_{|\gen}$ denotes its corresponding {\Renyi} divergence quantity conditioned on being a generation round, and $\sfP$ denotes the distribution it induces on $\CS$. Finally, since each $\sfQ \in S_\Acc$ is a probability distribution, we write $\sfQ(\gen)$ to denote its weight on the $\gen$ term, and $D\Dargs{\sfQ}{\sfP}$ is the Kullback-Leibler (KL) divergence between $\sfQ$ and $\sfP$.

We can intuitively understand~\eqref{eq:vibe_penalty} by noting that it can be viewed as a relaxed version of the constrained optimization
\begin{align}
\sup_{\rchi \in \outputset} \quad & n \, \sfP(\gen) \renD(A)_{|\gen} \nonumber\\
\suchthat \quad & \sfP\in S_\Acc,
\label{eq:vibe_constrained}
\end{align}
i.e.~maximizing \emph{only} over states $\rchi$ that produce distributions $\sfP$ in the accepted set $S_\Acc$. Basically, notice that in~\eqref{eq:vibe_penalty} we do not directly impose ``hard constraints'' on the supremum over $\rchi$, but we do impose that $\sfQ$ lies within $S_\Acc$. Therefore, the $\frac{1}{\alpha-1} D\Dargs{\sfQ}{{\sfP}}$ term instead acts as a ``soft penalty'' on any states $\rchi$ that don't satisfy $\sfP \in S_\Acc$, ``nudging'' the optimal $\rchi$ to stay near that set. 

For sufficiently ``well-behaved'' {\Renyi} divergences $\renD$ and a nonlocal game with single-round rigidity properties, the hard-constrained optimization~\eqref{eq:vibe_constrained} can ``obviously'' be forced to be $O(n)$ with an arbitrarily small constant, simply by choosing $S_\Acc$ to be sufficiently small (but still independent of $n$). Hence the only remaining task is to ensure that the actual soft-penalty version in~\eqref{eq:vibe_penalty} can be made sufficiently close to~\eqref{eq:vibe_constrained}. Intuitively, this should be the case whenever the penalty term is ``large'' relative to the $\sfQ(\gen) \renD(A)_{|\gen}$ term --- we can indeed formalize this intuition by splitting the optimization domain $\rchi\in\outputset$ into those states where $\sfP$ is ``near'' $S_\Acc$ (so $\sfQ(\gen) \renD(A)_{|\gen}$ is already small) and those where $\sfP$ is ``far from'' $S_\Acc$ (so the penalty term is sufficiently large).

In this work, we use two approaches to enforce that the penalty is sufficiently large relative to $\sfQ(\gen) \renD(A)_{|\gen}$. First, for the generic sequential analysis that chooses $\renD$ to be the sharp-divergence $\Dsharp_\alpha$, we simply note that since $\sfQ$ is constrained to lie in $S_\Acc$, we can have $S_\Acc$ impose that the fraction of generation rounds lies below some sufficiently small (but constant) value, which already suffices to force $\sfQ(\gen) \renD(A)_{|\gen}$ to be small relative to the penalty term. 
The main limitation of this approach is that we cannot freely choose the testing probability --- in fact, heuristically we expect that in most cases it will have to be set close to $1$, yielding a poor albeit constant rate.

For the no-memory case that chooses $\renD$ to be the sandwiched divergence $D_\alpha$, we take a different approach that avoids this issue. Specifically, we follow a more ``standard'' entropy-accumulation tactic and take $\alpha$ sufficiently close to $1$ to make the penalty term $\frac{1}{\alpha-1} D\Dargs{\sfQ}{{\sfP}}$ large. We could not do this for the sharp-divergence $\Dsharp_\alpha$ because that is \emph{not} monotone in $\alpha$, and the $\alpha\to 1$ limit is the BS divergence, which exhibits an unpleasant ``discontinuity'' behaviour discussed in~\cite[Appendix~B]{arx_MD23}. 

{Thus both approaches can bound the comparison exponent by $km$ for any desired $k>0$, but only the no-memory approach, with the sandwiched-divergence analysis, lets us choose the generation fraction freely.}\footnote{The ``price'' of taking $\alpha$ near $1$ instead shows up elsewhere, in terms of making the final bound on bad-event probabilities worse. But more broadly, we have a simple general fact that our bound for the no-memory case is always tighter, just because $D_\alpha \leq \Dsharp_\alpha$. Hence any scaling behaviour achievable with the former is achievable with the latter; here we are simply arguing that the latter has ``more available options''.}

\subsection{Future work}
\label{subsec:futurework}

\subsubsection{Iterated-mean divergences}

Ideally, we would like to have a result that can combine the advantages of both approaches we discussed above; i.e., it handles the sequential scenario with memory, while still allowing us to freely choose the testing fraction. 
To the best of our knowledge, currently there is no quantum {\Renyi} divergence that is explicitly known to have all the required properties (putting aside the prospect of finding a dramatically different proof approach from those we have used here). 
However, as an interesting potential candidate, we study the iterated-mean (IM) divergence introduced in~\cite{BFF21}. 

The IM divergence is rather similar in spirit to the sharp-divergence, but is defined by ``iterating'' the usual geometric mean $\mathbin{\#}$ (see Definition~\ref{def:IM}), rather than considering the generalized geometric mean $\mathbin{\#_t}$ used to define the sharp-divergence (see Definition~\ref{def:divergences}). As a consequence, it easily satisfies the property of being monotone in its {\Renyi} parameter $\alpha$ (albeit one that can only take on a discrete set of values, corresponding to the number of ``iterations''), and moreover does not face the issue of converging to the BS-divergence as $\alpha \to 1$ (see~\cite[Sec.~5]{BFF21}, or~\eqref{eq:IM_vs_sharp2} in the appendix). These properties were useful for its original purpose of evaluating certain optimization problems that arise in DI protocols. 
However, it is not known whether it satisfies an analogous chain rule, which would be useful for our goals in this work.

While we have not managed to fully resolve the question of whether it suffices for our intended applications, we do show that it satisfies a form of chain rule similar but not identical to that for the sharp-divergence. 
We leave for future work a more in-depth analysis of whether this version suffices for cryptographic applications, or whether one can improve the steps in the analysis to obtain a tighter chain rule.

We give the proof of this chain rule in Appendix~\ref{app:IM-divergence}. It follows a similar argument used for the sharp-divergence in~\cite{FF21sharp}, after first rewriting the IM divergence in a form manifestly similar to the sharp-divergence. With this form, we are able to reduce the computation of the (stabilized) channel IM divergence from two nested optimizations to a single optimization of an operator norm, closely analogous to a result for the sharp-divergence. This is almost sufficient to obtain an analogous chain rule; however, we unfortunately need to introduce a relaxation at a critical step, and thus the final chain rule is only in terms of a ``relaxed channel IM divergence'' that we were unable to completely characterize. Still, we show that at least for measure-and-prepare channels using a common rank-1 projective measurement, this relaxed version coincides with the usual channel IM divergence.

\subsubsection{Other potential applications}

In the above presentation of our results, we have focused on acceptance sets satisfying some mild niceness conditions, allowing us to effectively ``optimize out'' the choices of the estimators $f_i$. However, as discussed in~\cite{ZFK20,arx_AHT24}, this is not strictly necessary --- for any general acceptance condition, one can obtain bounds on $\renD(\bA)_{|\Acc}$ (in the sense of~\eqref{eq:vibe_truncated}) in terms of explicit choices of these estimators. Furthermore, as presented in those works, this more elaborate perspective potentially allows for ``adaptive'' choices of these estimators as the protocol progresses, for instance to handle scenarios where the device performance drifts over time. In such an approach, the duality analysis for optimizing $f_i$ also implicitly yields a numerical approach for finding the optimal choices, even if the final result is less simple.
These ``adaptive'' options do seem to require some properties from the underlying chain rule, such as being able to handle sequential memory as with the sharp-divergences, but it seems promising to investigate in future work.

Moreover, in QKD applications, it has been found that it is unnecessary to restrict ourselves to  protocols that make a binary accept/abort decision: one can simply use the estimator perspective on $\sum_i f_i(c_i)$ to obtain protocols that return ``variable'' bounds on the entropy or divergence quantities, which suffice for QKD applications at least.\footnote{This method was initially described in presentations~\cite{inprep_weightentropy} without a public manuscript; see~\cite{arx_FKR+25,arx_KBT25} for expositions of the proofs.} 
It seems viable to consider whether analogous applications could be found in our setting. However, we do not attempt to address this here, as the QKD application here does seem to need some properties of sandwiched conditional entropy specifically (e.g.~the precise form of the leftover hashing lemma) --- we leave it as a prospect for future work.

Apart from changing up the estimator perspectives, we can also consider other cryptographic protocols --- the examples of applications we present here should not be considered as exhaustive. As something with a different flavour from these examples, we observe that~\cite{FM25,LFM26} made the critical observation that in spirit, an $n$-copy self-testing property ``should'' also imply some degree of independence between the classical outputs in different rounds. They drew on this insight to develop DI seedless extractors, based on the notion that the outputs across the rounds should be sufficiently independent, so to speak. It seems worth studying whether our approach could yield a similar result, since our ideal state also has independent outputs across (non-tested) rounds. However, one of the important gaps to close here is that seedless extraction does not rely solely on exponentially-small bounds on some ``bad'' event probability, so some adjustment to our approach may be needed. 

\subsubsection{Depolarizing noise versus infrequent-badness}

In this work, to impose the desired support-containment properties, we have chosen to define the ideal channels' behaviour in generation rounds as basically outputting slightly depolarized versions of the ideal states for the single-copy game, roughly similar to the results in the work~\cite{arx_MD24} (which we could not directly apply here because it introduced dimension dependencies). However, the work~\cite{arx_MD23} used a different approach, where for each round the ideal channel is instead defined by ``mixing'' with a small component of the real channel, e.g.~with probability $\pert$ it behaves like the real channel and with probability $1-\pert$ it outputs the ``perfect'' state \emph{without} depolarization. Our divergence bounds all straightforwardly translate to this version as well, though possibly with different scaling with respect to the single-round rigidity properties. 

However, we chose not to present that approach in that work because the resulting ideal state seems slightly more delicate to analyze --- rather than simply having depolarizing noise on every round, it is ``infrequently bad'' in each round, in the sense of behaving like the real channel. Still, since this infrequent-badness occurs in an IID fashion across the rounds, there should be no fundamental barriers to using a typical-sequence argument to say that with high probability a sufficiently large number of the rounds are ``good'', and then use them for DD protocols. (The main potential concern in the application is that the ``bad'' rounds are hidden, in that we do not know their locations.) It does have the slight advantage that in the ``good'' rounds, one has the perfectly noiseless ideal state, in case the DD protocol cannot in fact tolerate any depolarizing noise. Still, we note that depolarizing noise \emph{can} also in fact be viewed as a form of infrequent-badness behaviour, so any analysis that applies to that scenario should also be translatable to our depolarizing-noise model as well, albeit perhaps with different rates or performance.

\subsection{Organization of the paper}
In Section~\ref{sec:notation} we introduce quantum information concepts and state some known results we will be using throughout the work. In Section~\ref{sec:comparison}, we prove our sequential comparison theorem (Theorem~\ref{th:main-informal}). In Section~\ref{sec:DI} we give our generic DD-to-DI compiler for the general sequential case (Theorem~\ref{th:DI-informal}). In Section~\ref{sec:memoryless}, we prove our tensor-product comparison theorem (Theorem~\ref{th:tensor-informal}) and use it to derive the DD-to-DI compiler in the memoryless case (Theorem~\ref{th:memory-informal}). In Section~\ref{sec:applications}, we give our three applications for Theorem~\ref{th:DI-informal}, constructing DI uncloneable encryption, copy protection and blind randomness expansion. 

\section{Preliminaries}\label{sec:notation}

This section gathers the notation and analytic tools used throughout the manuscript. All Hilbert spaces are finite-dimensional, all logarithms and entropies are base $2$, while $\ln$ denotes the natural logarithm. $\mathbb{N}$ denotes natural numbers excluding zero. 

\subsection{Classical probability distributions}

Throughout the paper, registers and classical random variables are denoted by capital letters and their values by the corresponding lower-case letters; for a classical register of a quantum state, we use the same letter for its corresponding random variable. Bold symbols are reserved for strings or nontrivial subcollections. Thus
\[
\bCS=C_1\cdots C_n,
\qquad
\bcS=(c_1,\ldots,c_n),
\qquad
\bK_E=(K_i)_{i\in E}.
\]
The notation $\bCS_{<i}$ means $C_1\cdots C_{i-1}$.

We use $\sfP$ and $\sfQ$ for probability distributions.
The Kullback-Leibler (KL) divergence between these two distributions is defined as
\[
D(\sfQ\Vert\sfP)=\sum_c\sfQ(c)\log\frac{\sfQ(c)}{\sfP(c)},
\]
with value $+\infty$ if $\sfQ(c)>0=\sfP(c)$ for some $c$.

\begin{definition}[Frequency distribution]\label{def:freq}
For $\bz=(z_1,\ldots,z_n)\in\mathcal Z^n$, its frequency distribution is
\[
\freq_{\bz}(z)=\frac1n\bigl|\{i:z_i=z\}\bigr|.
\]
\end{definition}

\begin{fact}[Hoeffding's inequality, {\cite{seqHoeffding63}}]\label{th:Hoeffding}
For independent random variables $Y_i\in[0,1]$ and $u>0$,
\[
\Pr\left[\sum_{i=1}^n(Y_i-\mathbb E Y_i)\geq nu\right]\leq e^{-2nu^2},
\qquad
\Pr\left[\sum_{i=1}^n(Y_i-\mathbb E Y_i)\leq-nu\right]\leq e^{-2nu^2}.
\]
\end{fact}

\begin{definition}[Vanishing function]
A nonnegative, nondecreasing function $v:[0,a]\to\mathbb R$ is vanishing if $v(0)=0$ and $\lim_{\delta\to0^+}v(\delta)=0$. The notation $\vanf(\delta)$ denotes an unspecified vanishing function; subscripts indicate fixed auxiliary parameters.
\end{definition}

\begin{fact}[Minimax theorem, {\cite{seqSion58}}]\label{th:minimax}
Let $\mathcal X$ be compact and convex and $\mathcal Y$ convex. Suppose $L(x,y)$ is upper semicontinuous and concave in $x$ for each $y$, and lower semicontinuous and convex in $y$ for each $x$. Then
\[
\inf_{y\in\mathcal Y}\sup_{x\in\mathcal X}L(x,y)
=\sup_{x\in\mathcal X}\inf_{y\in\mathcal Y}L(x,y).
\]
\end{fact}

\subsection{Quantum states and measurements}
We use $\dop{=}(A)$ and $\dop{\leq}(A)$ for normalized and subnormalized states on the register $A$, and $\Pos(A)$ for positive semidefinite operators. $\id_A$ is the identity operator on $A$ and $\idmap_A$ for the identity channel on $A$. The dimension of a register $A$ is denoted by $|A|$, so its maximally mixed state is $\id_A/|A|$.

We will use $\CPTP{Q}{Q'}$ to denote the set of CPTP maps from the quantum register $Q$ to the register $Q'$. Throughout this work we will often leave tensor products with identity channels implicit; e.g.~given a channel $\mathcal{E}\in\CPTP{Q}{Q'}$, we often use the compact notation
\begin{align}\label{eq:impliedidmap}
\mathcal{E}[\rho_{QR}] = (\mathcal{E} \otimes \idmap_R)[\rho_{QR}].
\end{align}

For a state $\rho_A$, the notation $\ket\rho_{AR}$ always means a normalized purification of $\rho_A$. We will sometimes talk about extensions of a state that are not pure, i.e., a mixed state $\rchi_{AW}$ such that $\Tr_{W}[\rchi_{AW}] = \rho_A$. We will generally use different letters for such extensions.
 
\begin{definition}[CQ states and conditioning]\label{def:cond}
A state on registers $CQ$ is classical on $C$ when it has the form
\[
\rho_{CQ}=\sum_c\pure c_C\otimes\rho^c_Q,
\]
where the blocks $\rho^c_Q$ are positive and subnormalized. If $\rho$ is normalized, write $\sfP(c)=\tr{\rho^c_Q}$ for its classical distribution and, when $\sfP(c)>0$, write $\rho_{Q|c}=\rho^c_Q/\sfP(c)$. Equivalently,
\[
\rho_{CQ}=\sum_c\sfP(c)\pure c_C\otimes\rho_{Q|c}.
\]
We refer to the vector $(\sfP(c))_c$ as the distribution induced by $\rho$ on $C$.

For an event $\Omega$ on $C$, we define its partial state and probability by
\[
\rho_{\land\Omega}=\sum_{c\in\Omega}\pure c_C\otimes\rho^c_Q,
\qquad
p_\Omega=\tr{\rho_{\land\Omega}}.
\]
When $p_\Omega>0$, the normalized conditional state is $\rho_{|\Omega}=\rho_{\land\Omega}/p_\Omega$. Intersections of classical events are read associatively, for example $\rho_{\land\Omega\land\Omega'}=\rho_{\land(\Omega\cap\Omega')}$. We never assign an artificial conditional state to a zero-probability event.
\end{definition}

\begin{definition}[Types of channels]
A measurement channel $\mapM\in\CPTP{Q}{C}$ has a classical output and can be written as
\[
\mapM[\rho_Q]=\sum_j\tr{M_j\rho_Q}\pure j_C
\]
for a POVM $\{M_j\}_j$.

A pinching channel is a special case of a measurement channel where $\mapM[\rho]=\sum_j\pure j\rho\pure j$, i.e., the $M_j$-s are measurements in the same computational basis the outcome is recorded in.

A projective measure-and-prepare channel is a channel of the form
\[
\mapE[\rho_Q]=\sum_j(P_j\rho_QP_j)\otimes\sigma_{O|j},
\]
where $(P_j)_j$ is a projective measurement. If $Q$ is classical and the projectors read its classical basis, we call this a read-and-prepare channel. Tracing out the prepared register then returns the original state on $Q$.
\end{definition}

\begin{definition}[Generalized trace distance]\label{def:generalized-trace-distance}
For subnormalized states $\rho$ and $\sigma$, their generalized trace distance is
\[
\trnorm{\rho-\sigma}
=\frac12\norm{\rho-\sigma}_1
+\frac12\left|\tr{\rho}-\tr{\sigma}\right|.
\]
For normalized states this reduces to the usual trace distance
$\trnorm{\rho-\sigma}=\frac12\norm{\rho-\sigma}_1$. When $\rho$ and $\sigma$ are classical, i.e., two probability distributions $\sfP$ and $\sfQ$, the trace distance has the simple form
\[ \trnorm{\sfP-\sfQ} = \frac12\sum_c|\sfP(c)-\sfQ(c)|.\]
\end{definition}

\begin{definition}[Fidelity and purified distance]
For subnormalized states, their fidelity and purified distance are distance measures defined as
\[
\overline F(\rho,\sigma)=\norm{\sqrt\rho\sqrt\sigma}_1+
\sqrt{(1-\tr{\rho})(1-\tr{\sigma})},
\qquad
\pd(\rho,\sigma)=\sqrt{1-\overline F(\rho,\sigma)^2}.
\]
\end{definition}

\begin{fact}[Fuchs--van de Graaf inequalities, {\cite{Tom16}}]\label{th:FvdG}
For subnormalized states $\rho$ and $\sigma$,
\[
1-\overline F(\rho,\sigma)
\leq \trnorm{\rho-\sigma}
\leq \pd(\rho,\sigma).
\]
\end{fact}

\begin{fact}[Transfer of event probabilities]\label{lem:event-transfer}
Let $\rho$ and $\sigma$ be subnormalized states and let $0\leq H\leq\id$ be an effect. Then
\[
\left|\tr{H\rho}-\tr{H\sigma}\right|
\leq\trnorm{\rho-\sigma}.
\]
\end{fact}

\subsection{Sandwiched divergences and entropies}

\begin{definition}[Sandwiched \Renyi{} divergence]\label{def:sandwiched divergence}
For $\rho, \sigma\in\Pos(A)$, and $\alpha\in(0,1)\cup(1,\infty)$, define (for $\rho \neq 0$)
\[
Q_\alpha(\rho\Vert\sigma)=\tr{(\sigma^{(1-\alpha)/(2\alpha)}\rho\sigma^{(1-\alpha)/(2\alpha)})^\alpha},
\qquad
D_\alpha(\rho\Vert\sigma)=\frac1{\alpha-1}\log\frac{Q_\alpha(\rho\Vert\sigma)}{\tr{\rho}},
\]
and define $Q_\alpha(0\Vert\sigma)=0$.
For $\alpha>1$, both quantities are set to $+\infty$ unless $\supp(\rho)\subseteq\supp(\sigma)$; negative powers use the Moore--Penrose pseudoinverse. For $0<\alpha<1$, the formula for $Q_\alpha$ applies without a support restriction, and $D_\alpha$ is finite exactly when $\rho\not\perp\sigma$; otherwise $Q_\alpha=0$ and $D_\alpha=+\infty$. The divergence is extended to $\alpha\in\{0,1,\infty\}$ by the respective limits. At $\alpha=1$, it becomes the Umegaki divergence
\[
D(\rho\Vert\sigma)=
\frac{\tr{\rho\log\rho-\rho\log\sigma}}{\tr{\rho}}
\]
when $\supp(\rho)\subseteq\supp(\sigma)$, and $+\infty$ otherwise.
For $\alpha=\infty$, we have
\[
\Dmax(\rho\Vert\sigma) = \inf\{\lambda: \rho \leq 2^\lambda \sigma\}.
\]
The same notation is used for classical distributions by identifying them with diagonal states; for classical distributions, the Umegaki divergence is the same as the KL divergence.
\end{definition}

\begin{definition}[Smooth divergences and entropies]
Throughout this work, we will consider smoothing in purified distance. For $\alpha\in[0,\infty]$, nonzero $\rho\in\dop{\leq}(A)$, $\sigma \in \Pos(A)$, and $\esmooth \in [0, \sqrt{\Tr[\rho]})$, the $\esmooth$-smoothed $\alpha$-divergence between $\rho$ and $\sigma$ is given by
\[ D_\alpha^{\esmooth}(\rho\Vert\sigma) = \inf_{\substack{\rho'\in \dop{\leq}(A) \\ \pd(\rho',\rho)\leq {\esmooth}}}D_\alpha(\rho'\Vert\sigma).\]
\end{definition}

We now list some well-known properties of the sandwiched divergences that we will be using.  We will state some of the properties with the intermediate quantity $Q_\alpha$ since this will be useful for subnormalized states.
\begin{fact}[Sandwiched \Renyi{} divergence properties, {\cite{Tom16}}]\label{th:renyi-properties}
The sandwiched Renyi divergences satisfy the following properties, assuming $\alpha > 1$ unless otherwise specified.
\begin{enumerate}[label=\textup{(\roman*)},ref=\textup{(\roman*)}]
\item\label{prop:renyi-scaling} \emph{Scaling}: For $a,b>0$, the scaling identities are
\[
Q_\alpha(a\rho\Vert b\sigma)=a^\alpha b^{1-\alpha}Q_\alpha(\rho\Vert\sigma),
\qquad
D_\alpha(a\rho\Vert b\sigma)=D_\alpha(\rho\Vert\sigma)+\log(a/b).
\]

\item\label{fact:DPI} \emph{Data processing}: For positive $\rho,\sigma$ with $\rho$ nonzero, $\alpha\in[1/2,\infty]$, and a channel $\mapE$,
\[
D_\alpha(\mapE[\rho]\Vert\mapE[\sigma]) \leq D_\alpha(\rho\Vert\sigma),
\]
with equality holding if $\mapE$ is an isometry.

\item\label{fact:classmix} \emph{Classical direct sum}: For states that are direct sums on orthogonal blocks,
\[
Q_\alpha\left(\bigoplus_j\rho^j\middle\Vert\bigoplus_j\sigma^j\right)
=\sum_j Q_\alpha(\rho^j\Vert\sigma^j).\]
Consequently, if $\rho_{CQ}$ and $\sigma_{CQ}$ are normalized and classical on $C$, and if $\sfP$ and $\sfQ$ are their respective distributions on $C$, then
\[
D_\alpha(\rho\Vert\sigma)=\frac1{\alpha-1}\log
\sum_c\sfP(c)^\alpha\sfQ(c)^{1-\alpha}
2^{(\alpha-1)D_\alpha(\rho_{Q|c}\Vert\sigma_{Q|c})}.
\]
Here any term with $\sfP(c)=0$ is defined to be zero, and any term with $\sfQ(c)=0$ but $\sfP(c)>0$ is defined to be $+\infty$.

\item\label{fact:eventconditioning} \emph{Monotonicity}: $Q_\alpha$ is monotone increasing in the first argument and monotone decreasing in the second argument. That is,
\[ \rho'\geq\rho \; \Longrightarrow \; Q_\alpha(\rho'\Vert\sigma) \geq Q_\alpha(\rho\Vert\sigma), \qquad \sigma'\geq\sigma\; \Longrightarrow Q_\alpha(\rho\Vert\sigma') \leq Q_\alpha(\rho\Vert\sigma). \]
Consequently, $D_\alpha$ is also monotone decreasing in the second argument, and if normalized states satisfy $\rho\geq\lambda\rho'$ with $\lambda>0$, then
\[
D_\alpha(\rho'\Vert\sigma)
\leq D_\alpha(\rho\Vert\sigma)+\frac{\alpha}{\alpha-1}\log\frac1\lambda.
\]
\item\label{prop:renyi-purified-concavity} \emph{Purified-input concavity}: Let $\mapE$ and $\mapF$ be completely positive maps with the same input and output spaces. Suppose that $\rho_M=\sum_jt_j\rho^j_M$, where the $\rho^j_M$ are normalized states and the $t_j$ form a probability distribution. Choose an arbitrary purification $\ket\rho_{MR}$ of $\rho_M$, and independently choose arbitrary purifications $\ket{\rho^j}_{MR'}$ of the component states on a common sufficiently large register $R'$. Then
\begin{align*}
& Q_\alpha\!\left(
(\mapE\otimes\idmap_R)\!\left[\pure\rho\right]
\middle\Vert
(\mapF\otimes\idmap_R)\!\left[\pure\rho\right]
\right)\\
& \geq\sum_jt_j
Q_\alpha\!\left(
(\mapE\otimes\idmap_{R'})\!\left[\pure{\rho^j}\right]
\middle\Vert
(\mapF\otimes\idmap_{R'})\!\left[\pure{\rho^j}\right]
\right).
\end{align*}
\item\label{th:smoothing} \emph{\Renyi-to-max conversion}~\cite[Corollary~14, Eq.~(110)]{RLD26}:\footnote{Here we state the result from~\cite{RLD26}, which has $\smf{\esmooth}$ in place of the slightly worse value ${\log(2/\esmooth^2)}$ from earlier works. The practical improvement gained by this should be very small, but it allows a cleaner statement in our eventual Theorem~\ref{th:scalingdemo}.} For normalized $\rho$, positive $\sigma$, $\alpha>1$, and $0<\esmooth<1$,
\[
\Dmax^{\esmooth}(\rho\Vert\sigma)
\leq D_\alpha(\rho\Vert\sigma)+\frac{\smf{\esmooth}}{\alpha-1}.
\]
\end{enumerate}
\end{fact}


\begin{lemma}[Probability transfer]\label{lem:probability}
Let $\rho$ be normalized, $\sigma$ positive, and $0\leq H\leq\id$.
\begin{enumerate}[label=\textup{(\roman*)},ref=\textup{(\roman*)}]
\item If $D_2(\rho\Vert\sigma)\leq d$, then
\[
\tr{H\rho}\leq 2^{d/2}\sqrt{\tr{H\sigma}}.
\]
\item If $\Dmax^{\esmooth}(\rho\Vert\sigma)\leq d$, then
\[
\tr{H\rho}\leq2^d\tr{H\sigma}+\esmooth.
\]
\end{enumerate}
\end{lemma}
\begin{proof}
For (i), write $V=\sigma^{1/4}H\sigma^{1/4}$ and $W=\sigma^{-1/4}\rho\sigma^{-1/4}$ on the support of $\sigma$. The Cauchy--Schwarz inequality gives $(\tr{H\rho})^2\leq\tr{V^2}\tr{W^2}$. Since $0\leq V\leq\sigma^{1/2}$, we have by the cyclicity of trace, $\tr{V^2}\leq\tr{V\sigma^{1/2}}=\tr{H\sigma}$; the other factor is $2^{D_2(\rho\Vert\sigma)}$. This proves (i).

For (ii), choose a subnormalized $\tau$ in the smoothing ball with $\tau\leq2^{d+\zeta}\sigma$, allowing an arbitrary $\zeta>0$ if the infimum is not attained. Fact~\ref{th:FvdG} and Fact~\ref{lem:event-transfer} give
\[
|\tr{H(\rho-\tau)}|
\leq\trnorm{\rho-\tau}
\leq\pd(\rho,\tau)
\leq\esmooth.
\]
Therefore $\tr{H\rho}\leq2^{d+\zeta}\tr{H\sigma}+\esmooth$. Let $\zeta$ decrease to zero.
\end{proof}

\subsection{Sharp-divergences}
\begin{definition}[Sharp-divergences]\label{def:divergences}
For positive operators $\rho$ and $\sigma$ with $\rho$ nonzero, and $\alpha>1$, the sharp-divergence of order $\alpha$ between them is defined as follows:
\[
\Qsharp_\alpha(\rho\Vert\sigma)
=\inf\{\tr{H}:H\geq0,\ \rho\leq \sigma\mathbin{\#_{1/\alpha}}H\},
\qquad
\Dsharp_\alpha(\rho\Vert\sigma)
=\frac{1}{\alpha-1}\log\frac{\Qsharp_\alpha(\rho\Vert\sigma)}{\tr{\rho}}.
\]
Here $\#_t$ denotes the generalized geometric mean, defined as $\sigma\mathbin{\#_t}H=\sigma^{1/2}(\sigma^{-1/2}H\sigma^{-1/2})^t\sigma^{1/2}$ for positive-definite operators, and extended to positive operators via a decreasing sequence of positive-definite approximations. Set $\Qsharp_\alpha(0\Vert\sigma)=0$.\footnote{Fawzi and Fawzi~\cite{FF21sharp}, who defined the sharp-divergences, define $\Qsharp_\alpha$ the same way as us, but omit the factor of $\tr{\rho}$ in $\Dsharp_\alpha$.}

We will be particularly interested in $\alpha=2$, for which the optimization in $\Qsharp_\alpha$ takes the particular form
\begin{equation}\label{eq:order2}
\Qsharp_2(\rho\Vert\sigma)=\inf_{H\geq\rho}\tr{H\sigma^{-1}H},
\end{equation}
with the optimization restricted to the support of $\sigma$.
\end{definition}

We list some properties of the sharp-divergences that are similar to sandwiched divergences, as well as their comparison with the sandwiched divergences.

\begin{fact}[Sharp-divergence properties, {\cite[Section~3]{FF21sharp}}]\label{th:sharp-facts}
For $\alpha>1$, the following properties hold.
\begin{enumerate}[label=\textup{(\roman*)},ref=\textup{(\roman*)}]
\item\label{prop:sharp-scaling} \emph{Scaling}: For $a,b>0$, the scaling identities are
\[
\Qsharp_\alpha(a\rho\Vert b\sigma)=a^\alpha b^{1-\alpha}\Qsharp_\alpha(\rho\Vert\sigma),
\qquad
\Dsharp_\alpha(a\rho\Vert b\sigma)=\Dsharp_\alpha(\rho\Vert\sigma)+\log(a/b).
\]
\item\label{prop:sharp-data-processing} \emph{Data processing}: For a channel $\mapE$,
\[
\Qsharp_\alpha(\mapE[\rho]\Vert\mapE[\sigma])
\leq\Qsharp_\alpha(\rho\Vert\sigma).
\]
\item\label{prop:sharp-direct-sum} \emph{Classical direct sum:} For states that are direct sums on orthogonal blocks,
\[
\Qsharp_\alpha\left(\bigoplus_j\rho^j\middle\Vert\bigoplus_j\sigma^j\right)
=\sum_j\Qsharp_\alpha(\rho^j\Vert\sigma^j).
\]
Consequently, if $\rho_{CQ}$ and $\sigma_{CQ}$ are normalized and classical on $C$, with $\sfP$ and $\sfQ$ being their respective distributions on $C$, then
\[
\Dsharp_\alpha(\rho\Vert\sigma)=\frac1{\alpha-1}\log
\sum_c\sfP(c)^\alpha\sfQ(c)^{1-\alpha}
2^{(\alpha-1)\Dsharp_\alpha(\rho_{Q|c}\Vert\sigma_{Q|c})}.
\]
Here any term with $\sfP(c)=0$ is defined to be zero, and any term with $\sfQ(c)=0$ but $\sfP(c)>0$ is defined to be $+\infty$.
\item\label{prop:sharp-monotonicity} \emph{Monotonicity}: 
For $\alpha > 1$, $\Qsharp_\alpha$ is monotone increasing in the first argument and monotone decreasing in the second argument. That is,
\[ \rho'\geq\rho \; \Longrightarrow \; \Qsharp_\alpha(\rho'\Vert\sigma) \geq \Qsharp_\alpha(\rho\Vert\sigma), \qquad \sigma'\geq\sigma\; \Longrightarrow \Qsharp_\alpha(\rho\Vert\sigma') \leq \Qsharp_\alpha(\rho\Vert\sigma). \]
Consequently, $\Dsharp_\alpha$ is also monotone decreasing in the second argument, and if normalized states satisfy $\rho\geq\lambda\rho'$ with $\lambda>0$, then for $\alpha>1$,
\[
\Dsharp_\alpha(\rho'\Vert\sigma)
\leq \Dsharp_\alpha(\rho\Vert\sigma)+\frac{\alpha}{\alpha-1}\log\frac1\lambda.
\]
\item\label{prop:sharp-purified-concavity} \emph{Purified-input concavity}: Let $\mapE$ and $\mapF$ be completely positive maps with the same input and output spaces. Suppose that $\rho_M=\sum_jt_j\rho^j_M$, where the $\rho^j_M$ are normalized states and the $t_j$ form a probability distribution. Choose an arbitrary purification $\ket\rho_{MR}$ of $\rho_M$, and independently choose arbitrary purifications $\ket{\rho^j}_{MR'}$ of the component states on a common sufficiently large register $R'$. Then
\begin{align*}
& \Qsharp_\alpha\!\left(
(\mapE\otimes\idmap_R)\!\left[\pure\rho\right]
\middle\Vert
(\mapF\otimes\idmap_R)\!\left[\pure\rho\right]
\right)\\
& \geq\sum_jt_j
\Qsharp_\alpha\!\left(
(\mapE\otimes\idmap_{R'})\!\left[\pure{\rho^j}\right]
\middle\Vert
(\mapF\otimes\idmap_{R'})\!\left[\pure{\rho^j}\right]
\right).
\end{align*}

\item\label{prop:sharp-renyi-comparison} \emph{Relationship between {\Renyi} divergences}: For nonzero $\rho$, and for normalized $\rho$ in the second inequality,
\[
D_\alpha(\rho\Vert\sigma)\leq\Dsharp_\alpha(\rho\Vert\sigma),
\qquad
\Dsharp_\alpha(\rho\Vert\sigma)\leq\Dmax(\rho\Vert\sigma).
\]
\end{enumerate}
\end{fact}

For item~\ref{prop:sharp-purified-concavity}, a channel from $R$ to $R'J$ maps $\pure\rho_{MR}$ to the extension $\sum_jt_j\pure{\rho^j}_{MR'}\otimes\pure j_J$, since both states have marginal $\rho_M$. Apply sharp-divergence data-processing to this channel, followed by scaling and direct-sum additivity (items~\ref{prop:sharp-data-processing}, \ref{prop:sharp-scaling}, and~\ref{prop:sharp-direct-sum}, proved in~\cite[Propositions~3.2 and~3.5]{FF21sharp}), to obtain the stated weighted sum.


Finally, we define the sharp-divergence of channels, and state a chain rule for it.
\begin{definition}[Channel sharp-divergence]\label{def:channel-div}
For $\alpha>1$, a channel $\mapE\in\CPTP{\inQ}{\outQ}$, and a completely positive map $\mapF$ with the same input and output,
\[
\Dsharp_\alpha(\mapE\Vert\mapF)
=\sup_{\rho\in\dop{=}(\inQ)}
\Dsharp_\alpha\Dargs{\mapE[\rho]}{\mapF[\rho]}.
\]
We highlight that in our definition we do \emph{not} insert a purification of the input state (sometimes called a ``stabilizer'' register); see Remark~\ref{rem:stab}.
\end{definition}

\begin{remark}
\label{rem:stab}
In information theory, it is common to instead define channel divergences with an additional purifying register (a.k.a.~``stabilizer'') exactly analogous to the definition of diamond norm. That is, given a divergence $\mathbb{D}$ for states, the channel divergence is instead defined as $\sup_{\rho\in\dop{=}(\inQ R)}
\mathbb{D}\Dargs{\mapE[\rho]}{\mapF[\rho]}$ where $R$ is a register large enough to purify $\inQ$ (recall from our convention~\eqref{eq:impliedidmap} that the channels on the right-hand-side here have implied identity maps on $R$). However, throughout this work we do not inherently include a stabilizer register in our definitions of channel divergences; when we require such a register, we shall specify it by explicitly writing an identity channel on a suitable register.
\end{remark}

\begin{fact}[Sharp-divergence channel chain rule, {\cite[Proposition~4.5]{FF21sharp}}]\label{th:FF}
Let $\mapE$ be CPTP and $\mapF$ CP, with the same input and output domains. For positive $\rho_{\stabE\inQ},\sigma_{\stabE\inQ}$, with $|\stabE|\geq|\inQ|$, and $\alpha>1$,
\[
\Qsharp_\alpha\Dargs{\mapE[\rho]}{\mapF[\sigma]}
\leq
2^{(\alpha-1)\Dsharp_\alpha(\mapE\otimes\idmap_{\stabE}\Vert\mapF\otimes\idmap_{\stabE})}
\Qsharp_\alpha(\rho\Vert\sigma).
\]
For nonzero $\rho$, the equivalent trace-normalized form is
\[
\Dsharp_\alpha\Dargs{\mapE[\rho]}{\mapF[\sigma]}
\leq\Dsharp_\alpha(\rho\Vert\sigma)+\Dsharp_\alpha(\mapE\otimes\idmap_{\stabE}\Vert\mapF\otimes\idmap_{\stabE}).
\]
\end{fact}
Note that this chain rule does not hold for many other types of divergences, such as sandwiched divergences. The fact that it does hold for sharp-divergences will turn out to be a very important property for us. 


\section{Sequential comparison theorem}\label{sec:comparison}

In this section, we prove our sequential comparison theorem, which is formally stated in Theorem~\ref{th:main}.
We first formally describe the sequential processes, and the postulate they must satisfy.
\begin{postulates}[Form of sequential channel pairs]\label{posts:seq}
For $i=1,\ldots,n$, let the $i$-th round channels be
\[
\mapE_i,\mapF_i\in\CPTP{\inM_i}{\CS_i\outQ_i\inM_{i+1}}.
\]
We require:
\begin{enumerate}[label=\textup{(\roman*)},ref=\textup{(\roman*)}]
\item\label{post:common-record} The register $\CS_i$ is classical, taking values in $\alphCS$, and has the same marginal under $\mapE_i$ and $\mapF_i$ on every input.
\item\label{post:protected} All reusable information available at the start of round $i$ is in $\inM_i$; $\CS_j\outQ_j$, for $j<i$, are protected registers that are acted on only by identity maps. But $\inM_{i+1}$ may contain copies of the current classical record $\CS_i$ and any classical information stored in $\outQ_i$.
\item\label{post:domination} For some finite constant $\Lambda$, every finite reference $R_i$ and every normalized pure input $\ket{\sigma^i}_{\inM_iR_i}$ satisfy
\[
(\mapE_i\otimes\idmap_{R_i})\!\left[\pure{\sigma^i}\right]
\leq\Lambda\,(\mapF_i\otimes\idmap_{R_i})\!\left[\pure{\sigma^i}\right].\footnote{It is enough to take $R_i\cong\inM_i$.}
\]
\item\label{post:acceptance} Let $S_\Omega$ be a nonempty closed convex subset of the probability simplex on $\alphCS$. The accept condition on the protected classical record $\bCS$ is $\freq_\bcS \in S_\Omega$, and we denote this event by $\Omega$.
\end{enumerate}
\end{postulates}

 We use $\bmapE=\mapE_n\circ\cdots\circ\mapE_1$ and $\bmapF=\mapF_n\circ\cdots\circ\mapF_1$, to refer to the whole sequential process, with identities on earlier outputs understood.


\subsection{The penalized comparison exponent}\label{sec:compexp}

 Our $\mapF_i$ produces a state which is classical on $\CS_i$. Using this, for a real function $f$ on $\alphCS$, we define another completely positive, generally non-trace-preserving map in which the $\cS$-th block of the output of $\mapF_i$ is multiplied by $2^{f(\cS)}$:
\begin{equation}\label{eq:mapF_f}
\mapF_i^f[\sigma]
=\sum_{\cS\in\alphCS}2^{f(\cS)}
\left(\pure{\cS}_{\CS_i}\otimes\id_{\outQ_i\inM_{i+1}}\right)
\mapF_i[\sigma]
\left(\pure{\cS}_{\CS_i}\otimes\id_{\outQ_i\inM_{i+1}}\right).
\end{equation}
For brevity, we use the following notation for the stabilized channel sharp-divergence between $\mapE_i$ and $\mapF_i^f$:
\begin{equation}\label{eq:R_f-1}
R^\#_{\alpha,\mapE_i,\mapF_i}(f)=\Dsharp_\alpha(\mapE_i\otimes\idmap_{R_i}\Vert\mapF_i^f\otimes\idmap_{R_i}),
\end{equation}
where $R_i$ is a reference register with $|R_i|\geq|\inM_i|$.

For a normalized input $\ket{\sigma^i}_{\inM_iR_i}$, write
\[
\sigma^{i,\mapE_i}=(\mapE_i\otimes\idmap_{R_i})\!\left[\pure{\sigma^i}\right],
\qquad
\sigma^{i,\mapF_i}=(\mapF_i\otimes\idmap_{R_i})\!\left[\pure{\sigma^i}\right],
\]
and let $\sfP_{C_i}$ be their common distribution on $\CS_i$, as required by Postulate Set~\ref{posts:seq}\ref{post:common-record}. Write $\sigma^{i,\mapE_i}_{|\cS}$ and $\sigma^{i,\mapF_i}_{|\cS}$ for the corresponding normalized conditional states on $R_i\outQ_i\inM_{i+1}$ whenever $\sfP_{C_i}(\cS)>0$. These are trial one-round states, not marginals of a global state obtained by applying $\bmapE$ and $\bmapF$ to an initial state. Finally, using these states, we define the penalized comparison exponent:
\begin{equation}\label{eq:rate}
\compexp{\alpha}
=\sup_{\substack{\ket{\sigma^1},\ldots,\ket{\sigma^n}\\
\sfQ_{C_1},\ldots,\sfQ_{C_n}:\ \frac1n\sum_i\sfQ_{C_i}\in S_\Omega}}
\sum_{i=1}^n
\left[
\sum_{\cS}\sfQ_{C_i}(\cS)\Dsharp_\alpha(\sigma^{i,\mapE_i}_{|\cS}\Vert\sigma^{i,\mapF_i}_{|\cS})
-\frac{D(\sfQ_{C_i}\Vert\sfP_{C_i})}{\alpha-1}
\right],
\end{equation}
where a summand with $\sfQ_{C_i}(\cS)>0=\sfP_{C_i}(\cS)$ is interpreted to have value $-\infty$; a zero-probability, zero-weight summand is zero.

This quantity appears in the exponent when comparing event probabilities for the real and ideal processes. We call it penalized due to the KL-divergence term, which serves as a ``soft penalty'' as discussed in Sec.~\ref{subsubsec:vibes_opt}.

\subsection{Sequential testing comparison theorem}
\label{subsec:sequential_with_testing}
We will now state and prove our sequential comparison theorem with testing. First, we fix a normalized initial state $\rho^0_{\inM_1R}$, where $R$ is arbitrary external side information and all channels act as the identity on $R$. We set
\[
\realstate=\bmapE[\rho^0],
\qquad
\idealstate=\bmapF[\rho^0].
\]
Both these states are normalized. For the accept event $\Omega$, let $p_\Omega=\Pr_{\realstate}[\Omega]$.

\begin{theorem}[Sequential testing comparison]\label{th:main}
Under Postulate Set~\ref{posts:seq}, let $\alpha>1$ and $p_\Omega>0$. Then
\[
D_\alpha\Dargs{\realstate_{|\Omega}}{\idealstate}
\leq \compexp{\alpha}+\frac{\alpha}{\alpha-1}\log\frac1{p_\Omega}.
\]
Consequently, for every $0<\esmooth<1$,
\[
\Dmax^{\esmooth}\Dargs{\realstate_{|\Omega}}{\idealstate}
\leq \compexp{\alpha}
+\frac{\smf{\esmooth}+\alpha\log(1/p_\Omega)}{\alpha-1}.
\]
\end{theorem}

\begin{proof}
As noted in the introduction, we follow a proof technique basically developed in~\cite{DFR20,ZFK20,arx_AHT24}. First consider an arbitrary real-valued function $f$ on $\alphCS$.
We write $\bmapE_{<i}=\mapE_{i-1}\circ\cdots\circ\mapE_1$ and $\bmapF^f_{<i}=\mapF^f_{i-1}\circ\cdots\circ\mapF^f_1$, with identities on protected registers understood and the empty compositions equal to the identity. We apply Fact~\ref{th:FF} in round $i$, taking $R\bCS_{<i}\boutQ_{<i}$ as the reference and appending a fixed pure ancilla to this reference if needed. This gives
\[
\Qsharp_\alpha\left(
(\idmap\otimes\mapE_i)\!\left[\bmapE_{<i}[\rho^0]\right]
\middle\Vert
(\idmap\otimes\mapF^f_i)\!\left[\bmapF^f_{<i}[\rho^0]\right]
\right)
\leq 2^{(\alpha-1)R^\#_{\alpha,\mapE_i,\mapF_i}(f)}
\Qsharp_\alpha\left(\bmapE_{<i}[\rho^0]\middle\Vert\bmapF^f_{<i}[\rho^0]\right).
\]
Induction, starting from $\Qsharp_\alpha(\rho^0\Vert\rho^0)=1$, gives
\begin{equation}\label{eq:Qsharp-fchain}
\Qsharp_\alpha\Dargs{\realstate}{\bmapF^f[\rho^0]}
\leq2^{(\alpha-1)\sum_iR^\#_{\alpha,\mapE_i,\mapF_i}(f)}.
\end{equation}

By the protection condition in Postulate Set~\ref{posts:seq}\ref{post:protected}, $C_i$ is not touched in later rounds. Therefore,
\[ \left(\bmapF^f[\rho^0]\right)_{\land(\bCS=\bcS)}= 2^{\sum_i f(\cS_i)}\idealstate_{\land(\bCS=\bcS)}.\]
Direct-sum additivity and scaling for $\Qsharp$ (Fact~\ref{th:sharp-facts}\ref{prop:sharp-direct-sum} and \ref{prop:sharp-scaling}), along with its non-negativity, now give
\begin{align}
\Qsharp_\alpha(\realstate\Vert\bmapF^f[\rho^0])
 & = \sum_{\bcS}\Qsharp_\alpha\left(\realstate_{\land(\bCS=\bcS)}\middle\Vert\left(\bmapF^f[\rho^0]\right)_{\land(\bCS=\bcS)}\right) \nonumber \\
 & \geq \sum_{\bcS \in \Omega}\Qsharp_\alpha\left(\realstate_{\land(\bCS=\bcS)}\middle\Vert\left(\bmapF^f[\rho^0]\right)_{\land(\bCS=\bcS)}\right) \nonumber \\
 & = \sum_{\bcS \in \Omega}2^{-(\alpha-1)\sum_{i=1}^nf(c_i)}\Qsharp_\alpha\left(\realstate_{\land(\bCS=\bcS)}\middle\Vert\idealstate_{\land(\bCS=\bcS)}\right). \label{eq:Qsharp-penalty}
 \end{align}
Now for any $\bcS\in\Omega$,
\[ \sum_{i=1}^nf(c_i) = \sum_c\sum_{i:c_i=c}f(c) = n\sum_c\freq_\bcS(c)f(c) \leq n\cdot\sup_{\sfQ\in S_\Omega}\sum_c\sfQ(c)f(c),\]
where the inequality follows because $\freq_\bcS\in S_\Omega$ for any $\bcS\in\Omega$. Since $\alpha>1$, the factor $2^{-(\alpha-1)\sum_i f(c_i)}$ is decreasing in $\sum_i f(c_i)$. Substituting this upper bound into \eqref{eq:Qsharp-penalty} therefore gives
\begin{align*}
\Qsharp_\alpha(\realstate\Vert\bmapF^f[\rho^0])
 & \geq 2^{-(\alpha-1)n\sup_{\sfQ\in S_\Omega}\sum_c\sfQ(c)f(c)}\sum_{\bcS \in \Omega}\Qsharp_\alpha\left(\realstate_{\land(\bCS=\bcS)}\middle\Vert\idealstate_{\land(\bCS=\bcS)}\right) \\
 & = 2^{-(\alpha-1)n\sup_{\sfQ\in S_\Omega}\sum_c\sfQ(c)f(c)}\Qsharp_\alpha(\realstate_{\land\Omega}\Vert\idealstate_{\land\Omega}).
\end{align*}

Finally, from \eqref{eq:Qsharp-fchain}, we get,
\[
\Qsharp_\alpha(\realstate_{\land\Omega}\Vert\idealstate_{\land\Omega}) \leq 2^{(\alpha-1)\left(\sum_{i=1}^nR^\#_{\alpha,\mapE_i,\mapF_i}(f)+n\sup_{\sfQ\in S_\Omega}\sum_c\sfQ(c)f(c)\right)}.
\]
Since this holds for every real function $f$ on $\alphCS$, we can say,
\begin{align}
\Qsharp_\alpha(\realstate_{\land\Omega}\Vert\idealstate_{\land\Omega})
& \leq 2^{(\alpha-1)\inf_f\left(\sum_{i=1}^nR^\#_{\alpha,\mapE_i,\mapF_i}(f)+n\sup_{\sfQ\in S_\Omega}\sum_c\sfQ(c)f(c)\right)} \\
& = 2^{(\alpha-1)\compexp{\alpha}}, \label{eq:subminimax}
\end{align}
where the last line follows from a minimax property of $\compexp{\alpha}$: it is somewhat tedious to establish but is basically a standard result in entropy accumulation proofs, so we defer its statement and proof to Lemma~\ref{lemma:dual_sharp} in the appendices. Now, $\idealstate\geq\idealstate_{\land\Omega}$ because the latter is obtained by discarding the non-accepting classical blocks. Fact~\ref{th:sharp-facts}\ref{prop:sharp-monotonicity} thus gives
\[
\Qsharp_\alpha(\realstate_{\land\Omega}\Vert\idealstate)
\leq \Qsharp_\alpha(\realstate_{\land\Omega}\Vert\idealstate_{\land\Omega}).
\]
Dividing the first argument by $p_\Omega$ gives us the normalized state $\realstate_{|\Omega}$, and the scaling rule for $\Qsharp_\alpha$ (Fact~\ref{th:sharp-facts}\ref{prop:sharp-scaling}) gives
\[
\Qsharp_\alpha(\realstate_{|\Omega}\Vert\idealstate)
\leq p_\Omega^{-\alpha}2^{(\alpha-1)\compexp{\alpha}}.
\]
Since $\realstate_{|\Omega}$ is normalized, taking the logarithm and using the definition of $\Dsharp_\alpha$ gives
\[
\Dsharp_\alpha(\realstate_{|\Omega}\Vert\idealstate)
\leq \compexp{\alpha}+\frac{\alpha}{\alpha-1}\log\frac1{p_\Omega}.
\]

The comparison between the sandwiched and sharp divergences in
Fact~\ref{th:sharp-facts}\ref{prop:sharp-renyi-comparison} gives
\[
D_\alpha\left(\realstate_{|\Omega}\Vert\idealstate\right)
\leq \Dsharp_\alpha\left(\realstate_{|\Omega}\Vert\idealstate\right).
\]
This proves the first inequality in the theorem. Applying Fact~\ref{th:renyi-properties}\ref{th:smoothing} to $\realstate_{|\Omega}$ and $\idealstate$ then proves the smooth max-divergence bound.
\end{proof}

\subsection{Depolarization on a finite ideal system}\label{sec:depol}

Theorem~\ref{th:main} requires Postulate~\ref{posts:seq}\ref{post:domination}, which may not be true for any ideal channel sequence $\bmapF$ that we might want to apply the result to. In particular, the ideal channels we have in mind replace the real state on a finite system within the round with an ideal state, before applying the remaining operations. This system can contribute to both the protected output and the continuing memory. In particular this ideal state may be a pure state, in which case we cannot really ensure that $\mapE_i[\sigma]$ is support-contained in $\mapF_i[\sigma]$, so that the divergence is finite.

However, even if our application ideally requires a pure state, in most applications we are able to tolerate small amounts of noise in these states. In particular, we can usually tolerate depolarizing noise, which mixes the input state with the maximally mixed state. The resulting depolarized state has full support on the relevant register, and so the problem previously described may not arise. In this subsection, we prove some results to show exactly how such depolarization helps, by computing an upper bound on $\Dsharp_2$ between the states. In particular, the bounds we get here only depend on the ideal dimension and the depolarization parameter, but not on the dimension of the continuing quantum memory registers $\inM_{i+1}$. 

The bound here is similar in spirit to~\cite[Lemma~5.3]{arx_MD23}, and the proof structure is similar as well. However, there is a subtlety we discuss in Remark~\ref{rem:Dsharptotrace}. 

\begin{lemma}[Depolarized pure-state comparison]\label{lem:depol}
Let $\rho_{SJ}$ and pure $\ket\phi_S$ be normalized, let $|S|>1$, and let $0<\pert<1$. Define
\[
\omega_S^\pert=(1-\pert)\pure\phi_S+\pert\frac{\id_S}{|S|}.
\]
Then, the $\Dmax$ and $\Dsharp_2$ between $\rho_{SJ}$ and $\omega_S^\pert\otimes\rho_J$ can be upper bounded as
\begin{gather*}
 \Dmax(\rho_{SJ}\Vert\omega_S^\pert\otimes\rho_J) \leq\log\frac{|S|^2}{\pert} \\
 \Dsharp_2(\rho_{SJ}\Vert\omega_S^\pert\otimes\rho_J) \leq
\log\left(\min\left\{\frac{|S|^2}{\pert},
\left[a_\pert^{-1/3}
+\left(\frac{|S|(|S|-1)\eta}{\pert}\right)^{1/3}\right]^3\right\}\right),
\end{gather*}
where
\[ a_\mu = 1 - \mu + \frac\mu{|S|}, \qquad \eta = 1 - \bra\phi\rho_S\ket\phi.\]
\end{lemma}


\begin{proof}
We need the following operator inequality (see e.g.~\cite[Lemma~5.11]{Tom16}): for any positive operator $\Gamma_{SJ}$,
\begin{equation}\label{eq:pinching_corollary}
    \Gamma_{SJ}\leq |S|\,\id_S\otimes\Gamma_J.
\end{equation}

This operator inequality gives $\rho_{SJ}\leq |S|\id_S\otimes\rho_J$, and by definition $\omega_S^\pert\geq(\pert/|S|)\id_S$. Combining them proves
\[ \rho_{SJ} \leq \frac{|S|^2}{\pert}\omega^\pert_S\otimes\rho_J \]
and thus the $\Dmax$ bound. Fact~\ref{th:sharp-facts}\ref{prop:sharp-renyi-comparison} then gives the first term inside the $\min$ for the $\Dsharp_2$ bound. We then only need to show that the $\Dsharp_2$ is upper bounded by the $\log$ of the other quantity inside the $\min$.

To do this, let $\Pi=\pure\phi_S\otimes\id_J$ and $\rho^\perp_{SJ}=(\id-\Pi)\rho_{SJ}(\id-\Pi)$. Then $\tr{\rho^\perp_{SJ}}=1-\bra\phi\rho_S\ket\phi=\eta$. 
The $\eta=0$ case holds straightforwardly, so we focus only on $\eta>0$. 
For any real number $u>0$,
\[ \left(\sqrt u\,\Pi\sqrt{\rho_{SJ}}-\frac1{\sqrt{u}}(\id-\Pi)\sqrt{\rho_{SJ}}\right)
\left(\sqrt u\,\Pi\sqrt{\rho_{SJ}}-\frac1{\sqrt{u}}(\id-\Pi)\sqrt{\rho_{SJ}}\right)^\dagger \]
is a positive operator. Expanding this operator gives
\[ u\Pi\rho_{SJ}\Pi+\frac1u(\id-\Pi)\rho_{SJ}(\id-\Pi) - \Pi\rho_{SJ}(\id-\Pi) - (\id-\Pi)\rho_{SJ}\Pi \geq 0.\]
Rearranging this and using the fact that $\rho_{SJ}=\Pi\rho_{SJ}\Pi+(\id-\Pi)\rho_{SJ}\Pi+\Pi\rho(\id-\Pi)+\rho_{SJ}^\perp$, we get,
\begin{equation}\label{eq:u-ineq}
\rho_{SJ}\leq(1+u)\Pi\rho_{SJ}\Pi+\left(1+\frac1u\right)\rho_{SJ}^\perp.
\end{equation}

Expanding $\ket{\phi}_S$ to an orthogonal basis on $S$, we can write,
\[ \rho_J = \bra{\phi}\rho_{SJ}\ket{\phi} + \sum_{a \perp \phi}\bra{a}\rho_{SJ}\ket{a} \geq \bra{\phi}\rho_{SJ}\ket{\phi}.\]
From this it follows that
\[ \pure{\phi}_S\otimes\rho_J \geq \Pi\rho_{SJ}\Pi,\]
and putting this in \eqref{eq:u-ineq},
\[ \rho_{SJ}\leq(1+u)\pure{\phi}_S\otimes\rho_J+\left(1+\frac1u\right)\rho_{SJ}^\perp\]
By the alternative definition of $\Qsharp_2$ in \eqref{eq:order2},\footnote{All inverse operations below are restricted to $\supp(\omega^\mu_S\otimes\rho_J)$. This is legitimate since we have a finite $\Dmax$ bound, so the supports of all the relevant operators are contained in $\supp(\omega^\mu_S\otimes\rho_J)$.}
\[ \Qsharp_2(\rho_{SJ}||\omega^\mu_S\otimes\rho_J) = \inf_{H\geq \rho_{SJ}}\Tr[H(\omega^\mu_S\otimes\rho_J)^{-1}H].\]
The choice
\[
H=(1+u)\pure{\phi}_S\otimes\rho_J+\left(1+\frac1u\right)\rho^\perp_{SJ}
\]
is then feasible in this optimization. It remains to upper bound $\Tr[H(\omega^\mu_S\otimes\rho_J)^{-1}H]$.

The two summands of $H$ occupy orthogonal $S$ blocks corresponding to $\ket{\phi}$ and its orthogonal space. Moreover,
\begin{align*}
(\omega_S^\pert\otimes\rho_J)^{-1}
&=(\omega^\pert_S)^{-1}\otimes\rho_J^{-1}\\
&=\left(a_\pert\pure{\phi}+\frac\pert{|S|}(\id-\pure{\phi})\right)^{-1}
\otimes\rho_J^{-1}\\
&=\left(\frac1{a_\pert}\pure{\phi}
+\frac{|S|}\pert(\id-\pure\phi)\right)\otimes\rho_J^{-1}
\end{align*}
is also block-diagonal in the same decomposition. So there are no cross terms between $\pure{\phi}$ and its orthogonal block in $\Tr[H(\omega^\mu_S\otimes\rho_J)^{-1}H]$. We can compute the block-diagonal terms as,
\begin{align*}
& \tr{H(\omega_S^\pert\otimes\rho_J)^{-1}H} \\
& =\frac{(1+u)^2}{a_\pert}\Tr\left[\pure{\phi}_S\otimes\rho_J\rho_J^{-1}\rho_J\right]
+\frac{|S|}{\pert}\left(1+\frac1u\right)^2
\tr{\rho^\perp_{SJ}((\id-\pure\phi)\otimes\rho_J^{-1})\rho^\perp_{SJ}} \\
& = \frac{(1+u)^2}{a_\pert}
+\frac{|S|}{\pert}\left(1+\frac1u\right)^2
\tr{\rho^\perp_{SJ}((\id-\pure\phi)\otimes\rho_J^{-1})\rho^\perp_{SJ}}.
\end{align*}
Now we apply the operator inequality~\eqref{eq:pinching_corollary} on the $(|S|-1)$-dimensional subspace orthogonal to $\ket{\phi}$ to get
\begin{equation}\label{eq:nu-phi-1}
\rho^\perp_{SJ}\leq(|S|-1)(\id-\pure\phi)\otimes\rho^\perp_J \leq(|S|-1)(\id-\pure\phi)\otimes\rho_J.
\end{equation}

We claim
\begin{equation}\label{eq:nu-phi-2}
\rho^\perp_{SJ}((\id-\pure\phi)\otimes\rho_J^{-1})\rho^\perp_{SJ}
\leq (|S|-1)\rho^\perp_{SJ}.
\end{equation}
In fact, consider the general operator inequality $B\leq cA$, and let $X=A^{-1/2}BA^{-1/2}$. From the inequality we get $X^2\leq cX$. Hence,
\[ BA^{-1}B = A^{1/2}X^2A^{1/2} \leq A^{1/2}cXA^{1/2} = cB.\]
This, along with, \eqref{eq:nu-phi-1}, shows \eqref{eq:nu-phi-2}. Thus,
\[ \tr{\rho^\perp_{SJ}((\id-\pure\phi)\otimes\rho_J^{-1})\rho^\perp_{SJ}} \leq (|S|-1)\Tr[\rho^\perp_{SJ}] = (|S|-1)\eta.\]
The objective function for $\Qsharp_2$ is then upper bounded by
\[
a_\pert^{-1}(1+u)^2+
\frac{|S|(|S|-1)\eta}{\pert}\left(1+\frac1u\right)^2
\]
for any $u$, and in particular it holds for the $u$ that minimizes the expression.

For positive $x,y$, the minimum of $x(1+u)^2+y(1+u^{-1})^2$ over $u>0$ occurs at $u=(y/x)^{1/3}$ and equals $(x^{1/3}+y^{1/3})^3$. Substituting this,
\[ \Qsharp_2(\rho_{SJ}||\omega^\pert_S\otimes\rho_J) \leq \Tr[H(\omega^\pert_S\otimes\rho_J)^{-1}H] \leq \left(a_\pert^{-1/3}+\left(\frac{|S|(|S|-1)\eta}{\pert}\right)^{1/3}\right)^3\]
which gives the required upper bound on $\Dsharp_2$.
\end{proof}

\begin{remark}\label{rem:Dsharptotrace}
We can also obtain a bound similar to Lemma~\ref{lem:depol} via \cite[Lemma~5.3]{arx_MD23}: while the result has worse dependency on $\eta$, it does not change any scaling conclusions in our main results. 

Note however that this is not achieved simply by applying that lemma directly to $\rho_{SJ}$ and $\omega_S^\pert\otimes\rho_J$ with our max-divergence bound (the resulting bound is not sufficently strong). Rather, let us first introduce the intermediate state
\[
\rho^\pert_{SJ}=(1-\pert)\rho_{SJ}+\pert\frac{\id_S}{|S|}\otimes\rho_J.
\]
Letting $\eps_\pert$ denote the trace distance between this state and $\omega_S^\pert\otimes\rho_J$, 
the depolarizing contributions cancel when computing the trace distance, so we have (for an arbitrary purification $\ket\rho_{SJR}$ of $\rho_{SJ}$):
\begin{align*}
\eps_\pert
=\frac12\norm{\rho^\pert_{SJ}-\omega_S^\pert\otimes\rho_J}_1&=\frac{1-\pert}{2}\norm{\rho_{SJ}-\pure\phi_S\otimes\rho_J}_1\\
&\leq\frac{1-\pert}{2}\norm{\pure\rho_{SJR}-\pure\phi_S\otimes\rho_{JR}}_1 \\
&\leq 2(1-\pert)\sqrt{\eta}.
\end{align*}
where the last inequality holds since $\ket\phi_S$ is pure, allowing us pick an Uhlmann purification of the form $\ket\phi_S\ket\tau_{JR}$ and use triangle inequality to convert from $\rho_{JR}$ to $\tau_{JR}$.
The intermediate state $\rho^\pert_{SJ}$ retains the marginal $\rho_J$, so the operator inequality~\eqref{eq:pinching_corollary} also gives
$\rho^\pert_{SJ}\leq\frac{|S|^2}{\pert}\,\omega_S^\pert\otimes\rho_J$.
Applying \cite[Lemma~5.3]{arx_MD23} to this intermediate pair bounds its $\Dsharp_2$ in terms of $\eps_\pert$. Since $\rho^\pert_{SJ}\geq(1-\pert)\rho_{SJ}$, monotonicity and homogeneity of $\Qsharp_2$ in its first argument then give
\begin{align*}
\Dsharp_2(\rho_{SJ}\Vert\omega_S^\pert\otimes\rho_J)
&\leq \Dsharp_2(\rho^\pert_{SJ}\Vert\omega_S^\pert\otimes\rho_J)
+2\log\frac1{1-\pert}\\
&\leq 2\log\frac1{1-\pert}
+3\log\left[(1+\sqrt{\eps_\pert})^{2/3}
+\left(\frac{|S|^2}{\pert}\right)^{2/3}\eps_\pert^{1/6}\right].
\end{align*}
For fixed $\pert>0$, the second term is $O_{\pert,|S|}(\eta^{1/12})$ as $\eta\to0$, compared with the $O_{\pert,|S|}(\eta^{1/3})$ vanishing term in Lemma~\ref{lem:depol}. Thus this approach gives weaker quantitative dependence, but still allows the divergence to be made arbitrarily small by first choosing $\pert$ small and then choosing $\eta$ small.
\end{remark}


\section{Device-independent replacement of trusted preparation}\label{sec:DI}

In this section, we prove our sequential DD-to-DI conversion result. We first specify the trusted preparation resource, the security statement assumed of the device-dependent protocol, and the self-testing hypothesis required on the trusted preparation in order for a device-independent replacement to be feasible. We then describe the sequential device-independent protocol and state the main theorem relating its security directly to that of the original device-dependent protocol. The remainder of the section proves the theorem.

\subsection{Requirements on the device-dependent protocol}

Let $\mathcal X$ and $\mathcal A$ be finite alphabets. Fix a distribution $\sfP_{XA}$ on pairs $(x,a)\in\mathcal X\times\mathcal A$ and a pure state $\ket{\psi_{x,a}}_{\idealB}$ for every pair. The register $\idealB$ is the single quantum register sent to Bob. We consider prepare-and-measure protocols of the form in Protocol~\ref{prot:DD}. The security experiment relevant for the protocol may impose additional restrictions on Bob at specified times; for example, it may require him to divide his memory between two noncommunicating receivers before Alice releases some classical information.

\begin{prot}[!ht]
\caption{DD protocol with trusted state preparation.}\label{prot:DD}
\begin{enumerate}[leftmargin=*,itemsep=3pt]
\item For each $i\in\{1,\ldots,m\}$, Alice samples $(x_i,a_i)$ from $\sfP_{XA}$, prepares a fresh copy of $\ket{\psi_{x_i,a_i}}_{\idealB_i}$, and sends $\idealB_i$ to Bob. Bob may maintain arbitrary private memory and apply any channels permitted by the security experiment.
\item Alice and Bob run a specified cryptographic continuation using Alice's $\bx,\ba$, Bob's final memory, and any fresh trusted randomness required by the protocol. The security experiment defines an adversarial success event, denoted by $\Bad$.
\end{enumerate}
\end{prot}

We write $\ddP$ for any protocol having the form of Protocol~\ref{prot:DD}. The state preparation procedure followed by Alice for a single $i$ produces
\begin{equation}\label{eq:trusted-source}
\rho^{\mathrm{prep}}_{XA\idealB}
=\sum_{x,a}\sfP_{XA}(x,a)\pure{x,a}_{XA}\otimes\pure{\psi_{x,a}}_{\idealB}.
\end{equation}
Thus trusted preparation means that the DD security statement applies to $m$ independent instances of the CQ state \eqref{eq:trusted-source}, together with Bob's permitted interspersed channels and the prescribed continuation.

Our security proof for the DI version of $\ddP$ will later produce the following depolarized version of the trusted states:
\begin{equation}\label{eq:noisy-source}
\rho^{\mathrm{prep},\pert}_{XA\idealB}
=(1-\pert)\rho^{\mathrm{prep}}_{XA\idealB}
+\pert\sum_x\sfP_X(x)\pure x_X\otimes
\frac{\id_A}{|A|}\otimes\frac{\id_{\idealB}}{|\idealB|},
\end{equation}
where $0<\pert<1$. The depolarizing component keeps the prescribed distribution of $X$, makes $A$ uniform, and replaces Bob's transmitted state by the maximally mixed state. This state is used only in the ideal process of the security proof. The first postulate requires the DD security bound to remain valid for it.

\Needspace{6\baselineskip}
\begin{postulates}[DD security with depolarized noise]\label{post:DD}
Let $\ddP$ be a protocol of the form given in Protocol~\ref{prot:DD}. For every relevant $0<\pert<1$, replacing its $m$ trusted preparations by independent copies of \eqref{eq:noisy-source} gives
\begin{equation}\label{eq:DD}
\sup_{\textnormal{allowed attacks}}\Pr_{\ddP}[\Bad]
\leq q_{\mathrm{dep}}(\pert)^m.
\end{equation}
The bound includes arbitrary initial auxiliary information independent of Alice's preparation, arbitrary channels by Bob interspersed between each preparation, and all continuation messages and their timing.
\end{postulates}

If $\ddP$ only has an upper bound on $\Pr[\Bad]$ when no depolarization happens, this does not automatically imply \eqref{eq:DD}. This is because the second term in \eqref{eq:noisy-source} need not have the original distribution on $XA$. If the trusted outcome is uniform for every input, then $\sfP_{XA}(x,a)=\sfP_X(x)/|A|$ and
\begin{equation}\label{eq:Bob-depol}
\rho^{\mathrm{prep},\pert}
=(\idmap_{XA}\otimes\mapD_\pert)[\rho^{\mathrm{prep}}],
\qquad \text{where} \quad
\mapD_\pert[\sigma]=(1-\pert)\sigma+\pert\,\tr{\sigma}\frac{\id_{\idealB}}{|\idealB|}.
\end{equation}
Whenever the depolarizing noise channel $\mapD_\pert$ is an allowed Bob attack, the original DD bound applies with $q_{\mathrm{dep}}(\pert)=q$. This observation applies to the UE and copy-protection protocols in Sections~\ref{sec:UE} and~\ref{sec:copy}.

To certify the state \eqref{eq:trusted-source} by a nonlocal test in the DI setting, it must first admit an entanglement-based realization. That is, Alice should be able to prepare it by measuring her share of one fixed bipartite state. Let $\sfP_X(x)=\sum_a\sfP_{XA}(x,a)$ be Alice's marginal distribution, and remove any $x$ with $\sfP_X(x)=0$ from the input alphabet. The following fact is a standard result:

\begin{fact}[Pure-state ensemble realization, {\cite{HJW93,Wat18}}]\label{th:HJW}
There are a pure state $\ket\phi_{\idealA\idealB}$ and projective measurements $\{T_a^x\}_a$ on $\idealA$ satisfying
\begin{equation}\label{eq:source-realization}
\sfP_{XA}(x,a)\pure{\psi_{x,a}}_{\idealB}
=\sfP_X(x)\tr[\idealA]{
(T_a^x\otimes\id)\pure\phi(T_a^x\otimes\id)}
\end{equation}
for every $x,a$ if and only if the average transmitted state
\begin{equation}\label{eq:common-average}
\sum_a\frac{\sfP_{XA}(x,a)}{\sfP_X(x)}\pure{\psi_{x,a}}_{\idealB}
\end{equation}
is independent of $x$. The space $\idealA$ may be enlarged to make the measurements projective.
\end{fact}
The entanglement-based realization of the prepare-and-measure protocol $\ddP$ will then involve Alice and Bob sharing the state $\ket{\phi}_{S^\A S^\B}$, holding the registers $S^\A$ and $S^\B$ respectively, Alice sampling $x$ according to $\sfP_X$, and measuring the $S^\A$ register with $\{T^x_a\}$. The outcome $a$ here comes out of the measurement --- Alice does not choose it classically. The probability of obtaining $a$ for a fixed $x$ is given by
\[ \sfP_{A|x}(a) = \tr{
(T_a^x\otimes\id)\pure\phi(T_a^x\otimes\id)}.\]

We will also need the state $\ket{\phi}_{S^\A S^\B}$ and Alice's preparation measurements to occur in an optimal strategy for a nonlocal game that self-tests them. We first recall the standard terminology.

\begin{definition}[Nonlocal game and quantum value]\label{def:game}
A two-player nonlocal game $\game$ is specified by finite input alphabets $\mathcal X,\mathcal Y$, finite output alphabets $\mathcal A,\mathcal B$, an input distribution $\sfP_{XY}$, and a predicate
\[
\sfV:\mathcal X\times\mathcal Y\times\mathcal A\times\mathcal B
\longrightarrow\{0,1\}.
\]
The verifier samples $(x,y)$ from $\sfP_{XY}$, sends $x$ to Alice and $y$ to Bob, receives outputs $a,b$, and accepts exactly when $\sfV(x,y,a,b)=1$. A quantum strategy
\[
\strat=(\rho_{\memA\memB},\{H_a^x\}_{x,a},\{K_b^y\}_{y,b})
\]
consists of a bipartite state and local measurements on $\memA$ and $\memB$. Its winning probability is
\[
\sum_{x,y,a,b}\sfP_{XY}(x,y)\sfV(x,y,a,b)
\tr{(H_a^x\otimes K_b^y)\rho}.
\]
The quantum value $\omega^\ast(\game)$ is the supremum of this probability over finite-dimensional quantum strategies.
\end{definition}

For the self-testing statement alone, a standard Naimark dilation lets us restrict to projective measurements on enlarged local registers without loss of generality \cite[Section~2.3]{Wat18}. The sequential protocol also retains postmeasurement device memories, so the stronger stored-output form of this reduction is recorded below.

\begin{definition}[Strategy deficit]\label{def:deficit}
For a game $\game$ and a strategy $\strat=(\rho,\{H_a^x\}_{x,a},\{K_b^y\}_{y,b})$, its deficit is
\[
\delta_{\game}(\strat)
=\omega^\ast(\game)-
\sum_{x,y,a,b}\sfP_{XY}(x,y)\sfV(x,y,a,b)
\tr{(H_a^x\otimes K_b^y)\rho}.
\]
Thus the deficit is a property of the state and both measurement families, not of the state alone.
\end{definition}

We can now state the second postulate required from the protocol $\ddP$.

\Needspace{6\baselineskip}
\begin{postulates}[Self-testable entanglement-based realization]\label{post:selftest}
Let $\ddP$ be a protocol of the form Protocol~\ref{prot:DD}.
\begin{enumerate}[label=\textup{(\roman*)},ref=\textup{(\roman*)}]
\item\label{post:realize} The distribution $\sfP_{XA}$ and states $\ket{\psi_{x,a}}_{S^\B}$ used in $\ddP$ satisfy the condition in Fact~\ref{th:HJW}. Hence they admit an entanglement-based realization with shared state $\ket{\phi}_{S^\A S^\B}$ and measurements $(\{T^x_a\}_a)_x$ on Alice's side.

\item\label{post:game} There exists a nonlocal game
\[
\game=(\mathcal X,\mathcal Y,\mathcal A,\mathcal B,\sfP_{XY},\sfV)
\]
whose Alice input and output alphabets are those of the trusted preparation in $\ddP$ and $X$-marginal is $\sfP_X$, such that
 $\ket{\phi}_{S^\A S^\B}$ and $(\{T^x_a\}_a)_x$ are part of an optimal strategy for $\game$. More precisely, there are projective measurements $\{N_b^y\}_b$ on $\idealB$ such that the strategy consisting of $\ket\phi_{\idealA\idealB}$, Alice's measurements $\{T_a^x\}_a$, and Bob's measurements $\{N_b^y\}_b$ wins with probability $\omega^\ast(\game)$:
\[
\sum_{x,y,a,b}\sfP_{XY}(x,y)\sfV(x,y,a,b)
\bra\phi(T_a^x\otimes N_b^y)\ket\phi
=\omega^\ast(\game).
\]

\item\label{post:rigid} There is a vanishing function $\vanf_{\game}$ such that, for every pair of finite-dimensional registers $\memA,\memB$ and every pair of projective measurement families $\{H_a^x\}_a,\{K_b^y\}_b$ on these registers, there are local isometries
\[
V^\A:\memA\longrightarrow\idealA\junkA,
\qquad
V^\B:\memB\longrightarrow\idealB\junkB
\]
that depend only on the measurement families. For every purified state $\ket\rho_{R\memA\memB}$, let
\[
\strat_\rho=
\left(\tr[R]{\pure\rho},\{H_a^x\}_{x,a},\{K_b^y\}_{y,b}\right).
\]
There is a normalized state $\ket\tau_{R\junkA\junkB}$ satisfying
\begin{equation}\label{eq:rigidity}
\sum_{x,a}\sfP_X(x)
\norm{(\id_R\otimes V^\A\otimes V^\B)H_a^x\ket{\rho}_{RM^\A M^\B}
-(T_a^x\otimes\id)\ket{\phi}_{S^\A S^\B}\ket{\tau}_{RJ^\A J^\B}}_2^2
\leq\vanf_{\game}(\delta_{\game}(\strat_\rho))^2.
\end{equation}
Only $\ket\tau$ may depend on the state $\ket\rho$; the isometries are fixed once the measurement families are fixed.
\end{enumerate}
\end{postulates}

Note that even though Postulate~\ref{post:selftest}\ref{post:rigid} requires the isometries $V^\A$ and $V^\B$ to not depend on the state $\ket{\rho}$, there is a sense in which it is weaker than existing self-testing results in the literature. Equation~\eqref{eq:rigidity} only requires closeness on average over $x,a$, whereas existing self-testing results usually provide this for every $x, a$.

We write $\omega^\ast=\omega^\ast(\game)$ for the fixed game $\game$ from the above postulates from now on.

\subsection{The sequential DI protocol}

Alice is the trusted party in the DI protocol. She holds one component of an untrusted device behind a shielded interface, while Bob holds the other component (we will also refer to their components as their individual devices). Immediately before the inputs in round $i$ are supplied, $\memA_i$ and $\memB_i$ denote the complete quantum memories of these two components. They include their internal workspaces and everything retained from earlier rounds. The components may begin in an arbitrary joint state and may retain arbitrary history-dependent memory; in particular, no dimension bound or tensor-product structure across rounds is assumed.

We do require a tensor-product structure between the two device components within each round. Operationally, Alice enforces this by shielding her device: she can supply an input and retrieve its output, but no unprescribed classical or quantum system can cross the shield. Alice herself may sample inputs and send the protocol's prescribed classical messages to Bob. Thus the locality assumption concerns the two untrusted device components: it does not forbid legitimate communication by Alice as the trusted party.

\begin{postulates}[Shielding]\label{post:shield}
The protocols considered below have the following properties.
\begin{enumerate}[label=\textup{(\roman*)},ref=\textup{(\roman*)}]
\item Conditional on the classical inputs supplied in round $i$, the device channel factors as $\mapE_i^{\A}\otimes\mapE_i^{\B}$ where $\mapE^\A_i$ only acts on $M^\A_i$ and $\mapE^\B_i$ only acts on $M^\B_i$. There is no additional communication between Alice and Bob's devices.
\item Alice is able to stow away her device at any point, so that future operations do not touch it.
\end{enumerate}
\end{postulates}

Our sequential device-independent version of the DD protocol $\ddP$ is given in Protocol~\ref{prot:DI}, and is denoted by $\diP$. Before describing the protocol, we introduce some notation we will be using in it.

Each round of $\diP$ has a single classical record $C_i\in\alphCS=\{\gen,0,1\}$. A generation round has record $\gen$, while a test round has record $1$ or $0$ according to the predicate $\sfV$ evaluated on the inputs and outputs $(x_i, y_i, a_i, b_i)$. For a fixed protocol parameter $\gamma$, the ideal distribution is
\[
\sfP_C(\gen)=1-\ptest,
\qquad
\sfP_C(1)=\ptest\omega^\ast,
\qquad
\sfP_C(0)=\ptest(1-\omega^\ast).
\]
The accepting set is
\begin{equation}\label{eq:accept-ball}
S_\Acc=\{\sfQ_C:\trnorm{\sfQ_C-\sfP_C}\leq\epstol\}.
\end{equation}
The corresponding event is $\Acc$. The distribution generated by Alice and Bob's devices is tested against $\sfP_C$ with the condition \eqref{eq:accept-ball}; this simultaneously tests the game-winning probability and the frequency of generation rounds.

\begin{prot}[!ht]
\caption{Sequential DI replacement of trusted preparation.}\label{prot:DI}

\textbf{Parameters.}
\begin{itemize}
\item A protocol $\ddP$ of the form given in Protocol~\ref{prot:DD}, its entanglement-based realization and game $\game$ satisfying Postulate Set~\ref{post:selftest}
\item Number of rounds $n$, testing probability $\ptest$, and tolerance $\epstol$, satisfying
\[ m\leq\floor{(1-\ptest-\epstol)n}.\]
\end{itemize}
\begin{enumerate}[leftmargin=*,itemsep=3pt]
\item Alice and Bob have devices which take inputs from the input sets $\mathcal{X}$ and $\mathcal{Y}$, and produce outputs in the output sets $\mathcal{A}$ and $\mathcal{B}$ of $\game$ respectively.
\item For each $i\in[n]$ in sequence:
\begin{enumerate}
\item Alice samples $x_i$ from $\sfP_X$ independently and supplies it as input to her device. She gets outputs $a_i$ from her device, and stores $(x_i,a_i)$ outside the device.
\item Alice samples $u_i \in \{\test, \gen\}$ s.t. $u_i=\test$ w.p. $\ptest$.
 \begin{itemize}
\item  If $u_i=\test$, Alice samples $y_i$ from $\sfP_{Y|x_i}$ and sends $(u_i,y_i)$ to Bob. Bob returns $b_i$ to Alice. Alice sets $c_i=\sfV(x_i,y_i,a_i,b_i)$ and stores this outside the device.
\item If $u_i=\gen$, Alice only sends $u_i$ to Bob. She sets $c_i=\gen$ and stores it outside.
\end{itemize}
\end{enumerate}
\item Alice checks if $\trnorm{\freq_{\bcS}-\sfP_C}\leq\epstol$ and aborts if this condition is not satisfied. She sends this decision to Bob.
\begin{itemize}
\item If Alice aborts, the protocol ends.
\item If Alice accepts:
\begin{enumerate}
\item Alice selects the lexicographically first subset $E\subseteq[n]$ of size $m$ such that $c_i=\gen$ for every $i\in E$. She sends $E$ to Bob.
\item Alice stows away her device and Alice and Bob run the continuation of $\ddP$ using $\bX_E,\bA_E$ and Bob's final memory.
\end{enumerate}
\end{itemize}
\end{enumerate}
\end{prot}

\paragraph{Honest behaviour and completeness.} We first describe the intended honest behaviour of the devices and show that if the devices behave honestly, Alice does not abort with high probability. The noiseless honest behaviour is as follows:
\begin{itemize}
\item Alice and Bob's devices share $n$ copies of the optimal state $\ket{\phi}_{S^\A S^\B}$ for $\game$. On inputs $x,y$, they perform the measurements $\{T^x_a\}_a$ and $\{N^y_b\}_b$, respectively.
\item In every test round $i$, Bob supplies the input $y_i$ received from Alice to his device and returns the resulting output $b_i$ to Alice.
\end{itemize}
We also allow noisy honest implementations. We model their noise only through the winning probability of $\game$: the devices have honest noise $\dhon$ if they use a strategy $\strat$ satisfying
\[
\delta_\game(\strat)=\dhon.
\]
The parameter $\dhon$ need not decrease with $n$.

\begin{theorem}[Completeness]\label{th:complete}
Suppose the honest rounds are independent and in every round the devices use the same strategy $\strat$ with deficit $\delta_\game(\strat)=\dhon$. If $\ptest\dhon<\epstol$, then
\begin{equation}\label{eq:complete}
\Pr_{\diP}[\trnorm{\freq_{\bcS}-\sfP_C}>\epstol]
\leq2|\alphCS|\exp\left(
-\frac{8n(\epstol-\ptest\dhon)^2}{|\alphCS|^2}
\right)
=6\exp\left(-\frac{8n(\epstol-\ptest\dhon)^2}{9}\right).
\end{equation}
On every accepting execution, the set $E$ in Protocol~\ref{prot:DI} is well defined and has size $m$.
\end{theorem}
\begin{proof}
The winning probability of the honest strategy is $\omega^\ast-\dhon$. Since Alice independently chooses a test round with probability $\ptest$, the common distribution of each record is
\[
\sfQ_C(\gen)=1-\ptest,
\qquad
\sfQ_C(1)=\ptest(\omega^\ast-\dhon),
\qquad
\sfQ_C(0)=\ptest(1-\omega^\ast+\dhon).
\]
Consequently,
\[
\trnorm{\sfQ_C-\sfP_C}
=\frac12\left(0+\ptest\dhon+\ptest\dhon\right)
=\ptest\dhon.
\]
If Alice aborts, the triangle inequality therefore gives
\begin{equation}\label{eq:complete-triangle}
\trnorm{\freq_{\bcS}-\sfQ_C}
>\epstol-\ptest\dhon.
\end{equation}

Fix $c\in\alphCS$. The variables $\mathbf{1}\{C_i=c\}$ are independent Bernoulli variables with mean $\sfQ_C(c)$, and
\[
\freq_{\bcS}(c)=\frac1n\sum_{i=1}^n\mathbf{1}\{C_i=c\}.
\]
By the two one-sided bounds in Fact~\ref{th:Hoeffding}, for every $u>0$,
\[
\Pr_{\diP}\!\left[
\left|\freq_{\bcS}(c)-\sfQ_C(c)\right|>u
\right]
\leq 2e^{-2nu^2}.
\]
If all $|\alphCS|$ coordinates differed by at most
$2(\epstol-\ptest\dhon)/|\alphCS|$, their absolute differences would sum to at most $2(\epstol-\ptest\dhon)$, contradicting \eqref{eq:complete-triangle}. Thus an abort implies that this threshold is exceeded for at least one $c$. Substituting the threshold into the preceding bound and taking a union bound over $c\in\alphCS$ proves \eqref{eq:complete}.

Finally, on acceptance,
\[
\freq_{\bcS}(\gen)
\geq\sfP_C(\gen)-\epstol
=1-\ptest-\epstol.
\]
Hence there are at least $n(1-\ptest-\epstol)\geq m$ generation rounds, so the lexicographically first set $E$ of size $m$ exists.
\end{proof}

\paragraph{Soundness.} We now state our main soundness theorem for the protocol $\diP$, which upper bounds the probability of Alice not aborting and the event $\Bad$ occurring in the continuation of $\ddP$. The rest of the section proves this theorem.
\begin{theorem}[Generic DD-to-DI conversion]\label{th:DI}
Let $\ddP$ be a DD protocol of the form given in Protocol~\ref{prot:DD}, satisfying Postulate Sets~\ref{post:DD} and~\ref{post:selftest}. Let $\diP$ be its DI replacement from Protocol~\ref{prot:DI} in a sequential setting where Postulate Set~\ref{post:shield} holds. Then:
\begin{itemize}
\item for every $k>0$,
\item there exist constants $\pert\in\Theta(k)$, $\ptest\in(1/2,1)$ such that $\ptest\to1$ as $k\to0$, and $\epstol\in(0,(1-\ptest)/2]$, independent of $n$, such that
\item for every integer $n\geq4/(1-\ptest)$, every allowed implementation and attack, and every normalized initial state $\rho^0_{\inM_1R}$, the following holds with $m=\floor{(1-\ptest)n/2}$.
\end{itemize}
The soundness bound is
\begin{equation}\label{eq:DI-rate}
\Pr_{\diP}[\Acc\land\Bad]
\leq\left(2^kq_{\mathrm{dep}}(\pert)\right)^{m/2}.
\end{equation}
If $\limsup_{\pert\to0}q_{\mathrm{dep}}(\pert)\leq q$ for some $0<q<1$, then for every $0<k<-\log q$ the constants can be chosen so that $2^kq_{\mathrm{dep}}(\pert)<1$, giving an exponentially small soundness error in $m$.
\end{theorem}

\subsection{Real and ideal channels}
\label{subsec:realandideal}
We will invoke Theorem~\ref{th:main} in order to do the soundness proof. Before doing so, we need to identify what the ideal and real channels in each round of the protocol are. 
\subsubsection{Real channels}
We first note that the real channels we will be modelling here are not exclusively what is going on in the device, but also Alice's actions outside of the device. This is because the protocol $\diP$ requires Alice to sample inputs in each round randomly, and the device operations act on the respective inputs and $M^\A_i, M^\B_i$. However, the channels in Theorem~\ref{th:main} don't really take external inputs with trusted distributions and only act on their memory. Therefore, our channels $\mapE_i$ representing Alice's operations as well as the internal machinery of the device will first generate the state $\sum_{u_i,x_i,y_i}\sfP_{UXY}(u_i,x_i,y_i)\pure{u_i,x_i,y_i}_{U_iX_iY_i}$, perform a measurement depending on $u_i,x_i,y_i$ on $M^\A_iM^\B_i$ and produce some classical output registers that will not be touched in future rounds. Here we are using $\sfP_{UXY}$ to denote the joint distribution of the $u_i$ and $x_i,y_i$ that are generated in each round; technically no $y_i$ is generated if $u_i=\gen$, but it will be convenient to keep the $Y_i$ register around even in this case, with an empty value $\testQ$. Sampling these variables first gives the same channel as the protocol's ordering because Alice's device is supplied only $x_i$, and $u_i,y_i$ are sampled independently of $a_i$ conditional on $x_i$.

All the classical variables produced in round $i$ are $u_i, x_i, y_i, a_i, b_i$ and the record $c_i$. Therefore, $\mapE_i\in\CPTP{M^\A_iM^\B_i}{U_iX_iY_iA_iB_iC_iM^\A_{i+1}M^\B_{i+1}}$. Using the Postulate Set~\ref{post:shield}, and the structure of Protocol~\ref{prot:DI}, we can write the channel $\mapE_i$ as
\begin{equation}\label{eq:E_i-form}
\mapE_i[\rho_{M^\A_iM^\B_i}] = \mapR_i\left[\sum_{u_i,x_i,y_i}\sfP_{UXY}(u_i,x_i,y_i)\pure{u_i,x_i,y_i}_{U_iX_iY_i}\otimes(\mapE^{\A,x_i}_i\otimes\mapE^{\B,u_i,y_i}_i)[\rho_{M^\A_iM^\B_i}]\right].
\end{equation}
Here the recording channel $\mapR_i\in\CPTP{U_iX_iY_iA_iB_i}{U_iX_iY_iA_iB_iC_i}$ is a trusted operation whose action on the classical records is
\begin{align*}
&\mapR_i[\pure{u_i,x_i,y_i,a_i,b_i}_{U_iX_iY_iA_iB_i}]\\
&=\pure{u_i,x_i,y_i,a_i,b_i}_{U_iX_iY_iA_iB_i}\otimes\pure{g(u_i,x_i,y_i,a_i,b_i)}_{C_i},
\end{align*}
where
\[ g(u_i,x_i,y_i,a_i,b_i) = \begin{cases} \gen & \text{ if } u_i=\gen \\ \sfV(x_i,y_i,a_i,b_i) & \text{ if } u_i =\test. \end{cases} \]
The registers $A_i, B_i$ that $\mapR_i$ acts on are produced by the internal device channels
\[
\mapE^{\A,x_i}_i\in\CPTP{M^\A_i}{A_iM^\A_{i+1}},
\qquad
\mapE^{\B,u_i,y_i}_i\in\CPTP{M^\B_i}{B_iM^\B_{i+1}},
\]
by doing measurements on their internal memories (note that Alice's device channel never gets $u_i$).
For $u_i=\gen$, we write $\mapE^{\B,\gen}_i$ for $\mapE^{\B,\gen,\testQ}_i$. This is an arbitrary channel from $M_i^\B$ to $M^\B_{i+1}$, which also attaches the register $B_i$ with the empty value $\testQ$. This way the outputs of the internal device channels are on the same registers in all cases.

We now describe what the channels $\mapE^{\A,x_i}_i$ and $\mapE^{\B,\test,y_i}_i$ are doing. These channels are producing both a classical output by doing a measurement, and storing a post-measurement state as memory for future rounds. The stored-output form of Naimark dilation represents any such instrument by projective measurements followed by outcome-dependent channels \cite[Appendix~A]{Tan21}. For Alice, we first append an ancilla $\ancA_i$ in the fixed state $\ket0$, and use projectors $\{H^{x_i}_{i,a_i}\}_{a_i}$ on $M_i^\A\ancA_i$. The ancilla and its initial state are the same for all inputs $x_i$. Similarly, Bob's test measurements use an ancilla $\ancB_i$ and projectors $\{K^{y_i}_{i,b_i}\}_{b_i}$ on $M_i^\B\ancB_i$. These ancillas are internal to the round and are not additional inputs to $\mapE_i$. Therefore, we can decompose $\mapE^{\A,x_i}_i$ and $\mapE^{\B,\test,y_i}_i$ as measurement channels, including the addition of the ancillas, composed with channels
\[
\mapL_i^{\A,x_i}\in\CPTP{A_iM_i^\A\ancA_i}{A_iM_{i+1}^\A},
\qquad
\mapL_i^{\B,y_i}\in\CPTP{B_iM_i^\B\ancB_i}{B_iM_{i+1}^\B},
\]
where the latter maps use the $A_i$ and $B_i$ registers only as classical control registers. Explicitly,
\begin{equation}\label{eq:EAB-form-2}
\mapE^{\A,x_i}_i = \mapL^{\A,x_i}_i\circ\mapM^{\A,x_i}_i,
\qquad
\mapE^{\B,\test,y_i}_i = \mapL^{\B,y_i}_i\circ\mapM^{\B,y_i}_i,
\end{equation}
where
\begin{gather*}
\mapM^{\A,x_i}_i[\rho_{M_i^\A}] = \sum_{a_i}\pure{a_i}_{A_i}\otimes H^{x_i}_{i,a_i}(\rho_{M_i^\A}\otimes\pure0_{\ancA_i})H^{x_i}_{i,a_i},\\
\mapM^{\B,y_i}_i[\rho_{M_i^\B}] = \sum_{b_i}\pure{b_i}_{B_i}\otimes K^{y_i}_{i,b_i}(\rho_{M_i^\B}\otimes\pure0_{\ancB_i})K^{y_i}_{i,b_i},\\
\mapL^{\A,x_i}_i[\pure{a_i}_{A_i}\otimes\sigma] = \pure{a_i}_{A_i}\otimes\mapL_i^{\A,x_i,a_i}[\sigma],\\
\mapL^{\B,y_i}_i[\pure{b_i}_{B_i}\otimes\sigma] = \pure{b_i}_{B_i}\otimes\mapL_i^{\B,y_i,b_i}[\sigma].
\end{gather*}
Here $\mapL_i^{\A,x_i,a_i}\in\CPTP{M_i^\A\ancA_i}{M_{i+1}^\A}$ and $\mapL_i^{\B,y_i,b_i}\in\CPTP{M_i^\B\ancB_i}{M_{i+1}^\B}$ are arbitrary CPTP channels.

For use in the ideal channel, we also extend Bob's generation continuation to the enlarged memory by defining
\[
\mapL_i^{\B,\gen}\in\CPTP{M_i^\B\ancB_i}{B_iM_{i+1}^\B},
\qquad
\mapL_i^{\B,\gen}[\sigma]=\mapE_i^{\B,\gen}[\tr[\ancB_i]{\sigma}].
\]
Appending $\ket0_{\ancB_i}$ and applying this continuation gives exactly the original generation channel. Finally, we will often write
\[ \mapE_i= \gamma\pure{\test}_{U_i}\otimes\mapE^\test_i + (1-\gamma)\pure{\gen}_{U_i}\otimes\mapE^\gen_i,\]
where the maps $\mapE^\gen_i, \mapE^\test_i\in\CPTP{M^\A_iM^\B_i}{X_iY_iA_iB_iC_iM^\A_{i+1}M^\B_{i+1}}$ can be identified from \eqref{eq:E_i-form}, noticing that $\mapR_i$ leaves the register $U_i$ unchanged.

\paragraph{Register identification.}
Finally, in order to apply Theorem~\ref{th:main}, we need to identify which output registers of $\mapE_i$ will take the roles of $C_i, O_i$ and $M_{i+1}$ in the theorem. $M^\A_{i+1}M^\B_{i+1}$ need to be included as part of $M_{i+1}$, but there is some freedom in choosing what we consider its protected registers. The accept event is defined on $c_i$ in both $\diP$ and Theorem~\ref{th:main}, so we identify $C_i$ in the protocol with $C_i$ in Theorem~\ref{th:main}. The registers $U_iX_iY_iA_iB_i$ are all stored by Alice externally from the devices, so we could consider all of these the other protected register $\outQ_i$ of Theorem~\ref{th:main}. But for the purposes of the proof, we will only need the registers $U_iX_iA_i$ to be protected. Thus we include $Y_iB_i$ in the theorem's continuing memory; in the actual protocol, later channels will act identically on these registers, but this will not matter since this is a special case of what is allowed in Theorem~\ref{th:main}. Thus the mapping to Theorem~\ref{th:main} is
\begin{gather*}
C_i\mapsto C_i,\qquad U_iX_iA_i\mapsto O_i,\\
\bY_{<i}\bB_{<i}M_i^\A M_i^\B\mapsto M_i,\qquad
\bY_{\leq i}\bB_{\leq i}M_{i+1}^\A M_{i+1}^\B\mapsto M_{i+1}.
\end{gather*}
Here $\bY_{<i}\bB_{<i}$ denotes all the earlier $Y,B$ records and is trivial for $i=1$. We use the same notation $\mapE_i$ for the extension of \eqref{eq:E_i-form} by the identity on these earlier records. The incoming and outgoing device memories remain $M_i^\A M_i^\B$ and $M_{i+1}^\A M_{i+1}^\B$; the ancillas $\ancA_i,\ancB_i$ are used only within the round.

\subsubsection{Ideal channels}

We next define the ideal channels using the real channels above. To display their dependence on the depolarizing parameter, we denote them by $\mapF_i^\pert$; within the construction and its analysis below, we abbreviate $\mapF_i=\mapF_i^\pert$ for fixed $\pert$.

For each round $i$, applying Postulate Set~\ref{post:selftest}\ref{post:rigid} to the full projector families in \eqref{eq:EAB-form-2} gives fixed isometries
\[
V_i^\A:M_i^\A\ancA_i\longrightarrow S_i^\A J_i^\A,
\qquad
V_i^\B:M_i^\B\ancB_i\longrightarrow S_i^\B J_i^\B.
\]
These isometries depend only on the full measurement families, not on the incoming state or the inputs $x_i,y_i$. The rigidity bound holds for every incoming state with its corresponding deficit, including strategies with large deficit.
We also recall the state $\ket{\phi}_{S^\A_i S^\B_i}$ and projectors $T^{x_i}_{a_i}$ from the entanglement-based implementation of $\ddP$ from Postulate Set~\ref{post:selftest}.

Denote $S_i=A_iS_i^\A S_i^\B$, with dimension $|S_i|=|\mathcal A|\,|S^\A|\,|S^\B|=|S|$ --- this is fixed and independent of $i$ since $S^\A_i$ and $S^\B_i$ are copies of the registers $S^\A, S^\B$. For each $x_i$, define the state
\begin{equation}\label{eq:replace-state}
\omega_{S_i}^{x_i}
=\sum_{a_i}\pure{a_i}_{A_i}\otimes
(T_{a_i}^{x_i}\otimes\id)\pure\phi_{S_i^\A S_i^\B}(T_{a_i}^{x_i}\otimes\id).
\end{equation}
This state is normalized because $\{T_{a_i}^{x_i}\}_{a_i}$ is projective and complete, and is classical on $A_i$. We will actually use the depolarized version of the states $\omega_{S_i}^{x_i}$ for $0<\mu<1$:
\[
\omega_{S_i}^{\pert,x_i}=(1-\pert)\omega_{S_i}^{x_i}+\pert\frac{\id_{S_i}}{|S_i|}.
\]
 Using these states, define the replacement channel $\mapP_i^{\mu,x_i}\in\CPTP{S_iJ_i^\A J_i^\B}{S_iJ_i^\A J_i^\B}$ by
\begin{equation}\label{eq:replace-noisy}
\mapP_i^{\pert,x_i}[\sigma_{S_iJ_i^\A J_i^\B}]
=\omega_{S_i}^{\pert,x_i}\otimes\tr[S_i]{\sigma_{S_iJ_i^\A J_i^\B}},
\end{equation}

The actual joint marginal on $J_i^\A J_i^\B$ is retained under $\mapP_i^{\mu,x_i}$, only the state in $S_i$ is replaced.

Finally, we fix CPTP recovery channels
\[
\mapN_i^{\A}\in\CPTP{S_i^\A J_i^\A}{M_i^\A\ancA_i},
\qquad
\mapN_i^{\B}\in\CPTP{S_i^\B J_i^\B}{M_i^\B\ancB_i},
\]
which are left inverses of the isometries $V^\A_i, V^\B_i$: $\mapN_i^\A[V_i^\A\sigma(V_i^\A)^\dagger]=\sigma$, and similarly for Bob. Such channels can be completed arbitrarily on the orthogonal complements of the isometry images.

The ideal generation channel $\mapF_i^\gen$ has the same inputs and outputs as $\mapE_i^\gen$. Its operations are as follows:
\begin{enumerate}
\item Sample $x_i$ from $\sfP_X$, store it in $X_i$, and initialize $Y_i$ to the empty value $\testQ$, as in the real generation channel.
\item Apply Alice's measurement channel $\mapM_i^{\A,x_i}$ to $M_i^\A$ and append $\ket0_{\ancB_i}$ on Bob's side. Then apply $V_i^\A,V_i^\B$ to $M^\A_iM'^\A_i$ and $M^\B_iM'^\B_i$ respectively, getting the registers $A_iS_i^\A J_i^\A S_i^\B J_i^\B$.
\item Apply $\mapP_i^{\pert,x_i}$ from \eqref{eq:replace-noisy}, then apply $\mapN_i^\A\otimes\mapN_i^\B$. The registers are now $A_iM_i^\A\ancA_iM_i^\B\ancB_i$.
\item Apply the real continuation channels  $\mapL_i^{\A,x_i}\otimes\mapL_i^{\B,\gen}$  from \ref{eq:EAB-form-2}. These produce $A_iB_iM_{i+1}^\A M_{i+1}^\B$, with $B_i$ fixed to $\testQ$.
\item Apply the recording operation $\mapR_i$ with $u_i=\gen$, appending $C_i=\gen$. The output is on $X_iY_iA_iB_iC_iM_{i+1}^\A M_{i+1}^\B$.
\end{enumerate}
The register $A_i$ remains classical throughout, since the replacement state is classical on $A_i$ and the recovery channels act only on the $S_i,J^\A_iJ^\B_i$ registers. If the replacement in step~3 is omitted, each extraction isometry cancels with its recovery channel in step~3, by the left-inverse identities for $\mapN_i^\A, \mapN_i^\B$. Bob's added ancilla is discarded by $\mapL_i^{\B,\gen}$. The resulting channel is therefore exactly $\mapE_i^\gen$.

On test rounds no replacement is made. The full ideal channel is
\[
\mapF_i=\ptest\pure{\test}_{U_i}\otimes\mapE_i^\test
+(1-\ptest)\pure{\gen}_{U_i}\otimes\mapF_i^\gen.
\]
It has the same domain $M_i^\A M_i^\B$ and codomain $U_iX_iY_iA_iB_iC_iM_{i+1}^\A M_{i+1}^\B$ as $\mapE_i$. As in the real process, the channels act identically on earlier $\bY_{<i}\bB_{<i}$ when these are included in the continuing memory for Theorem~\ref{th:main}.

For any normalized initial state $\rho^0_{\memA_1\memB_1R}$, including the external side information in $R$, the sequential comparison theorem's two output states are therefore
\[
\realstate=\bmapE[\rho^0],
\qquad
\idealstate=\bmapF[\rho^0],
\]
with the identity on $R$ understood. We write $p_\Acc=\Pr_{\realstate}[\Acc]$.

\subsection{Single-round bound}
\label{subsec:single-round bound}

In order to apply Theorem~\ref{th:main}, we need to upper bound the $\Dsharp_2$ divergence between the real and ideal channels $\mapE_i$ and $\mapF_i$ we have just defined. The $\test$ components of both channels are identical, so only the $\gen$ components of the channels matter in the $\Dsharp_2$. This subsection bounds $\Dsharp_2(\mapE_i^\gen[\pure\rho]\Vert\mapF_i^\gen[\pure\rho])$ for a certain class of test states $\ket\rho$.

We assume Postulate Sets~\ref{post:selftest} and~\ref{post:shield} throughout this subsection. We first specify the game strategy associated with a round and its input state. 
Consider some round $i$ and a normalized purified input $\ket\rho_{RM_i^\A M_i^\B}$. The projectors $H_{i,a_i}^{x_i}$ and $K_{i,b_i}^{y_i}$ from \eqref{eq:EAB-form-2} act on $M_i^\A\ancA_i$ and $M_i^\B\ancB_i$, respectively. Consider the strategy 
\[
\strat=\left(
\tr[R]{\pure\rho}\otimes\pure0_{\ancA_i}\otimes\pure0_{\ancB_i},
\{H_{i,a_i}^{x_i}\}_{x_i,a_i},
\{K_{i,b_i}^{y_i}\}_{y_i,b_i}
\right).
\]
for $\game$, and its deficit $\delta=\delta_\game(\strat)$. Thus the shared state in $\strat$ is our reference state $\ket{\rho}$, with both ancillas appended so that the measurement families from our channels can act on it.

The decomposition \eqref{eq:EAB-form-2} identifies the test statistics of this strategy with those of the real instruments. Indeed, the channels $\mapL_i^{\A,x_i,a_i}$ and $\mapL_i^{\B,y_i,b_i}$ are trace preserving, and the classical outcomes $a_i,b_i$ are retained. Hence, when the real round channel is applied to the trial input, every test input pair in the support of $\sfP_{XY}$ satisfies
\[
\Pr[A_i=a_i,B_i=b_i\mid X_i=x_i,Y_i=y_i,U_i=\test]\\
 =
\norm{(\id_R\otimes H_{i,a_i}^{x_i}\otimes K_{i,b_i}^{y_i})
(\ket\rho\otimes\ket0_{\ancA_i}\otimes\ket0_{\ancB_i})}_2^2.
\]
The right side is the outcome probability in $\strat$. Since the input distribution on the test round is $\sfP_{XY}$, the definition of the game deficit gives
\[
\Pr[C_i=1\mid U_i=\test]=\omega^\ast-\delta.
\]

With this, the main result we will need can be stated as:
\begin{lemma}[Single-round bound]\label{lem:round}
Suppose Postulate Sets~\ref{post:selftest} and~\ref{post:shield} hold, and $0<\pert<1$. For every normalized purified input $\ket\rho_{RM_i^\A M_i^\B}$, let $\delta$ be the deficit of the strategy $\strat$ using $\ket{\rho}$. Then
\begin{gather*}
\mapE_i^\gen[\pure\rho]
\leq\frac{|S|^2}{\pert}\mapF_i^\gen[\pure\rho],\\
\Dsharp_2\left(\mapE_i^\gen[\pure\rho]\middle\Vert
\mapF_i^\gen[\pure\rho]\right)
\leq\log a_\pert^{-1}+\vanf_{\pert,\game}(\delta),
\end{gather*}
where $a_\pert=1-\pert+\pert/|S| \in (0,1)$ and the detailed expression for $\vanf_{\pert,\game}$ is given in \eqref{eq:round-cost}. The function $\vanf_{\pert,\game}$ is vanishing for every fixed $\pert>0$.
\end{lemma}

\begin{proof}
The proof mainly consists of building on Lemma~\ref{lem:depol} by establishing a slightly tedious sequence of continuity and data-processing properties; we defer it to Appendix~\ref{app:sharpdivbounds}.
\end{proof}

\subsection{Bounding the accumulated comparison exponent}
\label{subsec:bound_compexp}

The real channel $\mapE_i$ in \eqref{eq:E_i-form} and the channel $\mapF_i$ defined above agree on test rounds and induce the same distribution on $C_i$ for every input. On generation rounds both set $C_i=\gen$, so every trial distribution satisfies $\sfP_{C_i}(\gen)=1-\ptest$. Lemma~\ref{lem:round} verifies the domination condition with $\Lambda=|S|^2/\pert$, including after adjoining an arbitrary reference. The order of operations in Protocol~\ref{prot:DI} and the storage of $C_iU_iX_iA_i$ outside the shielded device supply the protection condition for $C_iO_i$. Finally, $S_\Acc$ in \eqref{eq:accept-ball} is nonempty, closed, convex, and contains a full-support distribution. Hence the register identification above, with $S_\Omega=S_\Acc$, verifies every condition in Postulate Set~\ref{posts:seq}.

We first simplify $\compexp{2}$ slightly: currently, it is a rather complicated optimization across all the channels at once. As is a typical tactic in DI security proofs, we relax it slightly to a simpler form, in terms of a single ``large'' channel $\mapE_\oplus$ that is essentially a direct sum of all the $\mapE_i$ channels (and analogously for the $\mapF_i$ channels), up to some domain or codomain identifications. 
This reduction uses only the agreement on test rounds and the common generation probability; it does not use any rigidity properties. 
While it is rather tedious to formalize, the result should be intuitive enough (see Lemma~\ref{lem:compexp-sum} for the main conclusion).

Explicitly, we shall construct channels $\mapE_\oplus,\mapF_\oplus$ on the direct sum of the input spaces $M_i$, with the orthogonal summands labelled by a classical register $I$. The output space is the direct sum of the spaces $C_iO_iM_{i+1}$, with the same label $I$ retained and all score registers $C_i$ identified with $C$. Both channels first remove coherence between different summands, in the sense that letting $\tau_{ii}$ denotes the diagonal block of an input $\tau$ on the summand $M_i$, their actions are
\[
\mapE_\oplus[\tau]=\sum_{i=1}^n\pure{i}_I\otimes\mapE_i[\tau_{ii}],
\qquad
\mapF_\oplus[\tau]=\sum_{i=1}^n\pure{i}_I\otimes\mapF_i[\tau_{ii}].
\]

We embed the reference spaces $R_i$ into a common reference $R$, on which both channels act identically. We consider normalized input states of the form
\[
\sigma^\oplus=\frac1n\sum_{i=1}^n\pure{i}_I\otimes\pure{\sigma^i}.
\]
The corresponding output states are
\begin{align*}
\sigma^{\oplus,\mapE}&=(\mapE_\oplus\otimes\idmap_R)[\sigma^\oplus]
=\frac1n\sum_{i=1}^n\pure{i}_I\otimes\sigma^{i,\mapE_i},\\
\sigma^{\oplus,\mapF}&=(\mapF_\oplus\otimes\idmap_R)[\sigma^\oplus]
=\frac1n\sum_{i=1}^n\pure{i}_I\otimes\sigma^{i,\mapF_i}
\end{align*}
where $\sigma^{i,\mapE_i}$ and $\sigma^{i,\mapF_i}$ are defined in the same way as Section~\ref{sec:compexp}. Their common distribution on $C$ is $\overline{\sfP}_C=\frac1n\sum_i\sfP_{C_i}$. We also use $\overline{\sfQ}_C \in S_\Acc$ to denote $\frac1n\sum_i\sfQ_{C_i}$, for any feasible distributions in \eqref{eq:rate}.

We first prove the following lemma, upper bounding the comparison exponent in terms of divergences and distributions on the enlarged states.
\begin{lemma}[Direct-sum comparison exponent]\label{lem:compexp-sum}
With $\sigma^\oplus, \sigma^{\oplus,\mapE}, \sigma^{\oplus,\mapF}, \overline{\sfP}_C, \overline{\sfQ}_C$ as defined above,
\[
\compexp{2}
\leq n\sup_{\substack{\sigma^\oplus\\\overline{\sfQ}_C\in S_\Acc}}
\left[\overline{\sfQ}_C(\gen)\Dgen-D(\overline{\sfQ}_C\Vert\overline{\sfP}_C)\right].
\]
The optimization over $\sigma^\oplus$ should be interpreted to mean optimization over states of the specific form $\frac1n\sum_i\pure{i}\otimes\pure{\sigma^i}$.
\end{lemma}

\begin{proof}
Because the real and ideal states coincide in each test block, the objective in \eqref{eq:rate}, divided by $n$, is
\begin{equation}\label{eq:payoff}
\frac1n\sum_i\sfQ_{C_i}(\gen)\Dgeni{i}-\frac1n\sum_iD(\sfQ_{C_i}\Vert\sfP_{C_i}).
\end{equation}
Moreover, note that 
\begin{equation}\label{eq:D2-directsum}
\Dsharp_2\left(\sigma^{\oplus,\mapE}_{|\gen}\middle\Vert\sigma^{\oplus,\mapF}_{|\gen}\right)
=\log\left(\frac1n\sum_{i=1}^n2^{\Dgeni{i}}\right),
\end{equation}
via the classical direct-sum identity in Fact~\ref{th:sharp-facts}\ref{prop:sharp-direct-sum}.

We introduce the joint distributions
\[
\overline{\sfP}_{IC}(i,c)=\frac1n\sfP_{C_i}(c),
\qquad
\overline{\sfQ}_{IC}(i,c)=\frac1n\sfQ_{C_i}(c).
\]
Both have a uniform marginal on $I$, and their marginals on $C$ are the $\overline{\sfP}_C$ and $\overline{\sfQ}_C$ in the lemma statement.
The chain rule for KL divergence therefore gives
\begin{align}\label{eq:joint-KL}
\frac1n\sum_iD(\sfQ_{C_i}\Vert\sfP_{C_i})
& =D(\overline{\sfQ}_C\Vert\overline{\sfP}_C)
+\sum_c\overline{\sfQ}_C(c)D(\overline{\sfQ}_{I|c}\Vert\overline{\sfP}_{I|c}) \nonumber \\
& \geq D(\overline{\sfQ}_C\Vert\overline{\sfP}_C)
+\overline{\sfQ}_C(\gen)
D(\overline{\sfQ}_{I|\gen}\Vert\overline{\sfP}_{I|\gen}).
\end{align}

For $\overline{\sfQ}_C(\gen)>0$, then by definition of $\overline{\sfQ}_{IC}$,
\[
\overline{\sfQ}_{I|\gen}(i)
=\frac{\overline{\sfQ}_{IC}(i,\gen)}{\overline{\sfQ}_C(\gen)}
=\frac{\sfQ_{C_i}(\gen)}{n\overline{\sfQ}_C(\gen)}.
\]
Consequently, the first sum in \eqref{eq:payoff} can be rewritten as
\[
\frac1n\sum_i\sfQ_{C_i}(\gen)\Dgeni{i}
=\overline{\sfQ}_C(\gen)\sum_i\overline{\sfQ}_{I|\gen}(i)\Dgeni{i}.
\]
Since every round in the protocol has the same test and generation probabilities $\overline{\sfP}_{I|\gen}(i) = \frac1n$ for all $i$. Expanding the second term in \eqref{eq:joint-KL} and using the uniformity of $\overline{\sfP}_{I|\gen}$ then gives
\[
D(\overline{\sfQ}_{I|\gen}\Vert\overline{\sfP}_{I|\gen})
=\sum_{i:\overline{\sfQ}_{I|\gen}(i)>0}\overline{\sfQ}_{I|\gen}(i)\log\bigl(n\overline{\sfQ}_{I|\gen}(i)\bigr).
\]
Therefore:
\begin{align}
& \sum_i\overline{\sfQ}_{I|\gen}(i)\Dgeni{i}
-D(\overline{\sfQ}_{I|\gen}\Vert\overline{\sfP}_{I|\gen}) \nonumber \\
&=\sum_{i:\overline{\sfQ}_{I|\gen}(i)>0}\overline{\sfQ}_{I|\gen}(i)
\log\left(\frac{2^{\Dgeni{i}}}{n\overline{\sfQ}_{I|\gen}(i)}\right) \nonumber\\
&\leq\log\left(
\sum_{i:\overline{\sfQ}_{I|\gen}(i)>0}\overline{\sfQ}_{I|\gen}(i)
\frac{2^{\Dgeni{i}}}{n\overline{\sfQ}_{I|\gen}(i)}
\right) \nonumber\\
&=\log\left(\frac1n
\sum_{i:\overline{\sfQ}_{I|\gen}(i)>0}2^{\Dgeni{i}}\right) \nonumber\\
&\leq\log\left(\frac1n\sum_i2^{\Dgeni{i}}\right)
=\Dgen \label{eq:D2-directsum-ineq}
\end{align}
where the first inequality is Jensen's, for the concave logarithm function, and the last equality uses \eqref{eq:D2-directsum}.

We now plug in \eqref{eq:D2-directsum-ineq} into \eqref{eq:payoff}, using \eqref{eq:joint-KL}:
\begin{align}
&\frac1n\sum_i\left[\sfQ_{C_i}(\gen)\Dgeni{i}
-D(\sfQ_{C_i}\Vert\sfP_{C_i})\right]\nonumber\\
&\leq\overline{\sfQ}_C(\gen)
\left[\sum_i\overline{\sfQ}_{I|\gen}(i)\Dgeni{i}
-D(\overline{\sfQ}_{I|\gen}\Vert\overline{\sfP}_{I|\gen})\right]
-D(\overline{\sfQ}_C\Vert\overline{\sfP}_C)\nonumber\\
&\leq\overline{\sfQ}_C(\gen)\Dgen
-D(\overline{\sfQ}_C\Vert\overline{\sfP}_C).
\label{eq:block-payoff}
\end{align}
If instead $\overline{\sfQ}_C(\gen)=0$, then every $\sfQ_{C_i}(\gen)=0$. The generation contribution in~\eqref{eq:payoff} therefore vanishes, and moreover, from \eqref{eq:joint-KL},
\[
-\frac1n\sum_iD(\sfQ_{C_i}\Vert\sfP_{C_i})
\leq-D(\overline{\sfQ}_C\Vert\overline{\sfP}_C),
\]
which is exactly \eqref{eq:block-payoff} in this case, without needing to define $\overline{\sfQ}_{I|\gen}$.

Taking the supremum over the trial inputs and distributions in \eqref{eq:block-payoff} proves the lemma.
\end{proof}

Essentially, the simplified upper bound on $\compexp{2}$ from the above lemma is optimizing over ``direct-sum convex combinations'' of states across the various protocol rounds --- a roughly analogous concept was introduced in~\cite{arx_AHT24} under the term ``convex range''. 
With this simplified form, we now prove we can make the comparison exponent arbitrarily small, using suitable parameter choices.

\begin{theorem}[Comparison exponent with arbitrarily small linear dependence at constant tolerated noise]\label{th:small-rate}
Assume Postulate Sets~\ref{post:selftest} and~\ref{post:shield}. For any $k'>0$ and any $\pert\in(0,1-2^{-k'/4})$, one can choose constants
\[
0<\epstol<1,
\qquad
\frac12<\ptest<1,
\]
independently of $n$ and the physical memory dimensions, such that
\begin{equation}\label{eq:small-rate}
\compexp{2}\leq k'(1-\ptest)n.
\end{equation}
Consequently, for every $p_\Acc>0$,
\begin{equation}\label{eq:small-rate-div}
D_2\left(\realstate_{|\Acc}\middle\Vert\idealstate\right)
\leq k'(1-\ptest)n+2\log\frac1{p_\Acc}.
\end{equation}
\end{theorem}
\begin{proof}
To simplify the presentation, we first declare that we will eventually choose $\epstol$ satisfying $\epstol\leq1-\ptest$ (this will indeed be explicitly enforced in our final parameter choices below). Then for any feasible point in the upper bound from Lemma~\ref{lem:compexp-sum}, we have
\begin{equation}\label{eq:qgen}
\overline{\sfQ}_C(\gen)\leq1-\ptest+\epstol\leq2(1-\ptest),
\end{equation}
since the acceptance condition gives $\trnorm{\overline{\sfQ}_C-\sfP_C}\leq\epstol$.
With this, we factor $1-\ptest$ out of the bound in Lemma~\ref{lem:compexp-sum} to obtain
\begin{align}
\compexp{2}&\leq(1-\ptest)n
\sup_{\substack{\sigma^\oplus\\\overline{\sfQ}_C\in S_\Acc}}
\left[\frac{\overline{\sfQ}_C(\gen)}{1-\ptest}\Dgen
-\frac{1}{1-\ptest} D(\overline{\sfQ}_C\Vert\overline{\sfP}_C)\right]\nonumber\\
&\leq(1-\ptest)n
\sup_{\substack{\sigma^\oplus\\\overline{\sfQ}_C\in S_\Acc}}
\left[2\,\Dgen
-\frac{1}{1-\ptest}D(\overline{\sfQ}_C\Vert\overline{\sfP}_C)\right].\label{eq:block-rate-normalized}
\end{align}
We will choose the parameters so that the objective function in this supremum is at most $k'$ for every feasible point in the optimization. The intuition is simple: we split the feasible points into two cases, those where $\overline{\sfP}_C$ is ``near'' $S_\Acc$ and those where it is not. For the former case, $\Dgen$ can be bounded by a small value via the rigidity and continuity results in Lemma~\ref{lem:round}, while $-D(\overline{\sfQ}_C\Vert\overline{\sfP}_C)$ can be trivially upper bounded by zero. For the latter, we upper bound $\Dgen$ by a ``trivial'' worst-case value, then choose $\gamma$ sufficiently close to $1$ to ensure that the $-\frac{1}{1-\ptest}D(\overline{\sfQ}_C\Vert\overline{\sfP}_C)$ term still makes the objective function sufficiently small.

We now formalize this intuition.
For each input state $\ket{\sigma^i}$, define
\[
\delta_i=\omega^\ast-\frac{\sfP_{C_i}(1)}{\ptest}.
\]
Here $\sfP_{C_i}(1)/\ptest$ is the probability that Alice and Bob win $\game$ when the test round is chosen: their inputs have the prescribed game distribution, and they apply their respective test measurements to this input state, with the local ancillas initialized as in $\mapE_i$. In particular, this is the deficit of the strategy to which Lemma~\ref{lem:round} applies, i.e.~we can write
$2^{\Dgeni{i}}\leq B_{\pert,\game}(\delta_i)$ where $B_{\pert,\game}$ is a function of the form
\[
B_{\pert,\game}(\delta)
= a_\pert^{-1}\, 2^{\vanf_{\pert,\game}(\delta)} \qquad \text{for } 0\leq\delta\leq\omega^\ast,
\]
and we extend $B_{\pert,\game}$ constantly beyond $\delta \geq \omega^\ast$ to cover the full interval $[0,1]$ (i.e.~define $B_{\pert,\game}(\delta) = B_{\pert,\game}(\omega^\ast)$ for $\delta \in [\omega^\ast,1]$).

For every fixed $\pert>0$, the function $B_{\pert,\game}$ satisfies
\[
B_{\pert,\game}(0)=
\lim_{\delta\to0^+}B_{\pert,\game}(\delta)=a_\pert^{-1},
\]
and is bounded and nondecreasing on $[0,1]$. If it is not already concave, these properties suffice (see Remark~\ref{rem:concavemajorant} for details) to relax it to a {concave} upper bound $\widehat B_{\pert,\game}$ that inherits the properties
\begin{equation} \label{eq:majorantlimits}
\widehat B_{\pert,\game}(0)=
\lim_{\delta\to0^+}\widehat B_{\pert,\game}(\delta)=a_\pert^{-1}.
\end{equation}

Since $\sfP_C(1)=\ptest\omega^\ast$, if $\trnorm{\overline{\sfP}_C-\sfP_C}\leq\Delta$ for some value $\Delta$ and $\ptest\geq1/2$, then
\[
\frac1n\sum_i\delta_i
=\frac{\sfP_C(1)-\overline{\sfP}_C(1)}{\ptest}
\leq2\Delta.
\]
When $2\Delta\leq1$, we get
\begin{align}
\Dgen&=\log\left(\frac1n\sum_i2^{\Dgeni{i}}\right)
\leq\log\left(\frac1n\sum_i\widehat B_{\pert,\game}(\delta_i)\right)\nonumber\\
&\leq\log\widehat B_{\pert,\game}\left(\frac1n\sum_i\delta_i\right)
\leq\log\widehat B_{\pert,\game}(2\Delta)\label{eq:flagged-rigidity}
\end{align}
by applying the individual bounds $2^{\Dgeni{i}}\leq B_{\pert,\game}(\delta_i)\leq\widehat B_{\pert,\game}(\delta_i)$, followed by concavity and monotonicity of $\widehat B_{\pert,\game}$.

We now choose the remaining parameters for the prescribed $\pert$. Since $a_\pert=1-\pert+\pert/|S|\geq1-\pert>2^{-k'/4}$, we have $\log a_\pert^{-1}<k'/4$. We then choose $0<\Delta\leq1/2$ sufficiently small that
\begin{equation}\label{eq:near-choice}
\log\widehat B_{\pert,\game}(2\Delta)\leq\frac{k'}{2}.
\end{equation}
Such a choice is possible because $\widehat B_{\pert,\game}(\delta)\to a_\pert^{-1}$ as $\delta\to0^+$. Finally, we choose
\begin{equation}\label{eq:choice}
0<1-\ptest<\min\left\{\frac12,
\frac{\Delta^2}{4\ln2\,\log(|S|^2/\pert)}\right\},
\qquad
0<\epstol\leq\min\left\{\frac{1-\ptest}{2},\frac\Delta2\right\}.
\end{equation}
All these choices are independent of $n$ and of the physical memory dimensions.

With these choices, the desired result follows by a case analysis over all feasible points $\sigma^\oplus,\overline{\sfQ}_C$ in \eqref{eq:block-rate-normalized}. In the first case, we suppose that $\trnorm{\overline{\sfP}_C-\sfP_C}\leq\Delta$. By non-negativity of relative entropy and \eqref{eq:flagged-rigidity} and \eqref{eq:near-choice}, we obtain
\begin{align}
2\,\Dgen-\frac{D(\overline{\sfQ}_C\Vert\overline{\sfP}_C)}{1-\ptest}
&\leq2\,\Dgen\label{eq:near-case}\\
&\leq2\log\widehat B_{\pert,\game}(2\Delta)\leq k'.\nonumber
\end{align}

The remaining case is $\trnorm{\overline{\sfP}_C-\sfP_C}>\Delta$. The triangle inequality and acceptance imply
\[
\trnorm{\overline{\sfQ}_C-\overline{\sfP}_C}>\Delta-\epstol\geq\frac\Delta2,
\]
and thus by Pinsker's inequality:
\begin{equation}\label{eq:far-penalty}
D(\overline{\sfQ}_C\Vert\overline{\sfP}_C)
\geq\frac{2}{\ln2}\trnorm{\overline{\sfQ}_C-\overline{\sfP}_C}^2
\geq\frac{\Delta^2}{2\ln2}.
\end{equation}
Recall that Lemma~\ref{lem:round} also gives a ``trivial'' universal bound of $\log\frac{|S|^2}{\pert}$ on every $\Dgeni{i}$, and thus also on $\Dgen$. Therefore, the parameter choice \eqref{eq:choice} gives
\[
2\,\Dgen-\frac{D(\overline{\sfQ}_C\Vert\overline{\sfP}_C)}{1-\ptest}
\leq2\log\frac{|S|^2}{\pert}-\frac{\Delta^2}{2(1-\ptest)\ln2}
\leq0
\]
by the choice of $\gamma$.

In summary, in both cases the expression inside the supremum in \eqref{eq:block-rate-normalized} is at most $k'$, giving the desired bound~\eqref{eq:small-rate}. Theorem~\ref{th:main} at order $2$ gives \eqref{eq:small-rate-div}.
\end{proof}

\begin{remark}
\label{rem:concavemajorant}
While it is intuitive that we can upper bound $B_{\pert,\game}$ with some concave function $\widehat B_{\pert,\game}$ (on $[0,1]$), it may be helpful to see an explicit choice. One option is to choose the latter to be the least concave majorant of the former, also sometimes known as the \emph{concave envelope}. Explicitly, this means we set
\[
\widehat B_{\pert,\game}(\delta)
=\sup_{\substack{\lambda_j\geq0,\ \delta_j\in[0,1]\\
\sum_j\lambda_j=1,\ \sum_j\lambda_j\delta_j=\delta}}
\sum_j\lambda_j B_{\pert,\game}(\delta_j),
\]
where the supremum ranges over all finite mixtures. From the definition it is immediate that $\widehat B_{\pert,\game}\geq B_{\pert,\game}$. Combining two finite mixtures shows that $\widehat B_{\pert,\game}$ is concave. Conversely, Jensen's inequality shows that every concave majorant of $B_{\pert,\game}$ bounds all the averages in this supremum, which proves the claimed least-majorant property. The function $\widehat B_{\pert,\game}$ is also nondecreasing: a mixture with mean less than $1$ can be combined with the deficit $1$ to obtain any larger mean, and this does not decrease the average because $B_{\pert,\game}(1)$ is the largest value of $B_{\pert,\game}$.

Also, the majorant preserves the value and limit at $\delta=0$ as claimed: for the value, note that the only mixtures satisfying $\sum_j\lambda_j\delta_j=0$ are those with all $\delta_j=0$. For the limit, note that as $\sum_j\lambda_j\delta_j$ tends to zero, Markov's inequality shows that the total weight on $\delta_j$ values bounded away from zero tends to zero, uniformly over the mixtures. On the remaining sufficiently small $\delta_j$ values, $B_{\pert,\game}$ is arbitrarily close to $a_\pert^{-1}$, while its boundedness controls the contribution from the other terms. Consequently, the properties in~\eqref{eq:majorantlimits} hold.
\end{remark}

\subsection{Finishing the soundness proof}

We fix an allowed DI attack, including Bob's operations in the cryptographic continuation. To finish the soundness proof, we define a common final measurement channel $\mapM$ on the outputs of the real and ideal processes as follows.
\begin{enumerate}[leftmargin=*]
\item If fewer than $m$ of the registers $U_i$ equal $\gen$, terminate and output $0$. Otherwise, let $E$ contain the first $m$ indices with $U_i=\gen$.
\item Discard Alice's device memory, her records outside $E$, and the test records $\bCS$. Retain $\bX_E,\bA_E$ for Alice and Bob's memory and copies of his classical transcript for Bob.
\item Give Bob $E$ and the message that Alice accepted, and run the DD continuation with the fixed attack. Output $1$ if the completed experiment satisfies $\Bad$ and $0$ otherwise.
\end{enumerate}
The channel $\mapM$ agrees with the actual continuation on $\Acc$. On other inputs it is defined without checking the acceptance condition. In particular, applying it to $\idealstate$ does not condition the ideal process on acceptance.

\begin{lemma}[Reduction to the DD protocol]\label{lem:DD}
Under Postulate Sets~\ref{post:shield} and~\ref{post:DD}, the channel $\mapM$ defined above satisfies
\[
\Pr_{\mapM[\idealstate]}[\Bad]\leq q_{\mathrm{dep}}(\pert)^m.
\]
\end{lemma}
\begin{proof}
We condition on $U_1=u_1,\ldots,U_n=u_n$. If fewer than $m$ values equal $\gen$, then $\mapM$ outputs $0$. Otherwise, $E$ is fixed.

For every $i$, we use the abbreviation $R_i=R\bX_{E\cap[i-1]}\bA_{E\cap[i-1]}\bY_{<i}\bB_{<i}$, retaining any initial external side information in $R$. We will compute the reduced map in $\CPTP{R_iM_i^\B}{R_{i+1}M_{i+1}^\B}$ induced by $\mapF_i$, and show that it uses only Bob's local operations together with one preparation of \eqref{eq:noisy-source} when $i\in E$.

We consider an arbitrary input $\sigma_{R_iM_i^\A M_i^\B}$ and define the CPTP map
\[
\mathcal J_i^\B[\xi_{M_i^\B}]
=
\tr[S_i^\B]{
V_i^\B(\xi_{M_i^\B}\otimes\pure0_{\ancB_i})(V_i^\B)^\dagger
}.
\]
The three cases for $i$ are as follows.
\begin{enumerate}[leftmargin=*]
\item \emph{$i\in E$.} We have
\begin{align*}
\tr[M_{i+1}^\A C_i]{\mapF_i^\gen[\sigma]} =
\pure{\testQ}_{Y_i}\otimes
(\mapL_i^{\B,\gen}\circ\mapN_i^\B)
\left[
\rho^{\mathrm{prep},\pert}_{X_iA_iS_i^\B}
\otimes
\mathcal J_i^\B[\sigma_{R_iM_i^\B}]
\right].
\end{align*}
For fixed $x_i$, write the state before replacement as
\[
\zeta^{x_i}=(V_i^\A\otimes V_i^\B)
\left(\mapM_i^{\A,x_i}[\sigma]\otimes\pure0_{\ancB_i}\right)
(V_i^\A\otimes V_i^\B)^\dagger.
\]
Its marginal satisfies
\begin{align*}
\zeta^{x_i}_{R_iJ_i^\B}
&=\tr[S_i^\B]{V_i^\B
\left(\tr[A_iM_i^\A\ancA_i]{\mapM_i^{\A,x_i}[\sigma]}\otimes\pure0_{\ancB_i}\right)(V_i^\B)^\dagger}\\
&=\mathcal J_i^\B[\sigma_{R_iM_i^\B}],
\end{align*}
where the first equality uses that $V_i^\A$ is an isometry, and the second uses that $\mapM_i^{\A,x_i}$ is trace preserving. Replacement discards the original $A_i$ and prepares a new $A_i$ as part of $\omega_{S_i}^{\pert,x_i}$. Consequently,
\[
\sum_{x_i}\sfP_X(x_i)\pure{x_i}_{X_i}\otimes
\tr[S_i^\A J_i^\A]{\mapP_i^{\pert,x_i}[\zeta^{x_i}]}
=\rho^{\mathrm{prep},\pert}_{X_iA_iS_i^\B}\otimes\mathcal J_i^\B[\sigma_{R_iM_i^\B}].
\]
Finally, for every $x_i,a_i$, trace preservation gives
\[
\Tr_{M_{i+1}^\A}\circ\mapL_i^{\A,x_i,a_i}\circ\mapN_i^\A
=\Tr_{S_i^\A J_i^\A}.
\]
Applying this identity in each classical $x_i,a_i$ block proves the displayed reduced action of $\mapF_i^\gen$.

Thus the right-hand side consists of one DD preparation and CPTP maps on Bob's registers.

\item \emph{$u_i=\gen$ and $i\notin E$.} In this case, Alice's input-output records are discarded. Tracing out $X_iA_i$ in the preceding identity gives
\begin{align*}
\tr[X_iA_iM_{i+1}^\A C_i]{\mapF_i^\gen[\sigma]} =
\pure{\testQ}_{Y_i}\otimes
(\mapL_i^{\B,\gen}\circ\mapN_i^\B)
\left[
\tr[X_iA_i]{\rho^{\mathrm{prep},\pert}_{X_iA_iS_i^\B}}
\otimes
\mathcal J_i^\B[\sigma_{R_iM_i^\B}]
\right].
\end{align*}
The state $\tr[X_iA_i]{\rho^{\mathrm{prep},\pert}_{X_iA_iS_i^\B}}$ is fixed, so Bob can prepare it locally.

\item \emph{$u_i=\test$.} The ideal process applies $\mapE_i^\test$, and
\[
\tr[X_iA_iC_iM_{i+1}^\A]{\mapE_i^\test[\sigma]}
=
\sum_{y_i}\sfP_Y(y_i)\pure{y_i}_{Y_i}
\otimes
\mapE_i^{\B,\test,y_i}[\sigma_{R_iM_i^\B}].
\]
The trace removes Alice's trace-preserving map $\mapE_i^{\A,x_i}$, and
$\sum_{x_i}\sfP_{XY}(x_i,y_i)=\sfP_Y(y_i)$.
The right-hand side is therefore a local channel for Bob.
\end{enumerate}

Composing these identities implements the required marginal of the ideal process using $m$ independent copies of \eqref{eq:noisy-source} and Bob's interspersed channels. The remaining action of $\mapM$ is the prescribed DD continuation, with $E$ and the acceptance message fixed. Postulate Set~\ref{post:DD} bounds its bad-event probability by $q_{\mathrm{dep}}(\pert)^m$. Averaging over $u_1,\ldots,u_n$ proves the claim.
\end{proof}

Finally, we are ready to finish our proof of Theorem~\ref{th:DI}.
\begin{proof}[Proof of Theorem~\ref{th:DI}]
For the given $k$, we set $k'=k/4$ and choose $0<\pert<1-2^{-k/16}$; for example, $\pert=(1-2^{-k/16})/2=\Theta(k)$ as $k\to0$. Theorem~\ref{th:small-rate}, with these $k'$ and $\pert$, gives $\ptest$ and $\epstol$. Decreasing the tolerance if necessary ensures $\epstol\leq(1-\ptest)/2$ while retaining
\[
\compexp{2}\leq\frac{k(1-\ptest)n}{4}.
\]
For $(1-\ptest)n\geq4$, the prescribed choice $m=\floor{(1-\ptest)n/2}$ gives $m\geq2$ and $(1-\ptest)n/2<m+1$. Consequently,
\[
\compexp{2}\leq\frac{k(1-\ptest)n}{4}
\leq\frac{k(m+1)}{2}\leq km,
\]
so Theorem~\ref{th:main} at order $2$ gives 
the bound we need
on $D_2(\realstate_{|\Acc}\Vert\idealstate)$ whenever $p_\Acc>0$.

The deterministic argument at the end of the proof of Theorem~\ref{th:complete}, together with $\epstol\leq(1-\ptest)/2$, ensures that $E$ exists on every accepted execution. Thus the $\mapM$ we have just defined agrees there with the actual cryptographic continuation.

If $p_\Acc=0$, the claimed security bound is immediate. Otherwise, data processing under $\mapM$ gives
\[
D_2\left(
\mapM[\realstate_{|\Acc}]
\middle\Vert
\mapM[\idealstate]
\right)
\leq
\compexp{2}+2\log\frac1{p_\Acc}.
\]
The event that $\mapM$ outputs $1$ has probability $\Pr_{\diP}[\Bad\mid\Acc]$ on the first state and $\Pr_{\mapM[\idealstate]}[\Bad]$ on the second. Lemma~\ref{lem:probability}(i) therefore gives
\[
\Pr_{\diP}[\Bad\mid\Acc]
\leq
\frac{2^{\compexp{2}/2}}{p_\Acc}
\sqrt{\Pr_{\mapM[\idealstate]}[\Bad]}.
\]
Using
$\Pr_{\diP}[\Acc\land\Bad]
=p_\Acc\Pr_{\diP}[\Bad\mid\Acc]$
and applying Lemma~\ref{lem:DD}, we obtain
\begin{equation}\label{eq:DI-transfer}
\Pr_{\diP}[\Acc\land\Bad]
\leq
2^{\compexp{2}/2}q_{\mathrm{dep}}(\pert)^{m/2}.
\end{equation}
Substituting $\compexp{2}\leq km$ proves \eqref{eq:DI-rate}.

For the final claim, we choose $0<k<-\log q$. The assumption
$\limsup_{\pert\to0}q_{\mathrm{dep}}(\pert)\leq q$
implies that $2^kq_{\mathrm{dep}}(\pert)<1$ for all sufficiently small positive $\pert$. We may therefore first choose $0<\pert<1-2^{-k/16}$ small enough to satisfy this inequality, and then apply Theorem~\ref{th:small-rate} with $k'=k/4$ and this $\pert$ to choose the remaining parameters. The example $\pert=(1-2^{-k/16})/2$ already satisfies the additional inequality for all sufficiently small $k$, so the choice $\pert=\Theta(k)$ as $k\to0$ is preserved. All these parameters are constants independent of $n$. Since $m$ grows linearly with $n$, the bound
\[
\Pr_{\diP}[\Acc\land\Bad]
\leq
\left(2^kq_{\mathrm{dep}}(\pert)\right)^{m/2}
\]
decays exponentially.
\end{proof}
\begin{remark}[Using smooth max-divergence]\label{rem:smooth}
Smoothing before conditioning on acceptance also gives
\[
\Pr_{\diP}[\Acc\land\Bad]
\leq\frac1{\esmooth^2}2^{\compexp{2}}q_{\mathrm{dep}}(\pert)^m+\esmooth.
\]
To see this, we fix $f$ in \eqref{eq:rate-tilt}. Equation~\eqref{eq:Qsharp-fchain}, the comparison between sharp and sandwiched divergences, and the smoothing bound in Fact~\ref{th:renyi-properties}\ref{th:smoothing} give a subnormalized $\tau$ within purified distance $\esmooth$ of $\realstate$, with
\[
\tau\leq\frac1{\esmooth^2}2^{\sum_iR^\#_{2,\mapE_i,\mapF_i}(f)}\bmapF^f[\rho^0].
\]
We may pinch $\tau$ in the classical basis of $\bCS$. On any accepting block, the weight introduced by $f$ is at most $2^{n\sup_{\sfQ\in S_\Acc}\sum_c\sfQ(c)f(c)}$. Consequently,
\[
\tau_{\land\Acc}\leq\frac1{\esmooth^2}2^{\sum_iR^\#_{2,\mapE_i,\mapF_i}(f)+n\sup_{\sfQ\in S_\Acc}\sum_c\sfQ(c)f(c)}\idealstate.
\]
The probability of the joint event $\Acc\land\Bad$, evaluated using $\mapM$ after the acceptance test, changes by at most $\esmooth$ when $\realstate$ is replaced by $\tau$, by Fact~\ref{th:FvdG} and Fact~\ref{lem:event-transfer}. Applying $\mapM$ to the operator inequality, using Lemma~\ref{lem:DD}, and taking the infimum over $f$ proves the displayed bound. If $\compexp{2}\leq km$, its multiplicative term is $(1/\esmooth^2)(2^kq_{\mathrm{dep}}(\pert))^m$.
\end{remark}


\section{Improving the rate under no-memory conditions}\label{sec:memoryless}
\label{sec:nomemory}

In this section, we prove our tensor-product comparison theorem, and using it, the DD-to-DI compiler with improved rate, under the memoryless condition.
The formal statements of our main results are in Theorems~\ref{th:robustselftest} and~\ref{th:scalingdemo}; we first present the necessary framework.

\subsection{``Almost-subadditive'' channel divergences}

The key idea here is to define a variant of channel divergence, denoted as $D^\downarrow_\alpha$, that minimizes over states in the second argument rather than maximizing. This allows for obtaining some properties similar to a sandwiched entropy usually denoted as $H^\uparrow_\alpha$ (see e.g.~\cite{Tom16}), such as the critical fact that its {regularized} counterpart $D^{\reg,\downarrow}_\alpha$ is upper bounded by the usual channel divergence. However, this minimization also introduces issues regarding properties that were previously ``obvious'' for the usual channel divergences (for instance, whether we have $D^\downarrow_\alpha \leq D^{\reg,\downarrow}_\alpha$), and hence we take special care to find suitable circumstances under which such properties hold, in Lemmas~\ref{lemma:measadd}--\ref{lemma:weakadd}. 

\begin{definition}
\label{def:channeldiv}
We again highlight (see Remark~\ref{rem:stab}) that our definition here does not inherently include a purifying/stabilizer register.
For any $\alpha\in[0,\infty]$ and $\mapE,\mapF \in \CPTP{\inQ}{\outQ}$,
\begin{align*}
D_\alpha\Dargs{\mapE}{\mapF} &\defvar \sup_{\rho\in\dop{=}(\inQ)} D_\alpha\Dargs{\mapE[\rho]}{\mapF[\rho]} , \nonumber\\
D^\downarrow_\alpha\Dargs{\mapE}{\mapF} &\defvar \sup_{\rho\in\dop{=}(\inQ)}
\inf_{\nu\in\dop{=}(\inQ)} D_\alpha\Dargs{\mapE[\rho]}{\mapF[\nu]},
\end{align*}
and for the latter we define its regularized version
\begin{align*}
D^{\reg,\downarrow}_\alpha\Dargs{\mapE}{\mapF} \defvar \mylimsup_{m\to\infty} \frac{1}{m} D^\downarrow_\alpha\Dargs{\mapE^{\otimes m}}{\mapF^{\otimes m}} .
\end{align*}
Since the limit superior of any sequence always exists\footnote{This is our main reason for using it rather than just the usual limit, because $D^\downarrow_\alpha$ is not obviously superadditive or subadditive and hence it is unclear if the usual limit always exists. Moreover, the choice to use $\mylimsup$ rather than $\myliminf$ here is simply to ensure an inequality direction we need in this work (Lemma~\ref{lemma:weakadd}); the other choice may be useful in other contexts.} in the extended reals, this is indeed well-defined for any $\mapE,\mapF$. 
\end{definition}

The main result of this section is the statement that $D^{\downarrow}_\alpha$ is (weakly) ``almost-subadditive'', in that the $D^{\downarrow}_\alpha$ divergence of $n$ IID copies of a channel pair {followed by a ``global'' measurement}\footnote{While one could attempt to present this result in terms of a ``measured channel divergence'' concept, one would have to use a definition where the optimization over measurements is ``outside'' the optimization over states, which is not necessarily standard.} can be upper bounded by $n$ times the usual channel divergence.
\begin{theorem}[Weak ``almost-subadditivity'' of $D^{\downarrow}_\alpha$]\label{th:almostadd} 
Consider any $n\in\mathbb{N}$, any $\alpha>1$, any channels $\mapE,\mapF \in \CPTP{\inQ}{\outQ}$, any register $\stabE$ such that $|\stabE|\geq|\inQ|$, and any measurement channel $\mapM \in \CPTP{\boutQ \bR}{F}$. Then 
\begin{align}
D^{\downarrow}_\alpha\Dargs{\mapM \circ (\mapE \otimes \idmap_{\stabE})^{\otimes n}}{\mapM \circ (\mapF \otimes \idmap_{\stabE})^{\otimes n}} 
\leq n D_\alpha\Dargs{\mapE \otimes \idmap_{\stabE}}{\mapF \otimes \idmap_{\stabE}} .\label{eq:almostadd}
\end{align}
\end{theorem}
\begin{proof}
The proof follows by chaining a series of bounds that we prove below as various lemmas, most of which are rather intuitive:
\begin{align*}
D^{\downarrow}_\alpha\Dargs{\mapM \circ (\mapE \otimes \idmap_{\stabE})^{\otimes n}}{\mapM \circ (\mapF \otimes \idmap_{\stabE})^{\otimes n}} 
&\leq D^{\reg,\downarrow}_\alpha\Dargs{\mapM \circ (\mapE \otimes \idmap_{\stabE})^{\otimes n}}{\mapM \circ (\mapF \otimes \idmap_{\stabE})^{\otimes n}} \nonumber\\
&\leq D^{\reg,\downarrow}_\alpha\Dargs{(\mapE \otimes \idmap_{\stabE})^{\otimes n}}{(\mapF \otimes \idmap_{\stabE})^{\otimes n}} \nonumber\\
&\leq n D^{\reg,\downarrow}_\alpha\Dargs{\mapE \otimes \idmap_{\stabE}}{\mapF \otimes \idmap_{\stabE}} \nonumber\\
&\leq n D_\alpha\Dargs{\mapE \otimes \idmap_{\stabE}}{\mapF \otimes \idmap_{\stabE}} ,
\end{align*}
where the first line is the fact that $D^{\downarrow}_\alpha \leq D^{\reg,\downarrow}_\alpha$ holds at least for measurement channels\footnote{Here we also use the fact that any sequence of channels ending in a measurement channel (that measures all registers produced up to that point) is itself a measurement channel.} (Lemma~\ref{lemma:measadd}), the second line is data-processing (Lemma~\ref{lemma:downDPI}), and the third line is weak subadditivity of the regularized version (Lemma~\ref{lemma:weakadd}). 

The main nonintuitive property is the final line, which states that (when stabilized) even the {regularized} version of $D^{\downarrow}_\alpha$ is upper bounded by the usual channel divergence. This was basically proven in~\cite[Lemma~3.5]{MFSR24}, but we provide a full statement and proof as Lemma~\ref{lemma:regbound}.
\end{proof}

We now provide the required lemmas for the above proof.

\begin{lemma}[Data-processing]\label{lemma:downDPI}
 For any channels $\mapE,\mapF \in \CPTP{\inQ}{Q'}$ and $\mapN \in \CPTP{Q'}{Q''}$, and any $\alpha\in[1/2,\infty]$, we have
\begin{gather*}
D^\downarrow_\alpha\Dargs{\mapN \circ \mapE}{\mapN \circ \mapF} \leq D^\downarrow_\alpha\Dargs{\mapE}{\mapF}, \nonumber\\
D^{\reg,\downarrow}_\alpha\Dargs{\mapN \circ \mapE}{\mapN \circ \mapF} \leq D^{\reg,\downarrow}_\alpha\Dargs{\mapE}{\mapF}.
\end{gather*}
\end{lemma}
\begin{proof}
This immediately follows from data-processing for the underlying state divergence, i.e.~for any feasible $\rho,\nu$ in the optimizations we have $D_\alpha\Dargs{\mapN \circ \mapE[\rho]}{\mapN \circ \mapF[\nu]} \leq D_\alpha\Dargs{\mapE[\rho]}{\mapF[\nu]}$, and this inequality is preserved by taking an infimum and also when next taking a supremum. Similarly, this inequality is preserved when taking $\mylimsup$ and hence it also holds for the regularized versions.
\end{proof}

\begin{lemma}\label{lemma:classicalinf}
For any $\alpha>1$, any classical states $\rho^1 \in \dop{=}(C_1)$, $\rho^2 \in \dop{=}(C_2)$, and measurement channels $\mapM_1 \in\CPTP{\inQ_1}{C_1}$, $\mapM_2 \in\CPTP{\inQ_2}{C_2}$ (with a consistent choice of classical basis throughout for $C_1,C_2$), we have 
\begin{align*}
\inf_{{\nu}\in\dop{=}(\inQ_1 \inQ_2)} D_\alpha\Dargs{\rho^1 \otimes \rho^2}{\mapM_1 \otimes \mapM_2[{\nu}]} =
\inf_{{\nu^1}\in\dop{=}(\inQ_1 )} D_\alpha\Dargs{\rho^1}{\mapM_1[{\nu^1}]} + \inf_{{\nu^2}\in\dop{=}(\inQ_2 )} D_\alpha\Dargs{\rho^2}{\mapM_2[{\nu^2}]}.
\end{align*}
\end{lemma}

The above lemma is an essentially straightforward consequence of~\cite[Theorem~4 and Corollary~5]{arx_RT24}, apart from technicalities to handle support issues. We first state that result in a form sufficient for our purposes (the full result covers the more general family known as $\alpha$-$z$ divergences):
\begin{fact}[Special case of {\cite[Theorem~4 and Corollary~5]{arx_RT24}}]\label{fact:renyi-opt}
Let $C$ be a finite-dimensional register, and let $\mathcal F\subseteq\dop{=}(C)$ be closed and convex and contain a full-support state. Consider $\alpha>1$, a state $\rho\in\dop{=}(C)$, and a feasible state $\tau\in\mathcal F$ satisfying $[\rho,\tau]=0$.
Then $\tau$ attains the infimum in $\inf_{\nu\in\mathcal F} D_\alpha\Dargs{\rho}{\nu}$ if and only if $\supp(\rho)\subseteq\supp(\tau)$ and
\[
\tr{\Xi\nu}\leq\kappa
\qquad\text{for every }\nu\in\mathcal F,
\]
where, whenever the support condition holds, we define
\[
\Xi=\rho^\alpha\tau^{-\alpha},
\qquad
\kappa=Q_\alpha\Dargs{\rho}{\tau}
=\tr{\rho^\alpha\tau^{1-\alpha}},
\]
with negative powers defined via the Moore--Penrose pseudoinverse.
\end{fact}

With this we can state the proof:
\begin{proof}
For $i\in\{1,2\}$, let
$\omega^i=\mapM_i[\id_{\inQ_i}/|\inQ_i|]$ and
$S_i=\supp(\omega^i)$. Every output of $\mapM_i$ is supported on
$S_i$, and every output of $\mapM_1\otimes\mapM_2$ is supported on
$S_1\otimes S_2$. Hence, if $\supp(\rho^i)\not\subseteq S_i$ for
either $i$, both sides of the claimed identity are $+\infty$.
Otherwise, restrict $C_i$ to $S_i$. The feasible output sets
\begin{align*}
\mathcal F_i
&\defvar\{\mapM_i[\nu^i]:\nu^i\in\dop{=}(\inQ_i)\},
\\
\mathcal F_{12}
&\defvar\{(\mapM_1\otimes\mapM_2)[\nu]:
\nu\in\dop{=}(\inQ_1\inQ_2)\}
\end{align*}
are then compact and convex and contain full-support states
$\omega^i$ and $\omega^1\otimes\omega^2$, respectively. Thus they
satisfy the assumptions on the feasible set in
Fact~\ref{fact:renyi-opt}.

Choose minimizers $\tau^i\in\mathcal F_i$ of
$D_\alpha\Dargs{\rho^i}{\tau^i}$. They exist by compactness and
lower semicontinuity, and their values are finite, so
$\supp(\rho^i)\subseteq\supp(\tau^i)$. Since these states commute, Fact~\ref{fact:renyi-opt} gives
\begin{gather*}
\tr{\Xi_i\mapM_i[\nu^i]}\leq\kappa_i
\qquad\forall\nu^i\in\dop{=}(\inQ_i),
\\
\text{where} \quad \Xi_i\defvar(\rho^i)^\alpha(\tau^i)^{-\alpha},
\qquad
\kappa_i \defvar Q_\alpha\Dargs{\rho^i}{\tau^i}.
\end{gather*}
This can be rewritten equivalently in terms of the adjoints
$\mapM_i^\dagger$ of the measurement channels:
\begin{align*}
    0\leq\mapM_i^\dagger[\Xi_i]\leq\kappa_i\id_{\inQ_i}.
\end{align*}
Consequently, for every $\nu\in\dop{=}(\inQ_1\inQ_2)$,
\begin{align*}
&\tr{(\Xi_1\otimes\Xi_2)
(\mapM_1\otimes\mapM_2)[\nu]}
\\
&=\tr{(\mapM_1^\dagger[\Xi_1]\otimes\mapM_2^\dagger[\Xi_2])\nu}\\
&\leq\kappa_1\kappa_2
\\
&=Q_\alpha\Dargs{\rho^1\otimes\rho^2}{\tau^1\otimes\tau^2}.
\end{align*}
The state $\tau^1\otimes\tau^2$ belongs to $\mathcal F_{12}$,
its support contains that of $\rho^1\otimes\rho^2$, and its
optimality operator in Fact~\ref{fact:renyi-opt} is $\Xi_1\otimes\Xi_2$.
Applying Fact~\ref{fact:renyi-opt} again therefore shows that $\tau^1\otimes\tau^2$
minimizes the joint divergence over $\mathcal F_{12}$, yielding the claimed identity.
\end{proof}

\begin{lemma}[Superadditivity for measurement channels]\label{lemma:measadd}
 For any $\alpha>1$ and any measurement channels $\mapM,\mapN \in \CPTP{Q}{C}$, $\mapM',\mapN'\in \CPTP{Q'}{C'}$,
\begin{align}
D^\downarrow_\alpha\Dargs{\mapM \otimes \mapM'}{\mapN \otimes \mapN'} \geq D^\downarrow_\alpha\Dargs{\mapM}{\mapN} + D^\downarrow_\alpha\Dargs{\mapM'}{\mapN'}  , \label{eq:measadd}
\end{align}
and thus for any $m\in\mathbb{N}$,
\begin{gather}
\frac{1}{m} D^\downarrow_\alpha\Dargs{\mapM^{\otimes m}}{\mapN^{\otimes m}} \geq D^\downarrow_\alpha\Dargs{\mapM}{\mapN} \label{eq:measweakadd}, \\
D^{\reg,\downarrow}_\alpha\Dargs{\mapM}{\mapN} \geq D^\downarrow_\alpha\Dargs{\mapM}{\mapN}. \notag
\end{gather}
\end{lemma}
\begin{proof}
Lemma~\ref{lemma:classicalinf} immediately gives~\eqref{eq:measadd}; the remaining bounds follow by inductively applying~\eqref{eq:measadd} to obtain
\begin{align*}
D^\downarrow_\alpha\Dargs{\mapM^{\otimes m}}{\mapN^{\otimes m}} \geq D^\downarrow_\alpha\Dargs{\mapM^{\otimes (m-1)}}{\mapN^{\otimes (m-1)}} + D^\downarrow_\alpha\Dargs{\mapM}{\mapN} \geq \dots \geq  m D^\downarrow_\alpha\Dargs{\mapM}{\mapN},
\end{align*}
and then taking $\mylimsup$ on the LHS of~\eqref{eq:measweakadd} to obtain the regularized version.
\end{proof}

\begin{lemma}[Weak subadditivity of regularized version]\label{lemma:weakadd}

For any $n\in\mathbb{N}$, $\alpha\in[0,\infty]$, and any channels $\mapE,\mapF\in \CPTP{\inQ}{\outQ}$,
\begin{align*}
\frac{1}{n} D^{\reg,\downarrow}_\alpha\Dargs{\mapE^{\otimes n}}{\mapF^{\otimes n}} \leq D^{\reg,\downarrow}_\alpha\Dargs{\mapE}{\mapF} .
\end{align*}
\end{lemma}
\begin{proof}
The relevant sequences in the definitions of the LHS and the RHS are respectively
\begin{align*}
\left\{\frac{1}{mn} D^{\downarrow}_\alpha\Dargs{\mapE^{\otimes mn}}{\mapF^{\otimes mn}} \right\}_{m\in\mathbb{N}}
\quad\text{and}\quad
\left\{\frac{1}{m} D^\downarrow_\alpha\Dargs{\mapE^{\otimes m}}{\mapF^{\otimes m}}\right\}_{m\in\mathbb{N}} 
.
\end{align*}
Since the former is a subsequence of the latter, its $\mylimsup$ is upper bounded by that of the latter.
\end{proof}

\newcommand{\allE}{\bR}
\begin{lemma}\label{lemma:regbound}
For any $\alpha>1$, channels $\mapE,\mapF \in \CPTP{\inQ}{\outQ}$ and stabilizer register $\stabE$ such that $|\stabE|\geq|\inQ|$,
\begin{align*}
D^{\reg,\downarrow}_\alpha\Dargs{\mapE \otimes \idmap_{\stabE}}{\mapF \otimes \idmap_{\stabE}} \leq D_\alpha\Dargs{\mapE \otimes \idmap_{\stabE}}{\mapF \otimes \idmap_{\stabE}}.
\end{align*}
\end{lemma}

\begin{proof}
The proof is similar to~\cite[Lemma~3.5]{MFSR24}; the broad idea has also been used in previous work, e.g.~\cite{GW21}.
Consider any $m\in\mathbb{N}$.
We first argue that it suffices to restrict the supremum over $\rho$ to states that are permutation-invariant (on $\allQ \allE$). While this is a ``well-known'' argument for the usual version of regularized stabilized channel divergence, here we take extra care to spell out the details due to the infimum in the second argument.\footnote{In the following analysis we make use of the fact that in the symmetrization we construct later in~\eqref{eq:symmstate}, all ``components'' of it satisfy a symmetry property in~\eqref{eq:alltermsequal}, which simplifies a later minimization over distributions on the classical register $P$. It should be possible to construct a generalization of this argument that allows different values for each term in the first line of~\eqref{eq:alltermsequal} and instead proves some form of concavity property, similar to analogous properties of $H^\uparrow_\alpha$~\cite{arx_AT25}, but we leave this for future work.}

To do so, we use the standard approach that for any feasible $\rho_{\allQ \allE}$ we shall construct a (pure) permutation-invariant counterpart that achieves a higher or equal value in the supremum, as follows.
Let us write $\pi$ to denote an arbitrary permutation map on $\allQ$; in a minor abuse of notation, let us also use that to denote ``the same permutation'' on $\allE$. 
First define the following ``symmetrized'' version of $\rho_{\allQ \allE}$, with a classical register $P$ that tracks the permutation used in the symmetrization:
\begin{align}
\rho^{\mathrm{sym}}_{\allQ \allE P} \defvar \sum_{\pi} \frac{1}{m!} (\pi\otimes\pi)[{\rho}_{\allQ \allE}] \otimes \pure{\pi}_P,\label{eq:symmstate}
\end{align}
Note that for each $\pi$, the corresponding conditional state on $\allQ \allE$ in the above state satisfies
\begin{align}
&\inf_{\nu_{\allQ \allE}} D_\alpha\Dargs{\mapE^{\otimes m} \circ (\pi\otimes\pi) [{\rho}_{\allQ \allE}]}{\mapF^{\otimes m} [\nu_{\allQ \allE}]} \nonumber\\
=&\inf_{\nu_{\allQ \allE}} D_\alpha\Dargs{ (\pi\otimes\pi) \circ \mapE^{\otimes m}  [{\rho}_{\allQ \allE}]}{\mapF^{\otimes m} [\nu_{\allQ \allE}]} \nonumber\\
=& \inf_{\nu_{\allQ \allE}} D_\alpha\Dargs{\mapE^{\otimes m}  [{\rho}_{\allQ \allE}]}{(\pi^{-1}\otimes\pi^{-1}) \circ \mapF^{\otimes m} [\nu_{\allQ \allE}]} \nonumber\\
=& \inf_{\nu_{\allQ \allE}} D_\alpha\Dargs{\mapE^{\otimes m}  [{\rho}_{\allQ \allE}]}{\mapF^{\otimes m} [\nu_{\allQ \allE}]},\label{eq:alltermsequal}
\end{align}
where the second line holds since $\pi$ commutes with the permutation-symmetric map $\mapE^{\otimes m}$, 
the third line is simply by applying $(\pi^{-1}\otimes\pi^{-1})$ on both arguments via unitary invariance of $D_\alpha$, and for the fourth line we similarly commute $(\pi^{-1}\otimes\pi^{-1})$ with $\mapF^{\otimes m}$ and then reparametrize the resulting optimization. In particular, this means that all these divergences (minimized over $\nu$) have the same value for every $\pi$.

Moreover, since $\rho^{\mathrm{sym}}_{\allQ}$ is permutation-invariant, it has a permutation-invariant purification $\ket{\psi^{\mathrm{sym}}}_{\allQ\allE}$. Since $\rho^{\mathrm{sym}}_{\allQ\allE P}$ in~\eqref{eq:symmstate} was also an extension of $\rho^{\mathrm{sym}}_{\allQ}$, there exists a channel $\mapN \in \CPTP{\allE}{\allE P}$ such that $\mapN [\pure{\psi^{\mathrm{sym}}}] = \rho^{\mathrm{sym}}_{\allQ \allE P}$. This channel commutes with $\mapE^{\otimes m},\mapF^{\otimes m}$ because it does not act on $\allQ$. Thus for any $\nu_{\allQ \allE}$, we have 
\begin{align*}
D_\alpha\Dargs{\mapE^{\otimes m}[\pure{\psi^{\mathrm{sym}}}]}{\mapF^{\otimes m}[\nu_{\allQ \allE}]} &\geq D_\alpha\Dargs{\mapN\circ\mapE^{\otimes m}[\pure{\psi^{\mathrm{sym}}}]}{\mapN\circ\mapF^{\otimes m}[\nu_{\allQ \allE}]} \nonumber\\
&=D_\alpha\Dargs{\mapE^{\otimes m}[\rho^{\mathrm{sym}}_{\allQ \allE P}]}{\mapF^{\otimes m} \circ \mapN[\nu_{\allQ \allE}]}\nonumber\\
&\geq \inf_{\nu \in \dop{=}(\allQ \allE P)} D_\alpha\Dargs{\mapE^{\otimes m}[\rho^{\mathrm{sym}}_{\allQ \allE P}]}{\mapF^{\otimes m} [\nu_{\allQ \allE P}]},
\end{align*}
where the first line is by data-processing,
in the second line we commuted $\mapN$ with $\mapE^{\otimes m},\mapF^{\otimes m}$, and in the third line we simply enlarged the infimum domain. Furthermore, for each point in that final infimum we have (letting $\mapP$ denote a pinching channel in the classical basis of $P$, and letting $\mathbb{P}_P$ denote the space of probability distributions on that classical register):
\begin{align*}
&D_\alpha\Dargs{\mapE^{\otimes m}[\rho^{\mathrm{sym}}_{\allQ \allE P}]}{\mapF^{\otimes m} [\nu_{\allQ \allE P}]} \nonumber\\
\geq& D_\alpha\Dargs{\mapE^{\otimes m} \circ \mapP [\rho^{\mathrm{sym}}_{\allQ \allE P}]}{\mapF^{\otimes m} \circ \mapP [\nu_{\allQ \allE P}]} \nonumber\\
\geq& \inf_{\sfP\in\mathbb{P}_P} \inf_{\left\{\nu^{(\pi)}_{\allQ \allE}\right\}} D_\alpha\Dargs{\sum_\pi \frac{1}{m!} \pure{\pi}_P \otimes \mapE^{\otimes m} \circ (\pi\otimes\pi) [{\rho}_{\allQ \allE}] }{\sum_\pi \sfP(\pi) \pure{\pi}_P \otimes \mapF^{\otimes m} [\nu^{(\pi)}_{\allQ \allE}]}\nonumber\\
=& \inf_{\sfP\in\mathbb{P}_P} \inf_{\left\{\nu^{(\pi)}_{\allQ \allE}\right\}} \frac{1}{\alpha-1}\log \left(
\sum_\pi \left(\frac{1}{m!}\right)^\alpha \sfP(\pi)^{1-\alpha}
2^{(\alpha-1)D_\alpha\Dargs{\mapE^{\otimes m} \circ (\pi\otimes\pi) [{\rho}_{\allQ \allE}]}{\mapF^{\otimes m} [\nu^{(\pi)}_{\allQ \allE}]}}
\right) \nonumber\\
=& \inf_{\sfP\in\mathbb{P}_P}  \frac{1}{\alpha-1}\log \left(
\sum_\pi \left(\frac{1}{m!}\right)^\alpha \sfP(\pi)^{1-\alpha}
2^{(\alpha-1) \inf_{\nu_{\allQ \allE}} D_\alpha\Dargs{\mapE^{\otimes m} \circ (\pi\otimes\pi)  [{\rho}_{\allQ \allE}]}{\mapF^{\otimes m} [\nu_{\allQ \allE}]}}
\right) \nonumber\\
=& \inf_{\sfP\in\mathbb{P}_P}  \frac{1}{\alpha-1}\log \left(
\sum_\pi \left(\frac{1}{m!}\right)^\alpha \sfP(\pi)^{1-\alpha}
2^{(\alpha-1) \inf_{\nu_{\allQ \allE}} D_\alpha\Dargs{\mapE^{\otimes m}  [{\rho}_{\allQ \allE}]}{\mapF^{\otimes m} [\nu_{\allQ \allE}]}}
\right) \nonumber\\
=& \inf_{\sfP\in\mathbb{P}_P}  \frac{1}{\alpha-1}\log \left(
\sum_\pi \left(\frac{1}{m!}\right)^\alpha \sfP(\pi)^{1-\alpha}
\right) + \inf_{\nu_{\allQ \allE}} D_\alpha\Dargs{\mapE^{\otimes m}  [{\rho}_{\allQ \allE}]}{\mapF^{\otimes m} [\nu_{\allQ \allE}]}\nonumber\\
=& \inf_{\nu_{\allQ \allE}} D_\alpha\Dargs{\mapE^{\otimes m}  [{\rho}_{\allQ \allE}]}{\mapF^{\otimes m} [\nu_{\allQ \allE}]}.
\end{align*}
In the above, the second line is data-processing together with the fact that the pinching on $P$ commutes with the channels on $\allQ$. The third line is by applying $\mapP$ to $\nu_{\allQ \allE P}$ and writing the resulting state in the form $\sum_\pi \sfP(\pi) \pure{\pi}_P \otimes \nu^{(\pi)}_{\allQ \allE}$, and similarly noting that $\mapP$ leaves $\rho^{\mathrm{sym}}_{\allQ \allE P}$ invariant and then substituting its definition. The fourth line is just Fact~\ref{th:renyi-properties}\ref{fact:classmix} again. In the fifth line we brought each $\nu^{(\pi)}$ infimum in to the only term involving it (noting that there is no change in the optimization direction because $\alpha>1$) and dropped the $(\pi)$ indexing in these ``dummy variables'' $\nu^{(\pi)}$ for the infima (since now the ordering of the summation and infimum already indicates that there is a separate infimum for each $\pi$).
In the sixth line we recall the property~\eqref{eq:alltermsequal} that each of these infima in fact has the same value, so in the seventh line we can factorize out this quantity that is independent of the sum over $\pi$. Finally in the eighth line we observe that the infimum over $\sfP$ is simply minimizing the classical {\Renyi} divergence between it and the uniform distribution, and thus its minimal value is just zero.

Summarizing, the two bounds we obtained above give for any $\nu_{\allQ \allE}$:
\begin{align*}
D_\alpha\Dargs{\mapE^{\otimes m}[\pure{\psi^{\mathrm{sym}}}]}{\mapF^{\otimes m}[\nu_{\allQ \allE}]} &\geq
\inf_{\nu_{\allQ \allE}} D_\alpha\Dargs{\mapE^{\otimes m}  [{\rho}_{\allQ \allE}]}{\mapF^{\otimes m} [\nu_{\allQ \allE}]},
\end{align*}
and taking the infimum on the left gives the desired claim that $\pure{\psi^{\mathrm{sym}}}$ yields a value in the supremum over $\rho$ at least as high as that of the original state $\rho_{\allQ\allE}$. Thus it suffices to restrict the supremum to permutation-invariant states (moreover they can be taken to be pure, though this is not needed for the remaining argument). 

Now, let ${\rho}_{\allQ\allE}$ be any feasible state in this restricted supremum. Since it is permutation-invariant, there exists~\cite{CKR09} a de~Finetti state $\tau_{\allQ\allE}$ such that $\tau_{\allQ\allE} \geq \frac{1}{\operatorname{poly}(m) }{\rho}_{\allQ\allE} $. Thus we have
\begin{align*}
\inf_{\nu_{\allQ \allE}}
D_\alpha\Dargs{\mapE^{\otimes m}[{\rho}_{\allQ\allE}]}{\mapF^{\otimes m}[\nu_{\allQ \allE}]} &\leq
\inf_{\nu_{\allQ \allE}}
D_\alpha\Dargs{\mapE^{\otimes m}[{\tau}_{\allQ\allE}]}{\mapF^{\otimes m}[\nu_{\allQ \allE}]} + \frac{\alpha}{\alpha-1} \log \operatorname{poly}(m)\nonumber\\
&\leq
D_\alpha\Dargs{\mapE^{\otimes m}[{\tau}_{\allQ\allE}]}{\mapF^{\otimes m}[\tau_{\allQ \allE}]} + \frac{\alpha}{\alpha-1} \log \operatorname{poly}(m) \nonumber\\
&\leq m
D_\alpha\Dargs{\mapE \otimes \idmap_{\stabE}}{\mapF \otimes \idmap_{\stabE}} + \frac{\alpha}{\alpha-1} \log \operatorname{poly}(m),
\end{align*}
where the first line is by Fact~\ref{th:renyi-properties}\ref{fact:eventconditioning}, in the second line we just pick a feasible point of the infimum, and the last line is simply by noting a de~Finetti state is a mixture of IID states and then applying quasi-convexity of $D_\alpha$. 

Since that final bound holds for any  permutation-invariant $\rho$ and any $m$, we can take the supremum of the LHS over such $\rho$, then divide by $m$ and take the $\mylimsup$ to obtain the claimed result. 
\end{proof}

\subsection{Tensor-product comparison with testing}
\label{subsec:applicationIID}

As sketched in the introduction, we shall again incorporate ``testing'' and condition on the frequency distribution observed on the testing registers. 
To do so, we first specify the relevant postulate set for incorporating testing.

\begin{postulates}\label{posts:overall}
For $i=1,\ldots,n$, let the $i$-th round channels be
\[
\mapE_i,\mapF_i\in\CPTP{\inN_i}{\CS_i\outQ_i}.
\]
We require:
\begin{enumerate}[label=\textup{(\roman*)},ref=\textup{(\roman*)}]
\item The register $\CS_i$ is classical, taking values in $\alphCS$, and has the same marginal under $\mapE_i$ and $\mapF_i$ on every input.

\item All $\CS_i\outQ_i$ are protected registers that are acted on only by identity maps after being produced.

\item Let $S_\Omega$ be a nonempty closed convex subset of the probability simplex on $\alphCS$, containing at least one full-support distribution. The accept condition on the protected classical record $\bCS$ is $\freq_\bcS \in S_\Omega$, and we denote this event by $\Omega$.

\item Under fixed identifications of their respective input and output spaces, the channel pairs $(\mapE_i,\mapF_i)$ are isomorphic copies of a single pair $(\mapE,\mapF)$.
\end{enumerate}
\end{postulates}

These channels can thus act on a state of the form $\rho^0_{\allN\bR}$, where $\allN=\inN_1\cdots\inN_n$ and $\bR=R_1\cdots R_n$ is an external reference register, with each $R_i$ a copy of a sufficiently large register $R$. This tensor-product notation does not impose any independence: all the $\allN\bR$ registers can be arbitrarily entangled. The resulting output state therefore has the form
\begin{align}
\realstate_{\bCS\boutQ\bR}
=\mapE^{\otimes n}[\rho^0_{\allN\bR}]. \label{eq:tensorrealstate}
\end{align}
\begin{remark}\label{remark:IIDchanns}
While at first sight the restriction of all the channels being isomorphic might seem too strong, it in fact suffices to handle tensor products of non-isomorphic channel pairs, i.e.~$\realstate_{\bCS \boutQ \bR} = (\mapE_1 \otimes \mapE_2 \otimes \dots \otimes \mapE_n)\left[\rho^0_{\allN \bR}\right]$. This is because we can embed all these distinct channels into a single ``larger'' channel $\mapE$ via a direct-sum construction analogous to that in Sec.~\ref{subsec:bound_compexp}, which we discuss further below, and hence capture all such behaviours within the $\realstate_{\bCS \boutQ \bR} = \mapE^{\otimes n}\left[\rho^0_{\allN \bR}\right]$ formulation. Thus it is essentially without loss of generality, albeit at the cost of a slightly loose final bound due to the direct-sum relaxation.
\end{remark}

Given the above, we specify some additional notation.
In analogy to $\compexp{\alpha}$ for the sharp-divergences, define
\begin{align}
\compexpsand{\alpha}
= n 
\sup_{\substack{\sfQ\in S_\Omega,\\
\sigma\in\dop{=}(\inN)}}
\left( \left(\sum_{\cS} \sfQ(\cS) D_\alpha\Dargs{\sigma^{\mapE}_{\outQ\stabE|\cS}}{\sigma^{\mapF}_{\outQ\stabE|\cS}}\right) -\frac{1}{\alpha-1} D\Dargs{\sfQ}{{\sfP}} \right), \label{eq:alpharatedefn}
\end{align}
where the optimization is to be interpreted as follows:
\begin{itemize}

\item $\sigma^{\mapE}, \sigma^{\mapF} \in \dop{=}(\CS \outQ \stabE)$ denote the following states for any purification $\ket\sigma_{\inN\stabE}$ of $\sigma\in\dop{=}(\inN)$:
\begin{align*}
\sigma^{\mapE} = \mapE[\pure\sigma], \qquad  \sigma^{\mapF} = \mapF[\pure\sigma].
\end{align*}

\item $\sfP$ denotes the distribution induced by either of these states on $\CS$ (by hypothesis, these distributions are the same for both).

\item For any $\cS$, if $\sfP(\cS)=0$ then we assign $D_\alpha\Dargs{\sigma^{\mapE}_{\outQ\stabE|\cS}}{\sigma^{\mapF}_{\outQ\stabE|\cS}}=0$, and if $\sfQ(\cS)=0$ then we assign $\sfQ(\cS) D_\alpha\Dargs{\sigma^{\mapE}_{\outQ\stabE|\cS}}{\sigma^{\mapF}_{\outQ\stabE|\cS}}=0$ even if the divergence term is infinite.
The entire objective is interpreted as $+\infty$ whenever
its first sum is infinite, regardless of the value of
$D\Dargs{\sfQ}{\sfP}$.
\end{itemize}
The value of $\compexpsand{\alpha}$ is independent of the chosen purification $\ket\sigma_{\inN\stabE}$: any two purifications are related by an isometry on the purifying register, under which the quantities are invariant by Fact~\ref{th:renyi-properties}\ref{fact:DPI}.

With this, we again state a result that incorporates testing:
\begin{theorem}[Tensor-product testing comparison]\label{th:robustselftest}
Consider channels and parameters satisfying Postulate Set~\ref{posts:overall}. Now take any $n \in \mathbb{N}$, any $\alpha>1$, and any $\esmooth\in(0,1)$. Also, take any initial state $\rho^0 \in \dop{=}(\allN \bR)$ and any measurement channel $\mapM \in \CPTP{\boutQ \bCS \bR}{F}$. Let $\realstate_{\bCS \boutQ \bR} \defvar \mapE^{\otimes n}\left[\rho^0_{\allN \bR}\right]$, and let $p_\Omega$ denote the probability of the event $\Omega$ with respect to this state; we require $p_\Omega>0$.

Then there exists some $\nu^0 \in \dop{=}(\allN \bR)$ such that for $\idealstate_{\bCS \boutQ \bR} \defvar \mapF^{\otimes n}\left[\nu^0_{\allN \bR}\right]$, we have
\begin{align}
D_\alpha\Dargs{\mapM\left[\realstate_{|\Omega} \right]}{\mapM\left[\idealstate\right]} &\leq \compexpsand{\alpha} + \frac{\alpha\log(1/p_\Omega)}{\alpha-1}, \label{eq:Dalphaselftest} \\
\Dmax^{\esmooth}\Dargs{\mapM\left[\realstate_{|\Omega} \right]}{\mapM\left[\idealstate\right]} &\leq \compexpsand{\alpha} + \frac{\smf{\esmooth}}{\alpha-1} + \frac{\alpha\log(1/p_\Omega)}{\alpha-1}.
\end{align}
\end{theorem}
\begin{proof}
The proof is similar to the previous approach to incorporate testing in Sec.~\ref{subsec:sequential_with_testing}, so we defer it to Appendix~\ref{app:IIDtesting}. The main technical difference is that a little care is needed to handle the global measurement channel at the end.
\end{proof}

We now wish to apply the above result to Protocol~\ref{prot:DI} under a no-memory assumption. We formalize this via the following postulate set. Again, we stress the point (from Remark~\ref{remark:no-memory}) that this condition is not assuming the devices behave independently across rounds, but rather only that the \emph{measurements} performed by the devices are in tensor product across rounds --- the \emph{state} can be arbitrarily entangled. Also, the restriction of the channels in different rounds being isomorphic is basically without loss of generality (Remark~\ref{remark:IIDchanns}).
\begin{postulates}[No-memory condition]\label{post:no-memory}
The protocol is modelled using the one-round channels from Sec.~\ref{subsec:realandideal}, before their extension by identities on earlier classical records, with the following modifications and assumptions.
\begin{enumerate}
\item For each round $i$, the channels $\mapE_i,\mapF_i$ have their input domains modified to a new register $\inN_i$, and on the output domains the $M_{i+1}^\A M_{i+1}^\B$ registers are now grouped into the protected registers $\outQ_i$, i.e.
\[
\mapE_i,\mapF_i\in\CPTP{\inN_i}{\CS_i\outQ_i},
\qquad
\outQ_i=U_iX_iY_iA_iB_iM_{i+1}^\A M_{i+1}^\B.
\]
The structure of each $\inN_i$ register is that it is isomorphic to the previously defined registers $M_{i}^\A M_{i}^\B$; notice however that now the output registers $M_{i+1}^\A M_{i+1}^\B$ of each channel are distinct from the input register $\inN_{i+1}$ of the next channel.

\item Under fixed identifications of their respective input and output spaces, the channel pairs $(\mapE_i,\mapF_i)$ are isomorphic copies of a single pair $(\mapE,\mapF)$.
\end{enumerate}

Together, these conditions imply that the $\mapE_i$ channels act via a tensor-power structure to produce a final output state as in~\eqref{eq:tensorrealstate}. 
\end{postulates}

\paragraph{Real and ideal channels in the protocol.}
The construction of each channel $\mapE_i$ is the same as in Sec.~\ref{subsec:realandideal}, with the register changes specified in Postulate Set~\ref{post:no-memory}. We first address a technical point: if the ``physical'' channels in the original protocol are non-isomorphic, we shall argue it suffices to construct a single ``representative'' channel that can model any of them, similar (although not precisely identical) to Sec.~\ref{subsec:bound_compexp}.
Without loss of generality (due to our finite-dimensionality condition throughout this work, allowing us to embed any Hilbert space in some larger finite-dimensional one), suppose all the $\inN_i$ registers are isomorphic to a single register $\commN$; similarly take registers $\commO$ and $\commC$ for $\outQ_i$ and $\CS_i$ respectively. Let $I$ be a register of dimension $n$, let $\inN = I \commN$ (this is isomorphic to the direct sum of the $\inN_i$ registers). Then define the following channel, with input register $\inN$ and output registers $\commC,\commO$: 
\[
\mapE[\tau]=\sum_{i=1}^n\mapE_i[\tau_{ii}],
\]
where $\tau_{ii}$ denotes the diagonal block in the same sense as Sec.~\ref{subsec:bound_compexp}. In other words, $\mapE$ basically measures $I$ in its classical basis and performs $\mapE_i$ on the $\commN$ register.

With this, observe that the real state in the protocol can be rewritten via
\[ (\mapE_1\otimes\cdots\otimes\mapE_n)[\rho^0_{\allN\bR}] = \mapE^{\otimes n}\left[\pure{1}_{I_1}\otimes\ldots\otimes\pure{n}_{I_n}\otimes\rho^0_{\allN\bR}\right],
\]
i.e.~it is $\mapE^{\otimes n}$ acting on a different initial state. Moreover, observe that $\mapE$ has the critical property that (here for brevity we absorb the $U$ ``flag'' indicating $\test$ or $\gen$ into the outputs, instead of explicitly writing it):
\[
\mapE = \ptest \mapE^\test + (1-\ptest) \mapE^\gen, \quad \text{where } \mapE^\test[\tau]=\sum_{i=1}^n\mapE^\test_i[\tau_{ii}], \, \mapE^\gen[\tau]=\sum_{i=1}^n\mapE^\gen_i[\tau_{ii}], 
\]
i.e.~it inherits the test/generation structure of the original channels, with $\mapE^\test$ playing the nonlocal game and $\mapE^\gen$ storing the state for the DD cryptographic continuation. With this structure, we can construct the ideal channel $\mapF$ in terms of the actions performed by \emph{this} channel $\mapE$ in the test and generation branches (``forgetting'' the underlying direct-sum structure), after which all the relevant properties are inherited from the previous analysis.

With this in mind, the conditions of Postulate Set~\ref{posts:overall} are satisfied, allowing us to apply Theorem~\ref{th:robustselftest} (for suitable parameters). To demonstrate one example of the resulting scaling, we first bound the comparison exponent, but note that this time we can \emph{first} choose an arbitrary (constant) test probability $\ptest$:

\begin{corollary}[Comparison exponent with arbitrarily small linear dependence at constant tolerated noise]
\label{corr:practicaltest}
Assume Postulate Sets~\ref{post:selftest}, \ref{post:shield} and \ref{post:no-memory} hold in Protocol~\ref{prot:DI}.
Then for any $\constval'\in(0,\infty)$ any $\pert\in(0,1-2^{-\constval'/4})$, and any $\ptest\in(0,1)$, 
one can choose constants
\[
0<\epstol<1,
\qquad
1<\alphastar<2,
\]
independently of $n$ and the physical memory dimensions, such that
\begin{align*}
\compexpsand{\alphastar} \leq \constval' n.
\end{align*}
\end{corollary}
\begin{proof}
Since the sandwiched divergence is nondecreasing in $\alpha$ and upper bounded by the sharp-divergences, we have  $D_\alpha \leq D_2 \leq \Dsharp_2$. So we can immediately bound $\compexpsand{\alpha}$ for $\alpha<2$ via
\begin{align}
\frac{1}{n}\compexpsand{\alpha} &\leq
 \sup_{\sfQ\in S_\Omega} \sup_{\sigma\in\dop{=}(\inN)}
\left( \left(\sum_{\cS} \sfQ(\cS) \Dsharp_2\Dargs{\sigma^{\mapE}_{\outQ\stabE|\cS}}{\sigma^{\mapF}_{\outQ\stabE|\cS}}\right) -\frac{1}{\alpha-1} D\Dargs{\sfQ}{{\sfP}} \right) \nonumber\\
&\leq
 \sup_{\sfQ\in S_\Omega} \sup_{\sigma\in\dop{=}(\inN)}
\left( \Dsharp_2\Dargs{\sigma^{\mapE}_{\outQ\stabE|\gen}}{\sigma^{\mapF}_{\outQ\stabE|\gen}} -\frac{1}{\alpha-1} D\Dargs{\sfQ}{{\sfP}} \right), 
\end{align}
where in the second line we have eliminated terms with $\cS \neq \gen$ by recalling the real and ideal channels are identical on the test-round case, and used $\sfQ(\gen)\leq1$ and nonnegativity of the generation divergence.

It thus suffices to upper bound this by $\constval'$. This holds by an analogous argument to the Theorem~\ref{th:small-rate} proof, with the key difference being that here we have a prefactor of $\frac{1}{\alpha-1}$ on the KL-divergence penalty term, in contrast to the prefactor of $\frac{1}{1-\ptest}$ previously in~\eqref{eq:block-rate-normalized}. For the prescribed $\pert$ and $\ptest$, we have
\[
a_\pert=1-\pert+\pert/|S|\geq1-\pert>2^{-\constval'/4},
\qquad \log a_\pert^{-1}<\frac{\constval'}4.
\]
Using the function $\widehat B_{\pert,\game}$ from the proof of Theorem~\ref{th:small-rate} and its limit in~\eqref{eq:majorantlimits}, we choose $0<\Delta\leq\ptest$ sufficiently small that
\[
\log\widehat B_{\pert,\game}(\Delta/\ptest)\leq\frac{\constval'}2.
\]
We then choose
\[
0<\alphastar-1<\min\left\{1,\frac{\Delta^2}{2\ln2\,\log(|S|^2/\pert)}\right\},
\qquad
0<\epstol\leq\min\left\{\frac{1-\ptest}{2},\frac\Delta2\right\}.
\]
These choices are independent of $n$.

For any feasible trial input with $\trnorm{\sfP-\sfP_C}\leq\Delta$, its game deficit satisfies
\[
\delta=\omega^\ast-\frac{\sfP(1)}{\ptest}
=\frac{\sfP_C(1)-\sfP(1)}{\ptest}\leq\frac\Delta\ptest.
\]
The generation-divergence bound and monotonicity of $\widehat B_{\pert,\game}$ therefore give
\[
\Dsharp_2\Dargs{\sigma^{\mapE}_{\outQ\stabE|\gen}}{\sigma^{\mapF}_{\outQ\stabE|\gen}}
\leq\log\widehat B_{\pert,\game}(\Delta/\ptest)\leq\frac{\constval'}2,
\]
so the objective is at most $\constval'$ by nonnegativity of the KL penalty. In the remaining case, $\trnorm{\sfP-\sfP_C}>\Delta$, acceptance and the triangle inequality give $\trnorm{\sfQ-\sfP}>\Delta/2$. Pinsker's inequality and the universal bound above then give
\[
\Dsharp_2\Dargs{\sigma^{\mapE}_{\outQ\stabE|\gen}}{\sigma^{\mapF}_{\outQ\stabE|\gen}}
-\frac{D(\sfQ\Vert\sfP)}{\alphastar-1}
\leq\log\frac{|S|^2}{\pert}-\frac{\Delta^2}{2(\alphastar-1)\ln2}\leq0.
\]
Thus the objective is at most $\constval'$ in both cases, without needing to tune the value of $\ptest$.
\end{proof}

With this, we now present an example that demonstrates a specific scaling behaviour of Theorem~\ref{th:robustselftest}. Still, we highlight that Theorem~\ref{th:robustselftest} more generally allows for a much wider scope of scaling analysis, including ``second-order analysis'' (see~\cite{DFR20,DF19,arx_AHT24}) to some extent --- this is simply one example. We present the quantifier ordering in rather fine detail in order to emphasize the difference compared to the sharp-divergence results, and also because of a slightly subtle quantifier issue regarding the ``ideal'' state.

\begin{theorem}[DD-to-DI conversion under the no-memory condition]\label{th:scalingdemo}
Let $\ddP$ be a DD protocol of the form given in Protocol~\ref{prot:DD}, satisfying Postulate Sets~\ref{post:DD} and~\ref{post:selftest}. Let $\diP$ be its DI replacement from Protocol~\ref{prot:DI}, satisfying Postulate Sets~\ref{post:shield} and~\ref{post:no-memory}. Then:
\begin{itemize}
\item for every $k>0$ and every testing probability $\ptest\in(0,1)$,
\item there exist constants $\pert=\Theta(k)$, $\base\in(0,1)$ such that $b\to1$ as $k\to0$,and $\epstol\in(0,(1-\ptest)/2]$, independent of $n$, such that
\item for every integer $n\geq4/(1-\ptest)$, every allowed implementation and attack, every initial state $\rho^0_{\allN\bR}$, and every measurement channel $\mapM$ on $\bCS\boutQ\bR$,
\item there exists an initial state $\nu^0_{\allN\bR}$ such that the following holds with $m=\floor{(1-\ptest)n/2}$.
\end{itemize}
Let $\mapE$ and $\mapF^\pert$ be the real and ideal channels specified in Postulate Set~\ref{post:no-memory}, and define
\[
\realstate=\mapE^{\otimes n}[\rho^0_{\allN\bR}],
\qquad
\idealstate=(\mapF^\pert)^{\otimes n}[\nu^0_{\allN\bR}],
\]
with identities on $\bR=R_1\cdots R_n$ understood. If $\Pr_{\realstate}[\Acc]\geq\base^m$, then
\begin{equation}\label{eq:linearbound}
\Dmax^{\base^m}\Dargs{\mapM[\realstate_{|\Acc}]}{\mapM[\idealstate]}
\leq km.
\end{equation}
The soundness bound for the prescribed DD continuation is
\begin{equation}\label{eq:landbnd}
\Pr_{\diP}[\Acc\land\Bad]
\leq\left(2^kq_{\mathrm{dep}}(\pert)\right)^m+\base^m.
\end{equation}
If $\limsup_{\pert\to0}q_{\mathrm{dep}}(\pert)\leq q$ for some $0<q<1$, then for every $0<k<-\log q$ and every $\ptest\in(0,1)$ the constants can be chosen so that $2^kq_{\mathrm{dep}}(\pert)<1$, giving an exponentially small soundness error in $n$.

Moreover, the same completeness properties hold as in Theorem~\ref{th:complete}.
\end{theorem}

Here the testing probability $\ptest$ can be fixed in advance, whereas in Theorem~\ref{th:DI} it is chosen together with $\pert$ and $\epstol$. In fact, via a suitable ``second-order analysis'' similar to~\cite{DF19,arx_AHT24}, it seems likely that one can even pick $\ptest$ to be any function of $n$ that goes to zero slower than $1/n$; however, we do not aim to pin down these details here.

The state $\nu^0$ may depend on both $\rho^0$ and $\mapM$. The DD reduction applies uniformly to every $\nu^0$, because the ideal channels replace the generation-round states.

Unlike Theorem~\ref{th:DI}, this theorem bounds the smooth max-divergence after $\mapM$. Its proof uses an order $\alpha$ (for the sandwiched divergence) sufficiently close to $1$, rather than order $2$.

\begin{proof}[Proof of Theorem~\ref{th:scalingdemo}]
For the given $k$ and $\ptest$, we set $\constval'=k(1-\ptest)/8$ and choose $0<\pert<1-2^{-k(1-\ptest)/32}$; for example, $\pert=(1-2^{-k(1-\ptest)/32})/2=\Theta(k)$ as $k\to0$ for fixed $\ptest$. Corollary~\ref{corr:practicaltest}, with these $\constval'$, $\ptest$, and $\pert$, gives constants $\alphastar\in(1,2)$ and $\epstol$ such that
\begin{equation}\label{eq:halfrate}
\compexpsand{\alphastar}\leq\frac{k(1-\ptest)n}{8},
\end{equation}
where the ideal channel in the comparison exponent is $\mapF=\mapF^\pert$. Decreasing the tolerance if necessary ensures $\epstol\leq(1-\ptest)/2$ without increasing this exponent. The acceptance set contains a full-support distribution because it is a trace-distance ball of positive radius about $\sfP_C$. Next, we set
\begin{equation}\label{eq:chosenbase}
\base=2^{-(\alphastar-1)k/8},
\end{equation}
which lies in $(0,1)$ and tends to $1$ as $k\to0$, since $0<\alphastar-1<1$. These choices are independent of $n$, the physical memory dimensions, the initial state, and $\mapM$.

For every $n\geq4/(1-\ptest)$, the prescribed choice $m=\floor{(1-\ptest)n/2}$ gives $m\geq2$ and $(1-\ptest)n/2<m+1$. If $\Pr_{\realstate}[\Acc]\geq\base^m$, Theorem~\ref{th:robustselftest}, with $\Omega=\Acc$ and $\esmooth=\base^m$, yields an initial state $\nu^0$ such that
\begin{align*}
\Dmax^{\base^m}\Dargs{\mapM[\realstate_{|\Acc}]}{\mapM[\idealstate]}
&\leq\frac{k(1-\ptest)n}{8}
+\frac{(2+\alphastar)m\log(1/\base)}{\alphastar-1}\\
&\leq\frac{k(m+1)}{4}
+\frac{4m\log(1/\base)}{\alphastar-1}\\
&\leq\frac{km}{2}
+\frac{4m\log(1/\base)}{\alphastar-1}\\
&=km.
\end{align*}
This proves~\eqref{eq:linearbound}. If $\Pr_{\realstate}[\Acc]<\base^m$, the implication is vacuous and any initial state $\nu^0$ suffices.

For soundness, $\mapM$ is the continuation channel from the proof of Lemma~\ref{lem:DD}, applied to the corresponding no-memory outputs. 
Acceptance ensures at least $(1-\ptest-\epstol)n\geq m$ generation rounds, so this channel agrees with the actual continuation on $\Acc$. The reduced-channel identities in that lemma hold for every incoming state and with an arbitrary reference. They therefore apply to the tensor-product process with initial state $\nu^0$ as well: the unused input registers and earlier outputs can be retained as reference registers while each round acts. After Alice's device registers and discarded records are traced out, each retained generation round supplies an independent copy of~\eqref{eq:noisy-source}, and all remaining maps act on Bob's registers. Thus regardless of $\nu^0$, we have the uniform bound
\[
\Pr_{\mapM[\idealstate]}[\Bad]\leq q_{\mathrm{dep}}(\pert)^m.
\]
If $\Pr_{\realstate}[\Acc]\geq\base^m$, Lemma~\ref{lem:probability}(ii) applied to~\eqref{eq:linearbound} gives
\[
\Pr_{\diP}[\Bad\mid\Acc]
\leq 2^{km}\Pr_{\mapM[\idealstate]}[\Bad]+\base^m
\leq\left(2^kq_{\mathrm{dep}}(\pert)\right)^m+\base^m.
\]
Multiplying by $\Pr_{\realstate}[\Acc]\leq1$ gives the claimed bound on the joint event. In the other case, this joint probability is smaller than $\base^m$. These cases establish~\eqref{eq:landbnd}.

Finally, Corollary~\ref{corr:practicaltest} applies to every $0<\pert<1-2^{-k(1-\ptest)/32}$ before choosing $\alphastar$ and $\epstol$. Under the stated limit assumption, for $0<k<-\log q$ we may therefore first choose such a $\pert$ also satisfying $2^kq_{\mathrm{dep}}(\pert)<1$. The example $\pert=(1-2^{-k(1-\ptest)/32})/2$ already satisfies this additional inequality for all sufficiently small $k$, so the choice $\pert=\Theta(k)$ as $k\to0$ for fixed $\ptest$ is preserved. Both terms in~\eqref{eq:landbnd} then decay exponentially in $n$.
\end{proof}

Using $\log(2/\esmooth^2)$ instead of $\smf{\esmooth}$ in the \Renyi-to-max conversion would add $1/(\alphastar-1)$ to the displayed smooth max-divergence bound with the same parameter choices. The sharper conversion avoids that additive term, though this has fairly minimal practical consequences (since for instance this term is independent of $n$, and thus becomes irrelevant as $n$ grows).


\section{Concrete applications}\label{sec:applications}

\subsection{Single-key uncloneable encryption}\label{sec:UE}
\newcommand{\ueP}{\mathfrak{P}}
\newcommand{\BadS}{\Bad_{\mathrm{search}}}
\newcommand{\BadD}{\Bad_{\mathrm{dec}}}

Our first application is a sequential device-independent protocol for uncloneable encryption. We describe the security experiments for uncloneable encryption in a general way that encompasses both DD and sequential DI protocols. The protocols will involve a preparation phase, independent of the plaintext message. In the device-dependent case, this just consists of Alice sending a trusted state she prepared to Bob\footnote{Usually, DD protocols for UE involve Alice sending both the message-independent quantum state and the message-dependent part of the ciphertext to Bob at once. But we can break this up into two separate phases.}; in the DI case, the preparation phase is a sequential protocol like in Protocol~\ref{prot:DI}, at the end of which Alice accepts or aborts. We denote acceptance by $\Acc$. On acceptance, Alice uses her private preparation records to encrypt a plaintext and retains the complete classical decryption key. Bob's ciphertext consists of his quantum memory together with the classical information released by Alice.

An attack includes Bob's operations throughout preparation, the quantum channel he applies to the ciphertext in order to distribute it between the decoders, and operations by the two decoders. Bob may retain arbitrary quantum memory between rounds and adapt his operations to the information available to him.  The decoders may share entanglement and classical information established before the split, but cannot communicate afterwards. We place no computational restriction on the attack. In the DI setting, the attack also includes the initial device state and device operations allowed by the shielding postulate. Alice's device memory remains inaccessible to Bob and both decoders throughout the experiment. 


The following definitions apply both to the trusted-preparation protocol $\ueP$ and to its DI implementation $\diP$. As in the interactive formulation of~\cite[Section~5.2]{GMP22}, security concerns the joint event of acceptance and successful recovery, rather than recovery conditioned on acceptance.

\begin{definition}[Uncloneable search security]\label{def:UE-search}
The search experiment for uncloneable encryption of $\ell$-bit plaintexts is as follows.
\begin{enumerate}[leftmargin=*]
\item Alice samples the secret decryption key.
\item Alice and Bob perform the preparation phase, where Alice uses the secret key. In a DI protocol, Alice may abort at the end of this phase. If Alice aborts, the experiment terminates.
\item On acceptance, Alice samples a uniform plaintext $\bmsg\in\{0,1\}^{\ell}$, independently of the preparation records and all quantum systems, and completes its encryption.
\item Bob splits his information between the two decoders. After receiving the same complete decryption key, they output strings $\bmsg_1,\bmsg_2\in\{0,1\}^{\ell}$.
\end{enumerate}
The search bad event is
\[
\BadS=\{\bmsg_1=\bmsg_2=\bmsg\}.
\]
We declare this event false on an aborted run. The protocol is $\eps$-search secure if every allowed attack satisfies in two respective cases
\[
\Pr_{\mathrm{DD}}[\BadS]\leq \eps, \qquad  \Pr_{\mathrm{DI}}[\Acc\land\BadS]\leq\eps.
\]
\end{definition}

\begin{definition}[Uncloneable decision security]\label{def:UE-dec}
The decision experiment for uncloneable encryption is as follows.
\begin{enumerate}[leftmargin=*]
\item Alice samples the secret key.
\item Alice and Bob perform the preparation phase, with Alice using the key. In a DI protocol, Alice may abort at the end of this phase. If Alice aborts, the experiment terminates.
\item On acceptance, Bob supplies two distinct classical plaintexts $\bmsg_0,\bmsg_1\in\{0,1\}^{\ell}$. He may choose them using his preparation transcript and quantum memory, retaining the resulting quantum state.
\item Alice samples a uniform bit $b$, independently of everything generated so far, and completes the encryption of $\bmsg_b$.
\item Bob splits his information between the two decoders. After receiving the same complete decryption key, they output bits $b_1,b_2$.
\end{enumerate}
The decision bad event is
\[
\BadD=\{b_1=b_2=b\}.
\]
We declare this event false on an aborted run. The protocol is $\eps$-decision secure, or $\eps$-uncloneable-indistinguishable secure, if every allowed attack satisfies in the two respective cases,
\[
\Pr_{\mathrm{DD}}[\BadD] \leq \frac12+\eps, \qquad \Pr_{\mathrm{DI}}[\Acc\land\BadD]
\leq \frac12\Pr_{\mathrm{DI}}[\Acc]+\eps.
\]
The baseline is $\Pr_{\mathrm{DI}}[\Acc]/2$ because acceptance precedes the independent choice of $b$: both decoders can attain this baseline by always outputting the same fixed bit. For one-bit plaintexts, the challenge pair consists of $0,1$.
\end{definition}

\subsubsection{Protocol and search security}
The underlying DD protocol we will be using is the very first uncloneable encryption protocol by Broadbent and Lord \cite{BL20} using Wiesner states.  It has exactly the prepare-and-measure form of Protocol~\ref{prot:DD}. Here $X$ and $A$ are independent uniform bits, $\idealB$ is one qubit, and $\ket{\psi_{x,a}}=H^x\ket a$. Thus a single-copy trusted state prepared by Alice is
\begin{equation}\label{eq:BB84-source}
\rho^{\mathrm{prep}}_{XA\idealB}
=\frac14\sum_{x,a\in\{0,1\}}
\pure{x,a}_{XA}\otimes H^x\pure a H^x.
\end{equation}
This is the only quantum register prepared and sent by Alice in the DD protocol. For either value of $x$, its average state is $\id_{\idealB}/2$, so the common-average condition \eqref{eq:common-average} holds. The entanglement-based implementation takes $\ket\phi_{\idealA\idealB}=(\ket{00}+\ket{11})/\sqrt2$ and takes $T_a^0$ and $T_a^1$ to be the $Z$- and $X$-basis projectors on $\idealA$. Measuring $\idealA$ remotely prepares $H^x\ket a$ on $\idealB$, and hence \eqref{eq:source-realization} holds with $\sfP_{XA}(x,a)=1/4$. Together with Bob's usual CHSH measurements, this state and Alice's measurements form an optimal strategy for the CHSH non-local game, as required by Postulate Set~\ref{post:selftest}\ref{post:game}.

We now state the DD search-security guarantee for this protocol, and describe the DI protocol we get by applying the construction of Protocol~\ref{prot:DI} to it.
\begin{fact}[Search security of Wiesner-state UE, {\cite[Corollary~2 and Theorem~21]{BL20}}]\label{th:BL}
Let $\ueP$ be the uncloneable encryption protocol in which Alice samples a uniform key $\bx\in\{0,1\}^m$ and an independent uniform string $\ba\in\{0,1\}^m$, and sends Bob the state $\bigotimes_{i=1}^m H^{x_i}\pure{a_i}H^{x_i}$. To encrypt $\bmsg\in\{0,1\}^m$, she then sends $\bmsg\oplus\ba$. In the search experiment of Definition~\ref{def:UE-search}, every allowed attack satisfies
\[
\Pr_{\ueP}[\BadS]\leq (\omega^\ast)^m,\qquad \omega^\ast=\frac12+\frac1{2\sqrt2}.
\]
\end{fact}
We note that applying the depolarizing channel $\mapD_\mu$ to the qubits Alice sends first is a possible attack strategy by Bob in the uncloneable encryption protocol, and the protocol is secure against this attack. Therefore, the above fact has the following corollary.
\begin{corollary}[Search security of depolarized Wiesner-state UE]\label{cor:dep-UE}
Let $\ueP_\mu$ be the version of the protocol from Fact~\ref{th:BL} where Alice sends Bob the state $\mapD_\mu^{\otimes m}\left[\bigotimes_{i=1}^m H^{x_i}\pure{a_i}H^{x_i}\right]$ instead. Then,
\[
\Pr_{\ueP_\mu}[\BadS]\leq (\omega^\ast)^m.
\]
\end{corollary}

The DI UE protocol we get by applying the construction of Protocol~\ref{prot:DI} to the Wiesner-state UE protocol is given in Protocol~\ref{prot:UE}. The syndrome used in the protocol is for error correction in the presence of honest device noise. This is very standard in for example QKD; see \cite{KT25} for more details in the context of uncloneable encryption. The syndrome length $s$ depends on the noise tolerance $\epstol$ of the protocol.

\begin{prot}[!ht]
\caption{Sequential DI uncloneable encryption.}\label{prot:UE}

\textbf{Parameters.} The preparation parameters $n,m,\ptest,\epstol$ from Protocol~\ref{prot:DI}, and an $s$-bit error-correction syndrome.
\begin{enumerate}[leftmargin=*,itemsep=3pt]
\item Alice samples a uniform decryption key $\bk\in\{0,1\}^m$.
\item Alice and Bob run the preparation phase of Protocol~\ref{prot:DI} with the CHSH game.
\begin{enumerate}
\item Alice samples $u_i$ before choosing $x_i$.
\begin{itemize}
\item If $u_i=\gen$ and this is the $j$th generation round with $j\leq m$, she sets $x_i=k_j$;
\item Otherwise, she samples $x_i$ uniformly and independently.
\end{itemize}
\item The choice of $y_i$, interaction with Bob and acceptance criterion, and selection of $E$ for Alice are those of Protocol~\ref{prot:DI}, specialized to the CHSH game. On acceptance, $\bx_E=\bk$.
\end{enumerate}
\item Alice stows away her device. To encrypt an $m$-bit plaintext $\bmsg$, she publishes $\bc_{\mathrm{msg}}=\bmsg\oplus\ba_E$ and an $s$-bit error-correction syndrome of $\ba_E$. Bob retains his quantum memory as the quantum part of the ciphertext.
\item Honest decryption with key $\bk$ measures Bob's retained qubits in the corresponding bases, corrects errors using the syndrome, and XORs the resulting string with $\bc_{\mathrm{msg}}$.
\end{enumerate}
The noiseless honest implementation uses EPR pairs, Alice's $Z,X$ measurements, Bob's CHSH measurements on tests, and storage of his halves on generation rounds.
\end{prot}

We note that the joint distribution of $U_i,X_i,Y_i$ and the information supplied to the devices in Protocol~\ref{prot:UE} are exactly those of Protocol~\ref{prot:DI}; sampling $U_i$ first only determines whether Alice uses a key bit or a separate random bit for $X_i$.

The completeness of Protocol~\ref{prot:UE} has two components: for honest implementations with error below the tolerance threshold $\epstol$, upper bounding the abort probability, and upper bounding the probability that Bob is not able to decrypt Alice's encrypted message with the key $\bk$. The abort probability is already upper bounded by Theorem~\ref{th:complete}; Bob correctly decrypting the message is handled by the error correction syndrome --- see \cite{KT25} for further details. We henceforth focus on soundness for Protocol~\ref{prot:UE}, i.e., search security in this case.

Robust CHSH self-testing was established in~\cite{seqMYS12,RUV13}. We use the projector formulation from~\cite[Fact~10]{KT25}.
\begin{fact}[CHSH projector self-testing, {\cite[Fact~10]{KT25}}]\label{cor:CHSH-projectors}
For every pair of binary projective measurement families $\{H_a^x\}_a,\{K_b^y\}_b$ on $M^\A,M^\B$, there are local isometries
\[
V^\A:M^\A\longrightarrow S^\A J^\A,\qquad
V^\B:M^\B\longrightarrow S^\B J^\B,
\]
with $|S^\A|=|S^\B|=2$, depending only on the measurement families, with the following property. For every purified state $\ket\rho_{\stabE M^\A M^\B}$ on which these measurements win CHSH with probability $\omega^\ast-\delta$, there is a normalized $\ket\tau_{\stabE J^\A J^\B}$ satisfying
\begin{equation}\label{eq:CHSH-projectors}
\frac12\sum_{x,a\in\{0,1\}}
\norm{(\id_{\stabE}\otimes V^\A\otimes V^\B)H_a^x\ket\rho
-(T_a^x\otimes\id)\ket\phi\ket\tau}_2^2
\leq\min\{4,10000\sqrt\delta\},
\end{equation}
where $\ket\phi_{S^\A S^\B}=(\ket{00}+\ket{11})/\sqrt2$ and
\[
T_a^0=\frac{\id+(-1)^aZ}{2},\qquad
T_a^1=\frac{\id+(-1)^aX}{2}
\]
are Alice's ideal preparation projectors. The isometries leave $\stabE$ unchanged.
\end{fact}
The extraction isometries in the cited construction depend only on the measurements. The numerical constant above is a loose specialization of the explicit bounds in~\cite[Theorems~1 and~2]{seqMYS12}, on which~\cite[Fact~10]{KT25} is based, with CHSH expectation deficit $8\delta$.

\begin{lemma}[The explicit CHSH preparation cost]\label{lem:CHSH}
Postulate Set~\ref{post:selftest}\ref{post:rigid} holds for CHSH with $\vanf_{\mathrm{CHSH}}(\delta)=\min\{2,100\delta^{1/4}\}$. The finite coherent preparation system satisfies $|S|=8$. Hence
\begin{equation}\label{eq:CHSH-cost}
\Dsharp_2(\mapE_i^{\gen}[\pure\rho]\Vert\mapF_i^{\gen}[\pure\rho])
\leq\log\min\left\{\frac{64}{\pert},\left[\left(1-\frac{7\pert}{8}\right)^{-1/3}+\left(\frac{560000}{\pert}\right)^{1/3}\delta^{1/6}\right]^3\right\}.
\end{equation}
\end{lemma}
\begin{proof}
Equation~\eqref{eq:CHSH-projectors} is precisely the projector estimate required in \eqref{eq:rigidity}, since Alice's CHSH input is uniform. Hence Postulate Set~\ref{post:selftest}\ref{post:rigid} holds with
\[
\vanf_{\mathrm{CHSH}}(\delta)=\min\{2,100\delta^{1/4}\}.
\]
Each of $A_i,S_i^\A,S_i^\B$ has dimension $2$, so $|S|=8$. Substituting this dimension and
$\vanf_{\mathrm{CHSH}}(\delta)^2\leq10000\sqrt\delta$
into Lemma~\ref{lem:round} gives \eqref{eq:CHSH-cost}, with
$|S|(|S|-1)\cdot10000=560000$.
\end{proof}


\begin{theorem}[Sequential DI UE search security]\label{th:UE}
Under Postulate Set~\ref{post:shield}, for any $0<k<-\log \omega^\ast$, the parameters of Protocol~\ref{prot:UE} can be chosen as in Theorem~\ref{th:DI} so that
\[
\Pr_{\diP}[\Acc\land\BadS]\leq2^{s/2}(2^k\omega^\ast)^{m/2}
\]
for every allowed attack. Thus Protocol~\ref{prot:UE} is search-secure in the sense of Definition~\ref{def:UE-search}, with the displayed error bound. The bound decays exponentially when $s/m+k<-\log \omega^\ast$, while the testing tolerance remains a positive constant.
\end{theorem}
\begin{proof}
First, since the version of the protocol where Alice does not send the syndrome has $\Pr_{\ueP_\mu}[\BadS]\leq (\omega^*)^m$ by Corollary~\ref{cor:dep-UE}, we claim that the protocol with the syndrome has security $2^s(\omega^*)^m$. This is because an attack without the syndrome can guess its value uniformly from $2^s$ possibilities and run the attack that uses it. The guess is correct with probability $2^{-s}$ for every raw string, giving the factor $2^s$.

The Wiesner-state realization and Lemma~\ref{lem:CHSH} verify Postulate Set~\ref{post:selftest}. The probability-transfer argument in the proof of Theorem~\ref{th:DI}, with $\compexp{2}\leq km$, gives $2^{km/2}\sqrt{2^s (\omega^\ast)^m}$, as claimed.
\end{proof}

\subsubsection{Decision security from search security}

We obtain one-bit decision security in the sense of Definition~\ref{def:UE-dec} from the search bound, using the simultaneous Goldreich--Levin reduction of Coladangelo, Liu, and Xie.

\begin{fact}[Simultaneous Goldreich--Levin, {\cite[Theorem~3.4]{CLX26}}]\label{th:GL}
Let $p_{\mathrm{search}}$ be the optimal success probability in a search UE game with hidden string $\bmsg$. The parity-prediction decision game corresponding to this search game is identical to the search game for message $\bmsg$ up to the point when the keys are revealed to the decoders. In the reveal phase, the decoders are given a uniformly random string $\br$ along with the original game's secret key, and they must both output $\br\cdot\bmsg$, with the inner product taken over $\mathbb F_2$. If their optimal success probability is $p_{\mathrm{pred}}$, then
\[
p_{\mathrm{search}}\geq(2p_{\mathrm{pred}}-1)^4,
\qquad
p_{\mathrm{pred}}\leq\frac12+\frac12p_{\mathrm{search}}^{1/4}.
\]
\end{fact}

The decision-secure DI UE protocol constructed from applying the parity-prediction construction to Protocol~\ref{prot:UE} is given in Protocol~\ref{prot:decision}, and its security proof is given in Theorem~\ref{cor:decision}.
\begin{prot}[!ht]
\caption{One-bit decision-secure sequential DI UE.}\label{prot:decision}

\textbf{Parameters.} A search-secure instance of Protocol~\ref{prot:UE}, with raw message length $m$.
\begin{enumerate}[leftmargin=*,itemsep=3pt]
\item Alice samples the search-protocol key $\bk$, an independent uniform string $\br\in\{0,1\}^m$, and an independent uniform bit $z$. The complete decryption key is $(\bk,\br,z)$.
\item Alice and Bob run the preparation phase of Protocol~\ref{prot:UE} using $\bk$. On abort, the protocol ends. The variables $\br,z$ are not used during preparation.
\item On acceptance, to encrypt a bit $d$, Alice samples an independent uniform $\bmsg\in\{0,1\}^m$, encrypts it as in Protocol~\ref{prot:UE}, and appends the public bit
\[
e=d\oplus z\oplus(\br\cdot\bmsg).
\]
\item Honest decryption recovers $\bmsg$ using $\bk$ and outputs $e\oplus z\oplus(\br\cdot\bmsg)$.
\end{enumerate}
In the decision experiment, Bob chooses distinct bits $d_0,d_1$ after preparation, Alice encrypts $d_b$ for an independent uniform challenge bit $b$, and both decoders receive $(\bk,\br,z)$ only after the split. Their outputs $b_1,b_2$ define $\BadD$ as in Definition~\ref{def:UE-dec}.
\end{prot}

\begin{theorem}[Decision security with abort]\label{cor:decision}
Suppose the underlying search protocol is $\delta_{\mathrm{UE}}$-search secure. For every attack against Protocol~\ref{prot:decision}, with $p_\Acc=\Pr_{\diP}[\Acc]$,
\begin{equation}\label{eq:decision-bound}
\Pr_{\diP}[\Acc\land\BadD]-\frac{p_\Acc}{2}
\leq\frac12p_\Acc^{3/4}\delta_{\mathrm{UE}}^{1/4}
\leq\frac12\delta_{\mathrm{UE}}^{1/4}.
\end{equation}
In particular, the parameters of Theorem~\ref{th:UE} give decision-security error at most
\[
\frac12\,2^{s/8}(2^k\omega^\ast)^{m/8},
\]
which decays exponentially when $s/m+k<-\log \omega^\ast$. Correct search decryption implies correct decryption of the bit.
\end{theorem}
\begin{proof}
We fix an attack with challenge messages $d_0,d_1$. Since these are distinct bits, $d_b=d_0\oplus b$.

The claim is immediate if $p_\Acc=0$. Otherwise, we fix the attack's preparation and message-selection operations and condition on $\Acc$. The remaining search experiment has optimal success probability at most $\delta_{\mathrm{UE}}/p_\Acc$: any splitting and decoding strategy can be preceded by these fixed operations to give an allowed attack in the original search experiment.

Since $b$ is uniform, the public bit $e=d_0\oplus b\oplus z\oplus(\br\cdot\bmsg)$ is uniform and independent of the search ciphertext, key, $\bmsg,\br,z$, and the chosen message pair. We may therefore include $e$ as independent public randomness in the search experiment and $z$ as an independent key bit. These additions do not increase the probability of $\BadS$, since fixing their values gives an allowed strategy for the original search experiment. Bob can retain and copy $d_0$ to both decoders as part of his classical information.

In the parity-prediction game of Fact~\ref{th:GL}, each decoder uses its decision output $b_j$ to output $b_j\oplus d_0\oplus e\oplus z$. Both parity guesses are correct exactly when $b_1=b_2=b$. Hence
\[
\Pr_{\diP}[\BadD\mid\Acc]
\leq\frac12+\frac12\left(\frac{\delta_{\mathrm{UE}}}{p_\Acc}\right)^{1/4}.
\]
Multiplication by $p_\Acc$ proves \eqref{eq:decision-bound}. Substituting the search error from Theorem~\ref{th:UE} gives the explicit bound.
\end{proof}

\subsection{Copy-protection of multi-bit point functions}\label{sec:copy}

For our second application, we give a sequential DI protocol for quantum copy protection of multi-bit point functions.

For a marked input $\by\in\{0,1\}^{m}$ and marked value $\bv\in\{0,1\}^m$, the multi-bit point function $P_{\by,\bv}$ is defined as
\[
P_{\by,\bv}(\bx)=\begin{cases}\bv,&\bx=\by,\\0^m,&\bx\neq\by.\end{cases}
\]
The family is given by $\{P_{\by,\bv}\}_{(\by,\bv)}$.

The construction of copy protection for point functions from UE, by Gheorghiu, Metger, and Poremba~\cite[Section~6]{GMP22}, needs a wrong-key detection mechanism that uses a quantum random oracle. In our construction, we give a short hash function in the plain instead. Since the underlying UE scheme has exponential security, for the specific distribution used in the security game, we are still able to get exponential security with this construction. However, note that security would likely fail if a stronger security game with adaptively-chosen inputs was considered.

We describe the copy protection security game directly for sequential DI protocols, since our security proof will be based on a reduction to DI UE, instead of a DD copy protection scheme. 
\begin{definition}[Copy-protection security]\label{def:copy-security}
The copy-protection security experiment for the family $\{P_{\by,\bv}\}$ is as follows.
\begin{enumerate}[leftmargin=*]
\item Alice samples independent uniform $\by\in\{0,1\}^{m}$ and $\bv\in\{0,1\}^m$ and runs the protection protocol with Bob. If Alice aborts, the experiment terminates.
\item On acceptance, Bob applies an arbitrary channel to divide the program and all his retained information between two decoders. They may share entanglement established before this operation, but cannot communicate afterwards.
\item Independently for $j=1,2$, Alice samples $X^{(j)}=\by$ with probability $1/2$, and otherwise samples $X^{(j)}$ uniformly from $\{0,1\}^{m}\setminus\{\by\}$. She gives $X^{(j)}$ only to decoder $j$.
\item Decoder $j$ returns $V^{(j)}\in\{0,1\}^m$. The bad event is
\[
\Bad_{\mathrm{copy}}=(V^{(1)}=P_{\by,\bv}(X^{(1)}))\land (V^{(2)}=P_{\by,\bv}(X^{(2)})).
\]
\end{enumerate}
This event is false on abort. The protocol is $\eps$-secure for this distribution if every allowed attack satisfies
\[
\Pr[\Acc\land\Bad_{\mathrm{copy}}]\leq\frac12\Pr[\Acc]+\eps.
\]
\end{definition}
We note that it is reasonable to use other distributions for the challenge inputs $X^{(j)}$ in the copy-protection security game. We have made the choice that gives us a reduction for UE security most easily.

\subsubsection{Protocol and security}
Our protocol for sequential DI copy-protection of multi-bit point functions is given in Protocol~\ref{prot:copy}, its correctness is proved in Lemma~\ref{lem:copy-correct}, and its security is proved in Theorem~\ref{th:copy}. The matrix $R$  used in the protocol specifies a public universal hash function.

\begin{prot}[!ht]
\caption{DI copy-protection of a multi-bit point function.}\label{prot:copy}

\textbf{Parameters.} The parameters of Protocol~\ref{prot:UE}, including an $s$-bit error-correction syndrome, and a hash length $\tau$.

To protect $P_{\by,\bv}$, Alice uses the UE key $\bk=\by$ and plaintext $\bmsg=\bv$.
\begin{enumerate}[leftmargin=*,itemsep=3pt]
\item Alice and Bob run the preparation phase of Protocol~\ref{prot:UE}, using the prescribed key $\bk$ in place of the key sampled in its first step. If Alice aborts, the protection protocol terminates.
\item On acceptance, Alice stows away her device and completes the encryption of $\bmsg$ in the UE protocol, whose classical part consists of $\bc_{\mathrm{msg}}=\bmsg\oplus\ba_E$ and the prescribed syndrome.
\item Alice samples a uniform binary $\tau\times m$ matrix $R$, independently of everything generated so far, and publishes $(R,\bz=R\bk)$ along with the classical ciphertext from the last step. Bob's copy-protected program is the quantum registers retained by him along with this classical data.
\item To evaluate the program at $\bx\in\{0,1\}^m$, Bob first checks whether $R\bx=\bz$. If not, he returns $0^m$. Otherwise, he applies the honest decryption procedure of Protocol~\ref{prot:UE} using $\bx$ as the decryption key and returns the recovered plaintext. For repeated evaluation, the evaluation procedure is implemented coherently, followed by copying its answer and uncomputing.
\end{enumerate}
\end{prot}

\begin{lemma}[Pointwise correctness]\label{lem:copy-correct}
Suppose the honest implementation of Protocol~\ref{prot:UE} has probability of preparation abort or decryption failure at most $\eps_{\mathrm{corr}}$, uniformly over every fixed key $\bk$ and plaintext $\bmsg$. Then, for every fixed function $P_{\by,\bv}$ and every fixed evaluation input $\bx$, Protocol~\ref{prot:copy} has probability of abort or incorrect evaluation at most $\eps_{\mathrm{corr}}+2^{-\tau}$.
\end{lemma}

\begin{proof}
If $\bx=\by=\bk$, the hash check always succeeds. Correct UE decryption then returns $\bmsg=\bv$, as required.

If $\bx\neq\by$, each row of $R$ has a uniform inner product with the nonzero string $\bx\oplus\by$. Since the rows are independent,
\[
\Pr_R[R\bx=R\by]=2^{-\tau}.
\]
Outside this event, evaluation returns the correct value $0^m$ without acting on the quantum ciphertext. Combining this collision bound with the honest abort and decryption-error bound proves the claim.
\end{proof}

For independent noiseless honest devices, UE decryption is exact for every key, so $\eps_{\mathrm{corr}}$ can be taken to be the preparation-abort bound from Theorem~\ref{th:complete}. With honest generation noise, any error-correction procedure used here must satisfy the stated uniform correctness bound.

For a fixed list of $t$ evaluation inputs, the probability that any unmarked input passes the hash check is at most $t2^{-\tau}$. Outside this event, unmarked evaluations leave the quantum ciphertext unchanged. At the marked input, coherent evaluation, copying, and uncomputing preserve the honest program up to the gentle-measurement error. A hybrid argument therefore gives total error
\[
O(t\sqrt{\eps_{\mathrm{corr}}})+t2^{-\tau}.
\]
Note however that these statements concern inputs fixed independently of $R$; they do not assert correctness for inputs chosen adaptively from the published hash.

The following security proof is basically the same as \cite[Proposition~6.11]{GMP22}, with the only difference being we keep track of acceptance more explicitly.
\begin{theorem}[Copy-protection security]\label{th:copy}
Suppose Protocol~\ref{prot:UE} is $\delta_{\mathrm{UE}}$-search secure. Then, in the experiment of Definition~\ref{def:copy-security}, every attack against Protocol~\ref{prot:copy} satisfies
\begin{equation}\label{eq:copy-security}
\Pr[\Acc\land\Bad_{\mathrm{copy}}]
\leq\frac{p_\Acc}{2}+2^{\tau-1}\delta_{\mathrm{UE}},
\qquad p_\Acc=\Pr[\Acc].
\end{equation}
In particular, Theorem~\ref{th:UE} bounds the piracy advantage by
\[
2^{\tau-1+s/2}(2^k\omega^\ast)^{m/2}.
\]
The choice
\[
\tau=\floor{\frac{(-\log\omega^\ast-k)m-s}{4}}
\]
gives exponentially small piracy advantage and hash-collision error whenever $s/m+k$ stays below $-\log\omega^\ast$ by a positive constant.
\end{theorem}

\begin{proof}
Fix a piracy attack. We consider four experiments that use these same operations and differ only in the challenges given to the decoders. We denote their probabilities by
\[
\Pr_{\mathrm{MM}},\qquad
\Pr_{\mathrm{MU}},\qquad
\Pr_{\mathrm{UM}},\qquad
\Pr_{\mathrm{UU}},
\]
where the first and second letters specify decoder 1's and decoder 2's challenges, respectively, so $\mathrm{MU}$ means decoder $1$ receives the marked input and decoder $2$ receives an unmarked input. The decoders receive only their respective inputs, without being told which experiment is performed.

All four experiments have the same distribution of functions and the same operations before the challenges. In particular, their acceptance probability is the same $p_\Acc$. Since Definition~\ref{def:copy-security} chooses each challenge to be marked or unmarked independently with probability $1/2$, its success probability is the equally weighted average of the success probabilities in these four experiments.

In the experiment where both challenges are marked, we define
\begin{align*}
p_j
&=\Pr_{\mathrm{MM}}[\Acc\land (V^{(j)}=\bv)],
\qquad j\in\{1,2\},\\
p_{1,2}
&=\Pr_{\mathrm{MM}}[\Acc\land (V^{(1)}=\bv)\land (V^{(2)}=\bv)].
\end{align*}
Thus $p_j$ is an accepting marginal success probability, whereas $p_{1,2}$ requires both decoders to succeed.

We first bound $p_{1,2}$. With the hash omitted, the experiment where both challenges are marked is precisely the UE search experiment: the key $\bk=\by$ and plaintext $\bmsg=\bv$ are independent and uniform, and both decoders receive $\bk$ after the split. Sampling $\bmsg$ before preparation makes no difference, since preparation does not use it. The public matrix $R$ is independent randomness, and the classical guessing argument already used in Theorem~\ref{th:UE} accounts for the $\tau$-bit hash value. Consequently,
\[
p_{1,2}\leq2^\tau\delta_{\mathrm{UE}}.
\]

We next relate these probabilities to success in the other challenge experiments. In the marked/unmarked experiment, the required answers are $\bv$ and $0^m$. Therefore,
\begin{align*}
\Pr_{\mathrm{MU}}[\Acc\land\Bad_{\mathrm{copy}}]
&=\Pr_{\mathrm{MU}}[\Acc\land V^{(1)}=\bv\land V^{(2)}=0^m]\\
&\leq\Pr_{\mathrm{MU}}[\Acc\land V^{(1)}=\bv]
=p_1.
\end{align*}
For the last equality, we consider each fixed function and the joint state before the challenges. Decoder 1 receives the same input $\by$ in the MM and MU experiments and applies the same decoding operation. Since the decoders cannot communicate, changing decoder 2's input cannot change decoder 1's marginal output distribution. Averaging over functions and preparation outcomes gives the equality, including the acceptance event.

The same argument with the decoders interchanged gives
\begin{align*}
\Pr_{\mathrm{UM}}[\Acc\land\Bad_{\mathrm{copy}}]
 \leq p_2.
\end{align*}
For the two remaining experiments,
\[
\Pr_{\mathrm{UU}}[\Acc\land\Bad_{\mathrm{copy}}]\leq p_\Acc,
\qquad
\Pr_{\mathrm{MM}}[\Acc\land\Bad_{\mathrm{copy}}]=p_{1,2}.
\]

Finally, inclusion--exclusion within the MM experiment gives
\begin{align*}
p_1+p_2 =\Pr_{\mathrm{MM}}[\Acc\land((V^{(1)}=\bv) \lor (V^{(2)}=\bv))]
+p_{1,2}\leq p_\Acc+p_{1,2}.
\end{align*}
Combining the four challenge patterns, we obtain
\begin{align*}
\Pr[\Acc\land\Bad_{\mathrm{copy}}] \leq\frac14(p_\Acc+p_1+p_2+p_{1,2}) \leq \frac{p_\Acc}{2}+\frac{p_{1,2}}2 \leq \frac{p_\Acc}{2}+2^{\tau-1}\delta_{\mathrm{UE}},
\end{align*}
which proves \eqref{eq:copy-security}.

Substituting Theorem~\ref{th:UE} gives the stated piracy advantage. From there it is easy to see that the chosen hash output length satisfies
\[
2^\tau\delta_{\mathrm{UE}}
\leq2^{-\left[(-\log\omega^\ast-k)m-s\right]/4},
\qquad
2^{-\tau}\leq2^{1-\left[(-\log\omega^\ast-k)m-s\right]/4}.
\]
Thus both the piracy advantage and the hash-collision error decay exponentially.
\end{proof}

\subsection{Blind randomness expansion}\label{sec:blind}

Our final application will be to blind randomness expansion. The security condition for it is described below.

\begin{definition}[Blind randomness expansion security]\label{def:blind-security}
A protocol specifies Alice's trusted classical operations, her interaction with the device and Bob, an acceptance rule, and an output length $\ell$. Its security experiment is as follows.
\begin{enumerate}[leftmargin=*]
\item Bob supplies an arbitrary joint state of Alice's device and his own memory, and chooses the device instruments and his own operations. Alice's private random source is initially independent of this state.
\item Alice executes the protocol using her device, her private randomness and classical communication with Bob.
\item Alice either aborts, or accepts and outputs $L\in\{0,1\}^{\ell}$. Bob retains his complete final memory $M^\B_{\mathrm{fin}}$ and the public interaction transcript $C_{\mathrm{pub}}$.
\end{enumerate}
Let $R_{\mathrm{in}}$ record all of Alice's trusted random choices, including any postprocessing seed, and let $\Acc$ denote acceptance. The protocol has secrecy error $\eps_{\mathrm{sec}}$ if every allowed attack satisfies
\begin{equation}\label{eq:blind-secrecy}
\frac12\norm{\rho_{L R_{\mathrm{in}} C_{\mathrm{pub}} M^\B_{\mathrm{fin}}\land\Acc}-\frac{\id_L}{2^\ell}\otimes\rho_{R_{\mathrm{in}} C_{\mathrm{pub}} M^\B_{\mathrm{fin}}\land\Acc}}_1\leq\eps_{\mathrm{sec}}.
\end{equation}
Both states are subnormalized, of trace $p_\Acc=\Pr[\Acc]$. The conditioning registers include all communication, Bob's auxiliary information, and the input randomness whether or not it is disclosed. The completeness error for a specified honest implementation is its abort probability.
\end{definition}

\subsubsection{Protocol and security proof}

In the Magic Square game, independent uniform inputs $x,y\in\{0,1,2\}$ specify Alice's even-parity row $a$ and Bob's odd-parity column $b$. They win when $a[y]=b[x]$, so the predicate $\sfV$ outputs $1$ iff this is the case. The quantum value is $\omega^*=1$, and both output alphabets have size four.

\begin{prot}[!ht]
\caption{Sequential blind randomness expansion from Magic Square.}\label{prot:blind}

\textbf{Parameters.} The parameters $n, m, \ptest, \epstol$ from Protocol~\ref{prot:DI}, with $m=\lfloor(1-\gamma)n/2\rfloor$, and a quantum-proof strong seeded extractor $\operatorname{Ext}$ with output length $\ell$.
\begin{enumerate}[leftmargin=*,itemsep=3pt]
\item In round $i$, Alice samples independent uniform $x_i,y_i$, supplies $x_i$ to her device to get output $a_i$. She sends $y_i$ to Bob and gets $b_i$ back from him.
\item After receiving the outputs, Alice independently samples $U_i$, with $\Pr[U_i=\test]=\ptest$. On a test, she sets $c_i=\sfV(x_i,y_i,a_i,b_i)$ and uses dummy values for $k_i,t_i$. Otherwise, she sets $C_i=\gen$, samples $t_i$ uniformly from $\{0,1,2\}\setminus\{y_i\}$, and sets $k_i=a_i[t_i]$. She discards the unused coordinates of $a_i$.
\item Alice accepts if $\trnorm{\freq_{\bcS}-\sfP_C}\leq\epstol$ and there are at least $m$ generation positions; $E$ consists of the first $m$ such positions. Otherwise she aborts.
\item On acceptance, Alice stows away her device, samples an independent uniform extractor seed $d$, and outputs $L=\operatorname{Ext}(\bk_E,d)$.
\end{enumerate}
\end{prot}

Here $R_{\mathrm{in}}=\bX\bY\mathbf U\bT D$, $C_{\mathrm{pub}}=D\bX\bY\bB\bCS\bT$ on acceptance, and $M^\B_{\mathrm{fin}}=M^\B_{n+1}$. Since $U_i=\gen$ exactly when $C_i=\gen$, the side information in Definition~\ref{def:blind-security} is equivalent, up to redundant classical copies, to $DZ$, where $Z=M^\B_{n+1}\bX\bY\bB\bCS\bT$.
Thus \eqref{eq:blind-secrecy} becomes $\frac12\norm{\rho_{LDZ\land\Acc}-\tau_L\otimes\rho_{DZ\land\Acc}}_1\leq\eps_{\mathrm{sec}}$. Subsequent operations by Bob preserve this bound by data processing.

Robust self-testing for the Magic Square game was established in~\cite{WBM+16}. We use the projector formulation below.
\begin{fact}[Magic Square projector self-testing, {\cite[Lemma~5 and Appendix~B]{KST22}}]\label{cor:MS-joint}
Let $T_a^x,N_b^y$ be the rank-one projectors of a fixed perfect Magic Square strategy on two EPR pairs $\ket\phi_{S^\A S^\B}$, with $|S^\A|=|S^\B|=4$. For every pair of projective measurement families $H_a^x,K_b^y$, there are local extraction isometries $V^\A,V^\B$, depending only on these families, such that every purified state $\ket\rho_{RM^\A M^\B}$ with deficit $\delta$ admits a normalized $\ket\tau_{RJ^\A J^\B}$ satisfying
\begin{equation}\label{eq:MS-joint}
\frac19\sum_{x,y,a,b}\norm{(\id_R\otimes V^\A\otimes V^\B)H_a^x\otimes K_b^y\ket\rho-(T_a^x\otimes N_b^y)\ket\phi\ket\tau}_2^2\leq\vanf_{\mathrm{MS}}(\delta)^2,
\end{equation}
where $\vanf_{\mathrm{MS}}(\delta)=\min\{2,C_{\mathrm{MS}}\delta^{1/4}\}$ for an absolute constant $C_{\mathrm{MS}}$. The isometries leave $R$ unchanged.
\end{fact}
The cited result bounds the state and each party's projectors. The joint-projector form above follows by commuting operators on opposite parties and applying the triangle inequality; the fixed number of outcomes is absorbed into $C_{\mathrm{MS}}$.

In addition, we will need the following two well-known results for the security proof.
\begin{fact}[Local randomness in Magic Square, {\cite[Corollary~3]{FM18}}]\label{fact:MS-local}
Suppose a quantum strategy wins the Magic Square game with probability $1-\delta$ on independent uniform inputs. For fixed $x,y,t\in\{0,1,2\}$ with $t\neq y$, Alice and Bob measure on inputs $x,y$ and record their complete outputs $A,B$. Let $M^\B$ contain all of Bob's remaining quantum systems, including his auxiliary information, and let $\rho^{x,y,t}_{KBM^\B}$ be the resulting state with $K=A_t$. Then
\[
p_{\mathrm{guess}}(K\mid BM^\B)_{\rho^{x,y,t}}\leq\frac12+9\sqrt\delta.
\]
The guessing measurement may depend on $x,y,t$ and the recorded output $B$. In particular, for a perfect strategy,
\[
\rho^{x,y,t}_{KBM^\B}=\frac{\id_K}{2}\otimes\rho^{x,y,t}_{BM^\B}.
\]
\end{fact}

\begin{fact}[Quantum leftover hashing, {\cite[Proposition~7.7]{Tom16}}]\label{fact:blind-hashing}
Let $\rho_{\bK Z}$ be a normalized classical--quantum state with $\bK\in\{0,1\}^m$, and let $D$ be an independent uniform seed selecting a two-universal family of functions $h_d:\{0,1\}^m\to\{0,1\}^{\ell}$. Thus, for distinct $\bk,\bk'$, we have $\Pr_D[h_D(\bk)=h_D(\bk')]\leq2^{-\ell}$. If $L=h_D(\bK)$ and $\tau_L=2^{-\ell}\id_L$, then
\[
\frac12\norm{\rho_{LDZ}-\tau_L\otimes\rho_{DZ}}_1\leq\frac12\sqrt{2^\ell p_{\mathrm{guess}}(\bK\mid Z)_\rho}.
\]
Here $Z$ may contain both classical and quantum side information, and the secrecy bound includes the seed $D$.
\end{fact}

There is no explicit device-dependent theorem we can use a comparison to, in order to get the security of Protocol~\ref{prot:blind}, so the security below is derived using the construction from Protocol~\ref{prot:DI} in a non-blackbox way. We have to explicitly identify what the real and ideal channels are for Protocol~\ref{prot:blind}, and then we can bound their comparison exponent the same way as in Section~\ref{sec:DI}.

\begin{theorem}[Sequential blind randomness expansion security]\label{th:blind}
For every $0<k<1$, there are constants $0<\ptest<1$, $\epstol>0$, independent of $n$, such that Protocol~\ref{prot:blind}, for $(1-\ptest)n\geq4$, satisfies
\begin{equation}\label{eq:blind-guess}
p_\Acc\,p_{\mathrm{guess}}(\bK_E\mid Z)_{\realstate_{|\Acc}}\leq2^{-(1-k)m/2}.
\end{equation}
With an independent uniform seed selecting a two-universal hash $\{0,1\}^m\to\{0,1\}^{\ell}$, its secrecy error satisfies
\[
\eps_{\mathrm{sec}}\leq\frac12\,2^{[\ell-(1-k)m/2]/2}.
\]
For independent honest implementations using the same strategy with $\ptest\dhon<\epstol$, the completeness error is exponentially small.
\end{theorem}
\begin{proof}
Using the setup of Section~\ref{subsec:realandideal}, the real round channel $\mapE_i$ for Protocol~\ref{prot:blind} has continuing memory $M_i^\A M_i^\B$, test record $C_i$, and protected output $O_i=X_iY_iB_iT_iK_i$. The ideal channel is constructed as in that section, now replacing the register $S_i=A_iB_iS_i^\A S_i^\B$, of dimension $256$, with the state
\[
\omega_{S_i}^{\pert,x_i,y_i}=(1-\pert)\sum_{a_i,b_i}\pure{a_i,b_i}_{A_iB_i}\otimes(T_{a_i}^{x_i}\otimes N_{b_i}^{y_i})\pure\phi(T_{a_i}^{x_i}\otimes N_{b_i}^{y_i})+\pert\frac{\id_{S_i}}{256}
\]
in generation rounds. The replacement retains the actual joint marginal on the reference and $J_i^\A J_i^\B$. After the recovery channels and local continuations, here $\mapF_i^\gen$ preserves the replacement records $A_iB_i$, samples $T_i$, computes $K_i$ from $A_i$, discards $A_i$, and sets $C_i=\gen$, as in the real channel.

Equation~\eqref{eq:MS-joint} gives the self-testing bound for the coherently recorded joint outcomes. We follow the proof of Lemma~\ref{lem:round} in Appendix~\ref{app:sharpdivbounds}, now averaging over the input pair $(x_i,y_i)$ and using the outcome registers $A_iB_i$, so that $S_i=A_iB_iS_i^\A S_i^\B$ has dimension $256$. Lemma~\ref{lem:depol} and the same concavity argument bound the divergence before postprocessing. Pinching both outcome registers, applying the recovery channels, and performing the remaining generation-round operations then give the bounds in Lemma~\ref{lem:round}, with $|S|=256$ and $\vanf_{\mathrm{MS}}$ in place of $\vanf_{\game}$.
In particular, $\mapE_i^\gen\leq(256^2/\pert)\mapF_i^\gen$. We set $k'=k/8$ and choose any $0<\pert<1-2^{-k/32}$. Applying Theorem~\ref{th:small-rate} with these $k'$ and $\pert$, and decreasing the tolerance if necessary, gives parameters satisfying
\[
\compexp{2}\leq\frac{k(1-\ptest)n}{8},\qquad
\epstol\leq\min\left\{\frac12,\frac{1-\ptest}{2}\right\},\qquad
\binh(\epstol)\leq\frac{k(1-\ptest)}{8},
\]
where $\binh$ is the binary entropy function, and decreasing the positive tolerance cannot increase the comparison exponent. The prescribed choice $m=\floor{(1-\ptest)n/2}$ gives $m\geq2$ and $(1-\ptest)n/2<m+1$. Hence
\[
\compexp{2}\leq\frac{k(m+1)}4\leq\frac{km}{2},
\qquad
n\binh(\epstol)\leq\frac{k(m+1)}4\leq\frac{km}{2}.
\]

We next bound guessing in the ideal process. For fixed generation inputs $x_i,y_i,t_i$ with $t_i\neq y_i$, the replacement state after computing $K_i$ satisfies
\begin{equation}\label{eq:MS-uniform}
\omega^{\pert,x_i,y_i,t_i}_{K_iB_iS_i^\B}=\frac{\id_{K_i}}2\otimes\omega^{\pert,x_i,y_i,t_i}_{B_iS_i^\B}.
\end{equation}
For the perfect Magic Square component this is Fact~\ref{fact:MS-local}, with $M^\B=S_i^\B$. For the maximally mixed component, $A_i$ is uniform on the four even-parity rows, so each row bit is uniform and independent of $B_iS_i^\B$.

To check that later rounds preserve this independence, we first discard $\bCS$ and condition only on $\mathbf U=\mathbf u$, where $\mathbf u$ has at least $m$ generation positions. This fixes $E$. At this stage we do not condition on acceptance. We write
\[
R_i=\bK_{E\cap[i-1]}\bX_{<i}\bY_{<i}\bB_{<i}\bT_{<i},\qquad
\mathcal J_i^{\B,y}[\xi]=\tr[B_iS_i^\B]{V_i^\B\mapM_i^{\B,y}[\xi](V_i^\B)^\dagger}.
\]
For this calculation, we may compute $K_i$ before recovery, since Alice's recovery channel and $\mapL_i^{\A,x_i}$ preserve the classical record $A_i$. For every incoming state $\xi_{R_iM_i^\A M_i^\B}$, the reduced generation-round state on $R_iK_iB_iS_i^\B J_i^\B$ before recovery is
\[
\omega^{\pert,x_i,y_i,t_i}_{K_iB_iS_i^\B}\otimes\mathcal J_i^{\B,y_i}[\xi_{R_iM_i^\B}].
\]
Indeed, tracing out Alice's registers removes her trace-preserving local instrument from the retained junk marginal. Her recovery channel and $\mapL_i^{\A,x_i}$ also disappear under this trace. Bob's recovery channel followed by $\mapL_i^{\B,y_i}$ acts only on $B_iS_i^\B J_i^\B$, so this composition preserves the uniform factor in \eqref{eq:MS-uniform}; its action on the previous records depends only on the incoming marginal $\xi_{R_iM_i^\B}$. On tests, discarding $C_i$ similarly leaves only Bob's local channel. The test results are not used to select later inputs or device operations. Since each round uses fresh independent inputs, induction shows that the selected bits remain uniform and independent of Bob's memory and the other recorded outputs, even when Alice's memory contains earlier bits. This gives
\begin{equation}\label{eq:blind-before-C}
\tr[\bCS]{\idealstate_{\bK_E Z\mid\mathbf U=\mathbf u}}=\frac{\id_{\bK_E}}{2^m}\otimes\tr[\bCS]{\idealstate_{Z\mid\mathbf U=\mathbf u}}.
\end{equation}

We now include $\bCS$ in the side information and restrict to acceptance. For fixed $\mathbf u$, a compatible record has $c_i=\gen$ on generation positions and $c_i\in\{0,1\}$ on test positions. It is therefore completely determined by the set of losing test positions $\{i:c_i=0\}$. Moreover, $\sfP_C(0)=0$, so the trace-distance acceptance condition implies
\[
\frac1n\bigl|\{i:c_i=0\}\bigr|=\freq_{\bcS}(0)\leq\trnorm{\freq_{\bcS}-\sfP_C}\leq\epstol.
\]
Thus at most $\floor{\epstol n}$ tests lose on acceptance. The number of accepted records compatible with $\mathbf u$ is at most
\[
\sum_{j=0}^{\floor{\epstol n}}\binom nj\leq2^{n\binh(\epstol)},
\]
where the first bound allows losing positions anywhere in $[n]$ and hence also bounds the choices among the test positions. The second uses $\epstol\leq1/2$.

Discarding $\bCS$ averages the conditional states with their probabilities. Each weighted conditional state is a positive semidefinite summand, so
\begin{align*}
&\Pr_{\idealstate}[\bCS=\bcS\mid\mathbf U=\mathbf u]\,\tr[\bCS]{\idealstate_{\bK_E Z|\bcS,\mathbf u}}\\
&\leq\sum_{\bcS'}\Pr_{\idealstate}[\bCS=\bcS'\mid\mathbf U=\mathbf u]\,\tr[\bCS]{\idealstate_{\bK_E Z|\bcS',\mathbf u}}\\
&=\tr[\bCS]{\idealstate_{\bK_E Z|\mathbf u}}=\frac{\id_{\bK_E}}{2^m}\otimes\tr[\bCS]{\idealstate_{Z|\mathbf u}},
\end{align*}
where the last equality is \eqref{eq:blind-before-C}; zero-probability records contribute zero to the sum. Tensoring the $\pure{\bcS}_\bCS$ register to both sides of the above inequality gives us 
\[
\Pr_{\idealstate}[\bCS=\bcS\mid\mathbf U=\mathbf u]\,\idealstate_{\bK_E Z|\bcS,\mathbf u}\leq\frac{\id_{\bK_E}}{2^m}\otimes\left(\pure{\bcS}_{\bCS}\otimes\tr[\bCS]{\idealstate_{Z|\mathbf u}}\right).
\]

Let $\{G_{\bk}^{\mathbf u,\bcS}\}_{\bk\in\{0,1\}^m}$ be a POVM on $Z$ for guessing $\bk$. Its dependence on $\bcS,\mathbf u$ is allowed because the measurement can read $\bCS$, which also determines $\mathbf U$. Applying the preceding operator inequality to each positive guessing effect gives
\begin{align*}
&\Pr_{\idealstate}[\bCS=\bcS\land\widehat{\bK}_E=\bK_E\mid\mathbf U=\mathbf u]\\
&=\Pr_{\idealstate}[\bCS=\bcS\mid\mathbf U=\mathbf u]\sum_{\bk}\tr{(\pure{\bk}_{\bK_E}\otimes G_{\bk}^{\mathbf u,\bcS})\idealstate_{\bK_E Z|\bcS,\mathbf u}}\\
&\leq2^{-m}\sum_{\bk}\tr{G_{\bk}^{\mathbf u,\bcS}\left(\pure{\bcS}_{\bCS}\otimes\tr[\bCS]{\idealstate_{Z|\mathbf u}}\right)}=2^{-m}.
\end{align*}
The last equality uses $\sum_{\bk}G_{\bk}^{\mathbf u,\bcS}=\id$ and normalization of the state in parentheses.
Summing over the accepted records therefore gives
\[
\Pr_{\idealstate}[\Acc\land\widehat{\bK}_E=\bK_E\mid\mathbf U=\mathbf u]\leq2^{-m+n\binh(\epstol)}.
\]
If $\mathbf u$ has fewer than $m$ generation positions, acceptance is impossible. Averaging over $\mathbf u$ proves, for every guessing measurement on $Z$,
\begin{equation}\label{eq:blind-ideal-guess}
\Pr_{\idealstate}[\Acc\land\widehat{\bK}_E=\bK_E]\leq2^{-m+n\binh(\epstol)}.
\end{equation}

For $p_\Acc>0$, we apply Theorem~\ref{th:main}, followed by data processing under the channel that performs an optimal real-state guessing measurement on $Z$ and checks acceptance and correctness of the guess, and then Lemma~\ref{lem:probability}(i). The same measurement is used on the ideal state, where \eqref{eq:blind-ideal-guess} applies. Hence
\[
p_\Acc\,p_{\mathrm{guess}}(\bK_E\mid Z)_{\realstate_{|\Acc}}\leq2^{[\compexp{2}-m+n\binh(\epstol)]/2}\leq2^{-(1-k)m/2}.
\]
If $p_\Acc=0$, the accepted success probability is zero.

Finally, for $p_\Acc>0$, we apply Fact~\ref{fact:blind-hashing} to the normalized accepted state $\realstate_{\bK_E Z|\Acc}$. The seed $D$ is sampled independently after the acceptance decision, so it remains independent of this state. Returning to subnormalized accepted states and using \eqref{eq:blind-guess}, we obtain
\begin{align*}
\frac12\norm{\realstate_{LDZ\land\Acc}-\tau_L\otimes\realstate_{DZ\land\Acc}}_1
&\leq\frac{p_\Acc}{2}\sqrt{2^\ell p_{\mathrm{guess}}(\bK_E\mid Z)_{\realstate_{|\Acc}}}\\
&=\frac12\sqrt{2^\ell p_\Acc\bigl(p_\Acc p_{\mathrm{guess}}(\bK_E\mid Z)_{\realstate_{|\Acc}}\bigr)}\\
&\leq\frac12\sqrt{p_\Acc}\,2^{[\ell-(1-k)m/2]/2}\\
&\leq\frac12\,2^{[\ell-(1-k)m/2]/2}.
\end{align*}
Since $\mathbf U$ is determined by $\bCS$, the side-information registers required in the secrecy definition are obtained from $DZ$ by classical processing and discarding registers. Data processing therefore gives the claimed secrecy bound. Completeness follows from Theorem~\ref{th:complete}.
\end{proof}

Quantum-proof strong seeded extractors with polylogarithmic seed length and linear output also exist for sources whose guessing probability is exponentially small in their length~\cite[Corollary~5.3]{DPVR12}. To use one with inverse-polynomial error $\eps_{\mathrm{ext}}$, we choose it for sources satisfying $p_{\mathrm{guess}}(\bK_E\mid Z)\leq2^{-(1-k)m/4}$. If $p_\Acc\geq2^{-(1-k)m/4}$, this guessing-probability bound for the normalized accepted state follows from \eqref{eq:blind-guess}. Otherwise the secrecy error is at most $p_\Acc$. This gives linear output with secrecy error at most $\eps_{\mathrm{ext}}+2^{-(1-k)m/4}$, while the protocol's total input randomness remains linear.


\section*{Acknowledgements} 

We thank Omar Fawzi, Ashutosh Marwah and Yuming Zhao, for their advice on important details in the proofs. AA conducted research at the Institute for Quantum
Computing, at the University of Waterloo, which is supported by Innovation, Science, and Economic Development
Canada. Support was also provided by NSERC under the
Discovery Grants Program, Grant No. 341495.

\paragraph{AI disclosure.} We used GPT 6 Astra as research and writing assistants in this work. The chief conceptual contribution of the AI model was in the selection of parameters in Theorem~\ref{th:small-rate} in order to make $k'$ arbitrarily small. Aside from this, the AI model contributed in smaller proof improvements, e.g., in getting the supremum outside the sum in the definition of $\compexp{\alpha}$, and Lemma~\ref{lem:depol} (which gives a tighter bound than applying~\cite[Lemma~5.3]{arx_MD23}), and in working out the details of the blind randomness expansion application. The authors have verified all AI outputs, and take full responsibility for the final paper.

\bibliographystyle{Dmax_alpha}
\bibliography{biblio_Dmax_v1}

\appendix

\counterwithin*{equation}{section}
\renewcommand\theequation{\thesection\arabic{equation}}

\section{Minimax property of comparison exponent}
\label{app:minimax}

In this appendix, we establish that the comparison exponent has a particularly useful minimax/duality/variational property, which we used in order to obtain our sequential comparison theorem. The proof is based on essentially standard techniques in convex optimization, but in this case we pay some special attention to handle delicate issues regarding the supports of various operators, which can sometimes affect minimax properties.

This form of the comparison exponent gives the optimal bound obtainable by choosing $f$ in \eqref{eq:subminimax}, although explicit choices of $f$ may still help demonstrate particular scaling behaviours.

\begin{lemma}[Duality relation for $\compexp{\alpha}$]\label{lemma:dual_sharp}
Under Postulate Set~\ref{posts:seq}, the penalized comparison exponent has the equivalent form
\begin{equation}\label{eq:rate-tilt}
\compexp{\alpha}
=\inf_f\left\{\sum_{i=1}^n R^\#_{\alpha,\mapE_i,\mapF_i}(f)
+n\sup_{\sfQ\in S_\Omega}\sum_{\cS}\sfQ(\cS)f(\cS)\right\},
\end{equation}
with the same convention for the sums as in~\eqref{eq:rate}.
\end{lemma}

Before proving this lemma, we establish the properties of the one-round block
quantities that will be used in the minimax argument.
\begin{lemma}\label{lem:h-prop}
For a normalized pure input $\ket{\sigma^i}_{M_iR_i}$, let
\begin{equation}\label{eq:h-definition}
h_{i,\cS}\left(\ket{\sigma^i}_{M_iR_i}\right)
=\Qsharp_\alpha\!\left(
\sigma^{i,\mapE_i}_{\land(\CS_i=\cS)}
\middle\Vert
\sigma^{i,\mapF_i}_{\land(\CS_i=\cS)}
\right).
\end{equation}
When $\sfP_{C_i}(\cS)>0$, this is equivalently
\[
h_{i,\cS}\left(\ket{\sigma^i}_{M_iR_i}\right)
=\sfP_{C_i}(\cS)2^{(\alpha-1)\Dsharp_\alpha\!\left(
\sigma^{i,\mapE_i}_{|\cS}
\middle\Vert
\sigma^{i,\mapF_i}_{|\cS}
\right)}.
\]
$h_{i,\cS}$ satisfies the following properties:
\begin{enumerate}[label=\textup{(\roman*)},ref=\textup{(\roman*)}]
\item\label{prop:h-marginal} The value of $h_{i,\cS}$ depends only on the
marginal $\sigma^i_{M_i}$. We may therefore write it as
$h_{i,\cS}(\sigma^i_{M_i})$.
\item\label{prop:h-concavity} The function $h_{i,\cS}$ is concave in
$\sigma^i_{M_i}$.
\item\label{prop:h-continuity} The function $h_{i,\cS}$ is continuous at
every positive-definite $\sigma^i_{M_i}$.
\end{enumerate}
\end{lemma}
\begin{proof}
To prove item~\ref{prop:h-marginal}, consider two purifications of the same
input marginal. After enlarging one reference register if necessary, they are
related by an isometry on the reference register. Applying the real or ideal
map and selecting the block $\CS_i=\cS$ commutes with this isometry. Data
processing in
Fact~\ref{th:sharp-facts}\ref{prop:sharp-data-processing} gives one
inequality between the corresponding sharp-divergence quantities. Applying data
processing to a channel that reverses the isometry on its range gives the
opposite inequality. The two values are therefore equal, proving that
$h_{i,\cS}$ depends only on the input marginal.

To prove item~\ref{prop:h-concavity}, suppose that
$\sigma^i_{\inM_i}=\sum_jt_j\sigma^{i,j}_{\inM_i}$. Choose an arbitrary purification
$\ket{\sigma^i}_{\inM_iR_i}$ of $\sigma^i_{\inM_i}$ and, independently, arbitrary
purifications $\ket{\sigma^{i,j}}_{\inM_iR'_i}$ of the component states on one
sufficiently large register $R'_i$. Apply
Fact~\ref{th:sharp-facts}\ref{prop:sharp-purified-concavity} to the two
completely positive maps obtained by applying $\mapE_i$ and $\mapF_i$ and then
selecting the $\cS$ block. It gives
\[
h_{i,\cS}(\sigma^i_{\inM_i})
\geq\sum_jt_jh_{i,\cS}(\sigma^{i,j}_{\inM_i}).
\]

It remains to prove item~\ref{prop:h-continuity}. Projecting the operator
inequality in Postulate
Set~\ref{posts:seq}\ref{post:domination} onto the block $\CS_i=\cS$
gives
\[
\sigma^{i,\mapE_i}_{\land(\CS_i=\cS)}
\leq
\Lambda\sigma^{i,\mapF_i}_{\land(\CS_i=\cS)}.
\]
If $\sfP_{C_i}(\cS)>0$, the two partial states have the same trace, so division by
$\sfP_{C_i}(\cS)$ gives
\[
\Dmax\!\left(
\sigma^{i,\mapE_i}_{|\cS}
\middle\Vert
\sigma^{i,\mapF_i}_{|\cS}
\right)
\leq\log\Lambda.
\]
The comparison
Fact~\ref{th:sharp-facts}\ref{prop:sharp-renyi-comparison} and the
conditional-state expression following \eqref{eq:h-definition} now imply
\[
h_{i,\cS}(\sigma^i_{M_i})
\leq
\sfP_{C_i}(\cS)\Lambda^{\alpha-1}
\leq
\Lambda^{\alpha-1}.
\]
When $\sfP_{C_i}(\cS)=0$, both partial blocks are zero, and the same bound holds.
Consequently, item~\ref{prop:h-concavity} shows that $h_{i,\cS}$ is a
finite-valued concave function on the convex state space $S_=(M_i)$. Every
finite-valued concave function is continuous on the relative interior of its
domain. The relative interior of $S_=(M_i)$ consists exactly of the
positive-definite states, which proves item~\ref{prop:h-continuity}.
\end{proof}

\begin{proof}[Proof of Lemma~\ref{lemma:dual_sharp}]
The direct-sum and scaling identities in
Fact~\ref{th:sharp-facts}\ref{prop:sharp-direct-sum} and
Fact~\ref{th:sharp-facts}\ref{prop:sharp-scaling} give, for each fixed
$\ket{\sigma^i}_{\inM_iR_i}$,
\[
\Qsharp_\alpha\!\left(
\sigma^{i,\mapE_i}
\middle\Vert
\mapF_i^f\!\left[\pure{\sigma^i}\right]
\right)
=\sum_{\cS}\sfP_{C_i}(\cS)2^{-(\alpha-1)f(\cS)}
2^{(\alpha-1)\Dsharp_\alpha\!\left(
\sigma^{i,\mapE_i}_{|\cS}
\middle\Vert
\sigma^{i,\mapF_i}_{|\cS}
\right)}.
\]
Taking the logarithm and optimizing gives
\begin{equation}\label{eq:R_f}
R^\#_{\alpha,\mapE_i,\mapF_i}(f)
=\sup_{\ket{\sigma^i}}\frac1{\alpha-1}\log
\sum_{\cS}\sfP_{C_i}(\cS)\,
2^{(\alpha-1)\left(
\Dsharp_\alpha\!\left(
\sigma^{i,\mapE_i}_{|\cS}
\middle\Vert
\sigma^{i,\mapF_i}_{|\cS}
\right)-f(\cS)
\right)}.
\end{equation}

Now, fixing $\sigma^i_{M_i}$ and $f$, define
\begin{equation}\label{eq:wZ}
w_i(\cS)
=h_{i,\cS}(\sigma^i_{\inM_i})2^{-(\alpha-1)f(\cS)},
\qquad
Z_i=\sum_{\cS}w_i(\cS)
\end{equation}
where $h_{i,\cS}$ is as in Lemma~\ref{lem:h-prop}. The ratio
$w_i(\cS)/Z_i$ defines a probability distribution on $\alphCS$. For every
distribution $\sfQ_{C_i}$ satisfying
$\operatorname{supp}(\sfQ_{C_i})\subseteq\{\cS:w_i(\cS)>0\}$,
\begin{align*}
D\!\left(\sfQ_{C_i}\middle\Vert\frac{w_i}{Z_i}\right)
=\sum_{\cS}\sfQ_{C_i}(\cS)
\log\frac{\sfQ_{C_i}(\cS)}{w_i(\cS)/Z_i}=\log Z_i-
\sum_{\cS}\sfQ_{C_i}(\cS)\log\frac{w_i(\cS)}{\sfQ_{C_i}(\cS)}.
\end{align*}
Using the definition of $w_i$, the above rearranges to
\begin{align}
\log Z_i-D\!\left(\sfQ_{C_i}\middle\Vert w_i/Z_i\right) = -(\alpha-1)\sum_{\cS}\sfQ_{C_i}(\cS)f(\cS)
+\sum_{\cS}\sfQ_{C_i}(\cS)
\log\frac{h_{i,\cS}(\sigma^i_{\inM_i})}{\sfQ_{C_i}(\cS)}.
\label{eq:logmoment-kl}
\end{align}
The relative entropy on the lefthand side of
\eqref{eq:logmoment-kl} is non-negative. Dividing that equation by
$\alpha-1>0$ therefore gives
\[
-\sum_{\cS}\sfQ_{C_i}(\cS)f(\cS)
+\frac1{\alpha-1}\sum_{\cS}\sfQ_{C_i}(\cS)
\log\frac{h_{i,\cS}(\sigma^i_{\inM_i})}{\sfQ_{C_i}(\cS)}
\leq\frac{\log Z_i}{\alpha-1}.
\]
Equality is attained by $\sfQ_{C_i}=w_i/Z_i$, because then
$D(\sfQ_{C_i}\Vert w_i/Z_i)=0$. Hence the supremum of the lefthand side over
$\sfQ_{C_i}$ equals $\log Z_i/(\alpha-1)$. Rewriting \eqref{eq:R_f} in terms of
$h_{i,\cS}$ gives
\begin{align}
R^\#_{\alpha,\mapE_i,\mapF_i}(f)
& =\sup_{\ket{\sigma^i}}\frac1{\alpha-1}\log\left(
\sum_{\cS}h_{i,\cS}(\sigma^i_{M_i})2^{-(\alpha-1)f(\cS)}\right) \nonumber \\
& = \sup_{\ket{\sigma^i}}\frac1{\alpha-1}\log Z_i \nonumber \\
& = \sup_{\ket{\sigma^i},\sfQ_{C_i}}
\left\{-\sum_{\cS}\sfQ_{C_i}(\cS)f(\cS)
+\frac1{\alpha-1}\sum_{\cS}\sfQ_{C_i}(\cS)
\log\frac{h_{i,\cS}(\sigma^i_{\inM_i})}{\sfQ_{C_i}(\cS)}\right\} \nonumber \\
& = \sup_{\sigma^i_{M_i},\sfQ_{C_i}}
\left\{-\sum_{\cS}\sfQ_{C_i}(\cS)f(\cS)
+\frac1{\alpha-1}\sum_{\cS}\sfQ_{C_i}(\cS)
\log\frac{h_{i,\cS}(\sigma^i_{\inM_i})}{\sfQ_{C_i}(\cS)}\right\}. \label{eq:R-variational}
\end{align}
where the last line follows from
Lemma~\ref{lem:h-prop}\ref{prop:h-marginal}.

We will now substitute \eqref{eq:R-variational} into the righthand side of \eqref{eq:rate-tilt}. Introducing $\sfQ_C^0\in S_\Omega$ to express the inner supremum in
\eqref{eq:rate-tilt}, we get
\begin{align}
\inf_f\sup_{\substack{\sigma^1_{\inM_1},\ldots,\sigma^n_{\inM_n};\\ 
\sfQ_C^0\in S_\Omega;\\ \sfQ_{C_1},\ldots,\sfQ_{C_n}}}
\Bigg\{&
\sum_{i=1}^n\frac1{\alpha-1}\sum_{\cS}\sfQ_{C_i}(\cS)
\log\frac{h_{i,\cS}(\sigma^i_{\inM_i})}{\sfQ_{C_i}(\cS)}+\sum_{\cS}\left(n\sfQ_C^0(\cS)-\sum_{i=1}^n\sfQ_{C_i}(\cS)\right)f(\cS)
\Bigg\}.
\label{eq:minimax-expanded}
\end{align}
Denote the quantity in \eqref{eq:minimax-expanded} by $A_0$. We obtain it
as a limit of minimax problems whose input marginals are uniformly positive
definite. For $0<z<1$, let $\Pos_{\geq z}(M_i)$ be the set of normalized
states on $M_i$ that are lower bounded by $z\id_{M_i}/|M_i|$, and define
\begin{align}
A_z = \inf_f\sup_{\substack{\sigma^1_{\inM_1} \in \Pos_{\geq z}(M_1),\ldots,\sigma^n_{\inM_n} \in \Pos_{\geq z}(M_n);\\ 
\sfQ_C^0\in S_\Omega;\\ \sfQ_{C_1},\ldots,\sfQ_{C_n}}}
\Bigg\{&
\sum_{i=1}^n\frac1{\alpha-1}\sum_{\cS}\sfQ_{C_i}(\cS)
\log\frac{h_{i,\cS}(\sigma^i_{\inM_i})}{\sfQ_{C_i}(\cS)}\nonumber\\
&+\sum_{\cS}\left(n\sfQ_C^0(\cS)-\sum_{i=1}^n\sfQ_{C_i}(\cS)\right)f(\cS)
\Bigg\}.
\label{eq:minimax-z}
\end{align}
We first verify the limiting relation. Given any input marginal
$\sigma^i_{M_i}$, set
\[
\sigma^{i,z}_{M_i}=(1-z)\sigma^i_{M_i}+z\frac{\id_{M_i}}{|M_i|}.
\]
Then $\sigma^{i,z}_{M_i}\in\Pos_{\geq z}(M_i)$. Non-negativity and
Lemma~\ref{lem:h-prop}\ref{prop:h-concavity} give
\[
h_{i,\cS}(\sigma^{i,z}_{M_i})
\geq(1-z)h_{i,\cS}(\sigma^i_{M_i}).
\]
For fixed $f$, $\sfQ_C^0$, and $\sfQ_{C_1},\ldots,\sfQ_{C_n}$, replacing every
$\sigma^i_{M_i}$ by $\sigma^{i,z}_{M_i}$ therefore decreases the objective in
\eqref{eq:minimax-expanded} by at most
$-n\log(1-z)/(\alpha-1)$. More explicitly,
\[
\frac1{\alpha-1}\sum_{i=1}^n\sum_{\cS}\sfQ_{C_i}(\cS)
\log\frac{h_{i,\cS}(\sigma^{i,z}_{M_i})}{\sfQ_{C_i}(\cS)}
\geq
\frac1{\alpha-1}\sum_{i=1}^n\sum_{\cS}\sfQ_{C_i}(\cS)
\log\frac{h_{i,\cS}(\sigma^i_{M_i})}{\sfQ_{C_i}(\cS)}
+\frac{n}{\alpha-1}\log(1-z),
\]
where we used $\sum_{\cS}\sfQ_{C_i}(\cS)=1$ for every $i$. This bound is
independent of $f$ and of the distributions. Taking the relevant suprema and
then the infimum over $f$ gives
\begin{equation}\label{eq:Az-limit}
A_0+\frac{n}{\alpha-1}\log(1-z)
\leq A_z\leq A_0.
\end{equation}
In particular, $\lim_{z\to0}A_z=A_0$.

We now apply Fact~\ref{th:minimax} to $A_z$ for fixed $z>0$. The
domain of $f$, namely the vector space of real functions on $\alphCS$, is
convex. The product of $S_\Omega$, the probability simplices containing the
$\sfQ_{C_i}$, and the sets $\Pos_{\geq z}(M_i)$ is compact and convex. For fixed
distributions and input marginals, the objective in
\eqref{eq:minimax-z} is affine and continuous in $f$.

First we note that each $h_{i,\cS}$ is concave and continuous on
$\Pos_{\geq z}(M_i)$ by
Lemma~\ref{lem:h-prop}\ref{prop:h-concavity} and
Lemma~\ref{lem:h-prop}\ref{prop:h-continuity}. If $\sfP_{C_i}(\cS)=0$
for one positive-definite input $\sigma^i_{M_i}$, then it is zero for every input, because
a positive-definite state dominates a positive multiple of every state and
outcome probabilities are positive linear functions of the input. In this case
$h_{i,\cS}(\sigma^i_{M_i})=0$ for every input state $\sigma^i_{M_i}$, and we may
restrict $\sfQ_{C_i}(\cS)=0$, since otherwise the corresponding summand is
$-\infty$. We henceforth consider only the $i,\cS$ for which
$h_{i,\cS}(\sigma^i_{M_i})>0$ for every
$\sigma^i_{M_i}\in\Pos_{\geq z}(M_i)$.

The function
\[
g:(u,v)\longmapsto u\log\frac vu,
\qquad g(0,v)=0,
\]
is non-decreasing in $v$, continuous on $u\geq 0, v > 0$, and jointly concave in
both variables. Since $h_{i,\cS}$ takes positive values for all cases under consideration, and composing a continuous function with a continuous
function that has the correct domain produces a continuous function,
\[
(\sfQ_{C_i},\sigma^i_{M_i})\longmapsto
g\bigl(\sfQ_{C_i}(\cS),h_{i,\cS}(\sigma^i_{M_i})\bigr)
\]
is continuous. Moreover, composing an outer monotonically non-decreasing
function with an inner concave function produces a concave function, and
therefore this is also concave. Taking the sum over the $i$-s and $\cS$-s
preserves continuity and concavity, so the overall right-hand side of
\eqref{eq:minimax-z} for $z>0$ is continuous and concave in the
$\sfQ_{C_i}$-s and $\sigma^i$-s, and we can apply
Fact~\ref{th:minimax}.

After the minimax exchange from Fact~\ref{th:minimax}, the infimum over $f$ in $A_z$ for $z>0$ can be evaluated explicitly. If
\[
\sum_{i=1}^n\sfQ_{C_i}(\cS)=n\sfQ_C^0(\cS)
\quad\text{for every }\cS,
\]
the $f$-dependent part of \eqref{eq:minimax-z} is zero for every $f$. If this
equality fails at some $\cS_0$, we send $f(\cS_0)$ to the sign of infinity that makes
its nonzero coefficient contribute $-\infty$. Thus all and only the surviving
points must satisfy
$\frac1n\sum_i\sfQ_{C_i}=\sfQ_C^0\in S_\Omega$. This gives
\begin{equation}\label{eq:A_z-final}
A_z
=
\sup_{\substack{\sigma^1_{\inM_1}\in\Pos_{\geq z}(M_1),\ldots,
\sigma^n_{\inM_n}\in\Pos_{\geq z}(M_n);\\
\sfQ_{C_1},\ldots,\sfQ_{C_n}:\ \frac1n\sum_i\sfQ_{C_i}\in S_\Omega}}
\frac1{\alpha-1}\sum_{i=1}^n\sum_{\cS}
\sfQ_{C_i}(\cS)\log
\frac{h_{i,\cS}(\sigma^i_{\inM_i})}{\sfQ_{C_i}(\cS)}.
\end{equation}
Now using
\eqref{eq:h-definition}, for every $i$ we can rewrite
\[
\frac1{\alpha-1}\sum_{\cS}\sfQ_{C_i}(\cS)
\log\frac{h_{i,\cS}(\sigma^i_{M_i})}{\sfQ_{C_i}(\cS)}
=
\sum_{\cS}\sfQ_{C_i}(\cS)
\Dsharp_\alpha\!\left(
\sigma^{i,\mapE_i}_{|\cS}
\middle\Vert
\sigma^{i,\mapF_i}_{|\cS}
\right)
-\frac{D(\sfQ_{C_i}\Vert\sfP_{C_i})}{\alpha-1}.
\]
Consequently, the righthand side of \eqref{eq:A_z-final} is exactly the
definition of $\compexp{\alpha}$ with the input marginals restricted to
$\Pos_{\geq z}(M_i)$; here
Lemma~\ref{lem:h-prop}\ref{prop:h-marginal} allows the original optimization
over purified inputs to be written over their marginals. It follows
immediately that
\[
A_z\leq\compexp{\alpha}.
\]
Conversely, apply the approximation
$\sigma^i_{M_i}\mapsto\sigma^{i,z}_{M_i}$ to any feasible point in the definition
of the unrestricted comparison exponent.
Lemma~\ref{lem:h-prop}\ref{prop:h-concavity} and the calculation preceding
\eqref{eq:Az-limit} show that the objective decreases by at most
$-n\log(1-z)/(\alpha-1)$. Taking the supremum gives
\[
\compexp{\alpha}+\frac{n}{\alpha-1}\log(1-z)
\leq A_z\leq\compexp{\alpha}.
\]
Therefore $\lim_{z\to0}A_z=\compexp{\alpha}$. Combining this with
\eqref{eq:Az-limit} gives
\[
A_0
=\lim_{z\to0}A_z
=\compexp{\alpha},
\]
which proves \eqref{eq:rate-tilt}.
\end{proof}

\section{Bounding single-round sharp-divergences}
\label{app:sharpdivbounds}

This appendix is centred entirely on proving Lemma~\ref{lem:round}. The proof basically separates into a self-testing bound for a coherent state, a divergence bound after averaging over the input, and an identification of the operations that give the real and ideal generation channels. The desired lemma then follows quite straightforwardly by combining these results, as we describe at the end of this appendix.

We follow the notation presented in Sec.~\ref{subsec:single-round bound} leading up to Lemma~\ref{lem:round}, and assume Postulate Sets~\ref{post:selftest} and~\ref{post:shield} throughout this analysis.
Let $V_i^\A,V_i^\B$ denote the extraction isometries already chosen from the full families $(\{H^{x_i}_{i,a_i}\}_{a_i})_{x_i},(\{K^{y_i}_{i,b_i}\}_{b_i})_{y_i}$, and write $J=RJ_i^\A J_i^\B$. The following normalized vectors will be used in the proof:
\begin{align*}
\ket{\phi_{x_i}}_{S_i}
&=\sum_{a_i}\ket{a_i}_{A_i}(T_{a_i}^{x_i}\otimes\id)\ket\phi_{S_i^\A S_i^\B},\\
\ket{\zeta_{x_i}}_{S_iJ}
&=\sum_{a_i}\ket{a_i}_{A_i}
(\id_R\otimes V_i^\A\otimes V_i^\B)
H_{i,a_i}^{x_i}(\ket\rho\otimes\ket{00}_{\ancA_i\ancB_i}).
\end{align*}
Here $S_i=A_iS_i^\A S_i^\B$ as before. We write $\zeta_{S_iJ}^{x_i}=\pure{\zeta_{x_i}}$ and define
\[
\eta_{x_i}=1-\bra{\phi_{x_i}}\zeta_{S_i}^{x_i}\ket{\phi_{x_i}}.
\]

\begin{lemma}[Self-testing of the coherent state]\label{lem:coh}
There is a normalized vector $\ket{\tau_\rho}_J$ such that
\[
\sum_{x_i}\sfP_X(x_i)
\norm{\ket{\zeta_{x_i}}-\ket{\phi_{x_i}}\ket{\tau_\rho}}_2^2
\leq\vanf_\game(\delta)^2.
\]
Consequently,
\[
\sum_{x_i}\sfP_X(x_i)\eta_{x_i}\leq\vanf_\game(\delta)^2.
\]
\end{lemma}

\begin{proof}
The state $\pure{\zeta_{x_i}}$ can be obtained in the following way. Consider the coherent version of Alice's measurement channel
\[
\widetilde{\mapM}_i^{\A,x_i}[\sigma_{M_i^\A}]
=\sum_{a_i,a'_i}\ket{a_i}\bra{a'_i}_{A_i}\otimes
H_{i,a_i}^{x_i}(\sigma_{M_i^\A}\otimes\pure0_{\ancA_i})H_{i,a'_i}^{x_i}.
\]
Applying this channel to $\ket\rho$, appending Bob's ancilla, and then applying $V_i^\A,V_i^\B$ gives $\pure{\zeta_{x_i}}$. Pinching the $A_i$ register of $\pure{\zeta_{x_i}}$ gives the same state as using $\mapM_i^{\A,x_i}$ in place of $\widetilde{\mapM}_i^{\A,x_i}$ on $\ket{\rho}$, and then applying the isometries. Similarly, pinching $A_i$ in $\pure{\phi_{x_i}}$ gives $\omega_{S_i}^{x_i}$ from \eqref{eq:replace-state}.

We apply Postulate Set~\ref{post:selftest}\ref{post:rigid} to the strategy $\strat$, with purification $\ket\rho\ket{00}_{\ancA_i\ancB_i}$. It supplies a normalized $\ket{\tau_\rho}_J$ for which \eqref{eq:rigidity} bounds the action of each Alice projector. Orthogonality of the outcome labels in the definitions of $\ket{\zeta_{x_i}}$ and $\ket{\phi_{x_i}}$ gives
\begin{align*}
&\sum_{x_i}\sfP_X(x_i)
\norm{\ket{\zeta_{x_i}}-\ket{\phi_{x_i}}\ket{\tau_\rho}}_2^2\\
&=\sum_{x_i,a_i}\sfP_X(x_i)
\bigl\|(\id_R\otimes V_i^\A\otimes V_i^\B)H_{i,a_i}^{x_i}
\ket\rho\ket{00}_{M'^\A_iM'^\B_i}
-(T_{a_i}^{x_i}\otimes\id)\ket\phi\ket{\tau_\rho}\bigr\|_2^2\\
&\leq\vanf_\game(\delta)^2.
\end{align*}

For the second claim, notice that,
\begin{align*}
\eta_{x_i}
&=\norm{\bigl((\id_{S_i}-\pure{\phi_{x_i}})\otimes\id_J\bigr)
\ket{\zeta_{x_i}}}_2^2\\
&=\norm{\bigl((\id_{S_i}-\pure{\phi_{x_i}})\otimes\id_J\bigr)
(\ket{\zeta_{x_i}}-\ket{\phi_{x_i}}\ket{\tau_\rho})}_2^2\\
&\leq\norm{\ket{\zeta_{x_i}}-\ket{\phi_{x_i}}\ket{\tau_\rho}}_2^2
\end{align*}
where the second line is because $\id_{S_i}-\pure{\phi_{x_i}}$ is orthogonal to $\pure{\phi_{x_i}}$. Averaging over $x_i$ proves the second claim.
\end{proof}

For $0<\pert<1$, we define the depolarized pure states
\[
\widetilde{\omega}_{S_i}^{\pert,x_i}
=(1-\pert)\pure{\phi_{x_i}}_{S_i}+\pert\frac{\id_{S_i}}{|S_i|}.
\]
Including Alice's classical input register gives the two states
\begin{align*}
\zeta_{X_iS_iJ}
&=\sum_{x_i}\sfP_X(x_i)\pure{x_i}_{X_i}\otimes\zeta_{S_iJ}^{x_i},\\
\widetilde{\omega}_{X_iS_iJ}^{\pert}
&=\sum_{x_i}\sfP_X(x_i)\pure{x_i}_{X_i}\otimes
\widetilde{\omega}_{S_i}^{\pert,x_i}\otimes\zeta_J^{x_i}.
\end{align*}
We will now prove a lemma bounding the $\Dsharp_2$ between these two states. For this, it will be useful to define
\begin{equation}\label{eq:round-cost}
\vanf_{\pert,\game}(\delta)=
\min\left\{
\log\frac{a_\pert |S|^2}{\pert},
3\log\left[1+
\left(\frac{a_\pert |S|(|S|-1)
\vanf_{\game}(\delta)^2}{\pert}\right)^{1/3}\right]
\right\}
\end{equation}
where $a_\pert=1-\pert+\pert/|S|$.

\begin{lemma}[Divergence after averaging the input]\label{lem:input-div}
The states defined above satisfy
\begin{gather*}
\zeta_{X_iS_iJ}
\leq\frac{|S|^2}{\pert}\widetilde{\omega}_{X_iS_iJ}^{\pert},\\
\Dsharp_2\left(\zeta_{X_iS_iJ}\middle\Vert
\widetilde{\omega}_{X_iS_iJ}^{\pert}\right)
\leq\log a_\pert^{-1}+\vanf_{\pert,\game}(\delta).
\end{gather*}
For every fixed $\pert>0$, the function $\vanf_{\pert,\game}$ is vanishing.
\end{lemma}

\begin{proof}

For each $x_i$, the coherent replacement channel
\[
\widetilde{\mapP}_i^{\pert,x_i}[\sigma_{S_iJ_i^\A J_i^\B}]
=\widetilde{\omega}_{S_i}^{\pert,x_i}\otimes\tr[S_i]{\sigma_{S_iJ_i^\A J_i^\B}}
\]
has the same input and output registers as $\mapP_i^{\pert,x_i}$. It sends $\zeta_{S_iJ}^{x_i}$ to $\widetilde{\omega}_{S_i}^{\pert,x_i}\otimes\zeta_J^{x_i}$. Thus the state $\widetilde{\omega}_{X_iS_iJ}^{\pert}$ is obtained from $\zeta_{X_iS_iJ}$ by applying this channel controlled by $X_i$. The channel retains the actual marginal and does not use $\ket{\tau_\rho}$.

Lemma~\ref{lem:depol} applies to $\zeta_{S_iJ}^{x_i}$ and $\pure{\phi_{x_i}}_{S_i}\otimes\zeta_J^{x_i}$. Its error parameter is exactly $\eta_{x_i}$. This gives both
\[
\zeta_{S_iJ}^{x_i}
\leq \frac{|S|^2}{\pert}
\widetilde{\omega}_{S_i}^{\pert,x_i}\otimes\zeta_J^{x_i}
\]
and
\[
2^{\Dsharp_2\left(\zeta_{S_iJ}^{x_i}\middle\Vert
\widetilde{\omega}_{S_i}^{\pert,x_i}\otimes\zeta_J^{x_i}\right)}
\leq
\left[
a_\pert^{-1/3}
+\left(\frac{|S|(|S|-1)\eta_{x_i}}{\pert}\right)^{1/3}
\right]^3.
\]

Multiplying the operator inequalities by $\sfP_X(x_i)$ and summing over the orthogonal $X_i$ blocks proves
\[
\zeta_{X_iS_iJ}
\leq\frac{|S|^2}{\pert}\widetilde{\omega}_{X_iS_iJ}^{\pert}.
\]
Consequently,
\[
\Dsharp_2\left(\zeta_{X_iS_iJ}\middle\Vert
\widetilde{\omega}_{X_iS_iJ}^{\pert}\right)
\leq
\Dmax\left(\zeta_{X_iS_iJ}\middle\Vert
\widetilde{\omega}_{X_iS_iJ}^{\pert}\right)
\leq\log\frac{|S|^2}{\pert}.
\]

To get the sharper bound on $\Dsharp_2$, the direct-sum rule in Fact~\ref{th:sharp-facts}\ref{prop:sharp-direct-sum} gives
\begin{align*}
2^{\Dsharp_2\left(\zeta_{X_iS_iJ}\middle\Vert
\widetilde{\omega}_{X_iS_iJ}^{\pert}\right)}
&=\sum_{x_i}\sfP_X(x_i)
2^{\Dsharp_2\left(\zeta_{S_iJ}^{x_i}\middle\Vert
\widetilde{\omega}_{S_i}^{\pert,x_i}\otimes\zeta_J^{x_i}\right)}\\
&\leq\sum_{x_i}\sfP_X(x_i)
\left[
a_\pert^{-1/3}
+\left(\frac{|S|(|S|-1)\eta_{x_i}}{\pert}\right)^{1/3}
\right]^3.
\end{align*}

To bound the average, we expand the function of $\eta$ on the right:
\begin{align*}
\left[
a_\pert^{-1/3}
+\left(\frac{|S|(|S|-1)\eta}{\pert}\right)^{1/3}
\right]^3
&=a_\pert^{-1}
+3a_\pert^{-2/3}
\left(\frac{|S|(|S|-1)}{\pert}\right)^{1/3}\eta^{1/3}\\
&\quad
+3a_\pert^{-1/3}
\left(\frac{|S|(|S|-1)}{\pert}\right)^{2/3}\eta^{2/3}
+\frac{|S|(|S|-1)}{\pert}\eta.
\end{align*}
Each power of $\eta$ on the right is nondecreasing and concave on $[0,\infty)$, and every coefficient is nonnegative. Hence their sum is also nondecreasing and concave. Jensen's inequality and Lemma~\ref{lem:coh} therefore give
\begin{align*}
&\sum_{x_i}\sfP_X(x_i)
\left[
a_\pert^{-1/3}
+\left(\frac{|S|(|S|-1)\eta_{x_i}}{\pert}\right)^{1/3}
\right]^3\\
&\leq
\left[
a_\pert^{-1/3}
+\left(
\frac{|S|(|S|-1)}{\pert}
\sum_{x_i}\sfP_X(x_i)\eta_{x_i}
\right)^{1/3}
\right]^3\\
&\leq
\left[
a_\pert^{-1/3}
+\left(
\frac{|S|(|S|-1)\vanf_{\game}(\delta)^2}{\pert}
\right)^{1/3}
\right]^3.
\end{align*}
The first inequality uses concavity, and the second uses monotonicity together with $\sum_{x_i}\sfP_X(x_i)\eta_{x_i}\leq\vanf_\game(\delta)^2$.

Taking the logarithm and factoring $a_\pert^{-1/3}$ out of the bracket yields
\[
\Dsharp_2\left(\zeta_{X_iS_iJ}\middle\Vert
\widetilde{\omega}_{X_iS_iJ}^{\pert}\right)
\leq
\log a_\pert^{-1}
+3\log\left[
1+\left(
\frac{a_\pert|S|(|S|-1)\vanf_{\game}(\delta)^2}{\pert}
\right)^{1/3}
\right].
\]
The universal bound has the form
\[
\log\frac{|S|^2}{\pert}
=\log a_\pert^{-1}
+\log\frac{a_\pert|S|^2}{\pert}.
\]
Taking the smaller of the two bounds gives the claimed divergence estimate, with $\vanf_{\pert,\game}$ as defined in \eqref{eq:round-cost}. For fixed $\pert>0$, the second defining term of this function tends to zero as $\delta\to0$, since $\vanf_\game(\delta)\to0$. Both defining terms are nonnegative, so $\vanf_{\pert,\game}$ is vanishing.
\end{proof}

In the next lemma, we show how to convert the states from the previous lemma to $\mapE^\gen_i[\pure{\rho}]$ and $\mapF^\gen_i[\pure{\rho}]$, so that we can upper bound the sharp-divergence between the latter two states.

\begin{lemma}[Recovery of the generation channels]\label{lem:recover}
There is a postprocessing channel
\[
\mapT_i\in\CPTP{X_iS_iJ_i^\A J_i^\B}{X_iY_iA_iB_iC_iM_{i+1}^\A M_{i+1}^\B},
\]
fixed by the round instruments and independent of $\ket{\rho}$, such that
\[
\mapT_i[\zeta_{X_iS_iJ}] =\mapE_i^\gen[\pure\rho], \qquad
\mapT_i[\widetilde{\omega}_{X_iS_iJ}^{\pert}] =\mapF_i^\gen[\pure\rho].
\]
\end{lemma}

\begin{proof}
We define $\mapT_i$ by the following steps, recalling that $S_i=A_iS_i^\A S_i^\B$.
\begin{enumerate}
\item\label{step:pinch} Pinch $A_i$ in the basis $\{\ket{a_i}\}_{a_i}$. The registers remain $X_iA_iS_i^\A J_i^\A S_i^\B J_i^\B$.
\item\label{step:rec} Apply recovery channels $\mapN_i^\A$ on $S_i^\A J_i^\A$ and $\mapN_i^\B$ on $S_i^\B J_i^\B$, giving registers $X_iA_iM_i^\A\ancA_iM_i^\B\ancB_i$.
\item\label{step:cont} Conditional on $x_i$ in $X_i$, apply $\mapL_i^{\A,x_i}$ to $A_iM_i^\A\ancA_i$ and $\mapL_i^{\B,\gen}$ to $M_i^\B\ancB_i$. The registers are now $X_iA_iB_iM_{i+1}^\A M_{i+1}^\B$, with $B_i$ fixed to $\testQ$.
\item\label{step:record} Append $Y_i=\testQ$ and $C_i=\gen$, giving the output registers $X_iY_iA_iB_iC_iM_{i+1}^\A M_{i+1}^\B$.
\end{enumerate}

We first compute $\mapT_i[\zeta_{X_iS_iJ}]$. In a fixed $x_i$ block, step~\ref{step:pinch} removes the terms with distinct outcome labels $a_i,a'_i$ from $\pure{\zeta_{x_i}}$. Pinching $A_i$ commutes with isometries $V^\A_i$ and $V^\B_i$, which act on the device memories. The recovery channels in step~\ref{step:rec} then undo those isometries, giving
\[
\sum_{a_i}\pure{a_i}_{A_i}\otimes
(H_{i,a_i}^{x_i}\otimes\id)
\left(
\pure\rho\otimes\pure0_{\ancA_i}\otimes\pure0_{\ancB_i}
\right)
(H_{i,a_i}^{x_i}\otimes\id).
\]
This is exactly the state obtained by applying Alice's measurement channel $\mapM_i^{\A,x_i}$ to $\pure{\rho}$ and appending Bob's ancilla. In step~\ref{step:cont}, $\mapL_i^{\A,x_i}$ completes Alice's real instrument, including its input- and outcome-dependent operation on the postmeasurement memory. The channel $\mapL_i^{\B,\gen}$ discards $\ancB_i$ and applies Bob's real generation channel. Step~\ref{step:record} supplies the remaining generation records. Thus, after averaging over $x_i$,
\[
\mapT_i[\zeta_{X_iS_iJ}]=\mapE_i^\gen[\pure\rho].
\]

We next compute $\mapT_i[\widetilde{\omega}_{X_iS_iJ}^{\pert}]$. In each $x_i$ block, the pinching in step~\ref{step:pinch} sends the pure part of $\widetilde{\omega}_{S_i}^{\pert,x_i}$ to
\[
\sum_{a_i}\pure{a_i}_{A_i}\otimes
(T_{a_i}^{x_i}\otimes\id)\pure\phi
(T_{a_i}^{x_i}\otimes\id)
=\omega_{S_i}^{x_i}.
\]
The maximally mixed part is unchanged. The state after step~\ref{step:pinch} is therefore
\[
\sum_{x_i}\sfP_X(x_i)\pure{x_i}_{X_i}\otimes
\omega_{S_i}^{\pert,x_i}\otimes\zeta_J^{x_i}.
\]

To identify this with the state prepared by the replacement channel in $\mapF_i^\gen$, we check the retained marginal. The definition of $\mapF_i^\gen$ uses $\mapM_i^{\A,x_i}$, whereas $\zeta_{S_iJ}^{x_i}$ is obtained by coherently recording Alice's outcome. By the construction in Lemma~\ref{lem:coh}, these two states differ only by pinching $A_i$. Tracing out $S_i=A_iS_i^\A S_i^\B$ gives the same marginal before and after this pinching: the partial trace over $A_i$ already discards every term with distinct outcome labels. Hence the marginal retained in the definition of $\mapF_i^\gen$ is exactly $\zeta_J^{x_i}$.

It follows that $\omega_{S_i}^{\pert,x_i}\otimes\zeta_J^{x_i}$ is precisely the state produced by $\mapP_i^{\pert,x_i}$ in \eqref{eq:replace-noisy}. Steps~\ref{step:rec}--\ref{step:record} are the recovery channels, local continuations, and record preparation specified in $\mapF_i^\gen$. Consequently,
\[
\mapT_i[\widetilde{\omega}_{X_iS_iJ}^{\pert}]=\mapF_i^\gen[\pure\rho],
\]
which proves the second identity.
\end{proof}

With these results established, Lemma~\ref{lem:round} follows basically straightforwardly.
By Lemma~\ref{lem:recover}, a common CPTP map sends the two states in Lemma~\ref{lem:input-div} to the respective channel outputs in the statement. Applying the map $\mapT_i$ to the operator inequality in Lemma~\ref{lem:input-div} proves the first bound. Data processing for the sharp-divergence gives
\begin{align*}
\Dsharp_2\left(\mapE_i^\gen[\pure\rho]\middle\Vert
\mapF_i^\gen[\pure\rho]\right)
&\leq\Dsharp_2\left(\zeta_{X_iS_iJ}\middle\Vert
\widetilde{\omega}_{X_iS_iJ}^{\pert}\right)\\
&\leq\log a_\pert^{-1}+\vanf_{\pert,\game}(\delta).
\end{align*}
The vanishing property follows from Lemma~\ref{lem:input-div}.

\section{Incorporating testing for ``almost-subadditive'' channel divergences}
\label{app:IIDtesting}

Again, we follow the same proof structure as before. First, in Lemma~\ref{lemma:fbound} we derive a {family} of upper bounds on $D_\alpha\Dargs{\mapM\left[\realstate_{|\Omega} \right]}{\mapM\left[\idealstate\right]}$, parametrized by a vector $f\in\mathbb{R}^{|\alphCS|}$. Then, in Lemma~\ref{lemma:dual_sand} we show that optimizing over the choice of $f$ yields the claimed bound in Theorem~\ref{th:robustselftest}.

The actual technical difference compared to the Theorem~\ref{th:main} is mainly just some bookkeeping regarding the measurement channel, and the minimization over the input state to the second argument. However, to present a wider variety of proof styles, the proof here is presented in a slightly different fashion regarding the details, e.g.~by constructing ancillary registers to ``encode'' the $f$ values, and applying a known convex-conjugate formula rather than explicitly solving the minimization over $f$.

\begin{remark}
The results throughout this analysis can be made very slightly more general: rather than requiring $\realstate_{|\Omega}$ to genuinely be $\realstate$ conditioned on the specific event $\Omega$, we can (in a slight abuse of notation) merely require that $\realstate = p_\Omega \realstate_{|\Omega} + (1-p_\Omega) \realstate_{|\neg\Omega} $ for some $p_\Omega \in (0,1]$ and normalized states $\realstate_{|\Omega},\realstate_{|\neg\Omega}$ such that for every $\bCS$ with nonzero probability in $\realstate_{|\Omega}$, we have $\freq_{\bCS} \in S_\Omega$. See e.g.~\cite[Theorem~5.1]{arx_AHT24} for an example.
\end{remark}

Before presenting the main lemmas, we establish a technical result that is a minor variant of~\cite[Lemma~4.5]{MFSR24}.
\begin{fact}\label{lemma:extract_Z}
Consider any states $\rho , \sigma \in \dop{=}(Z Q \CS)$, with respective classical distributions $\sfP,\sfQ$ on $\CS$, of the form
\begin{align*}
\rho_{Z Q \CS}&=\sum_{\cS}\sfP(\cS)\pure{\cS }_{\CS}\otimes\rho_{Q|\cS}\otimes\rho_{Z|\cS}, \\
\sigma_{Z Q \CS}&=\sum_{\cS}\sfQ(\cS)\pure{\cS }_{\CS}\otimes\sigma_{Q|\cS}\otimes\sigma_{Z|\cS},
\end{align*}
i.e.~both of them satisfy the Markov condition $Z \leftrightarrow \CS \leftrightarrow Q$. Then for any $\alpha>1$, we have:
\begin{align*}
D_\alpha \Dargs{\rho_{Q\CS}}{\sigma_{Q\CS}} \leq D_\alpha \Dargs{\rho_{ZQ\CS}}{\sigma_{ZQ\CS}} 
- \min_{\cS\in\supp(\sfP)} D_\alpha \Dargs{\rho_{Z|\cS}}{\sigma_{Z|\cS}}. 
\end{align*}
The entire right-hand side is interpreted as $+\infty$ whenever
$D_\alpha\Dargs{\rho_{ZQ\CS}}{\sigma_{ZQ\CS}}=+\infty$,
without evaluating the conditional-state divergences in that case.
\end{fact}
\begin{proof}
If $D_\alpha\Dargs{\rho_{ZQ\CS}}{\sigma_{ZQ\CS}}=+\infty$,
the claim holds by the stated convention. Otherwise, finiteness implies
$\supp(\sfP)\subseteq\supp(\sfQ)$ and, for every
$\cS\in\supp(\sfP)$, finiteness of
$D_\alpha\Dargs{\rho_{ZQ|\cS}}{\sigma_{ZQ|\cS}}$.
By data processing under tracing out $Q$, each
$D_\alpha\Dargs{\rho_{Z|\cS}}{\sigma_{Z|\cS}}$ is also finite.
Thus all conditional states appearing in the minimum are well-defined,
and the minimum is finite.
Focusing on this case, we apply Fact~\ref{th:renyi-properties}\ref{fact:classmix} to get
\begin{align*}
D_\alpha \Dargs{\rho_{ZQ\CS}}{\sigma_{ZQ\CS}}
&=\frac{1}{\alpha-1}\log\left(\sum_{\cS\in\supp(\sfP)} \sfP(\cS)^\alpha\sfQ(\cS)^{1-\alpha} 2^{(\alpha-1)D_\alpha(\rho_{ZQ|\cS}\Vert\sigma_{ZQ|\cS})}\right) \nonumber \\
&=\frac{1}{\alpha-1}\log\Biggl(\sum_{\cS\in\supp(\sfP)}
\sfP(\cS)^\alpha\sfQ(\cS)^{1-\alpha}\\
&\hspace{35mm}\cdot 2^{(\alpha-1)\left(D_\alpha(\rho_{Z|\cS}\Vert\sigma_{Z|\cS})+D_\alpha(\rho_{Q|\cS}\Vert\sigma_{Q|\cS})\right)}\Biggr)\\
&\geq\min_{\cS\in\supp(\sfP)}D_\alpha(\rho_{Z|\cS}\Vert\sigma_{Z|\cS})\\
&\quad+\frac{1}{\alpha-1}\log\left(\sum_{\cS\in\supp(\sfP)}
\sfP(\cS)^\alpha\sfQ(\cS)^{1-\alpha}
2^{(\alpha-1)D_\alpha(\rho_{Q|\cS}\Vert\sigma_{Q|\cS})}\right)\\
&= D_\alpha \Dargs{\rho_{Q\CS}}{\sigma_{Q\CS}} + \min_{\cS\in\supp(\sfP)} D_\alpha \Dargs{\rho_{Z|\cS}}{\sigma_{Z|\cS}},
\end{align*}
where in the first equality we used the restriction to the $\supp(\sfP) \subseteq \supp(\sfQ)$ case, the second equality holds because $\rho_{Q Z|\cS} = \rho_{Q|\cS}\otimes\rho_{Z|\cS}$ and $\sigma_{Q Z|\cS} = \sigma_{Q|\cS}\otimes\sigma_{Z|\cS}$, and the inequality holds since $\frac{1}{\alpha-1}\log(\cdot)$ is a nondecreasing function given that $\alpha>1$.
\end{proof}

We now prove the two main lemmas we need.
In the following lemma, we state the LHS with an explicit minimization over $\nu$ in order to make it clear that the LHS has no dependence on the choice of $f$; therefore we are free to later minimize the RHS over the choice of $f$.

\begin{lemma}\label{lemma:fbound}
Consider channels and parameters satisfying the conditions of Theorem~\ref{th:robustselftest} (including in particular Postulate Set~\ref{posts:overall}), except that here we allow $S_\Omega$ to be an \emph{arbitrary} nonempty subset of the probability simplex on $\alphCS$.
Take any $f \in \mathbb{R}^{|\alphCS|}$, and define
\begin{align*}
\ubndaf = \sup_{\sigma\in\dop{=}(\inN)} \frac{1}{\alpha-1}\log\left(\sum_{\cS} \sfP(\cS) 2^{(\alpha-1) \left(D_\alpha\Dargs{\sigma^{\mapE}_{\outQ\stabE|\cS}}{\sigma^{\mapF}_{\outQ\stabE|\cS}}-f(\cS)\right) }\right).
\end{align*}
Then we have
\begin{align}
\min_{\nu\in\dop{=}(\allN \bR)}
D_\alpha\Dargs{\mapM\left[\realstate_{|\Omega} \right]}{\mapM\left[\mapF^{\otimes n}[\nu]\right]} &\leq n\left(\ubndaf + \sup_{\sfQ\in S_\Omega} f\cdot\sfQ\right) + \frac{\alpha\log(1/p_\Omega)}{\alpha-1}, \label{eq:f_Dalphaselftest} \\
\min_{\nu\in\dop{=}(\allN \bR)}
\Dmax^{\esmooth}\Dargs{\mapM\left[\realstate_{|\Omega} \right]}{\mapM\left[\mapF^{\otimes n}[\nu]\right]} &\leq n\left(\ubndaf + \sup_{\sfQ\in S_\Omega} f\cdot\sfQ\right) + \frac{\smf{\esmooth}}{\alpha-1} + \frac{\alpha\log(1/p_\Omega)}{\alpha-1},
\end{align}
where $f\cdot\sfQ$ denotes $\sum_{\cS} f(\cS) \sfQ(\cS)$.
\end{lemma}
\begin{proof}
In the notation throughout this proof we use the formal definition of the event $\Omega$ as being technically a subset of the alphabet of $\bCS$.
We focus only on proving the bound on $D_\alpha$ in~\eqref{eq:f_Dalphaselftest}; the bound on $\Dmax^{\esmooth}$ follows immediately by Fact~\ref{th:renyi-properties}\ref{th:smoothing}. Also, we first define a new measurement channel $\cmapM \in \CPTP{\boutQ \bCS \bR}{\bCS F}$ acting on states with classical $\bCS$ as follows:\footnote{To extend the action of this channel to all states in $\dop{=}(\boutQ \bCS \bR)$, we could simply say it always first applies a pinching channel on $\bCS$ in the classical basis; this makes it well-defined on general states without changing the way it acts on states that are already classical on $\bCS$.}  it first creates a copy of $\bCS$  in another classical register $\bL$, then it applies $\mapM$ to $\boutQ \bL \bR$ (using $\bL$ in lieu of $\bCS$).
For this channel, we straightforwardly have for any states $\omega,\omega' \in \dop{=}(\boutQ \bCS \bR)$ classical on $\bCS$:
\begin{align*}
D_\alpha\Dargs{\mapM [\omega]}{\mapM[\omega']} = D_\alpha\Dargs{\Tr_{\bCS} \circ \cmapM [\omega]}{\Tr_{\bCS} \circ \cmapM[\omega']} \leq D_\alpha\Dargs{\cmapM [\omega]}{\cmapM[\omega']}
\end{align*}
where the equality holds since the copying leaves the reduced states on $\boutQ \bCS \bR$ invariant (and $\bL$ is a copy of $\bCS$), and the inequality holds by data-processing.
Hence to prove our desired result, it suffices to prove it with $\cmapM$ in place of $\mapM$. 

We now again apply the proof technique developed in~\cite{DFR20,ZFK20,arx_AHT24}.
We choose $K \in \mathbb{R}$ such that $K-f(\cS)>0$ for all $\cS$, and a register $Z$ with dimension large enough that $\log|Z| > K-f(\cS)$ for all $\cS$. 
We choose any read-and-prepare channel $\mapD \in \CPTP{\CS}{\CS Z}$ such that the output states it prepares conditioned on $\cS$ are classical on $Z$ and satisfy
\begin{align}\label{eq:D_entropy}
D_\alpha \Dargs{\mapD[\rho]_{Z|\cS}}{\id_Z/|Z|} = K-f(\cS).
\end{align}
Note that such a channel indeed exists since $D_\alpha \Dargs{\mapD[\rho]_{Z|\cS}}{\id_Z/|Z|} = \log|Z| - H_\alpha(Z)_{\mapD[\rho]_{|\cS}}$, thus we are just requiring that for each $\cS$ the channel prepares some classical state with entropy $H_\alpha(Z)_{\mapD[\rho]_{|\cS}} = \log|Z|-K+f(\cS) \in [0, \log|Z|]$, which is easily doable (e.g.~a mixture between a point distribution and a uniform distribution). Furthermore, let $\mapA$ denote a channel that simply appends the maximally mixed state $\id_Z/|Z|$ to its input state, without any other effect.
With this, we extend the channels $\mapE,\mapF$ to new channels $\newmapE,\newmapF \in \CPTP{\inN}{\CS \outQ Z}$ as follows:
\begin{align*}
\newmapE \defvar \mapD \circ \mapE, \quad
\newmapF \defvar \mapA \circ \mapF,
\end{align*}
which have the property that their outputs are always classical on $Z$ and $\CS$. 

Let $\mapP$ denote a pinching channel on $\bZ$ in its classical basis. With this, we note that $\mapP \otimes \cmapM \in \CPTP{\bZ \boutQ \bCS \bR}{\bZ \bCS F}$ is a valid measurement channel. Therefore we can apply Theorem~\ref{th:almostadd} to the channels constructed above, to obtain the central bound in this proof:
\begin{align}\label{eq:applyalmostadd}
\inf_{\nu\in\dop{=}(\allN \bR)}
D_\alpha \Dargs{\mapP \otimes \cmapM \left[\newmapE^{\otimes n} [\rho^0]\right]}{\mapP \otimes \cmapM \left[\newmapF^{\otimes n} [\nu]\right]} \leq n D_\alpha\Dargs{\newmapE \otimes \idmap_{\stabE}}{\newmapF \otimes \idmap_{\stabE}}.
\end{align}
We now relate the LHS and RHS of this bound to our desired quantities of interest. 

We start by relating the LHS to the sandwiched divergence on the LHS of our final desired bound~\eqref{eq:f_Dalphaselftest}. First note that we can define a valid extension of the state $\realstate_{\bCS \boutQ \bR}$ as follows:
\begin{align*}
\realstate_{\bZ \bCS \boutQ \bR} &\defvar \newmapE^{\otimes n} [\rho^0] = \mapD^{\otimes n} \left[\realstate_{\bCS \boutQ \bR}\right],
\end{align*} 
since $\mapD$ is a read-and-prepare channel. Furthermore, the construction of the channel $\cmapM$ ensures that this extension has the following property: writing 
$\suppOmega=\{\bcS\in\Omega \mid
\Pr_{\realstate}[\bcS]>0\}$ to denote the support of the classical distribution of
$\realstate_{\bCS|\Omega}$ (mainly to avoid conditioning on zero-probability $\bcS$ values), we have
\begin{align}
\cmapM \left[\realstate_{\bZ \boutQ \bCS \bR|\Omega}\right] &= \sum_{\bcS\in\suppOmega} \frac{\Pr_{\realstate}[\bcS] }{p_\Omega} \cmapM \left[\pure{\bcS} \otimes \realstate_{\bZ \boutQ \bR|\bcS}\right] \nonumber\\
&= \sum_{\bcS\in\suppOmega} \frac{\Pr_{\realstate}[\bcS] }{p_\Omega} \pure{\bcS} \otimes \mapM \left[\pure{\bcS} \otimes \realstate_{\bZ \boutQ \bR|\bcS}\right] \nonumber\\
&= \sum_{\bcS\in\suppOmega} \frac{\Pr_{\realstate}[\bcS] }{p_\Omega}  \pure{\bcS} \otimes \realstate_{\bZ|\bcS} \otimes \mapM \left[\pure{\bcS} \otimes \realstate_{\boutQ \bR|\bcS}\right],\label{eq:rhoMarkov}
\end{align}
where the first line holds by simply filling in the definition of conditional states and applying linearity of $\cmapM$, the second line holds fundamentally from the definition of $\cmapM$, and in the third line we use the fact that $\realstate_{\bZ \boutQ \bR|\bcS} = \realstate_{\bZ|\bcS} \otimes \realstate_{\boutQ \bR|\bcS}$ by how $\realstate$ is constructed using the channels $\mapD$ (along with the fact that $\mapM$ does not act on $\bZ$).
This means that $\cmapM \left[\realstate_{\bZ \boutQ \bCS \bR|\Omega}\right]$ has the correct Markov structure to apply Fact~\ref{lemma:extract_Z} (noting that each $\realstate_{\bZ|\bcS} $ and $\mapM \left[\pure{\bcS}\otimes  \realstate_{\boutQ \bR|\bcS}\right]$ term is a normalized state and therefore must in fact be the conditional states of  $\cmapM \left[\realstate_{\bZ \boutQ \bCS \bR|\Omega}\right]$ on $\bZ$ and $F$ respectively).

With this in mind, for \emph{any}  $\nu\in\dop{=}(\allN \bR)$ we can write
\begin{align*}
& D_\alpha\Dargs{\cmapM \left[\realstate_{\boutQ \bCS \bR|\Omega}\right]}{\cmapM \left[\mapF^{\otimes n} [\nu]\right]}   \nonumber\\
\leq& D_\alpha\Dargs{\cmapM \left[\realstate_{\bZ \boutQ \bCS \bR|\Omega}\right]}{\cmapM \left[\newmapF^{\otimes n} [\nu]\right]} - \min_{\bcS \in \suppOmega} D_\alpha\Dargs{\realstate_{\bZ|\bcS}}{\id_{\bZ}/|\bZ|}   \nonumber\\
\leq& D_\alpha\Dargs{\cmapM \left[\realstate_{\bZ \boutQ \bCS \bR}\right]}{\cmapM \left[\newmapF^{\otimes n} [\nu]\right]} + \frac{\alpha\log(1/p_\Omega)}{\alpha-1} - \min_{\bcS \in \suppOmega} D_\alpha\Dargs{\realstate_{\bZ|\bcS}}{\id_{\bZ}/|\bZ|} \nonumber\\
=& D_\alpha \Dargs{\mapP \otimes \cmapM \left[\newmapE^{\otimes n} [\rho^0]\right]}{\mapP \otimes \cmapM \left[\newmapF^{\otimes n} [\nu]\right]}\nonumber\\ 
&\qquad + \frac{\alpha\log(1/p_\Omega)}{\alpha-1} - \min_{\bcS \in \suppOmega} D_\alpha\Dargs{\realstate_{\bZ|\bcS}}{\id_{\bZ}/|\bZ|}
.
\end{align*}
In the above, in the second line we applied Fact~\ref{lemma:extract_Z} by noting that the state in the first argument has the correct Markov structure as stated above in~\eqref{eq:rhoMarkov}, together with the fact that we straightforwardly have $\cmapM \left[\newmapF^{\otimes n} [\nu]\right] = \id_{\bZ}/|\bZ| \otimes \cmapM \left[\mapF^{\otimes n} [\nu] \right]$ and hence the state in the second argument also has that Markov structure in a ``trivial'' sense (in that $\id_{\bZ}/|\bZ|$ is independent of $\bcS$). 
To get the third line, we use the fact that $\realstate_{\bZ \boutQ \bCS \bR} \geq p_\Omega\, \realstate_{\bZ \boutQ \bCS \bR|\Omega}$ and thus $\cmapM \left[\realstate_{\bZ \boutQ \bCS \bR}\right] \geq p_\Omega\, \cmapM \left[\realstate_{\bZ \boutQ \bCS \bR|\Omega}\right] $ as well, then apply Fact~\ref{th:renyi-properties}\ref{fact:eventconditioning}. Finally, in the last line we filled in the definition of $\realstate_{\bZ \boutQ \bCS \bR}$ and used the fact that $
\newmapE^{\otimes n} [\rho^0]$ and $
\newmapF^{\otimes n} [\nu]$ are already classical on $\bZ$, hence they are invariant under the pinching channel $\mapP$ on those registers.

Since the above bound holds for any $\nu$, it still holds by taking the infimum on both sides.
Furthermore, the $D_\alpha\Dargs{\realstate_{\bZ|\bcS}}{\id_{\bZ}/|\bZ|}$ terms can be simplified as follows:
\begin{align*}
D_\alpha\Dargs{\realstate_{\bZ|\bcS}}{\id_{\bZ}/|\bZ|} = \sum_{i=1}^n D_\alpha\Dargs{\realstate_{Z_i|\cS_i}}{\id_{Z_i}/|Z_i|}
= nK - \sum_{i=1}^n f(\cS_i),
\end{align*}
where the first equality holds by recalling that $\realstate_{\bZ}$ was constructed by applying $\mapD^{\otimes n}$ and thus we have $\realstate_{\bZ|\bcS} = \bigotimes_{i=1}^n \realstate_{Z_i|\cS_i}$, while the second equality holds by the defining property \eqref{eq:D_entropy} of the $\mapD$ channels.
Putting everything together, we conclude that the relationship between the LHS expressions in~\eqref{eq:f_Dalphaselftest} and~\eqref{eq:applyalmostadd} is
\begin{align}
& \inf_{\nu\in\dop{=}(\allN \bR)} D_\alpha\Dargs{\cmapM \left[\realstate_{\boutQ \bCS \bR|\Omega}\right]}{\cmapM \left[\mapF^{\otimes n} [\nu]\right]}   \nonumber\\
\leq& \inf_{\nu\in\dop{=}(\allN \bR)} D_\alpha \Dargs{\mapP \otimes \cmapM \left[\newmapE^{\otimes n} [\rho^0]\right]}{\mapP \otimes \cmapM \left[\newmapF^{\otimes n} [\nu]\right]}\nonumber\\ 
&\qquad + \frac{\alpha\log(1/p_\Omega)}{\alpha-1} - nK + \max_{\bcS \in \suppOmega} \sum_{i=1}^n f(\cS_i)
.\label{eq:LHSbound}
\end{align}
Furthermore, the maximum in the last line can be bounded using the acceptance set:
\begin{align*}
\max_{\bcS \in \suppOmega} \sum_{i=1}^n f(\cS_i) = \max_{\bcS \in \suppOmega} n f \cdot \freq_{\bcS} \leq n\sup_{\sfQ \in S_\Omega} f \cdot \sfQ.
\end{align*}

We now turn to the RHS of~\eqref{eq:applyalmostadd}.
Similar to Appendix~\ref{app:minimax}, we highlight that all subsequent steps in this analysis (including Lemma~\ref{lemma:dual_sand}) are equalities, and do not admit any further sharpening.
Recall that the channel divergence is a supremum over input states to $\newmapE\otimes\idmap_{\stabE}$ and $\newmapF\otimes\idmap_{\stabE}$; by data processing it suffices to use purifications $\pure\sigma$ of input marginals on $\inN$. In terms of the $\sigma^{\mapE},\sigma^{\mapF}$ notation in the theorem statement, $\newmapE[\pure\sigma] = \mapD [\sigma^{\mapE}_{\CS \outQ \stabE}]$ and $\newmapF[\pure\sigma] = \mapA [\sigma^{\mapF}_{\CS \outQ \stabE}]$. By hypothesis, the channels $\mapE,\mapF$ are such that these one-round states have the same distribution $\sfP$ on $\CS$. Hence every feasible point in the supremum satisfies
\begin{align}
D_\alpha\Dargs{\mapD [\sigma^{\mapE}]}{\mapA [\sigma^{\mapF}]} &=\frac{1}{\alpha-1}\log\left(\sum_{\cS} \sfP(\cS) 2^{(\alpha-1)D_\alpha\Dargs{\mapD [\sigma^{\mapE}]_{Z\outQ\stabE|\cS}}{\mapA [\sigma^{\mapF}]_{Z\outQ\stabE|\cS}}}\right) \nonumber\\
&=\frac{1}{\alpha-1}\log\left(\sum_{\cS} \sfP(\cS) 2^{(\alpha-1) \left(D_\alpha\Dargs{\mapD [\sigma^{\mapE}]_{Z|\cS}}{\id_Z/|Z|} + D_\alpha\Dargs{\sigma^{\mapE}_{\outQ\stabE|\cS}}{\sigma^{\mapF}_{\outQ\stabE|\cS}}\right) }\right) \nonumber\\
&=\frac{1}{\alpha-1}\log\left(\sum_{\cS} \sfP(\cS) 2^{(\alpha-1) \left(K + D_\alpha\Dargs{\sigma^{\mapE}_{\outQ\stabE|\cS}}{\sigma^{\mapF}_{\outQ\stabE|\cS}}-f(\cS)\right) }\right) \nonumber\\
&=\frac{1}{\alpha-1}\log\left(\sum_{\cS} \sfP(\cS) 2^{(\alpha-1) \left(D_\alpha\Dargs{\sigma^{\mapE}_{\outQ\stabE|\cS}}{\sigma^{\mapF}_{\outQ\stabE|\cS}}-f(\cS)\right) }\right)  + K ,\label{eq:RHSbound}
\end{align}
where the first line is Fact~\ref{th:renyi-properties}\ref{fact:classmix} together with the fact that the two one-round states have the same distribution $\sfP$ on $\CS$, the second line is by recalling that since $\mapD$ is a read-and-prepare channel we have $\mapD[\sigma^{\mapE}]_{Z\outQ\stabE|\cS} = \mapD[\sigma^{\mapE}]_{Z|\cS} \otimes \sigma^{\mapE}_{\outQ\stabE|\cS}$ (and also we have $\mapA[\sigma^{\mapF}]_{Z\outQ\stabE|\cS} = \id_Z/|Z| \otimes \sigma^{\mapF}_{\outQ\stabE|\cS}$), and the third line is by the defining property~\eqref{eq:D_entropy} of $\mapD$.

Putting together the bounds~\eqref{eq:LHSbound}--\eqref{eq:RHSbound} with the central bound~\eqref{eq:applyalmostadd}, and noting that the value $K$ cancels off, we arrive at the claimed bound in~\eqref{eq:f_Dalphaselftest}. To address a final technicality we observe that it is indeed valid to write $\min$ instead of $\inf$ in that bound, because it is minimizing a convex lower semicontinuous function over a compact set and thus the infimum is attained.
\end{proof}

Since Lemma~\ref{lemma:fbound} holds for \emph{any} choice of $f$ on the RHS, we can optimize over it to get the tightest possible bound. The following lemma gives the resulting bound, which then immediately yields Theorem~\ref{th:robustselftest} (taking $\nu^0$ to be the state achieving the minimum in~\eqref{eq:f_Dalphaselftest}).

\begin{lemma}\label{lemma:dual_sand}
In Lemma~\ref{lemma:fbound}, if furthermore $S_\Omega$ is closed, convex, and contains at least one full-support distribution, then the following duality relation holds:
\begin{align}
\inf_{f\in\mathbb{R}^{|\alphCS|}}\left(\ubndaf + \sup_{\sfQ\in S_\Omega} f\cdot\sfQ\right)
= \frac{1}{n} \compexpsand{\alpha}, \label{eq:duality}
\end{align}
keeping in mind the convention following
\eqref{eq:alpharatedefn}.
\end{lemma}

\begin{proof}
Throughout this proof, whenever we state that some expression is concave in multiple variables, it should be understood that we mean \emph{joint} concavity (and analogously for convexity). The {\Renyi} parameter $\alpha>1$ should be viewed throughout as a constant.

We first handle a degenerate case: if there exists some $\sigma\in\dop{=}(\inN)$ such that $D_\alpha\Dargs{\sigma^{\mapE}_{\outQ\stabE|\cS}}{\sigma^{\mapF}_{\outQ\stabE|\cS}} = \infty$ for some $\cS$, by our convention following
\eqref{eq:alpharatedefn} this means we have $\sfP(\cS)>0$. Then we would have $\ubndaf[f]=\infty$ for all $f$ due to this feasible point in the supremum over $\sigma$, hence the LHS of~\eqref{eq:duality} would be $\infty$. As for the RHS, by hypothesis $S_\Omega$ contains a full-support distribution $\sfQ$, and thus for the aforementioned $\sigma$ and $\cS$ we have $\sfQ(\cS) D_\alpha\Dargs{\sigma^{\mapE}_{\outQ\stabE|\cS}}{\sigma^{\mapF}_{\outQ\stabE|\cS}} = \infty$ and the RHS of~\eqref{eq:duality} is $\infty$ as well, hence equality holds.
Thus for the remainder of the proof, we assume the channels have the property that all the  $D_\alpha\Dargs{\sigma^{\mapE}_{\outQ\stabE|\cS}}{\sigma^{\mapF}_{\outQ\stabE|\cS}}$ terms are always finite.

We begin by highlighting some properties of the following function of $f\in\mathbb{R}^{|\alphCS|}$ (which is always finite, by the above assumption):
\begin{align}\label{eq:defineF}
F_{\alpha,\sigma}(f) =
\frac{1}{\alpha-1}\log\left(\sum_{\cS} \sfP(\cS) 2^{(\alpha-1) \left(D_\alpha\Dargs{\sigma^{\mapE}_{\outQ\stabE|\cS}}{\sigma^{\mapF}_{\outQ\stabE|\cS}}-f(\cS)\right) }\right) ,
\end{align}
so that the objective function in the definition of $\ubndaf[f]$ can be written as $F_{\alpha,\sigma}(f)$.
First, note that following the same construction of the extended channels $\newmapE,\newmapF$ as in the Lemma~\ref{lemma:fbound} proof, from~\eqref{eq:RHSbound} we have
\begin{align*}
F_{\alpha,\sigma}(f) = D_\alpha\Dargs{\mapD [\sigma^{\mapE}]}{\mapA [\sigma^{\mapF}]} - K = D_\alpha\Dargs{\newmapE [\pure\sigma]}{\newmapF [\pure\sigma]} - K,
\end{align*}
and thus by Fact~\ref{th:renyi-properties}\ref{prop:renyi-purified-concavity} we see that (for fixed $f$) it is a concave function of $\sigma \in \dop{=}(\inN)$.

We now handle a delicate technical point. Recall we have restricted to the case where the channels are such that all  $D_\alpha\Dargs{\sigma^{\mapE}_{\outQ\stabE|\cS}}{\sigma^{\mapF}_{\outQ\stabE|\cS}}$ are always finite. It turns out this property implies that $D_\alpha\Dargs{\newmapE [\pure\sigma]}{\newmapF [\pure\sigma]}$ is a continuous function of $\sigma \in \dop{=}(\inN)$ (as long as we take the purifying function to be continuous), as follows. First, in particular it implies that $D_\alpha$ between the Choi states of these channels is finite. In turn this implies that there exists a constant $\Lambda >0$ such that every feasible input state (with the stabilizing register) yields output states satisfying the domination condition ${\newmapE [\pure\sigma]} \leq \Lambda \, {\newmapF [\pure\sigma]}$. This suffices to ensure continuity, by Lemma~\ref{lem:renyi-dominated-continuity} below. Then $F_{\alpha,\sigma}(f)$ is thus also continuous in $\sigma \in \dop{=}(\inN)$.

Moreover, another way to express $F_{\alpha,\sigma}(f)$ is as follows. For brevity, we write $\vecD$ for the ``vector'' of conditional divergence terms, i.e.~its components are
\[
\vecDc{\cS}=\begin{cases}
D_\alpha\Dargs{\sigma^{\mapE}_{\outQ\stabE|\cS}}{\sigma^{\mapF}_{\outQ\stabE|\cS}} & \text{if }\sfP(\cS)>0,\\
0 & \text{if }\sfP(\cS)=0,
\end{cases}
\]
to maintain consistency with the convention following~\eqref{eq:alpharatedefn}.
Let $f_{\fullsupp},\sfP_{\fullsupp},\vecD_{\fullsupp}$ denote the restrictions of $f,\sfP,\vecD$ to $\supp(\sfP)$ (which depends on $\sigma$). For brevity we also write ``componentwise'' logarithms of vectors, e.g.~$\ln\sfP_{\fullsupp}=(\ln(\sfP(\cS)))_{\cS\in\supp(\sfP)}$. With this we can write
\begin{align*}
F_{\alpha,\sigma}(f) &=
\frac{1}{(\alpha-1)\ln2}\ln\left(\sum_{\cS\in\supp(\sfP)} e^{\ln(\sfP(\cS)) + (\alpha-1) (\ln2) \left(D_\alpha\Dargs{\sigma^{\mapE}_{\outQ\stabE|\cS}}{\sigma^{\mapF}_{\outQ\stabE|\cS}}-f(\cS) \right) }\right) \nonumber\\
&=\frac{\lse\left(
\ln\sfP_{\fullsupp} + (\alpha-1) (\ln2) \left(\vecD_{\fullsupp}-f_{\fullsupp} \right)
\right)}{(\alpha-1)\ln2}
.
\end{align*}
where $\lse$ denotes the base-$e$ log-sum-exponential function (see e.g.~\cite{BV04v8}). Since $\lse$ is a convex function~\cite{BV04v8}, we see that (for fixed $\sigma$) $F_{\alpha,\sigma}(f)$ is a convex function of $f$.\footnote{While the $\lse$ function is not applied to all components of $f$ if $\sfP$ does not have full support, we see that in such cases $F_{\alpha,\sigma}(f)$ does not depend on those components of $f$ at all, hence it trivially maintains convexity with respect to the full vector $f$.} Also, from the definition~\eqref{eq:defineF} it is clear that it is continuous with respect to $f$.

With this we see that $F_{\alpha,\sigma}(f) + f\cdot\sfQ$ is always finite, concave and continuous with respect to $(\sigma,\sfQ)$, and convex and continuous with respect to $f$.
Furthermore, by hypothesis $S_\Omega$ is a compact convex set (and so is $\dop{=}(\inN)$, under the background finite-dimensionality assumption), thus when minimizing over $f$ to find the best possible bound, we can perform an optimization order swap via Sion's minimax theorem:
\begin{align}
\inf_{f\in\mathbb{R}^{|\alphCS|}}\left(\ubndaf[f] + \sup_{\sfQ\in S_\Omega} f\cdot\sfQ\right)
&= \inf_{f\in\mathbb{R}^{|\alphCS|}} \sup_{\sfQ\in {S}_\Omega} \sup_{\sigma\in\dop{=}(\inN)} \left( F_{\alpha,\sigma}(f) + f\cdot\sfQ \right)\nonumber\\
&= \sup_{\sfQ\in {S}_\Omega} \sup_{\sigma\in\dop{=}(\inN)} \inf_{f\in\mathbb{R}^{|\alphCS|}} \left( F_{\alpha,\sigma}(f) + f\cdot\sfQ \right).\label{eq:minimax}
\end{align}

Our goal will now be to show that for any fixed $\sfQ\in {S}_\Omega,\sigma \in \dop{=}(\inN)$, the minimization over $f$ simplifies as follows:
\begin{align}
\inf_{f\in\mathbb{R}^{|\alphCS|}} \left( F_{\alpha,\sigma}(f) + f\cdot\sfQ \right)
= \sfQ\cdot\vecD -\frac{1}{\alpha-1} D\Dargs{\sfQ}{{\sfP}} .
\label{eq:conjugate}
\end{align}
Note that the LHS is exactly $-F^*_{\alpha,\sigma}(-\sfQ)$ where $F^*_{\alpha,\sigma}$ denotes the convex conjugate~\cite{BV04v8} of $F_{\alpha,\sigma}(f)$, so we can view this result as finding a formula for $F^*_{\alpha,\sigma}$ at $-\sfQ$ (for $\sfQ\in {S}_\Omega,\sigma \in \dop{=}(\inN)$).\footnote{Note that as shown in~\cite[Lemma~5.3]{arx_AHT24}, if $\sfQ$ is not a probability distribution then~\eqref{eq:conjugate} does not hold (various steps can fail in our following argument, such as the existence of a feasible point $\lambda_{\fullsupp}=\sfQ_{\fullsupp}$), but this does not cause an issue in our context since we only considered $\sfQ\in {S}_\Omega$.}
This was done in Appendix~\ref{app:minimax}, and also essentially found in~\cite[Lemmas~5.3--5.4]{arx_AHT24}, together with in-depth discussion of its connection to a particular Lagrange dual problem; we defer details to that work. 

For the sake of variety we present here a proof using only convex-conjugate properties of $\lse$; specifically the fact that for any $d$-dimensional real vector $v$, we have~\cite{BV04v8}
\begin{align*}
\lse(v) = \sup_{\lambda\in\mathbb{P}_d} \left(\lambda\cdot v + H_{\mathrm{nat}}(\lambda)\right),
\end{align*}
where here we write $\mathbb{P}_d$ for the probability simplex on an alphabet of size $d$, and $H_{\mathrm{nat}}$ denotes base-$e$ Shannon entropy.

First focusing on the case where $\supp(\sfQ) \subseteq \supp(\sfP)$, let us analogously write $\sfQ_{\fullsupp}$ to denote the restriction of ${\sfQ}$ to $\supp(\sfP)$, and similarly we write $F_{\alpha,\sigma,\fullsupp}$  to denote $F_{\alpha,\sigma}$ with the summation restricted to $\supp(\sfP)$. Given that $\supp(\sfQ) \subseteq \supp(\sfP)$, observe that we have $F_{\alpha,\sigma}(f) + f\cdot\sfQ = F_{\alpha,\sigma,\fullsupp}(f_{\fullsupp}) + f_{\fullsupp}\cdot\sfQ_{\fullsupp}$ for all $f$.
With this we can write (letting $d\defvar\left|\supp(\sfP)\right|$, indexing the corresponding vector space $\mathbb{R}^{d}$ and probability simplex $\mathbb{P}_{d}$ the obvious way in terms of $\supp(\sfP)$):
\begin{align*}
&\inf_{f\in\mathbb{R}^{|\alphCS|}} \left( F_{\alpha,\sigma}(f) + f\cdot\sfQ \right) \\
=&\inf_{f_{\fullsupp}\in\mathbb{R}^{d}} \left( F_{\alpha,\sigma,\fullsupp}(f_{\fullsupp}) + f_{\fullsupp}\cdot\sfQ_{\fullsupp}\right) \nonumber\\
=& \inf_{f_{\fullsupp}\in\mathbb{R}^{d}} \left( \frac{
\lse\left(
\ln\sfP_{\fullsupp} + (\alpha-1) (\ln2) \left(\vecD_{\fullsupp}-f_{\fullsupp} \right)
\right)
}{(\alpha-1)\ln2}
 + f_{\fullsupp}\cdot\sfQ_{\fullsupp} \right) \nonumber\\
=& \inf_{f_{\fullsupp}\in\mathbb{R}^{d}} \sup_{\lambda_{\fullsupp}\in\mathbb{P}_{d}} \left( \frac{\lambda_{\fullsupp}\cdot\!\left(
\ln\sfP_{\fullsupp} + (\alpha-1) (\ln2) \left(\vecD_{\fullsupp}-f_{\fullsupp} \right)
\right) +
{H_{\mathrm{nat}}(\lambda_{\fullsupp})} }{(\alpha-1)\ln2} + f_{\fullsupp}\cdot\sfQ_{\fullsupp} \right) \nonumber\\
=& \inf_{f_{\fullsupp}\in\mathbb{R}^{d}}  \sup_{\lambda_{\fullsupp}\in\mathbb{P}_{d}} \left( \lambda_{\fullsupp}\cdot \vecD_{\fullsupp} -\frac{1}{\alpha-1} D\Dargs{\lambda_{\fullsupp}}{\sfP_{\fullsupp}} + f_{\fullsupp}\cdot(\sfQ_{\fullsupp}-\lambda_{\fullsupp}) \right) \nonumber\\
=& \sup_{\lambda_{\fullsupp}\in\mathbb{P}_{d}} \inf_{f_{\fullsupp}\in\mathbb{R}^{d}} \left( \lambda_{\fullsupp}\cdot \vecD_{\fullsupp} -\frac{1}{\alpha-1} D\Dargs{\lambda_{\fullsupp}}{\sfP_{\fullsupp}} + f_{\fullsupp}\cdot(\sfQ_{\fullsupp}-\lambda_{\fullsupp}) \right) \nonumber\\
=& \sfQ_{\fullsupp}\cdot \vecD_{\fullsupp} -\frac{1}{\alpha-1} D\Dargs{\sfQ_{\fullsupp}}{\sfP_{\fullsupp}}  \nonumber\\
=& \sfQ\cdot\vecD -\frac{1}{\alpha-1} D\Dargs{\sfQ}{{\sfP}} .
\end{align*}
The interchange of the infimum and supremum follows again from Sion's minimax theorem, noting that $-D\Dargs{\lambda_{\fullsupp}}{\sfP_{\fullsupp}}$ is continuous and concave in $\lambda_{\fullsupp}$ (and all other terms in the expression are affine in $\lambda_{\fullsupp}$ and $f_{\fullsupp}$). In the following equality we use the fact that there exists a specific feasible point $\lambda_{\fullsupp}=\sfQ_{\fullsupp}$ in the supremum where the objective becomes independent of $f_{\fullsupp}$, in which case the infimum over $f_{\fullsupp}$ trivially simplifies, and for all other points (i.e.~$\lambda_{\fullsupp}\neq\sfQ_{\fullsupp}$) the infimum of $f_{\fullsupp}\cdot(\sfQ_{\fullsupp}-\lambda_{\fullsupp})$ over $f_{\fullsupp}$ is $-\infty$ and the other terms are all independent of $f_{\fullsupp}$.
Finally, in the last line we recall that we are in the case  $\supp(\sfQ) \subseteq \supp(\sfP)$, hence enlarging the sums in $\sfQ_{\fullsupp}\cdot \vecD_{\fullsupp} $ and $D\Dargs{\sfQ_{\fullsupp}}{\sfP_{\fullsupp}}$ to terms outside of $\supp(\sfP)$ does not change the value.

For the other case $\supp(\sfQ) \not\subseteq \supp(\sfP)$, this means there exists some $\cS$ such that $\sfQ(\cS)>0$ but $\sfP(\cS)=0$. Observe that in that case,  $F_{\alpha,\sigma}(f)$ is independent of the corresponding $f(\cS)$ value, but $f\cdot\sfQ$ is not. Therefore by taking $f(\cS)\to-\infty$ for that value, the infimum over $f$ evaluates to $-\infty$, which means the above relation~\eqref{eq:conjugate} still holds in this case, in a ``degenerate'' sense:
\begin{align*}
\inf_{f\in\mathbb{R}^{|\alphCS|}} \left( F_{\alpha,\sigma}(f) + f\cdot\sfQ \right) =& -\infty \nonumber\\
=& \sfQ\cdot\vecD -\frac{1}{\alpha-1} D\Dargs{\sfQ}{{\sfP}}  ,
\end{align*}
where the second line holds since  $\supp(\sfQ) \not\subseteq \supp(\sfP)$ and thus $-\frac{1}{\alpha-1}D\Dargs{\sfQ}{{\sfP}}=-\infty$.

Since~\eqref{eq:conjugate} holds in both these cases (covering all $\sfQ\in {S}_\Omega,\sigma \in \dop{=}(\inN)$), we finally substitute this relation into~\eqref{eq:minimax}, which yields the desired result.
\end{proof}
\begin{lemma}[Continuity under uniform domination]
\label{lem:renyi-dominated-continuity}
Fix a finite-dimensional Hilbert space, $\alpha \in (1,\infty)$, and
$0<\Lambda <\infty$. The sandwiched divergence $D_\alpha(A\Vert B)$ is
jointly continuous on the set of pairs satisfying
\[
A,B\geq0,\qquad \operatorname{Tr}A=1,\qquad A\leq \Lambda  B.
\]
\end{lemma}
\begin{proof}
Let $(A_n,B_n)\to(A,B)$ be a convergent sequence in this set, and put
\[
X_n=B_n^{-1/2}A_nB_n^{-1/2},
\]
where the negative powers are interpreted via Moore--Penrose pseudoinverses.
Then $0\leq X_n\leq \Lambda \id$ and $A_n=B_n^{1/2}X_nB_n^{1/2}$.
By compactness, every subsequence admits a further subsequence along
which $X_n\to X$. Along this further subsequence,
$A=B^{1/2}XB^{1/2}$, and continuity of positive matrix powers gives
\begin{align*}
Q_\alpha(A_n\Vert B_n)
=\operatorname{Tr}\!\left[
  \left(B_n^{1/(2\alpha)}X_nB_n^{1/(2\alpha)}\right)^\alpha
  \right] \longrightarrow\operatorname{Tr}\!\left[
  \left(B^{1/(2\alpha)}XB^{1/(2\alpha)}\right)^\alpha
  \right]
=Q_\alpha(A\Vert B).
\end{align*}
The last equality follows from $A=B^{1/2}XB^{1/2}$, since the
surrounding powers of $B$ remove any components outside its support.
Thus the full sequence $Q_\alpha(A_n\Vert B_n)$ converges to
$Q_\alpha(A\Vert B)$. This limit is positive and finite, so taking the logarithm proves the claim.
\end{proof}

\section{Iterated-mean divergences}\label{app:IM-divergence}
\newcommand{\DIM}[1]{D^\mathrm{IM}_{(#1)}}
\newcommand{\QIM}[1]{Q^\mathrm{IM}_{(#1)}}
\newcommand{\hDIM}[1]{\widehat{D}^\mathrm{IM}_{(#1)}}
\newcommand{\hQIM}[1]{\widehat{Q}^\mathrm{IM}_{(#1)}}
\newcommand{\bQIM}[1]{\bar{Q}^\mathrm{IM}_{(#1)}}
\newcommand{\inIM}{Q}

In this appendix, we prove a chain rule for the IM divergences, as a prospective candidate for future study.

\begin{definition}[Iterated-mean divergences]
\label{def:IM}
For any $\rho\in\dop{=}(\inIM)$ and $\sigma\in\Pos(\inIM)$, and $k \in \mathbb{N}$, with $\alpha_k=1+\frac{1}{2^k-1} \in (1,2]$, the iterated-mean divergence between $\rho,\sigma$ is defined as:
\begin{align}\label{eq:IM}
\DIM{\alpha_k}\Dargs{\rho}{\sigma} \defvar\begin{cases}
    \frac{1}{\alpha_k-1}\log\QIM{\alpha_k}\Dargs{\rho}{\sigma} &\supp(\rho)\subseteq\supp(\sigma)\\
    +\infty & \text{otherwise},
\end{cases}
\end{align}
where
\begin{align}\label{eq:IM_Q}
\QIM{\alpha_k}\Dargs{\rho}{\sigma} \defvar &\min_{A_1,\dots,A_k\in\Pos(\inIM)}\Tr[A_1]\nonumber\\
&\qquad\suchthat\quad\Tr[A_1]=\Tr[A_2]=\cdots=\Tr[A_k]\nonumber\\
&\qquad\qquad\quad\rho\leq A_1\#(A_2\#(\dots\#(A_k\#\sigma)\dots)).
\end{align}
Here $\#$ denotes the geometric mean, defined as $A\mathbin{\#}B=A^{1/2}(A^{-1/2}B A^{-1/2})^{1/2}A^{1/2}$ for positive-definite operators, and extended to positive operators via a decreasing sequence of positive-definite approximations.
\end{definition}
The IM divergence satisfies most\footnote{One somewhat nonstandard property is that it only has subadditivity rather than additivity across tensor products, similar to the sharp-divergences but unlike the sandwiched or Petz divergences. This does not affect our intended applications, which were already compatible with the subadditivity of sharp-divergences.} desirable properties of a {\Renyi} divergence, such as data-processing, reducing to the classical case when the states are diagonal, and so on; see~\cite[Proposition~3.3]{BFF21}. 
Moreover, by inspection we see from the definitions that $\DIM{2}\Dargs{\rho}{\sigma}
= \Dsharp_2\Dargs{\rho}{\sigma}$. Together with monotonicity in $\alpha_k$, this yields a promising property for \emph{any} value of the {\Renyi} parameter:
\begin{align}\label{eq:IM_vs_sharp2}
\forall k \in \mathbb{N}, \quad
\DIM{\alpha_k}\Dargs{\rho}{\sigma}
\leq \DIM{2}\Dargs{\rho}{\sigma}
= \Dsharp_2\Dargs{\rho}{\sigma}.
\end{align}
For the purposes of our analysis, this ensures that in single rounds we can immediately apply bounds such as Lemma~\ref{lem:round}; the main challenge is in the chain rule. More broadly, in~\cite[Sec.~5]{BFF21} it was observed that there exist states where all the IM divergences are smaller than the BS divergence; the above bound is another way to see this (when combined with~\cite[Appendix~B]{arx_MD23}, which implies the existence of states where $\Dsharp_2$ is smaller than BS divergence).

The above definition is one of three equivalent definitions of $\DIM{\alpha_k}$ from~\cite{BFF21}; we have presented the one most convenient for our purposes. To begin, we first note that the proof of equivalence between those definitions in that work straightforwardly yields a fourth equivalent definition, as follows:

\begin{lemma}\label{lemma:IM_equiv}
Definition~\ref{def:IM} is equivalent to
\begin{align}\label{eq:IM_equiv}
\DIM{\alpha_k}\Dargs{\rho}{\sigma} \defvar\begin{cases}
    \frac{1}{\alpha_k-1}\log\bQIM{\alpha_k}\Dargs{\rho}{\sigma} &\supp(\rho)\subseteq\supp(\sigma)\\
    +\infty & \text{otherwise},
\end{cases}
\end{align}
where
\begin{align}\label{eq:IM_equiv_Q}
\bQIM{\alpha_k}\Dargs{\rho}{\sigma} \defvar &\alpha_k\min_{A_1,\dots,A_k\in\Pos(\inIM)}\sum_{i=1}^k2^{-i}\Tr[A_i]\nonumber\\
&\qquad\suchthat\quad\rho\leq A_1\#(A_2\#(\dots\#(A_k\#\sigma)\dots)).
\end{align}
\end{lemma}
\begin{proof}
We follow the same proof structure as~\cite{BFF21}. If  $\supp(\rho)\nsubseteq\supp(\sigma)$, both feasible sets are empty, since the support of the nested geometric mean is contained in $\supp(\sigma)$. Hence, we may assume support inclusion. We initially compare the infima; attainment follows from Lemma~\ref{lemma:bounded-closed-feasible}. We first prove that $\bQIM{\alpha_k}\Dargs{\rho}{\sigma}\leq\QIM{\alpha_k}\Dargs{\rho}{\sigma}$. Let $(A_1,\dots,A_k)$ be feasible for Eq.~\eqref{eq:IM_Q}. Let $\Tr[A_i]=t$ be its corresponding objective value, then, 
\begin{align}
    \label{eq:<=}
    \sum_{i=1}^k2^{-i}\Tr[A_i]&=t\sum_{i=1}^k2^{-i}\nonumber\\
    &=t(1-2^{-k})\nonumber\\
    &=\frac{t}{\alpha_k}.
\end{align}
Therefore, every feasible point of~\eqref{eq:IM_Q}, is a feasible point of~\eqref{eq:IM_equiv_Q} with the same objective value, hence $\bQIM{\alpha_k}\Dargs{\rho}{\sigma}\leq\QIM{\alpha_k}\Dargs{\rho}{\sigma}$. 
Conversely, let $(A_1,\ldots,A_k)$ be feasible for \eqref{eq:IM_equiv_Q}, and set
\begin{align*}
 t\defvar \alpha_k\sum_{i=1}^k2^{-i}\Tr[A_i],\qquad
 \mu_i\defvar \frac{t}{\Tr[A_i]},\qquad B_i\defvar \mu_iA_i.
\end{align*}
Each $\Tr[A_i]$ is strictly positive: otherwise $A_i=0$, making the nested mean zero, contrary to $\Tr[\rho]=1$. The weighted arithmetic--geometric mean inequality gives
\begin{align*}
 t\geq\prod_{i=1}^k\bigl(\Tr[A_i]\bigr)^{\alpha_k2^{-i}}.
\end{align*}
Taking both sides to the power $1/\alpha_k$, we obtain
\begin{align*}
 t^{1/\alpha_k}\geq
 \prod_{i=1}^k\bigl(\Tr[A_i]\bigr)^{2^{-i}}.
\end{align*}
Since $\sum_{i=1}^k2^{-i}=1/\alpha_k$, it follows that
\begin{align*}
 \prod_{i=1}^k\mu_i^{2^{-i}}
 &=\frac{t^{\sum_{i=1}^k2^{-i}}}
         {\prod_{i=1}^k\bigl(\Tr[A_i]\bigr)^{2^{-i}}}\\
 &=\frac{t^{1/\alpha_k}}
         {\prod_{i=1}^k\bigl(\Tr[A_i]\bigr)^{2^{-i}}}
 \geq1.
\end{align*}
By homogeneity of the geometric mean,
\begin{align*}
  B_1\#(B_2\#(\dots\#(B_k\#\sigma)\dots))
 =\left(\prod_{i=1}^k\mu_i^{2^{-i}}\right) A_1\#(A_2\#(\dots\#(A_k\#\sigma)\dots))
 \geq\rho.
\end{align*}
Moreover, $\Tr[B_i]=t$ for every $i$. Thus $(B_1,\ldots,B_k)$ is feasible for \eqref{eq:IM_Q} with the same objective value as the original tuple in \eqref{eq:IM_equiv_Q}. Taking infima proves the reverse inequality.
\end{proof}

We next establish the existence of optimal tuples for the two state optimizations under support inclusion.
\begin{lemma}\label{lemma:bounded-closed-feasible}
Let $\rho\in\dop{=}(\inIM)$ and $\sigma\in\Pos(\inIM)$ satisfy $\supp(\rho)\subseteq\supp(\sigma)$. The optimization in Lemma~\ref{lemma:IM_equiv} may, without loss of generality, be restricted to a nonempty convex and compact set $\mathcal K\subseteq\Pos(\inIM)^k$. In particular, the minima in \eqref{eq:IM_Q} and \eqref{eq:IM_equiv_Q} are attained.
\end{lemma}
\begin{proof}
Choose $c>0$ such that $\rho\leq c\sigma$. The tuple $A_i=c^{\alpha_k}\sigma$ is feasible, since its nested geometric mean is $c\sigma$. Let
\begin{align*}
 t\defvar c^{\alpha_k}\Tr[\sigma],\qquad p_i\defvar \alpha_k2^{-i},
\end{align*}
and define
\begin{align*}
\mathcal K\defvar \left\{(A_1,\ldots,A_k)\in\Pos(\inIM)^k:
 \rho\leq A_1\#(A_2\#(\dots\#(A_k\#\sigma)\dots)),\quad
 \sum_{i=1}^kp_i\Tr[A_i]\leq t\right\}.
\end{align*}
This set is nonempty by the above feasible tuple. Restricting to it does not change the infimum, since every omitted tuple has objective value greater than $t$. Moreover,
\begin{align*}
 \norm{A_i}_\infty\leq\Tr[A_i]\leq\frac{t}{p_i}
\end{align*}
for every tuple in $\mathcal K$, so $\mathcal K$ is bounded.

To prove closedness, let $(A_1^{(n)},\ldots,A_k^{(n)})\in\mathcal K$ converge to $(A_1,\ldots,A_k)$. The limits are positive semidefinite. For every $\epsilon>0$, we have $A_i^{(n)}\leq A_i+\epsilon \id$ for all $i$ and all sufficiently large $n$. Monotonicity therefore gives
\begin{align*}
 \rho\leq(A_1+\epsilon \id)\#\bigl((A_2+\epsilon \id)\#
       (\cdots\#((A_k+\epsilon \id)\#\sigma)\cdots)\bigr).
\end{align*}
Letting $\epsilon\downarrow0$ and applying continuity from above successively yields the original feasibility inequality. The trace constraint is preserved by continuity. Thus $\mathcal K$ is closed, and therefore compact in finite dimension. Joint concavity and monotonicity of the geometric mean imply joint concavity of the nested mean, so $\mathcal K$ is convex.

Since the objective is continuous and $\mathcal K$ is nonempty and compact, its minimum is attained. Let $(A_1,\ldots,A_k)$ be such a minimizer, and write
\begin{align*}
    t_* \defvar \sum_{i=1}^k p_i\Tr[A_i]
    = \bQIM{\alpha_k}\Dargs{\rho}{\sigma}.
\end{align*}
Since each $\Tr[A_i]$ is strictly positive, we may define
\begin{align*}
    B_i \defvar \frac{t_*}{\Tr[A_i]}A_i,
    \qquad i\in\{1,\dots,k\}.
\end{align*}
The calculation in the proof of Lemma~\ref{lemma:IM_equiv} shows that this rescaling preserves feasibility, namely,
\begin{align*}
    \rho\leq B_1\#(B_2\#(\dots\#(B_k\#\sigma)\dots)).
\end{align*}
Moreover, $\Tr[B_i]=t_*$ for every $i$. Thus, $(B_1,\dots,B_k)$ is feasible for~\eqref{eq:IM_Q} with objective value $t_*$. By Lemma~\ref{lemma:IM_equiv}, the two optimization problems have the same infimum, so
\begin{align*}
    \QIM{\alpha_k}\Dargs{\rho}{\sigma}
    =\bQIM{\alpha_k}\Dargs{\rho}{\sigma}
    =t_*.
\end{align*}
Hence $(B_1,\dots,B_k)$ attains the minimum in~\eqref{eq:IM_Q} as well.
\end{proof}
We now turn to the corresponding channel divergences.
\begin{definition}[Channel iterated-mean divergences]
We again highlight (see Remark~\ref{rem:stab}) that our definition here does not inherently include a purifying/stabilizer register.
For any $k \in\mathbb{N}$, with $\alpha_k=1+\frac{1}{2^k-1}$, and $\mapE,\mapF \in \CPTP{\inIM}{\outQ}$, the channel IM divergence between the maps is defined as
\begin{align}
\DIM{\alpha_k} \Dargs{\mapE}{\mapF} &\defvar \sup_{\rho\in\dop{=}(\inIM)} \DIM{\alpha_k} \Dargs{\mapE[\rho]}{\mapF[\rho]} .
\end{align}
\end{definition}

Our main finding regarding the channel IM divergence is that {when there is a stabilizing register}, it can be reformulated in a ``dual'' form that is a minimization instead of a maximization. This essentially arises from its similarity to the $\#$-divergences defined in~\cite{FF21sharp}, allowing us to carry over the proof of the analogous result for the latter.
\begin{theorem}\label{thrm:channel div}
Let $k \in\mathbb{N}$, with $\alpha_k=1+\frac{1}{2^k-1}$,  and $\mapE,\mapF \in \CPTP{\inIM}{\outQ}$; let $\stabE$ be a register isomorphic to $\inIM$ with $d=\dim(\inIM)$. 
Then 
\begin{align}\label{eq:IM D channel div}
\DIM{\alpha_k} \Dargs{\mapE \otimes \idmap_{\stabE}}{\mapF \otimes \idmap_{\stabE}} = \frac{1}{\alpha_k-1}\log\QIM{\alpha_k}\Dargs{\mapE\otimes \idmap_{\stabE}}{\mapF \otimes \idmap_{\stabE}},
\end{align}
where
\begin{align}\label{eq:IM Qchannel div}
    \QIM{\alpha_k}\Dargs{\mapE\otimes \idmap_{\stabE}}{\mapF \otimes \idmap_{\stabE}}=&\inf_{A_1,\dots,A_k\in\Pos(\stabE\outQ)}\norm{\sum_{i=1}^kp_i\Tr_{\outQ}[A_i]}_\infty\nonumber\\
&\qquad\suchthat\quad J^\mapE\leq A_1\#(A_2\#(\dots\#(A_k\# J^\mapF)\dots)),
\end{align}
with $p_i=\alpha_k2^{-i}$ and $J^\mapE,J^\mapF\in\Pos(\stabE\outQ)$ being the unnormalized Choi matrices corresponding to $\mapE,\mapF$, respectively. In fixed corresponding orthonormal bases of $\inIM$ and $\stabE$, our convention is
\[
J^\mapE_{\stabE\outQ}=\sum_{a,b}\ketbra{a}{b}_{\stabE}\otimes\mapE[\ketbra{a}{b}_{\inIM}],
\]
and analogously for $J^\mapF$. In particular, $\Tr_{\outQ}[J^\mapE]=\Tr_{\outQ}[J^\mapF]=\id_{\stabE}$.
\end{theorem}
\begin{proof}
If $\supp(J^{\mapE})\nsubseteq\supp(J^{\mapF})$, the optimization in \eqref{eq:IM Qchannel div} is infeasible. Taking the normalized maximally entangled input gives output states $J^{\mapE}/d$ and $J^{\mapF}/d$, whose divergence is infinite. Hence both sides of \eqref{eq:IM D channel div} are infinite, and we may assume support inclusion.

Following the analysis in~\cite[Eq.~15]{FF21sharp}, the stabilized channel divergence can be written as
\begin{align*}
 \DIM{\alpha_k}\Dargs{\mapE\otimes\idmap_{\stabE}}{\mapF\otimes\idmap_{\stabE}}
 =\sup_{\rho\in\dop{=}(\stabE)}
 \DIM{\alpha_k}\Dargs{\rho^{1/2}J^{\mapE}\rho^{1/2}}{\rho^{1/2}J^{\mapF}\rho^{1/2}},
\end{align*}
where we have implicitly restricted to pure input states without loss of generality (via a standard data-processing argument on the stabilizing register, since the IM divergences satisfy data-processing~\cite[Proposition~3.3(6)]{BFF21}). Here $\rho$ acts on $\stabE$, and the factors $\rho^{1/2}$ in these expressions mean $\rho^{1/2}\otimes\id_{\outQ}$.

Let $\mathcal K$ be the full feasible set of tuples in \eqref{eq:IM Qchannel div}, and define, for $\rho\in\dop{=}(\stabE)$,
\begin{align*}
 g(\rho)\defvar
 \inf_{(A_1,\ldots,A_k)\in\mathcal K}
 \Tr\!\left[\rho\sum_{i=1}^kp_i\Tr_{\outQ}[A_i]\right].
\end{align*}
The set $\mathcal K$ is convex and nonempty. For nonemptiness, choose $c>0$ with $J^{\mapE}\leq cJ^{\mapF}$ and take $A_i=c^{\alpha_k}J^{\mapF}$.

For $\rho>0$, the change of variables $B_i=\rho^{1/2}A_i\rho^{1/2}$ is a bijection of positive semidefinite tuples. Congruence invariance of the geometric mean under invertible operators and Lemma~\ref{lemma:IM_equiv} therefore give
\begin{align*}
 \QIM{\alpha_k}\Dargs{\rho^{1/2}J^{\mapE}\rho^{1/2}}{\rho^{1/2}J^{\mapF}\rho^{1/2}}=g(\rho).
\end{align*}
The function $g$ is finite, nonnegative, and concave, being the infimum of linear nonnegative functions. For any $\rho\in\dop{=}(\stabE)$ and $\epsilon\in(0,1)$, let $\rho_\epsilon=(1-\epsilon)\rho+\epsilon \id_{\stabE}/d$. Then
\begin{align*}
 g(\rho_\epsilon)\geq(1-\epsilon)g(\rho)+\epsilon g(\id_{\stabE}/d)
 \geq(1-\epsilon)g(\rho).
\end{align*}
Thus
\begin{align*}
 \sup_{\substack{\rho\in\dop{=}(\stabE)\\\rho>0}}g(\rho)
 =\sup_{\rho\in\dop{=}(\stabE)}g(\rho).
\end{align*}

The set $\dop{=}(\stabE)$ is compact and convex, $\mathcal K$ is convex, and the map
\begin{align*}
 (\rho,(A_i))\longmapsto
 \Tr\!\left[\rho\sum_{i=1}^kp_i\Tr_{\outQ}[A_i]\right]
\end{align*}
is continuous and affine in each argument. Sion's minimax theorem~\cite{seqSion58} therefore yields
\begin{align*}
 &\sup_{\rho\in\dop{=}(\stabE)}
 \QIM{\alpha_k}\Dargs{\rho^{1/2}J^{\mapE}\rho^{1/2}}{\rho^{1/2}J^{\mapF}\rho^{1/2}}\\
 &\geq\sup_{\substack{\rho\in\dop{=}(\stabE)\\\rho>0}}g(\rho)\\
 &=\sup_{\rho\in\dop{=}(\stabE)}\inf_{(A_i)\in\mathcal K}
 \Tr\!\left[\rho\sum_{i=1}^kp_i\Tr_{\outQ}[A_i]\right]\\
 &=\inf_{(A_i)\in\mathcal K}\sup_{\rho\in\dop{=}(\stabE)}
 \Tr\!\left[\rho\sum_{i=1}^kp_i\Tr_{\outQ}[A_i]\right]\\
 &=\inf_{(A_i)\in\mathcal K}
 \norm{\sum_{i=1}^kp_i\Tr_{\outQ}[A_i]}_\infty.
\end{align*}
Compactness of $\mathcal K$ is not required.

For the reverse inequality, take any $(A_i)\in\mathcal K$ and any $\rho\in\dop{=}(\stabE)$, possibly singular, and set $B_i=\rho^{1/2}A_i\rho^{1/2}$. The transformer inequality, applied successively, gives
\begin{align*}
 \rho^{1/2}J^{\mapE}\rho^{1/2}
 \leq B_1\#\bigl(B_2\#(\cdots\#(B_k\#\rho^{1/2}J^{\mapF}\rho^{1/2})\cdots)\bigr).
\end{align*}
Hence $(B_i)$ is feasible for the state program in Lemma~\ref{lemma:IM_equiv}, and
\begin{align*}
 \QIM{\alpha_k}\Dargs{\rho^{1/2}J^{\mapE}\rho^{1/2}}{\rho^{1/2}J^{\mapF}\rho^{1/2}}
 &\leq\sum_{i=1}^kp_i\Tr[B_i]\\
 &=\Tr\!\left[\rho\sum_{i=1}^kp_i\Tr_{\outQ}[A_i]\right]\\
 &\leq\norm{\sum_{i=1}^kp_i\Tr_{\outQ}[A_i]}_\infty.
\end{align*}
Taking the supremum over $\rho$ and then the infimum over $(A_i)$ proves the required equality. Applying $\frac{1}{\alpha_k-1}\log$ proves \eqref{eq:IM D channel div}.
\end{proof}
Replacing the norm of the weighted sum
in~\eqref{eq:IM Qchannel div} by the weighted sum of the norms gives a relaxation of the channel optimization. The resulting quantity yields the following chain rule.
\begin{theorem}
    \label{thrm:IM ch div chain rule}
    Let $k\in\mathbb{N}$, with $\alpha_k=1+\frac{1}{2^k-1}$, and $\mapE,\mapF \in \CPTP{\inIM}{\outQ}$; let $\rho,\sigma\in\dop{=}(\inIM\stabE)$, where $\stabE$ is a register isomorphic to $\inIM$. Then, 
    \begin{align}
        \label{eq:IM ch div chain rule}
        \DIM{\alpha_k} \Dargs{\left(\mapE \otimes \idmap_{\stabE}\right)[\rho]}{\left(\mapF \otimes \idmap_{\stabE}\right)[\sigma]}\leq\DIM{\alpha_k}\Dargs{\rho}{\sigma}+\hDIM{\alpha_k} \Dargs{\mapE \otimes \idmap_{\stabE}}{\mapF \otimes \idmap_{\stabE}}.
    \end{align}
    Here
    \begin{align}\label{eq:IM Drelax}
\hDIM{\alpha_k} \Dargs{\mapE \otimes \idmap_{\stabE}}{\mapF \otimes \idmap_{\stabE}} = \frac{1}{\alpha_k-1}\log\hQIM{\alpha_k}\Dargs{\mapE\otimes \idmap_{\stabE}}{\mapF \otimes \idmap_{\stabE}},
\end{align}
where
\begin{align}\label{eq:IM Qrelax}
    \hQIM{\alpha_k}\Dargs{\mapE\otimes \idmap_{\stabE}}{\mapF \otimes \idmap_{\stabE}}=&\inf_{A_1,\dots,A_k\in\Pos(\stabE\outQ)}\sum_{i=1}^kp_i\norm{\Tr_{\outQ}[A_i]}_\infty\nonumber\\
&\qquad\suchthat\quad J^\mapE\leq A_1\#(A_2\#(\dots\#(A_k\# J^\mapF)\dots)),
\end{align}
with $p_i=\alpha_k2^{-i}$ and $J^\mapE,J^\mapF\in\Pos(\stabE\outQ)$ being the unnormalized Choi matrices corresponding to $\mapE,\mapF$, respectively, in the convention of Theorem~\ref{thrm:channel div}.
\end{theorem}
\begin{proof}
The IM divergence is nonnegative on normalized states (an immediate consequence of using data-processing to pinch the states and reduce to the classical case), 
so Theorem~\ref{thrm:channel div} and the triangle inequality for $\infty$-norm give
\begin{align*}
 1\leq\QIM{\alpha_k}\Dargs{\mapE\otimes\idmap_{\stabE}}{\mapF\otimes\idmap_{\stabE}}
 \leq\hQIM{\alpha_k}\Dargs{\mapE\otimes\idmap_{\stabE}}{\mapF\otimes\idmap_{\stabE}}.
\end{align*}
Thus both terms on the right-hand side of \eqref{eq:IM ch div chain rule} are nonnegative. If either is infinite, the claim is immediate. Otherwise, let $(B_1,\ldots,B_k)$, with $B_i\in\Pos(\inIM\stabE)$
for every $i$, be an optimal tuple for \eqref{eq:IM_Q} applied to $\rho,\sigma$. By Lemma~\ref{lemma:bounded-closed-feasible}, such a tuple exists, and
\begin{align*}
 \rho\leq B_1\#(B_2\#(\dots\#(B_k\#\sigma)\dots)),\qquad
 \Tr[B_1]=\cdots=\Tr[B_k]=\lambda\defvar \QIM{\alpha_k}\Dargs{\rho}{\sigma}.
\end{align*}
Let $(A_1,\ldots,A_k)$ be any feasible tuple in \eqref{eq:IM Qrelax}.

To distinguish the Choi input from the reference $\stabE$ of $\rho,\sigma$, we introduce a separate register ${\inIM}^{\prime}\simeq \inIM$ and relabel only the Choi input and the $A_i$ accordingly; the input states $\rho,\sigma$ and the $B_i$ remain on $\inIM\stabE$. Thus $A_i,J^{\mapE},J^{\mapF}$ now act on ${\inIM}^{\prime}\outQ$. Writing $d\defvar\dim(\inIM)$, we define
\begin{align*}
 C_i\defvar d\,
 \bigl(\bra{\Phi^+}_{{\inIM}{\inIM}^{\prime}}\otimes\id_{\stabE\outQ}\bigr)
 (B_i\otimes A_i)
 \bigl(\ket{\Phi^+}_{{\inIM}{\inIM}^{\prime}}\otimes\id_{\stabE\outQ}\bigr).
\end{align*}
Here $\ket{\Phi^+}_{{\inIM}{\inIM}^{\prime}}=d^{-1/2}\sum_a\ket a_{\inIM}\ket a_{{\inIM}^{\prime}}$ is normalized in the bases defining the Choi matrices, and $C_i\in\Pos(\stabE\outQ)$. In these contractions, tensor factors are reordered as $\inIM{\inIM}^{\prime}\stabE\outQ$; the contraction acts only on $\inIM{\inIM}^{\prime}$ and retains $\stabE\outQ$. By the unnormalized Choi convention, up to the canonical ordering of the output registers,
\begin{align*}
 (\mapE\otimes\idmap_{\stabE})[\rho]
 &=d\,\bigl(\bra{\Phi^+}_{{\inIM}{\inIM}^{\prime}}\otimes\id_{\stabE\outQ}\bigr)
 (\rho\otimes J^{\mapE})
 \bigl(\ket{\Phi^+}_{{\inIM}{\inIM}^{\prime}}\otimes\id_{\stabE\outQ}\bigr),\\
 (\mapF\otimes\idmap_{\stabE})[\sigma]
 &=d\,\bigl(\bra{\Phi^+}_{{\inIM}{\inIM}^{\prime}}\otimes\id_{\stabE\outQ}\bigr)
 (\sigma\otimes J^{\mapF})
 \bigl(\ket{\Phi^+}_{{\inIM}{\inIM}^{\prime}}\otimes\id_{\stabE\outQ}\bigr).
\end{align*}
Tensoring the two feasible inequalities and using the tensor-product identity for the geometric mean gives
\begin{align*}
 \rho\otimes J^{\mapE}
 \leq(B_1\otimes A_1)\#\bigl((B_2\otimes A_2)\#
       (\cdots\#((B_k\otimes A_k)\#(\sigma\otimes J^{\mapF}))\cdots)\bigr).
\end{align*}
Applying the map
\begin{align*}
 X\longmapsto d\,
 \bigl(\bra{\Phi^+}_{{\inIM}{\inIM}^{\prime}}\otimes\id_{\stabE\outQ}\bigr)
 X
 \bigl(\ket{\Phi^+}_{{\inIM}{\inIM}^{\prime}}\otimes\id_{\stabE\outQ}\bigr)
\end{align*}
and successively using the transformer inequality, we obtain
\begin{align*}
 (\mapE\otimes\idmap_{\stabE})[\rho]
 \leq C_1\#\bigl(C_2\#(\cdots\#(C_k\#(\mapF\otimes\idmap_{\stabE})[\sigma])\cdots)\bigr).
\end{align*}
Thus $(C_1,\ldots,C_k)$ is feasible for the state program in Lemma~\ref{lemma:IM_equiv}.

The contraction also gives
\begin{align*}
 \Tr[C_i]
 &=\Tr\!\left[(\Tr_{\stabE}[B_i])^{\mathsf T}\Tr_{\outQ}[A_i]\right]\\
 &\leq\norm{\Tr_{\outQ}[A_i]}_\infty\Tr[B_i]
 =\lambda\norm{\Tr_{\outQ}[A_i]}_\infty.
\end{align*}
Here $\Tr_{\stabE}[B_i]$ acts on $\inIM$, while $\Tr_{\outQ}[A_i]$ acts on ${\inIM}^{\prime}$; the transpose and the identification of $\inIM$ with ${\inIM}^{\prime}$ refer to the basis defining $\ket{\Phi^+}$. The inequality follows from positivity and $\Tr_{\outQ}[A_i]\leq\norm{\Tr_{\outQ}[A_i]}_\infty \id_{{\inIM}^{\prime}}$. Consequently,
\begin{align*}
 \QIM{\alpha_k}\Dargs{(\mapE\otimes\idmap_{\stabE})[\rho]}{(\mapF\otimes\idmap_{\stabE})[\sigma]}
 &\leq\sum_{i=1}^kp_i\Tr[C_i]\\
 &\leq\lambda\sum_{i=1}^kp_i\norm{\Tr_{\outQ}[A_i]}_\infty.
\end{align*}
Taking the infimum over feasible $(A_1,\ldots,A_k)$ and then applying $\frac{1}{\alpha_k-1}\log$ proves the claim.
\end{proof}
For rank-1 projective measure-and-prepare channels with a common measurement basis, the relaxation introduced above is exact. The following lemma expresses both channel quantities in terms of the prepared states.
\begin{lemma}
    Let $k\in\mathbb{N}$, with $\alpha_k=1+\frac{1}{2^k-1}$. Let $\mapE,\mapF\in\CPTP{\inIM}{\outQ}$ be such that $\mapE=\Tr_{\inIM}\circ\widetilde{\mapE}$ and $\mapF=\Tr_{\inIM}\circ\widetilde{\mapF}$ with $\widetilde{\mapE},\widetilde{\mapF}\in \CPTP{\inIM}{\inIM\outQ}$ being projective measure-and-prepare channels of the form
    \begin{align}
        \widetilde{\mapE}[\rho_{\inIM}]=\sum_j\ketbra{j}{j}_{\inIM}\rho_{\inIM}\ketbra{j}{j}_{\inIM}\otimes\sigma_{\outQ|j},\qquad\widetilde{\mapF}[\rho_{\inIM}]=\sum_j\ketbra{j}{j}_{\inIM}\rho_{\inIM}\ketbra{j}{j}_{\inIM}\otimes\omega_{\outQ|j},
    \end{align}
    for some $\sigma_{\outQ|j},\omega_{\outQ|j}\in\dop{=}(\outQ)$.
    Then, letting $\stabE$ be a register isomorphic to $\inIM$,
    \begin{align}\label{eq:im-block-equality}
        \hDIM{\alpha_k} \Dargs{\mapE \otimes \idmap_{\stabE}}{\mapF \otimes \idmap_{\stabE}}=\DIM{\alpha_k} \Dargs{\mapE \otimes \idmap_{\stabE}}{\mapF \otimes \idmap_{\stabE}}=\max_j\DIM{\alpha_k}\Dargs{\sigma_{\outQ|j}}{\omega_{\outQ|j}},
    \end{align}
    with $\hDIM{\alpha_k}$ as defined in~\eqref{eq:IM Drelax}.
\end{lemma}
\begin{proof}
First note that, for any $A_1,\ldots,A_k\in\Pos(\stabE\outQ)$, the triangle inequality gives
\begin{align*}
 \norm{\sum_{i=1}^kp_i\Tr_{\outQ}[A_i]}_\infty
 \leq\sum_{i=1}^kp_i\norm{\Tr_{\outQ}[A_i]}_\infty,
 \qquad p_i=\alpha_k2^{-i}.
\end{align*}
Taking the infimum over the common feasible set in~\eqref{eq:IM Qchannel div} and~\eqref{eq:IM Qrelax}, we conclude that
\begin{align}\label{eq:im-block-order}
 \QIM{\alpha_k}\Dargs{\mapE\otimes\idmap_{\stabE}}{\mapF\otimes\idmap_{\stabE}}
 \leq\hQIM{\alpha_k}\Dargs{\mapE\otimes\idmap_{\stabE}}{\mapF\otimes\idmap_{\stabE}}.
\end{align}
It now suffices to prove the following two inequalities:
\begin{align}
 \QIM{\alpha_k}\Dargs{\mapE\otimes\idmap_{\stabE}}{\mapF\otimes\idmap_{\stabE}}
 &\geq\max_j\QIM{\alpha_k}\Dargs{\sigma_{\outQ|j}}{\omega_{\outQ|j}},
 \label{eq:im-block-lower}\\
 \hQIM{\alpha_k}\Dargs{\mapE\otimes\idmap_{\stabE}}{\mapF\otimes\idmap_{\stabE}}
 &\leq\max_j\QIM{\alpha_k}\Dargs{\sigma_{\outQ|j}}{\omega_{\outQ|j}}.
 \label{eq:im-block-upper}
\end{align}

We begin with~\eqref{eq:im-block-lower}. Choose the basis defining the unnormalized Choi operators to be the measurement basis. This choice does not affect either channel optimization, since a change of Choi basis corresponds to a unitary conjugation on $\stabE$. Then
\begin{align}\label{eq:im-block-choi}
 J^{\mapE}=\sum_j\ketbra{j}{j}_{\stabE}\otimes\sigma_{\outQ|j},
 \qquad
 J^{\mapF}=\sum_j\ketbra{j}{j}_{\stabE}\otimes\omega_{\outQ|j}.
\end{align}
If the optimization in~\eqref{eq:IM Qchannel div} is infeasible, its value is $+\infty$ and~\eqref{eq:im-block-lower} holds immediately. Otherwise, let $(A_1,\ldots,A_k)$ be a feasible tuple, so that
\begin{align}\label{eq:im-block-channel-constraint}
 J^{\mapE}\leq A_1\#(A_2\#(\dots\#(A_k\#J^{\mapF})\dots)).
\end{align}
For every $i\in\{1,\ldots,k\}$ and every $j$, define
\begin{align}\label{eq:im-block-compression}
 A_i^j\defvar
 \bigl(\bra{j}_{\stabE}\otimes\id_{\outQ}\bigr)
 A_i
 \bigl(\ket{j}_{\stabE}\otimes\id_{\outQ}\bigr).
\end{align}
These operators belong to $\Pos(\outQ)$. Applying the transformer inequality and monotonicity successively to~\eqref{eq:im-block-channel-constraint}, we obtain
\begin{align*}
 \sigma_{\outQ|j}
 &=\bra{j}_{\stabE}J^{\mapE}\ket{j}_{\stabE}\\
 &\leq A_1^j\#\bigl(A_2^j\#(\dots\#(A_k^j\#
       \bra{j}_{\stabE}J^{\mapF}\ket{j}_{\stabE})\dots)\bigr)\\
 &=A_1^j\#(A_2^j\#(\dots\#(A_k^j\#\omega_{\outQ|j})\dots)),
\end{align*}
where identities on $\outQ$ are left implicit and the equalities follow from~\eqref{eq:im-block-choi}. Hence $(A_1^j,\ldots,A_k^j)$ is feasible for the weighted state optimization in~\eqref{eq:IM_equiv_Q}. By Lemma~\ref{lemma:IM_equiv},
\begin{align*}
 \QIM{\alpha_k}\Dargs{\sigma_{\outQ|j}}{\omega_{\outQ|j}}
 &\leq\sum_{i=1}^kp_i\Tr[A_i^j]\\
 &=\sum_{i=1}^kp_i\bra{j}_{\stabE}\Tr_{\outQ}[A_i]\ket{j}_{\stabE}\\
 &=\bra{j}_{\stabE}\sum_{i=1}^kp_i\Tr_{\outQ}[A_i]\ket{j}_{\stabE}\\
 &\leq\norm{\sum_{i=1}^kp_i\Tr_{\outQ}[A_i]}_\infty.
\end{align*}
Since this holds for every $j$, we have
\begin{align*}
 \max_j\QIM{\alpha_k}\Dargs{\sigma_{\outQ|j}}{\omega_{\outQ|j}}
 \leq\norm{\sum_{i=1}^kp_i\Tr_{\outQ}[A_i]}_\infty.
\end{align*}
Taking the infimum over all feasible tuples $(A_1,\ldots,A_k)$ proves~\eqref{eq:im-block-lower}.

We now prove~\eqref{eq:im-block-upper}. First, if
$\supp(\sigma_{\outQ|j})\nsubseteq\supp(\omega_{\outQ|j})$
for some $j$, then
$\QIM{\alpha_k}\Dargs{\sigma_{\outQ|j}}{\omega_{\outQ|j}}=+\infty$.
In this case,~\eqref{eq:im-block-lower} and~\eqref{eq:im-block-order} imply that both channel quantities are also infinite, and the claim follows. We may therefore assume
\begin{align*}
 \supp(\sigma_{\outQ|j})\subseteq\supp(\omega_{\outQ|j})
 \qquad\text{for every }j.
\end{align*}

For each $j$, let $(B_1^j,\ldots,B_k^j)$ be an optimal tuple for~\eqref{eq:IM_Q} applied to $\sigma_{\outQ|j},\omega_{\outQ|j}$. Such a tuple exists by Lemma~\ref{lemma:bounded-closed-feasible}. Thus
\begin{align}\label{eq:im-block-state-constraint}
 \sigma_{\outQ|j}\leq B_1^j\#(B_2^j\#(\dots\#(B_k^j\#\omega_{\outQ|j})\dots)),
\end{align}
and
\begin{align}\label{eq:im-block-equal-traces}
 \Tr[B_1^j]=\Tr[B_2^j]=\cdots=\Tr[B_k^j]
 =\QIM{\alpha_k}\Dargs{\sigma_{\outQ|j}}{\omega_{\outQ|j}}.
\end{align}
For every $i\in\{1,\ldots,k\}$, define
\begin{align}\label{eq:im-block-constructed-tuple}
 B_i\defvar\sum_j\ketbra{j}{j}_{\stabE}\otimes B_i^j.
\end{align}
Since the geometric mean respects direct sums, we have
\begin{align*}
 &B_1\#(B_2\#(\dots\#(B_k\#J^{\mapF})\dots))\\
 &=\sum_j\ketbra{j}{j}_{\stabE}\otimes
       \bigl(B_1^j\#(B_2^j\#(\dots\#(B_k^j\#\omega_{\outQ|j})\dots))\bigr)\\
 &\geq\sum_j\ketbra{j}{j}_{\stabE}\otimes\sigma_{\outQ|j}\\
 &=J^{\mapE},
\end{align*}
where the inequality follows from~\eqref{eq:im-block-state-constraint} and the last equality follows from~\eqref{eq:im-block-choi}. Thus $(B_1,\ldots,B_k)$ is feasible for~\eqref{eq:IM Qrelax}.

Furthermore, for every $i$,
\begin{align}\label{eq:im-block-norm}
 \norm{\Tr_{\outQ}[B_i]}_\infty
 &=\norm{\sum_j\Tr[B_i^j]\ketbra{j}{j}_{\stabE}}_\infty\nonumber\\
 &=\norm{\sum_j
    \QIM{\alpha_k}\Dargs{\sigma_{\outQ|j}}{\omega_{\outQ|j}}
    \ketbra{j}{j}_{\stabE}}_\infty\nonumber\\
 &=\max_j\QIM{\alpha_k}\Dargs{\sigma_{\outQ|j}}{\omega_{\outQ|j}}.
\end{align}
The first equality follows from~\eqref{eq:im-block-constructed-tuple}, the second follows from~\eqref{eq:im-block-equal-traces}, and the last uses that the operator is diagonal with nonnegative entries. Consequently,
\begin{align*}
 \hQIM{\alpha_k}\Dargs{\mapE\otimes\idmap_{\stabE}}{\mapF\otimes\idmap_{\stabE}}
 &\leq\sum_{i=1}^kp_i\norm{\Tr_{\outQ}[B_i]}_\infty\\
 &=\max_j\QIM{\alpha_k}\Dargs{\sigma_{\outQ|j}}{\omega_{\outQ|j}},
\end{align*}
where the equality follows from~\eqref{eq:im-block-norm} and $\sum_{i=1}^kp_i=1$. This proves~\eqref{eq:im-block-upper}. Combining~\eqref{eq:im-block-order},~\eqref{eq:im-block-lower}, and~\eqref{eq:im-block-upper}, and taking the logarithm, proves~\eqref{eq:im-block-equality}.
\end{proof}

\end{document}